\documentclass[11pt,english]{article}
\usepackage[T1]{fontenc}
\usepackage[latin9]{inputenc}
\usepackage{geometry}
\usepackage{color}
\usepackage{babel}
\usepackage{verbatim}
\usepackage{amsmath}
\usepackage{amsthm}
\usepackage{amssymb}
\usepackage{graphicx}
\usepackage{setspace}
\usepackage[authoryear]{natbib}
\PassOptionsToPackage{normalem}{ulem}
\usepackage{ulem}
\usepackage[unicode=true,pdfusetitle,
 bookmarks=true,bookmarksnumbered=false,bookmarksopen=false,
 breaklinks=false,pdfborder={0 0 0},pdfborderstyle={},backref=false,colorlinks=false]
 {hyperref}

\makeatletter
\theoremstyle{remark}
\newtheorem{rem}{\protect\remarkname}
\theoremstyle{plain}
\newtheorem{lem}{\protect\lemmaname}
\theoremstyle{plain}
\newtheorem{assumption}{\protect\assumptionname}
\theoremstyle{plain}
\newtheorem{prop}{\protect\propositionname}
\theoremstyle{plain}
\newtheorem{thm}{\protect\theoremname}
\theoremstyle{definition}
 \newtheorem{example}{\protect\examplename}
\theoremstyle{plain}
\newtheorem{cor}{\protect\corollaryname}

\usepackage{appendix}
\usepackage{multirow}
\usepackage{float}
\usepackage{booktabs}
\usepackage{longtable}
\usepackage{rotating}
\renewcommand{\hat}{\widehat}
\renewcommand{\tilde}{\widetilde}

\@ifundefined{showcaptionsetup}{}{%
 \PassOptionsToPackage{caption=false}{subfig}}
\usepackage{subfig}
\makeatother

\providecommand{\assumptionname}{Assumption}
\providecommand{\corollaryname}{Corollary}
\providecommand{\examplename}{Example}
\providecommand{\lemmaname}{Lemma}
\providecommand{\propositionname}{Proposition}
\providecommand{\remarkname}{Remark}
\providecommand{\theoremname}{Theorem}

\begin{document}
\title{On LASSO for High Dimensional Predictive Regression}
\author{Ziwei Mei and Zhentao Shi}
\date{}
\maketitle
\begin{abstract}
This paper examines LASSO, a widely-used $L_{1}$-penalized regression
method, in high dimensional linear predictive regressions, particularly
when the number of potential predictors exceeds the sample size and
numerous unit root regressors are present. The consistency of LASSO
is contingent upon two key components: the deviation bound of the
cross product of the regressors and the error term, and the restricted
eigenvalue of the Gram matrix. We present new probabilistic bounds
for these components, suggesting that LASSO's rates of convergence
are different from those typically observed in cross-sectional cases.
When applied to a mixture of stationary, nonstationary, and cointegrated
predictors, LASSO maintains its asymptotic guarantee if predictors
are scale-standardized. Leveraging machine learning and macroeconomic
domain expertise, LASSO demonstrates strong performance in forecasting
the unemployment rate, as evidenced by its application to the FRED-MD
database.
\end{abstract}
\thispagestyle{empty}

\vspace{0.8cm}



\vspace{0.8cm}

\small \noindent Ziwei Mei: \texttt{zwmei@link.cuhk.edu.hk}. Corresponding
author: Zhentao Shi: \texttt{zhentao.shi@cuhk.edu.hk}. Tel: (852)
3943 1432. Fax: (852) 2603 5805. Address: 928 Esther Lee Building,
the Chinese University of Hong Kong, Shatin, New Territories, Hong
Kong SAR, China. We thank Anna Bykhovskaya, Jinyuan Chang, Jianqing
Fan, Yingying Li, Alexey Onatskiy, Whitney Newey, Liangjun Su, and
Etienne Wijler for helpful comments.

\newpage{}

\normalsize
\onehalfspacing

\section{Introduction\label{sec:Introduction}}

Machine learning is a rapidly evolving field that has significantly
reshaped numerous academic disciplines. While statisticians often
concentrate on scenarios where the sample comprises independently
and identically distributed (i.i.d.)~observations, econometricians
pay special attention to settings where variables are gathered over
time. Temporal dependence plays a crucial role in these data generating
processes (DGP). In the context of parameter estimation, weakly dependent
data, under certain technical conditions, bear resemblance to i.i.d.~data
as the time span increases. However, many established theoretical
results under i.i.d.~data are inapplicable to instances where time
series exhibit high persistence.

Prediction is an important theme of empirical macroeconomics and finance.
Although forecasting the stock market is notoriously challenging,
recent advancements in machine learning offer some silver lining \citep{gu2020empirical};
forecasting macroeconomic variables, on the other hand, is more feasible
\citep{Stock2012,medeiros2021forecasting}. Macroeconomic time series
encompass a diverse range of dynamic patterns. GDP, industrial production
index, exchange rates, and money supply all provide perspectives on
the economy's current state and may offer insights into its future
trajectory.

Linear predictive regression is a straightforward model. However,
the presence of persistent regressors can pose a multitude of theoretical
and practical challenges in its estimation and inference. In recent
years, the advent of macroeconomic big data has sparked research interest
in digesting numerous potential variables in macroeconomics \citep{ng2013variable}.
The \emph{least absolute shrinkage and selection operator} (LASSO)
\citep{tibshirani1996regression}, an off-the-shelf machine learning
method for linear regressions, is one of such tools. While LASSO and
its variants have been extensively explored in statistics for i.i.d.~data,
only a handful of recent econometric papers have examined predictive
regression in the context of nonstationary regressors with growing
dimensions.

This paper serves as a stepping stone toward understanding LASSO in
high dimensional predictive regressions with persistent variables.
In particular, it considers the setting when a large number of unit
root regressors are present, and in scenarios where the number of
regressors ($p$) exceeds the sample size ($n$). We allow the innovation
processes to be time dependent and non-Gaussian. Under these conditions,
we introduce novel asymptotic rates for the \emph{deviation bound}
(DB) and the \emph{restricted eigenvalue} (RE), which will be discussed
in Sections \ref{subsec:DB} and \ref{subsec:RE}. DB and RE are two
pivotal conditions that govern LASSO's behavior. We establish convergence
rates under the nonstationary time series setting, which differ from
those of i.i.d.~data \citep{buhlmann2011statistics} and weakly dependent
data \citep{kock2015oracle,medeiros2016,Mogliani2021}.

To enhance the practical relevance of our theory, we expand the stylized
regression model with all unit root regressors in two ways. Firstly,
we examine not only \citet{tibshirani1996regression}'s original LASSO,
which imposes the same penalty level to all coefficients, but also
a LASSO variant that standardizes each regressor with its sample standard
deviation (s.d.). We refer to the former as \emph{Plain LASSO} (Plasso)
and the latter as \emph{Standardized LASSO} (Slasso), following \citet{lee2022lasso}.
While Plasso is more straightforward for theoretical analysis, Slasso
is commonly the default in applications. The asymptotic theory for
Slasso under i.i.d.~data can be easily extended from Plasso, as each
sample s.d.~is expected to converge \emph{in probability} to a positive
constant. However, the sample s.d.~of a unit root process, when divided
by $\sqrt{n}$, converges \emph{in distribution} to a non-degenerate
stochastic integral, introducing additional randomness and altering
the convergence rate. Secondly, we consider a mix of unit root, stationary,
and cointegrated regressors. The researcher maintains an agnostic
stance and includes all these variables in the regression without
pre-testing to categorize them. We find that while Slasso maintains
asymptotic guarantees, Plasso encounters multiple challenges. These
enhancements enrich the theory and broaden the applicability. 

We utilize the FRED-MD database to employ LASSO in forecasting the
unemployment rate in the United States. Initially, we include all
121 variables from the database as predictors. Adhering to a standard
practice in empirical studies, we transform each nonstationary time
series into a stationary one, and then compare the forecast results
with those based on the raw data without any transformation. Our findings
reveal that Plasso underperforms in comparison to Slasso when given
the same set of predictors. Moreover, the raw data are stronger than
the stationarized data as they better match the persistence of the
dependent variable. Subsequently, we experiment with a more comprehensive
setting, incorporating four lags of each predictor along with lagged
dependent variables and extracted factors. With a total of 504 regressors,
Slasso further reduces the prediction error, suggesting that macroeconomic
domain knowledge is beneficial in guiding initial specifications.

\medskip

This paper adds to a burgeoning literature concerning many nonstationary
time series. \citet{lee2022lasso} explore variable selection of (adaptive)
LASSO in the low dimensional setting where $p$ is fixed and highlights
that some well-known LASSO properties for i.i.d.~data collapse when
faced with nonstationary data. Several papers consider the minimum
eigenvalue or the RE of the Gram matrix of root unit processes as
$p/\sqrt{n}\to0$, which we call the case of \emph{moderate dimension}.
\citet{koo2020high} study a predictive regression with the unit root
regression forming cointegration systems, and they leave the symbol
of RE in the rate of convergence; they do not provide a lower bound
for the RE. \citet{fan2023predictive} work with quantile regressions
in a similar setting with an assumed RE. An explicit rate that bounds
the minimum eigenvalue is deduced in \citet{zhang2019identifying}
as a by-product of their exploration of cointegration systems, and
\citet{smeekes2021automated} use it to bound the RE in moderate dimensional
predictive regressions.

When we were preparing this manuscript, \citet{wijler2022restricted}
independently derived the RE of high dimensional ($p>n$) unit root
processes under the assumption that the innovations are i.i.d.~Gaussian,
as in \citet{kock2015oracle}. Our paper differs from \citet{wijler2022restricted}
in the following aspects. First, \citet{wijler2022restricted} uses
non-asymptotic tail bounds based on sub-Gaussian distributions \citep[Eq.(2.9)]{Wainwright2019high}
to obtain the rate of convergence of Plasso. Our approach, based on
the non-asymptotic deviation inequalities for the maximum and minimum
eigenvalues of Wishart random matrices \citep[Theorem 6.1]{Wainwright2019high}
offers sharper rates. Second, we go beyond i.i.d.~Gaussian and accommodate
sub-exponential and temporally dependent innovations by leveraging
the technique of Koml\'{o}s-Major-Tusn\'{a}dy coupling \citep[1976]{Komlos1975}
to achieve Gaussian approximation. Third, as in \citet{lee2022lasso}
our paper provides a comprehensive discussion of Plasso and Slasso
in the setting with a mix of unit root, stationary time series, and
cointegrated variables.

While high dimensional estimation counts on a well-behaved RE, testing
problems often involve the maximum eigenvalue or a few large eigenvalues.
Techniques of eigen-analysis for large random matrices are carried
over into nonstationary time series by \citet{zhang2018clt} for unit
root tests, by \citet{onatski2018alternative} and \citet{bykhovskaya2022asymptotics,bykhovskaya2022cointegration}
for cointegration tests, and by \citet{onatski2021spurious} for principal
component analysis and spurious regressions. 

Besides a handful of papers mentioned above, machine learning grows
fast in econometrics, covering i.i.d.~data \citep{chernozhukov2017double,caner2018asymptotically},
panel data \citep{su2016identifying,Shiforthcoming,shi2023forward},
weakly dependent time series \citep{yousuf2021boosting,babii2022machine},
and nonstationary time series \citep{phillips2021boosting,mei2022boosted,masini2022counterfactual},
to name a few.

\medskip

The rest of the paper is organized as follows. Section \ref{sec:Preliminaries}
introduces LASSO and two variants in implementation, namely Plasso
and Slasso. We put them into a unified framework, via a lemma that
highlights the two key building blocks. In Section \ref{sec:Theory},
we first focus on the low-level assumptions for DB and RE. We then
apply them to obtain the rates of convergence of Plasso and Slasso
given pure unit root regressors, respectively. To better match practical
circumstances, we further study LASSO given mixed regressors. Section
\ref{sec:Simulations} carries out Monte Carlo simulations and the
results corroborate the theoretical analysis. Section \ref{sec:Empirical-demo}
applies LASSO to predict the unemployment rate. Section \ref{sec:Conclusion}
concludes the paper. All technical proofs and additional simulation
results are relegated to the Online Appendices. 

\section{LASSO \label{sec:Preliminaries} }

Prior to formal presentation, we set up the notations. The set of
natural numbers, integers, real numbers, and complex numbers are denoted
as $\mathbb{N}$, $\mathbb{Z}$, $\mathbb{R}$, and $\mathbb{C}$,
respectively. The integer set $\{1,2,\cdots,n\}$ is denoted as $[n]$
for some $n\in\mathbb{N}$. The integer floor function and ceiling
function are denoted as $\left\lfloor \cdot\right\rfloor $ and $\left\lceil \cdot\right\rceil $,
respectively. For an $n$-dimensional vector $x=(x_{t})_{t\in[n]}$,
the $L_{2}$-norm is $\left\Vert x\right\Vert _{2}=\sqrt{\sum_{t=1}^{n}x_{t}^{2}}$,
the $L_{1}$-norm is $\left\Vert x\right\Vert _{1}=\sum_{t=1}^{n}\left|x_{t}\right|$,
and its sup-norm is $\|x\|_{\infty}=\sup_{t\in[n]}|x_{t}|$; we use
``double dots'' to denote the demeaned version $\ddot{x}=x-\bar{x}\cdot1_{n}$,
where $\bar{x}=n^{-1}\sum_{t=1}^{n}x_{t}$ and $1_{n}$ is a vector
of $n$ ones. Let $0_{n}$ be an $n\times1$ zero vector, and $I_{n}$
be the $n\times n$ identity matrix. For a generic index set $\mathcal{M}\subset[p]$
for some $p\in\mathbb{N}$, we use $\mathcal{M}^{c}=[p]\backslash\mathcal{M}$
to denote its complement, and $x_{\mathcal{M}}=\{x_{j}\}_{j\in\mathcal{M}}$
to denote the subvector of $x$ with coordinates located in $\mathcal{M}.$
For a generic matrix $B,$ let $B_{ij}$ be the $(i,j)$th element,
and $B^{\top}$ be its transpose. Let $\|B\|_{\max}=\max_{i,j}|B_{ij}|$,
and $\lambda_{\min}(B)$ and $\lambda_{\max}(B)$ be the minimum and
maximum eigenvalues, respectively. Define $a\wedge b:=\min\left\{ a,b\right\} $,
and $a\vee b:=\max\left\{ a,b\right\} $. An \emph{absolute constant}
is a positive, finite constant that is invariant with the sample size.
The abbreviation ``w.p.a.1'' is short for ``with probability approaching
one''. ``$a_{n}\stackrel{\mathrm{p}}{\preccurlyeq}b_{n}$'' means
that there is an absolute constant, say $c$, such that the event
$\left\{ a_{n}\leq cb_{n}\right\} $ holds w.p.a.1. Symmetrically,
``$a_{n}\stackrel{\mathrm{p}}{\succcurlyeq}b_{n}$'' means ``$b_{n}\stackrel{\mathrm{p}}{\preccurlyeq}a_{n}$''.

\subsection{Formulations}

Let $W_{t}=(W_{jt})_{j\in[p]}$ be a $p$-vector of regressors. At
time $n$, an econometrician is interested in using a linear combination
$\alpha+W_{n}^{\top}\theta$ to predict a future outcome $y_{n+1}$.
To learn the coefficients $\alpha$ and $\theta$, she collects historical
data $Y=\left(y_{t}\right)_{t\in[n]}$ ($n\times1$ vector) and $W=\left(W_{0},W_{1},\ldots,W_{n-1}\right)^{\top}=(W_{t-1}^{\top})_{t\in[n]}$
($n\times p$ matrix). When $p$ is close to $n$, or larger than
$n$, LASSO \citep{tibshirani1996regression} is one of the off-the-shelf
estimation methods. It minimizes the sum of squared residuals plus
an $L_{1}$ penalty 
\begin{equation}
(\widehat{\alpha}^{{\rm P}},\widehat{\theta}^{{\rm P}}):=\arg\min_{\alpha,\theta}\left\{ n^{-1}\left\Vert Y\boldsymbol{-}\alpha1_{n}-W\theta\right\Vert _{2}^{2}+\lambda\left\Vert \theta\right\Vert _{1}\right\} ,\label{eq:Lasso.theta.origin}
\end{equation}
where the intercept $\alpha$ is not penalized. The superscript ``$\mathrm{P}"$
of the estimator signifies \emph{P}lasso. Prediction is made as $\widehat{y}_{n+1}^{{\rm P}}=\widehat{\alpha}^{{\rm P}}+W_{n}^{\top}\widehat{\theta}^{{\rm P}}$.

Plasso is not scale-invariant, meaning that if we multiply a non-zero
constant $c_{j}$ to a regressor $W_{j,t-1}$, the corresponding LASSO
estimate will not change proportionally to $\widehat{\theta}_{j}^{\mathrm{P}}/c_{j}$.
Given that scale-invariance is a desirable property, a common practice
--- like the default option of LASSO via \texttt{glmnet::glmnet(x,y)}
in the \texttt{R} software --- scale-standardizes each regressor
by its sample s.d.~$\widehat{\sigma}_{j}=(n^{-1}\sum_{t=1}^{n}(W_{j,t-1}-\bar{W}_{j})^{2})^{1/2}$.
Let $D={\rm diag}(\widehat{\sigma}_{1},\widehat{\sigma}_{2},\cdots,\widehat{\sigma}_{p})$
be the diagonal matrix that stores the sample s.d., and the \emph{S}lasso
estimator is 
\begin{equation}
(\hat{\alpha}^{{\rm S}},\hat{\theta}^{{\rm S}}):=\arg\min_{\alpha,\theta}\left\{ n^{-1}\left\Vert Y\boldsymbol{-}\alpha1_{n}-W\theta\right\Vert _{2}^{2}+\lambda\left\Vert D\theta\right\Vert _{1}\right\} ,\label{eq:Lasso.theta.std}
\end{equation}
for which the prediction is made as $\widehat{y}_{n+1}^{{\rm S}}=\widehat{\alpha}^{{\rm S}}+W_{n}^{\top}\widehat{\theta}^{{\rm S}}$. 

To analyze Plasso and Slasso under the same framework, we write 
\begin{equation}
\left(\hat{\alpha},\hat{\theta}\right):=\arg\min_{\alpha,\theta}\left\{ n^{-1}\left\Vert Y\boldsymbol{-}\alpha1_{n}-W\theta\right\Vert _{2}^{2}+\lambda\left\Vert H\theta\right\Vert _{1}\right\} \label{eq:gene_lasso}
\end{equation}
where $H$ is a positive definite placeholder: $H=I_{p}$ in Plasso,
or $H=D$ in Slasso. In this paper, we will focus on the high dimensional
component $\hat{\theta}$, and the intercept is obviously $\hat{\alpha}=\bar{Y}-\bar{W}\hat{\theta}$
as it is unpenalized. Substitute $\widehat{\alpha}$ back to the criterion
function in (\ref{eq:gene_lasso}), the $\theta$ component is numerical
equivalent to 
\begin{equation}
\hat{\theta}:=\arg\min_{\theta}\left\{ n^{-1}\left\Vert \ddot{Y}-\ddot{W}\theta\right\Vert _{2}^{2}+\lambda\left\Vert H\theta\right\Vert _{1}\right\} \label{eq:demean_gene_lasso}
\end{equation}
where $\ddot{Y}=Y-\bar{Y}1_{n}$ and similar demeaning applies to
each column of $W$ to produce $\ddot{W}$. 

\subsection{Generic Convergence}

The above is the numerical programming independent of the DGP. Now,
suppose that the dependent variable is generated by 
\begin{equation}
y_{t}=\alpha^{*}+W_{t-1}^{\top}\theta^{*}+u_{t},\label{eq:DGP_W}
\end{equation}
where $\left(\alpha^{*},\theta^{*}\right)$ are the true parameters.
Sparsity means that most elements in $\theta^{*}$ are exactly zero.
Let $\mathcal{S}=\{j\in[p]:\theta_{j}^{*}\neq0\}$ be the true \emph{active
set}, i.e., the location of the non-zero components, with its cardinality
$s=|\mathcal{S}|.$ 

\medskip
\begin{rem}
Throughout this paper, we work with \emph{exact sparsity} for simplicity.
Extension to \emph{approximate sparsity} \citep[p.108-110]{buhlmann2011statistics}
is straightforward, although it will substantially complicate the
notations when we deal with the mixed regressors. One route of such
an extension is to follow \citet[Condition AS]{belloni2012sparse}
by modeling $y_{t}=\alpha^{*}+W_{t-1}^{\top}\theta_{n}^{*}+r_{t-1}+u_{t}$,
where the approximation error of the sparse coefficient is controlled
by $r_{t-1}$ which satisfies $\sum_{t=1}^{n}r_{t-1}^{2}=O_{p}(s)$.
For example, we can allow ``local-to-zero'' coefficients $\theta_{n}^{*}$
to relax exact sparsity in the form $\theta_{n}^{*}=\theta^{*}+\theta_{r}^{*}$
with $\theta^{*}$ is a sparse coefficient. When $W_{t-1}$ is a unit
root vector with i.i.d.~standard normal innovations, the sparse approximation
error $r_{t-1}=W_{t-1}^{\top}\theta_{r}^{*}\sim\mathcal{N}(0,(t-1)\|\theta_{r}^{*}\|_{2}^{2})$.
If the violation of exact sparsity is mild to the degree $\|\theta_{r}^{*}\|_{2}=O(\sqrt{s}/n)$,
then $\sum_{t=1}^{n}r_{t-1}^{2}=O_{p}\left(n^{2}\|\theta_{r}^{*}\|_{2}^{2}\right)=O_{p}(s)$
is satisfied. 
\end{rem}
\medskip

Well-known since \citet{bickel2009simultaneous}, the two essential
building blocks for the convergence of high dimensional LASSO are
the DB (See the condition in Lemma \ref{lem:gene-Lasso} below) and
the RE. Let $\hat{\Sigma}=\ddot{W}^{\top}\ddot{W}/n$ be the sample
covariance matrix of all regressors. For some $L>0$, the \emph{restricted
eigenvalue} is defined, in our context, as 
\begin{equation}
\kappa_{H}(\hat{\Sigma},L,s):=\inf_{\delta\in\mathcal{R}(L,s)}\dfrac{\delta^{\top}H^{-1}\hat{\Sigma}H^{-1}\delta}{\delta^{\top}\delta},\label{eq:RE}
\end{equation}
where $\mathcal{R}(L,s)=\{\delta\in\mathbb{R}^{p}\backslash\{0\}:\|\delta_{\mathcal{M}^{c}}\|_{1}\leq L\|\delta_{\mathcal{M}}\|_{1},\ \text{for all }|\mathcal{M}|\leq s\}.$
As our paper focuses on the rate of convergence, without loss of generality
we follow \citet[p.106]{buhlmann2011statistics} and \citet[Theorem 7.2]{bickel2009simultaneous}
by setting $L=3$ as a convenient choice of the constant and use $\widehat{\kappa}_{H}=\kappa_{H}(\hat{\Sigma},3,s)$
to simplify the notation. The following finite sample bounds hold
for the generic LASSO estimator (\ref{eq:demean_gene_lasso}). 
\begin{lem}
\label{lem:gene-Lasso} If $\lambda\geq4\|n^{-1}\sum_{t=1}^{n}H^{-1}\ddot{W}_{t-1}u_{t}\|_{\infty},$
then 
\begin{align*}
n^{-1}\|\ddot{W}(\ensuremath{\hat{\theta}-\theta^{*}})\|_{2}^{2} & \leq\frac{4\lambda^{2}s}{\widehat{\kappa}_{H}}\\
\|H(\hat{\theta}-\theta^{*})\|_{1} & \leq\frac{4\lambda s}{\widehat{\kappa}_{H}}\\
\|H(\hat{\theta}-\theta^{*})\|_{2} & \leq\frac{2\lambda\sqrt{s}}{\widehat{\kappa}_{H}}.
\end{align*}
\end{lem}
The condition requires that the tuning parameter $\lambda$ in the
LASSO estimation should be chosen above the deviation $\|n^{-1}\sum_{t=1}^{n}H^{-1}\ddot{W}_{t-1}u_{t}\|_{\infty}$,
which is governed by the DGP of $W_{t}$ and $u_{t}$. This is the
DB condition. The convergence rates of the LASSO estimator, signified
by the right-hand side expressions in Lemma \ref{lem:gene-Lasso},
are determined by the sparsity index $s$, the RE $\widehat{\kappa}_{H}$,
and the tuning parameter $\lambda$. We study in the next section
the conditions under which we can establish desirable rates for the
DB and RE, and then apply these two quantities to Plasso and Slasso
for their rates of convergence.

\section{Theory\label{sec:Theory}}

\subsection{Unit Root Regressors \label{sec:UnitRoot}}

This paper highlights unit root regressors. While $W_{j}$ is for
a generic regressor with coefficient $\theta_{j}$, we denote the
$j$th unit root regressor as $X_{j}=\left(X_{j0},\ldots,X_{j,t-1}\right)^{\top}$
and use $\beta_{j}$ as its coefficient. In this section we consider
a DGP 
\begin{equation}
y_{t}=\alpha^{*}+X_{t-1}^{\top}\beta^{*}+u_{t}\label{eq:DGP_X}
\end{equation}
where $X_{t}=\left(X_{1t},\ldots,X_{pt}\right)^{\top}$ is a vector
of $p$ unit root processes $X_{t}=X_{t-1}+e_{t}$, and for simplicity
let the initial value $\left\Vert X_{t=0}\right\Vert _{\infty}=O_{p}(1)$.
We concatenate it with the error term $u_{t}$ into a $(p+1)$-vector
$v_{t}=(e_{t}^{\top},u_{t})^{\top}$, and assume it is generated from
\begin{equation}
v_{t}=\Phi\varepsilon_{t},\label{eq:def-error}
\end{equation}
where $\ensuremath{\varepsilon_{t}=(\varepsilon_{jt})_{j\in[p+1]}}$
is a $(p+1)\times1$ random vector and $\Phi$ is a $(p+1)\times(p+1)$
deterministic matrix. For each $j\in[p+1]$, the shock
\begin{equation}
\varepsilon_{jt}=\sum_{d=0}^{\infty}\psi_{jd}\eta_{j,t-d}\label{eq:linrProc}
\end{equation}
follows a linear process \citep{phillips1992asymptotics}, which yields
temporal dependence.

We will use low-level assumptions to build up the two high-level asymptotic
properties DB and RE. Throughout this paper, we take the number of
regressors $p=p\left(n\right)$ and the sparsity index $s=s\left(n\right)$
as deterministic functions of the sample size $n$. In formal asymptotic
statements, we explicitly send $n\to\infty$ only, while it is understood
that $p(n)\to\infty$ as $n\to\infty$ whereas $s\left(n\right)$
is either fixed or divergent. 

\subsubsection{Deviation Bound \label{subsec:DB}}

{} We begin with the DB, which involves $n$ and $p$ only. As we allow
high dimensionality in that $p>n$, the model (\ref{eq:DGP_X}) must
be regularized by assumptions. We first impose Assumption \ref{assu:tail}
concerning the marginal distribution of the underlying shocks $\eta_{jt}$.
\begin{assumption}
\label{assu:tail} Suppose $(\eta_{jt})$ is i.i.d.~over the cross
section $j\in[p+1]$ and time $t\in\mathbb{Z}$ with $\mathbb{E}\eta_{jt}=0$
and $\mathbb{E}\eta_{jt}^{2}=1$. There exists an absolute constant
$C_{\mathrm{f}}$ such that 
\begin{equation}
\int_{-\infty}^{\infty}|f(x+a)-f(x)|dx\leq C_{\mathrm{f}}|a|,\ \forall a\in\mathbb{R},\label{eq:density}
\end{equation}
where $f$ is the density function of $\eta_{jt}$. For all $t\in\mathbb{Z}$
and $\mu>0$, there exist absolute constants $C_{\eta}$ and $b_{\eta}$
such that 
\begin{equation}
\Pr\left\{ |\eta_{jt}|>\mu\right\} \leq C_{\eta}\exp(-\mu/b_{\eta}).\label{eq:innovtail}
\end{equation}
 
\end{assumption}
In Assumption \ref{assu:tail} we assume i.i.d.~$\eta_{jt}$ with
density $f$ over both $j$ and $t$, following \citet{zhang2019identifying}
and \citet{smeekes2021automated}. It allows us to invoke concentration
inequalities in the high-dimensional setting. Condition (\ref{eq:innovtail})
is known as the \emph{sub-exponential} tail condition, which includes
the familiar \emph{sub-Gaussian} tail as a special case. 

Assumption \ref{assu:alpha} is concerning the coefficient $\psi_{jd}$
in the linear process (\ref{eq:linrProc}), which governs the temporal
dependence of $\varepsilon_{jt}$. For any $z\in\mathbb{C}$, we denote
the polynomial $\psi_{j}(z)=\psi_{j0}+\sum_{d=1}^{\infty}\text{\ensuremath{\psi_{jd}\cdot z^{d}}}=1+\sum_{d=1}^{\infty}\text{\ensuremath{\psi_{jd}\cdot z^{d}}}$,
where without loss of generality we normalize $\psi_{j0}=1$.
\begin{assumption}
\label{assu:alpha} For all $j\in[p+1]$, there exists some absolute
constants $C_{\psi}$, $c_{\psi}$ and $r$ such that the coefficients
of the linear processes
\begin{equation}
|\psi_{jd}|\leq C_{\psi}\exp\left(-c_{\psi}d^{r}\right),\ \ \forall d\in\mathbb{N},\label{eq:psi}
\end{equation}
and $\left|\psi_{j}(z)\right|>c_{\psi}>0$ for any $z\in\left\{ a\in\mathbb{C}:|a|\leq1\right\} $. 
\end{assumption}
Assumption \ref{assu:alpha} is a sufficient condition for linear
processes to satisfy the geometrically strong mixing ($\alpha$-mixing)
condition (See Lemma \ref{lem:mixing} in the Appendix). Finite-order
strictly stationary ARMA processes are special cases of (\ref{eq:psi})
as they admit MA($\infty$) representations with exponentially decaying
coefficients. The sub-exponential tail in Assumption \ref{assu:tail}
and strong mixing in Assumption \ref{assu:alpha} are common conditions
in high dimensional time series regressions \citep{fan2011high,fan2013large,ding2021high}. 

The cross-sectional dependence across the regressors, encoded in $\Phi$,
must be regularized as well. Assumption \ref{assu:covMat} is concerning
$\Phi$, which maps $\varepsilon_{t}$ into $v_{t}$ via (\ref{eq:def-error}).
Let $\Omega=\Phi\Phi^{\top}$.
\begin{assumption}
\label{assu:covMat} There are absolute constants $c_{\Omega}$, $C_{\Omega}$
and $C_{L}$ such that: (a) $c_{\Omega}\leq\lambda_{\min}(\Omega)\leq\lambda_{\max}(\Omega)\leq C_{\Omega}$;
(b) $\max_{j\in[p+1]}\sum_{\ell=1}^{p+1}\left|\Phi_{j\ell}\right|\leq C_{L}$. 
\end{assumption}
Assumption \ref{assu:covMat} (a) controls the magnitude of cross-sectional
correlation. It rules out the unfavorable cases where innovations
are very strongly correlated. Part (b), together with the condition
(\ref{eq:innovtail}), guarantees the sub-exponential tail of $v_{jt}=\sum_{\ell=1}^{p+1}\Phi_{j\ell}\varepsilon_{jt}$
for all $j\in[p+1]$. 
\begin{rem}
The literature on low dimensional regressions has developed a range
of general concepts to characterize dynamics in times series. Many
papers assume martingale difference sequence on $\eta_{jt}$, and
a vector moving average VMA($\infty$) process for $v_{t}$. In high
dimensional settings, we must ensure probabilistic results to hold
uniformly over a large $p$. For this purpose, we assume $\eta_{jt}$
i.i.d.~and form $v_{t}$ by linear combination of $(\varepsilon_{jt})_{t}$
to invoke existing concentration inequalities \citep{merlevede2011bernstein}
and coupling inequalities \citep[1976]{Komlos1975}. 
\end{rem}
The above assumptions have been prepared for DB.
\begin{prop}
\label{prop:UnitDB} Under Assumptions \ref{assu:tail}-\ref{assu:covMat},
if $(\log p)^{1+\frac{2}{r}}=o(n)$ , there exists an absolute constant
$C_{{\rm DB}}$ such that\textcolor{red}{{} }
\begin{equation}
4\,\bigg\Vert\dfrac{1}{n}\sum_{t=1}^{n}\ddot{X}_{t-1}u_{t}\bigg\Vert_{\infty}\leq C_{{\rm DB}}(\log p)^{1+\frac{1}{2r}}\label{eq:UnitDB}
\end{equation}
 w.p.a.1.~as $n\to\infty.$
\end{prop}
\begin{rem}
\citet[Lemma 13]{wong2020lasso} work with stationary mixing time
series, and their DB is a direct corollary of the Bernstein-type concentration
inequality for mixing sequences \citep[Theorem 1]{merlevede2011bernstein}.
For nonstationary regressors, our DB goes with a decomposition of
$n^{-1}\sum_{t=1}^{n}\ddot{X}_{t-1}u_{t}$ into three terms and they
are handled one by one. Due to the weak dependence of $\varepsilon_{t}$,
time series blocking techniques help to separate the observations
into groups across which the temporal dependence vanishes asymptotically.
\end{rem}
It is known that $n^{-1}\sum_{t=1}^{n}\ddot{X}_{j,t-1}u_{t}=O_{p}\left(1\right)$
as $n\to\infty$ for an individual unit root process $X_{j,t-1}$
and stationary error $u_{t}$ \citep{phillips1986understanding}.
Here to accommodate all $p$ unit root time series in a uniform matter,
the DB grows at a mild speed $(\log p)^{1+\frac{1}{2r}}$, where the
$r$ from Assumption \ref{assu:alpha} governs the rate of diminishing
temporal dependence. When $r$ is arbitrarily large, $\varepsilon_{t}$
will approach to temporal independence and the rate on the right-hand
side of (\ref{eq:UnitDB}) is reduced to $(\log p)^{1+\frac{1}{\infty}}=\log p$.
The constant ``4'' on the left-hand side replicates the same constant
required for $\lambda$ in Lemma \ref{lem:gene-Lasso}.

\subsubsection{Restricted Eigenvalue \label{subsec:RE}}

When $H$ in (\ref{eq:demean_gene_lasso}) is an identity matrix,
we study the RE $\widehat{\kappa}_{I}=\kappa_{I}(\widehat{\Sigma},3,s)$
for Plasso, associated with the sample Gram matrix $\widehat{\Sigma}=\ddot{X}^{\top}\ddot{X}/n$.
In the i.i.d.~case, it is easy to establish RE as the $(j,k)$th
entry $\widehat{\Sigma}_{jk}$ \emph{converges in probability} to
the population covariance ${\rm cov}(X_{jt},X_{kt})$ for any fixed
$j,k$, and then in high dimension we can apply concentration inequalities
to construct a uniform bound for $\|\widehat{\Sigma}-\Sigma_{X}\|_{\max}=o_{p}(1)$,
where $\Sigma_{X}:=\mathbb{E}(\widehat{\Sigma})$ has minimum eigenvalue
bounded away from 0; See \citet[Eq.(3.3)]{bickel2009simultaneous}.
This strategy does not carry over into non-stationary data. Recall
that $X_{t}=X_{t-1}+e_{t}$ and define $\Sigma_{e}:=\mathbb{E}\left(e_{t}e_{t}^{\top}\right).$
After scaling by $1/n$, for each fixed pair $(j,k)$ the random variable
\[
n^{-1}\widehat{\Sigma}_{jk}\stackrel{d}{\to}\mathcal{D}_{jk}=\int_{0}^{1}\mathcal{B}_{j}(r)\mathcal{B}_{k}(r)dr-\int_{0}^{1}\mathcal{B}_{j}(r)dr\int_{0}^{1}\mathcal{B}_{k}(r)dr
\]
where \textquotedblleft \textbf{$\stackrel{d}{\to}$}\textquotedblright{}
denotes \emph{convergence in distribution}, $\mathcal{B}_{j}$ and
$\mathcal{B}_{k}$ are two Brownian motions, and the limiting distribution
$\mathcal{D}_{jk}$ is a non-degenerate stable law. Since the diagonal
elements $\mathcal{D}_{jj}$ has non-trivial probability in any small
neighbor of zero, when the dimension $p$ accumulates $\widehat{\kappa}_{I}/n$
will shrink to 0. (\textcolor{black}{See Section \ref{sec:RE_example}
for elaboration.}) This is in sharp contrast with the case of i.i.d.~regressors,
where the RE is bounded away from 0. 

As the RE appears in the denominator of the error bounds in Lemma
\ref{lem:gene-Lasso}, the convergence of LASSO requests that the
RE shrinks to zero slowly enough. Lemma \ref{lem:Normal_RE} prepares
an RE condition when the underlying innovations are i.i.d.~normal,
and Proposition \ref{prop:UnitRE} allows non-Gaussian and time dependent
innovations. This is one of the main theoretical contributions of
this paper. 
\begin{lem}
\label{lem:Normal_RE} Suppose $\varepsilon_{t}\sim i.i.d.\ \mathcal{N}(0,I_{p+1})$
and $\Phi$ satisfies Assumption \ref{assu:covMat} (a). Then there
exists an absolute constant $c_{\kappa}$ such that 
\begin{equation}
\frac{\widehat{\kappa}_{I}}{n}\geq\dfrac{c_{\kappa}}{s\log p}\label{eq:RE-unit}
\end{equation}
holds w.p.a.1.~as $n\to\infty$ and $s/(n\wedge p)\to0$.
\end{lem}
With a fixed $p$, \citet[Lemma A.2]{phillips1990statistical} show
$\widehat{\Sigma}/n=\ddot{X}^{\top}\ddot{X}/n^{2}$ is positive-definite
w.p.a.1\@.~as $n\to\infty$. For unit root regressors the denominator
under $\ddot{X}^{\top}\ddot{X}$ is $n^{2}$, instead of $n$ as in
the i.i.d.~case.\footnote{In high dimensional regressions, the rates of convergence under weakly
dependent data are largely similar to those under the i.i.d.~data.
In the rest of the paper we mostly compare our results with what happens
under the i.i.d.~case for simplicity.} To align with this convention, we put $\widehat{\kappa}/n$ on the
left-hand side of (\ref{eq:RE-unit}). When $\varepsilon_{t}$ are
i.i.d.~normal, the right-hand side of (\ref{eq:RE-unit}) gives a
lower bound of RE proportional to $1/(s\log p)$. This result echoes
\citet[Theorem B.2-B.3]{smeekes2021automated} where they establish
$\lambda_{\min}(\widehat{\Sigma}/n)\stackrel{\mathrm{p}}{\succcurlyeq}1/p$
in the moderate dimensional case when $p/\sqrt{n}\to0$. Our (\ref{eq:RE-unit})
replaces $\lambda_{\min}(\widehat{\Sigma}/n)$ by the restricted version
$\widehat{\kappa}_{I}/n$ on the left-hand side, and replaces $p$
by $s\log p$ on the right-hand side. 

\medskip
\begin{rem}
\label{rem:sketch-proof-RE}Here we sketch the proof of Lemma \ref{lem:Normal_RE}.
The sparsity embodied by the restricted set $\mathcal{R}(3,s)$ reduces
the essential number of regressors from $p$ to the order of $s$.
As a unit root time series consists of partial sums of i.i.d.~shocks,
the minimum eigenvalue of the Gram matrix can be bounded below by
considering the largest $\ell$ eigenvalues of the deterministic transformation
matrix, where we choose $\ell\asymp s\log p$. This step shifts the
focus from the Gram matrix of unit root time series to that of Gaussian
random vectors, and in the same time reduces the sample size from
$n$ to effectively $\ell$. The deduction in both dimensions allows
us to invoke existing results about the Wishart matrices, in particular
the non-asymptotic deviation inequalities \citep[Theorem 6.1]{Wainwright2019high},
to bound away from 0 the sample minimum eigenvalues for any submatrix
of dimensions of the same order as $s$, and then extend the bound
uniformly to all such sub-matrices. 
\end{rem}
\medskip
\begin{rem}
\citet{wijler2022restricted} studies Plasso with data generated from
(\ref{eq:DGP_X}) with no intercept. Considering the i.i.d.~normal
$\varepsilon_{t}$ exclusively, \citet{wijler2022restricted} takes
advantage of the fact that the quadratic form of independent normal
distribution follows the $\chi^{2}$ distribution, and invokes the
tail probability bound of the maximum of sub-Gaussian random variables
\citep[Eq.(2.9)]{Wainwright2019high}. His main result \citep[Theorem 1]{wijler2022restricted}
have a slower diminishing rate than ours.
\end{rem}
The i.i.d.~normality in Lemma \ref{lem:Normal_RE} is a strong assumption.
To obtain an RE that accommodates more general innovations, we must
control the relative magnitude among $n$, $p$ and $s$.
\begin{assumption}
\label{assu:asym_n} (a) $p=O(n^{\nu})$ for an arbitrary absolute
constant $\nu\in\left(0,\infty\right)$; (b)\textcolor{red}{{} }$s=O\left(n^{1/4-\zeta}\wedge p^{1-\zeta}\right)$
for an arbitrary small constant $\zeta>0$.\textcolor{red}{{} }
\end{assumption}
Assumption \ref{assu:asym_n} (a) allows $p$ to be of high dimension.
The polynomial rate $n^{\nu}$ is for simplicity of presentation.\footnote{The proofs can still go through if we relax $p$ to grow at some exponential
rate of $n$. But the speed of such rates will be peculiar to each
of our convergence statement, thereby complicate the notations. For
example, the right-hand side of (\ref{eq:proof_logpn}) would involve
many specific terms.} In the meantime, by Assumption \ref{assu:asym_n} (b) the sparsity
index cannot grow faster than $n^{1/4}$ when the innovations are
non-normal. 

\medskip
\begin{prop}
\label{prop:UnitRE} If Assumptions \ref{assu:tail}-\ref{assu:asym_n}
hold, then (\ref{eq:RE-unit}) is satisfied w.p.a.1.~as $n\to\infty$. 
\end{prop}
\medskip

Proposition \ref{prop:UnitRE} substantially relaxes the distributional
and dependence conditions by substituting the normality in Lemma \ref{lem:Normal_RE}
with the sub-exponential tails in Assumptions \ref{assu:tail}, and
replacing i.i.d.~with the mixing condition in Assumption \ref{assu:alpha}.

\begin{rem}
The proof of Proposition \ref{prop:UnitRE} extends that of Lemma
\ref{lem:Normal_RE}. We use the Beveridge-Nelson decomposition to
obtain a leading term of the sum of independent innovations, which
asymptotically mimics the behavior of a Brownian motion. For each
$j$, the sub-exponential tail in Assumption \ref{assu:tail} allows
applying the Koml\'{o}s-Major-Tusn\'{a}dy coupling \citep{komlos1976approximation},
which is again a non-asymptotic inequality. Assisted by the union
bound, we carry the result in Lemma \ref{lem:Normal_RE} over into
the case of time dependent non-Gaussian $\varepsilon_{t}$. 
\end{rem}
\medskip

Given the two building blocks, DB and RE, we are ready to apply them
to study the LASSO estimators.

\subsubsection{Plain LASSO }

Consider Plasso 
\[
\hat{\beta}^{{\rm P}}=\text{\ensuremath{\arg\min_{\beta}\left\{ \dfrac{1}{n}\|\ddot{Y}-\ddot{X}\beta\|_{2}^{2}+\lambda\|\beta\|_{1}\right\} }}
\]
in the form of a special case of (\ref{eq:demean_gene_lasso}). Parallel
results to Lemma \ref{lem:gene-Lasso} immediately follows. 
\begin{thm}
\label{thm:LassoError}Suppose Assumptions \ref{assu:tail}-\ref{assu:asym_n}
hold. If we choose $\lambda=C_{\mathrm{DB}}(\log p)^{1+\frac{1}{2r}}$,
the Plasso estimator satisfies 
\begin{align}
\dfrac{1}{n}\|\ddot{X}(\ensuremath{\widehat{\beta}^{{\rm P}}-\beta}^{*})\|_{2}^{2} & =O_{p}\left(\dfrac{s^{2}}{n}(\log p)^{3+\frac{1}{r}}\right)\label{eq:UnitForecast}\\
\|\ensuremath{\hat{\beta}^{{\rm P}}-\beta}^{*}\|_{1} & =O_{p}\left(\dfrac{s^{2}}{n}(\log p)^{2+\frac{1}{2r}}\right)\label{eq:UnitL1}\\
\|\ensuremath{\hat{\beta}^{{\rm P}}-\beta}^{*}\|_{2} & =O_{p}\left(\dfrac{s^{3/2}}{n}(\log p)^{2+\frac{1}{2r}}\right).\label{eq:UnitL2}
\end{align}
\end{thm}
It is well-known that with high dimensional i.i.d.~data, Plasso's
$L_{1}$ and $L_{2}$ estimation error bounds are $s\sqrt{\left(\log p\right)/n}$
and $\sqrt{s\left(\log p\right)/n}$, respectively, under standard
conditions. Instead of the usual $\sqrt{n}$ for i.i.d.~data, in
the rates of convergence the denominators are $n$, yielding the familiar
\emph{super-consistency} when unit root regressors are present. The
numerators, on the other hand, are multiplied by an extra factor $s(\log p)^{\frac{3}{2}+\frac{1}{2r}}$.
This additional factor reflects the effect of the nonstationary time
series, where $s\log p$ comes from the denominator of the lower bound
of RE in (\ref{eq:RE-unit}), and another $(\log p)^{\frac{1}{2}+\frac{1}{2r}}$
term stems from DB. When $r$ is arbitrarily large, the temporal dependence
in $\varepsilon_{t}$ vanishes and the extra factor is reduced to
$s\left(\log p\right)^{3/2}$. For example, in the special case of
i.i.d.~normal $\varepsilon_{t}$, in (\ref{eq:UnitL1}) our convergence
rate under the $L_{1}$-norm is $\dfrac{s^{2}}{n}(\log p)^{2}$ as
$r=\infty$. This rate is faster than \citet{wijler2022restricted}'s
Corollary 1, which shows 
\[
\|\ensuremath{\hat{\beta}^{{\rm P}}-\beta}^{*}\|_{1}=O_{p}\left(\frac{s^{3}}{n^{1-\zeta_{1}}}(\log p)^{2}\right)=O_{p}\left(\frac{s^{2}}{n}(\log p)^{2}\times sn^{\zeta_{1}}\right)
\]
 for any $\zeta_{1}>0$.

The tuning parameter $\lambda$ in Theorem \ref{thm:LassoError} involves
an absolute constant $C_{\mathrm{DB}}$, which in turn depends on
the absolute constants in the assumptions that are unknown in practice.
Nevertheless, for all the left-hand side quantities in (\ref{eq:UnitForecast})--(\ref{eq:UnitL2})
to converge to zero in probability, it suffices if
\begin{equation}
\frac{(\log p)^{1+\frac{1}{2r}}}{\lambda}+\dfrac{s(\log p)}{\sqrt{n}}\lambda\to0,\label{eq:P_adm_rate}
\end{equation}
which specifies a wide range of admissible rates for $\lambda$. 

\subsubsection{Standardized LASSO }

Plasso is the prototype of the $L_{1}$-penalized regression. In practice,
Slasso is more often implemented in statistical software as scale-invariance
is a desirable property. Again, we focus on the high dimensional coefficient
$\beta$ in 
\[
\hat{\beta}^{{\rm S}}:=\text{\ensuremath{\arg}\ensuremath{\min_{\beta}\left\{ \dfrac{1}{n}\|\ddot{Y}-\ddot{X}\beta\|_{2}^{2}+\lambda\|D\beta\|_{1}\right\} }}.
\]

\begin{rem}
The only difference between Slasso and Plasso is that the former uses
$\widehat{\sigma}_{j}$ to scale-standardize each original regressor.
This transformation is theoretically uninteresting for i.i.d.~data,
where the sample s.d.~will converge to its population s.d as $n\to\infty$.
As a result, in this case Slasso shares the same rates of convergence
as Plasso, because the constant population s.d.~does not alter the
rates in DB and RE. The commonality breaks down when the regressors
are unit roots, since $\widehat{\sigma}_{j}/\sqrt{n}$ converging
\emph{in distribution} to a non-degenerate non-negative random variable,
which is the square root of an integral of the squared Brownian bridge.
 Since $\widehat{\sigma}_{j}/\sqrt{n}=O_{p}(1)$ appears in the penalty,
it incurs extra randomness. 
\end{rem}
Proposition \ref{prop:MinMax-Unit} establishes the bounds for $\widehat{\sigma}_{\min}$
and $\widehat{\sigma}_{\max}$, which refresh the DB for Slasso as
well as the RE $\widehat{\kappa}_{D}:=\kappa_{D}(\hat{\text{\ensuremath{\Sigma}}},3,s)$,
which is the restricted eigenvalue of the sample \emph{correlation
coefficient} matrix $D^{-1}\widehat{\Sigma}D^{-1}$ of the original
data.
\begin{prop}
\label{prop:MinMax-Unit}Suppose that Assumptions \ref{assu:tail}-\ref{assu:asym_n}
hold. As $n\to\infty$ w.p.a.1.~we have
\begin{enumerate}
\item Bounds for the sample s.d.:
\begin{equation}
n(\log p)^{-1}\stackrel{\mathrm{p}}{\preccurlyeq}\widehat{\sigma}_{\min}^{2}\leq\widehat{\sigma}_{\max}^{2}\stackrel{\mathrm{p}}{\preccurlyeq}n\log p.\label{eq:sigma-minmax}
\end{equation}
\item DB: There exists an absolute constant $\tilde{C}_{{\rm DB}}$ such
that
\begin{equation}
4\,\|\dfrac{1}{n}\sum_{t=1}^{n}D^{-1}\ddot{X}_{t-1}u_{t}\|_{\infty}\leq\dfrac{\tilde{C}_{{\rm DB}}}{\sqrt{n}}(\log p)^{\frac{3}{2}+\frac{1}{2r}}.\label{eq:SUnitDB}
\end{equation}
\item RE: there exists an absolute constant $c_{\kappa}$ such that
\begin{equation}
\widehat{\kappa}_{D}\geq\frac{c_{\kappa}}{s(\log p)^{4}}.\label{eq:SUnitRE}
\end{equation}
\end{enumerate}
\end{prop}
\begin{rem}
Compared with (\ref{eq:UnitDB}), the absolute constant $\tilde{C}_{{\rm DB}}$
in (\ref{eq:SUnitDB}) is Slasso's counterpart of $C_{{\rm DB}}$
for Plasso; Slasso incurs another extra factor $(n^{-1}\log p)^{1/2}$
on the right-hand side which comes from the probabilistic bounds of
$\widehat{\sigma}_{\min}$ in (\ref{eq:sigma-minmax}). The self-normalization
of $\hat{\Sigma}$ by $D^{-1}$ eliminates $n$ from (\ref{eq:SUnitRE}),
whereas an extra price $(\log p)^{3}$ is paid to cope with the randomness
in $D$. 

The DB and RE for Slasso ready another straightforward application
of Lemma \ref{lem:gene-Lasso}. 
\end{rem}
\begin{thm}
\label{thm:SlassoError} Specify $\lambda=\dfrac{\tilde{C}_{{\rm DB}}}{\sqrt{n}}(\log p)^{\frac{3}{2}+\frac{1}{2r}}$
given the same $\tilde{C}_{{\rm DB}}$ in (\ref{eq:SUnitDB}). Under
Assumptions \ref{assu:tail}-\ref{assu:asym_n}, we have 
\begin{align}
\dfrac{1}{n}\|\ddot{X}(\hat{\beta}^{{\rm S}}-\beta^{*})\|_{2}^{2} & =O_{p}\left(\dfrac{s^{2}}{n}(\log p)^{7+\frac{1}{r}}\right)\label{eq:UnitForecast-1}\\
\|\hat{\beta}^{{\rm S}}-\beta^{*}\|_{1} & =O_{p}\left(\dfrac{s^{2}}{n}(\log p)^{6+\frac{1}{2r}}\right)\label{eq:UnitL1-1}\\
\|\hat{\beta}^{{\rm S}}-\beta^{*}\|_{2} & =O_{p}\left(\dfrac{s^{3/2}}{n}(\log p)^{6+\frac{1}{2r}}\right).\label{eq:UnitL2-1}
\end{align}
\end{thm}
In terms of fitting, the leading term $s^{2}/n$ in (\ref{eq:UnitForecast-1})
is the same as that in (\ref{eq:UnitForecast}) for Plasso, up to
an extra logarithm term. Super-consistency is preserved in $\hat{\beta}^{{\rm S}}$
for the true original coefficient $\beta^{*}$. To counter the unknown
absolute constant $\tilde{C}_{{\rm DB}}$, it is sufficient to specify
$\lambda$ as 
\begin{equation}
\frac{(\log p)^{\frac{3}{2}+\frac{1}{2r}}}{\sqrt{n}\lambda}+s^{2}(\log p)^{\frac{9}{2}}\lambda\to0\label{eq:S_adm_rate}
\end{equation}
for the consistency of the quantities in Theorem \ref{thm:SlassoError}.

\subsection{Mixed Regressors\label{sec:Mixed-Regressors}}

In reality when we predict a target variable $y_{t}$ with many potential
regressors, the regressors are most likely to have various dynamic
patterns and we would not restrict ourselves by using nonstationary
regressors exclusively. In the low dimensional case it is possible
by pre-testing to classify variables into stationary and nonstationary
ones, but the power of unit root tests are known to be weak in finite
sample. What is worse, in high dimensional cases the individual test
errors will accumulate in multiple testing procedures. For these reasons,
\citet{lee2022lasso} study LASSO with mixed-root regressors without
pre-testing. That is, if we have a pool of mixed stationary and nonstationary
regressors, we keep an agnostic view and throw them all into LASSO.
It is in line with the attitude that a good machine learning method
should adapt to the complex nature of the regression\emph{.}

In this section we allow additional stationary regressors $Z_{t}$
in a DGP 
\begin{equation}
y_{t}=\alpha^{*}+X_{t-1}^{\top}\beta^{*}+Z_{t-1}^{\top}\gamma^{*}+u_{t}.\label{eq:DGP_XZ}
\end{equation}
The generic regressor $W_{t}$ and parameter $\theta^{*}$ in (\ref{eq:DGP_W})
represent $(X_{t}^{\top},Z_{t}^{\top})^{\top}$ and $(\beta^{*\top},\gamma^{*\top})^{\top}$,
respectively. Let $p_{x}$ be the length of $X_{t-1}$, and $p_{z}$
be the length of $Z_{j,t-1}$. We assume that the stationary components
$Z_{t}$, $e_{t}$ and $u_{t}$ are potentially correlated in the
following form
\begin{equation}
v_{t}=(e_{t}^{\top},Z_{t}^{\top},u_{t})^{\top}=\Phi\varepsilon_{t},\label{eq:def-I0-mix}
\end{equation}
where we slightly abuse the notations to keep using $v_{t}$, $\varepsilon_{t}$
and $\Phi$, understanding that they are adapted to the mixed root
case with the total number of regressors $p=p_{x}+p_{z}$.\footnote{The fact $E[\varepsilon_{t}]=0_{p+1}$ and (\ref{eq:def-I0-mix})
imply that the stationary regressor $E\left[Z_{jt}\right]=0$. This
restriction is merely for the conciseness of notation and there is
no loss of generality. In the model (\ref{eq:DGP_XZ}) the intercept
$\alpha^{*}$ can absorb the non-zero means of the stationary regressors.} In addition to Assumptions \ref{assu:tail}--\ref{assu:asym_n}
under the redefined symbols in this section, we impose one more condition. 
\begin{assumption}
\label{assu:EZu} $\mathbb{E}(Z_{t-1}u_{t})=0_{p_{z}}$ for all $t\in\mathbb{Z}$.
\end{assumption}
Assumption \ref{assu:EZu} is a necessary condition for identifying
the coefficient $\gamma^{*}$ in (\ref{eq:DGP_XZ}); otherwise $Z_{t-1}$
becomes endogenous and we must resort to external instrumental variables
for identification and consistent estimation. 
\begin{rem}
A zero-correlation assumption between $e_{t}$ and $u_{t}$ is not
needed, because the large variation of $X_{t-1}$ yields convergence
faster than the $\sqrt{n}$-rate. The bias caused by the endogeneity
of $e_{t}$ affects neither consistency nor the rate of convergence
for the unit root regressors, although it complicates hypothesis testing
asymptotically with nonstandard limiting distributions \citep{phillips2015halbert}.
\end{rem}
We study Slasso first in the mixed root case. We refresh $\hat{\Sigma}=n^{-1}\ddot{W}^{\top}\ddot{W}$
and $\widehat{\kappa}_{D}=\kappa_{D}(\hat{\Sigma},3,s)$ here. We
have established in Proposition \ref{prop:MinMax-Unit} (a) that \textbf{$\sqrt{n(\log p)^{-1}}\stackrel{\mathrm{p}}{\preccurlyeq}\widehat{\sigma}_{j}\stackrel{\mathrm{p}}{\preccurlyeq}\sqrt{n\log p}$
}w.p.a.1.~uniformly for all $j$ associated with the unit root variables.
Under Assumptions \ref{assu:tail}--\ref{assu:asym_n}\textbf{ $\max_{j}\left|\widehat{\sigma}_{j}-\mathrm{s.d.}(Z_{j,t-1})\right|\stackrel{p}{\to}0$}
for all $j$ associated the stationary variables, where $\mathrm{s.d.(\cdot)}$
denotes the population s.d.~of a stationary regressor. The scale-standardization
makes the DB and RE for the nonstationary variables comparable to
those of the stationary variables up to some $\log p$ terms. As a
result, convergence rates similar to Theorem \ref{thm:SlassoError}
follow.
\begin{thm}
\label{thm:SlassoError-Mix} Suppose that Assumptions \ref{assu:tail}-\ref{assu:EZu}
hold. There exists an absolute constant $C_{{\rm DB}}^{w}$ such that
if $\lambda=\dfrac{C_{{\rm DB}}^{w}}{\sqrt{n}}(\log p)^{\frac{3}{2}+\frac{1}{2r}}$,
we have 
\begin{align}
\dfrac{1}{n}\|\ddot{W}(\hat{\theta}^{{\rm S}}-\theta^{*})\|_{2}^{2} & =O_{p}\left(\dfrac{s^{2}}{n}(\log p)^{7+\frac{1}{r}}\right)\label{eq:MixForecast}\\
\left\Vert \hat{\beta}^{{\rm S}}-\beta^{*}\right\Vert _{1}+\sqrt{\dfrac{\log p}{n}}\left\Vert \hat{\gamma}^{{\rm S}}-\gamma^{*}\right\Vert _{1} & =O_{p}\left(\dfrac{s^{2}}{n}(\log p)^{6+\frac{1}{2r}}\right)\label{eq:MixL1}\\
\left\Vert \hat{\beta}^{{\rm S}}-\beta^{*}\right\Vert _{2}+\sqrt{\dfrac{\log p}{n}}\left\Vert \hat{\gamma}^{{\rm S}}-\gamma^{*}\right\Vert _{2} & =O_{p}\left(\dfrac{s^{3/2}}{n}(\log p)^{6+\frac{1}{2r}}\right)\label{eq:MixL2}
\end{align}
\end{thm}
For consistency, the admission rate for $\lambda$ in (\ref{eq:S_adm_rate})
remains valid. Slasso provides provable rates of convergence for the
mixed root case, thanks to the scale-standardization that aligns the
unit root variables with the stationary ones. Up to some logarithmic
term, the estimator $\hat{\beta}^{{\rm S}}$ is super-consistent and
$\hat{\gamma}^{{\rm S}}$ maintains the standard $\sqrt{n}$ rate.
Overall, $\widehat{\theta}^{\mathrm{S}}$ is consistent for $\theta^{*}$
under both the $L_{1}$ and $L_{2}$ norms.
\begin{rem}
\label{rem:var-sel} This paper does not attempt to formally develop
asymptotic results for variable selection, because LASSO in general
does not enjoy variable selection consistency \citep{Zou2006}. Here
we briefly discuss the selected variables. According to the Karush-Kuhn-Tucker
(KKT) condition for LASSO \citep[Lemma 2.1]{buhlmann2011statistics},
the solution to (\ref{eq:demean_gene_lasso}) must satisfy 
\begin{gather}
\frac{2}{n}H^{-1}\ddot{W}_{j}^{\top}\hat{u}=\lambda\times\mathrm{sign}(\hat{\theta}_{j})\ \ \text{ if }\hat{\theta}_{j}\neq0\label{eq:var_sel}\\
\big|\frac{2}{n}H^{-1}\ddot{W}_{j}^{\top}\hat{u}\big|\leq\lambda\ \ \text{ if }\hat{\theta}_{j}=0,\nonumber 
\end{gather}
where $\hat{u}=\ddot{Y}-\ddot{W}\hat{\theta}$ is the estimated residual.
Slasso transforms the stationary and nonstationary components into
comparable scales and the variable selection mechanism does not distinguish
these two parts. However, without the transformation the nonstationary
variables are of much larger scales in terms of $\widehat{\sigma}_{j}$
than the stationary variables. With a single tuning parameter for
both components, Plasso tends to select the nonstationary variables
more frequently according to the above KKT condition. This phenomenon
will be observed in our numerical exercises below.
\end{rem}
\begin{rem}
\label{rem:plasso-mix} In the low dimensional case \citet[Corollary 1]{lee2022lasso}
prove an inconvenient property: Plasso with mixed roots cannot achieve
variable estimation consistency and variable selection effect\footnote{\emph{Variable selection effect} in \citet{lee2022lasso} means that
an estimator of a true zero coefficient will be shrunk to zero with
non-trivial probability. This concept is weaker than \emph{variable
selection consistency}: an estimator asymptotically correctly distinguishes
the active and inactive coefficients w.p.a.1.} for both components simultaneously, due to the super-consistency
of the nonstationary component and the standard $\sqrt{n}$-consistency
of the stationary component. This dilemma in the low dimensional regression
naturally carries over into the high dimensional case. In our context,
if Plasso's tuning parameter is chosen according to the DB of the
nonstationary variables, then $\lambda\to\infty$ as in (\ref{eq:P_adm_rate}),
which is far larger than what is required for the DB of the stationary
variables in view of 
\[
\left\Vert \frac{1}{n}\sum_{t=1}^{n}Z_{t-1}u_{t}\right\Vert _{\infty}=\frac{1}{\sqrt{n}}\left\Vert \frac{1}{\sqrt{n}}\sum_{t=1}^{n}Z_{t-1}u_{t}\right\Vert _{\infty}\stackrel{\mathrm{p}}{\preccurlyeq}\frac{(\log p)^{\frac{3}{2}+\frac{1}{2r}}}{\sqrt{n}}\to0.
\]
On the other hand, if we choose Plasso's tuning parameter according
to the DB of the stationary regressors with $\lambda\to0$, then $\lambda$
will be too small to control the DB for the nonstationary components.
In a word, in the mixed root case we have no provable consistency
result for Plasso, and in our numerical studies Plasso performs poorly.
\end{rem}
Up to now we have addressed the asymptotic theory for a mix of $I(0)$
and $I(1)$ regressors. What happens if some nonstationary regressors
are actually cointegrated? We discuss it in the next section.

\subsection{Cointegration\label{subsec:Cointegration}}

To introduce the cointegrated variables into the predictive regression,
we consider a cointegration system of $p_{c}$ observable variables
with cointegration rank $p_{c1}$. We write the cointegration system
into the triangular representation \citep{Phillips1991cointinference}:
\begin{align}
X_{t}^{\mathrm{co}(1)} & =AX_{t}^{\mathrm{co}(2)}+v_{t}^{(1)}\label{eq:cointegration X1 X2}\\
X_{t}^{\mathrm{co}(2)} & =X_{t-1}^{\mathrm{co}(2)}+e_{t}^{(2)}\nonumber 
\end{align}
where $A^{(1)}$ ($p_{c1}\times p_{c2}$ matrix, where $p_{c2}:=p_{c}-p_{c1}$)
stores the $p_{c1}$ cointegration vectors, and the cointegration
error $v_{t}^{(1)}$ and the innovation $e_{t}^{(2)}$ are strictly
stationary. Following \citet[p.327]{lee2022lasso}, we consider a
model where the cointegration error $v_{t}^{(1)}$ enters the regression
linearly via a coefficient $\phi_{1}^{*}$:
\begin{align}
y_{t} & =\alpha^{*}+X_{t-1}^{\top}\beta^{*}+Z_{t-1}^{\top}\gamma^{*}+v_{t-1}^{(1)\top}\phi_{1}^{*}+u_{t},\label{eq:cointegration_y}
\end{align}
and $u_{t}$ is uncorrelated with $Z_{t-1}$ and $v_{t-1}^{(1)}$
to ensure the identification of their respective parameters $\gamma^{*}$
and $\phi_{1}^{*}$. 

In practice the econometrician has no knowledge about the nature of
the regressions \emph{a priori}. Without the identities of $X_{t}^{\mathrm{co}(1)}$
and $X_{t}^{\mathrm{co}(2)}$ she cannot identify or estimate $A$,
and thus $v_{t}^{(1)}$ is a latent variable, making (\ref{eq:cointegration_y})
an infeasible regression. She can, nevertheless, throw the $p$ ($=p_{x}+p_{z}+p_{c}$)
observable regressors into the feasible regression
\begin{align}
y_{t} & =\alpha^{*}+X_{t-1}^{\top}\beta^{*}+Z_{t-1}^{\top}\gamma^{*}+X_{t-1}^{\mathrm{co}(1)\top}\phi_{1}^{*}+X_{t-1}^{\mathrm{co}(2)\top}\phi_{2}^{*}+u_{t}\nonumber \\
 & =\alpha^{*}+X_{t-1}^{\top}\beta^{*}+Z_{t-1}^{\top}\gamma^{*}+X_{t-1}^{\mathrm{co}\top}\phi^{*}+u_{t}\nonumber \\
 & =\alpha^{*}+W_{t-1}^{\top}\theta^{*}+u_{t},\label{eq:cointegration_feasible}
\end{align}
where $\phi_{2}^{*}=-A^{\top}\phi_{1}^{*}$ by substituting (\ref{eq:cointegration X1 X2})
into (\ref{eq:cointegration_y}), $X_{t-1}^{\mathrm{co}}:=(X_{t-1}^{\mathrm{co}(1)\top},X_{t-1}^{\mathrm{co}(2)\top})^{\top}$
with associated parameter $\phi^{*}:=(\phi_{1}^{*\top},\phi_{2}^{*\top})^{\top}$,
and $W_{t-1}$ collects all regressors with the corresponding parameter
$\theta^{*}$ to fit into our framework (\ref{eq:DGP_W}). 

Penalized estimation methods face a generic challenge with cointegration
systems in the regressors, as demonstrated in the following example. 
\begin{example}
\label{exa:coint_DGP}Consider a toy DGP with two scalar regressors
only: 
\[
y_{t}=\phi_{1}^{*}v_{t}^{(1)}+u_{t}=\phi_{1}^{*(1)}x_{t}^{\mathrm{co}(1)}-\phi_{2}^{*(2)}x_{t}^{\mathrm{co}(2)}+u_{t},
\]
where $\phi_{1}^{*}\neq0$ (and $\phi_{2}^{*}=-A\phi_{1}^{*}\neq0$;
here $A$ is a non-zero scalar). Notice OLS is \emph{variable rotation
invariant}, meaning that given $v_{t-1}^{(1)}$ is a linear combination
of $X_{t-1}^{\mathrm{co}}$ the following two regressions produce
exactly the same residual vectors: (i) Regressing $y_{t}$ on $X_{t-1}^{\mathrm{co}}$;
(ii) Regressing $y_{t}$ on $v_{t-1}^{(1)}$. On the contrary, LASSO
estimators vary with variable rotations, because it penalizes different
$L_{1}$ norms on the coefficients in Regressions (i) and (ii). We
are unaware of any penalized method that is rotation invariant, including
ridge regression, smoothly clipped absolute deviation \citep[SCAD,][]{fan2001variable},
and minimax concave penalty \citep[MCP,][]{zhang2010nearly}. 
\end{example}
We continue the example with LASSO's specific issue.

\setcounter{example}{0}
\begin{example}[continue]
\label{exa:coint_Slasso} Slasso cannot consistently estimate the
coefficients $\left(\phi_{1}^{*},\phi_{2}^{*}\right)^{\top}$ under
the tuning parameter $\lambda=\tilde{C}_{{\rm DB}}^{w}\frac{\left(\log p\right)^{3/2}}{\sqrt{n}}$
(with $r=\infty$) as in Theorem \ref{thm:SlassoError-Mix}.\footnote{This choice of $\lambda$ is made because in a full model the tuning
parameter must also accommodate $X_{t-1}$ and $Z_{t-1}$.} Instead, Slasso will lead to $\widehat{\phi}^{S}\stackrel{p}{\to}(0,0)^{\top}$.
As a minimizer $\widehat{\phi}^{\mathrm{S}}$ must satisfy 
\begin{align}
\frac{1}{n}\left\Vert \ddot{Y}-\ddot{X}^{\mathrm{co}\top}\widehat{\phi}^{\mathrm{S}}\right\Vert _{2}^{2}+\lambda\left\Vert D\widehat{\phi}^{\mathrm{S}}\right\Vert _{1} & \leq\left\{ \frac{1}{n}\left\Vert \ddot{Y}-\ddot{X}^{\mathrm{co}\top}\phi\right\Vert _{2}^{2}+\lambda\left\Vert D\phi\right\Vert _{1}\right\} \bigg|_{\phi=(0,0)^{\top}}\nonumber \\
 & =\frac{1}{n}\left\Vert \ddot{Y}\right\Vert _{2}^{2}\stackrel{p}{\to}\left(\mathrm{s.d.}(y_{t})\right){}^{2}<\infty\label{eq:Slasso_0}
\end{align}
given that $y_{t}$ is stationary if the variances of the innovations
$(u_{t},v_{t}^{(1)})$ are finite. Notice that $x_{t}^{\mathrm{co}(2)}$
is a unit root process individually with its sample s.d.~$\widehat{\sigma}_{2}=O_{p}\left(\sqrt{n}\right)$,
and the sample s.d.~of $x_{t}^{\mathrm{co}(1)}$ is $\widehat{\sigma}_{1}/\sqrt{n}=|A|\cdot\widehat{\sigma}_{2}/\sqrt{n}+o_{p}(1)$
due to cointegration. Suppose $\widehat{\phi}^{S}$ converges in probability
to some non-zero constant, then it violates (\ref{eq:Slasso_0}) because
\begin{align*}
\frac{1}{n}\left\Vert \ddot{Y}-\ddot{X}^{\mathrm{co}\top}\widehat{\phi}^{\mathrm{S}}\right\Vert _{2}^{2}+\lambda\left\Vert D\widehat{\phi}^{\mathrm{S}}\right\Vert _{1} & \geq\lambda\left(\widehat{\sigma}_{1}\wedge\widehat{\sigma}_{2}\right)\left\Vert \widehat{\phi}^{\mathrm{S}}\right\Vert _{1}\asymp\left(\log p\right)^{3/2}\frac{\widehat{\sigma}_{2}}{\sqrt{n}}\left(1+o_{p}(1)\right)\to\infty.
\end{align*}
As a result, asymptotically (\ref{eq:Slasso_0}) holds only if $\widehat{\phi}^{\mathrm{S}}\stackrel{p}{\to}(0,0)^{\top}$. 

The same argument of inconsistency applies to Plasso when $\lambda=C_{{\rm DB}}\log p$
as in Theorem \ref{thm:LassoError} (with $r=\infty$). The minimizer
$\widehat{\phi}^{\mathrm{p}}$ satisfies 
\[
\frac{1}{n}\left\Vert Y-X^{\mathrm{co}\top}\widehat{\phi}^{\mathrm{P}}\right\Vert _{2}^{2}+\lambda\left\Vert \widehat{\phi}^{\mathrm{P}}\right\Vert _{1}\leq\left\{ \frac{1}{n}\left\Vert Y-X^{\mathrm{co}\top}\widehat{\phi}^{\mathrm{P}}\right\Vert _{2}^{2}+\lambda\left\Vert \phi\right\Vert _{1}\right\} \bigg|_{\phi=(0,0)^{\top}}=\frac{1}{n}\left\Vert Y\right\Vert _{2}^{2}=O_{p}(1)
\]
On the other hand, unless $\widehat{\phi}^{\mathrm{p}}\stackrel{p}{\to}(0,0)^{\top}$
the criterion function diverges as 
\[
\frac{1}{n}\left\Vert Y-X^{\mathrm{co}\top}\widehat{\phi}^{\mathrm{P}}\right\Vert _{2}^{2}+\lambda\left\Vert \widehat{\phi}^{\mathrm{P}}\right\Vert _{1}\geq C_{{\rm DB}}\log p\left\Vert \widehat{\phi}^{\mathrm{P}}\right\Vert _{1}\to\infty.
\]
 Plasso is also inconsistent in this toy model.
\end{example}
\bigskip

The above Example \ref{exa:coint_Slasso} implies that in (\ref{eq:cointegration_feasible})
LASSO cannot achieve consistent estimation for the whole parameter
$\theta^{*}$, which includes $\phi^{*}$ as a component. Instead,
we should benchmark it with the a tailored regression of $y_{t}$
on $X_{t-1}$ and $Z_{t-1}$ only, where $\phi^{*}$ is suppressed
to zero. In the DGP (\ref{eq:cointegration_y}) the latent variable
$v_{t-1}^{(1)}$, if correlated with $Z_{t-1}$, will induce the well-known
\emph{omitted variable bias }in the population model
\begin{align}
y_{t} & =\alpha^{*}+X_{t-1}^{\top}\beta^{*}+Z_{t-1}^{\top}(\gamma^{*}+\omega^{*})+\left(v_{t-1}^{(1)\top}\phi_{1}^{*}-Z_{t-1}^{\top}\omega^{*}\right)+u_{t}\nonumber \\
 & =\alpha^{*}+X_{t-1}^{\top}\beta^{*}+Z_{t-1}^{\top}\gamma^{*(1)}+u_{t}^{(1)}\label{eq:cointegration_y_2}
\end{align}
where $\omega^{*}:=\left[\mathbb{E}(Z_{t}Z_{t}^{\top})\right]^{-1}\mathbb{E}(Z_{t}v_{t}^{(1)\top})\phi_{1}^{*}$
is the projection coefficient of $v_{t-1}^{(1)\top}\phi_{1}^{*}$
onto the linear space spanned by $Z_{t-1}$, and the projection leads
to the new regression coefficient $\gamma^{(1)*}=\gamma^{*}+\omega^{*}$
for $Z_{t-1}$ to ensure that the new residual $u_{t}^{(1)}=u_{t}+v_{t-1}^{(1)\top}\phi_{1}^{*}-Z_{t-1}^{\top}\omega^{*}$
is orthogonal to $Z_{t-1}$. From the perspective of prediction, this
re-calibration of the population coefficient from $\gamma^{*}$ to
$\gamma^{(1)*}$ is desirable in that the predictive power of $v_{t}^{(1)}$
can be partially absorbed by the observable $Z_{t-1}$ to reduce the
variance of the error term as $\mathrm{var}\big(u_{t}^{(1)}\big)\leq\mathrm{var}\big(v_{t-1}^{(1)\top}\phi_{1}^{*}+u_{t}\big)$
by construction. 

Now we present the formal asymptotic analysis of Slasso when cointegrated
variables are present. We consider that the stationary components
$v_{t}^{(1)},$ $e_{t}^{(2)}$, $e_{t}$, $Z_{t}$, and $u_{t}$ are
potentially correlated in the form

\begin{equation}
v_{t}=(e_{t}^{(2)\top},e_{t}^{\top},v_{t}^{(1)\top},Z_{t}^{\top},u_{t})^{\top}=\Phi\varepsilon_{t},\label{eq:cointegration-def-I0-mix}
\end{equation}
understanding that $\Phi$ and $\varepsilon_{t}$ are redefined to
adapt to the DGP (\ref{eq:cointegration X1 X2}) and (\ref{eq:cointegration_y}).
Define $\|A\|_{r1}:=\max_{j}\sum_{k}|A_{jk}|$ as the maximum row-wise
$L_{1}$ norm. We further regularize the new coefficients $\gamma^{(1)*}$
and the cointegration matrix $A$. Let $\left\Vert \cdot\right\Vert _{0}$
be the cardinality of non-zero elements in a vector. 
\begin{assumption}
\label{assu:cointegration}Suppose that $\|\beta^{*}\|_{0}+\|\gamma^{*(1)}\|_{0}\leq s$
and $\|\omega^{*}\|_{1}+\|\phi_{1}^{*}\|_{1}\leq C_{0}$ for some
absolute constant $C_{0}$. Furthermore, there exist an absolute constant
$C_{A}$ such that $\|A\|_{r1}+\left[\lambda_{\min}(AA^{\top})\right]^{-1}+\lambda_{\max}(AA^{\top})\leq C_{A}$.
\end{assumption}
The restriction on $\|\beta^{*}\|_{0}+\|\gamma^{*(1)}\|_{0}$ controls
the sparsity of the coefficients in (\ref{eq:cointegration_y_2}).
The finite $L_{1}$-norm of $\|\omega^{*}\|_{1}+\|\phi_{1}^{*}\|_{1}$
governs the deviation bound for $u_{t}^{(1)}$. The restrictions on
the cointegration matrix $A$ regularize the high dimensional cointegration
system. 

Given the discussion in Remark \ref{rem:plasso-mix} that Slasso is
favored over Plasso when $X_{t-1}$ and $Z_{t-1}$ are present, we
apply Slasso to (\ref{eq:cointegration_feasible}) and obtain the
following results. 
\begin{thm}
\label{thm:cointegration}Suppose that Assumptions \ref{assu:tail}-\ref{assu:cointegration}
hold. There exists an absolute constant $\tilde{C}_{{\rm DB}}^{w}$
such that if $\lambda=\dfrac{\tilde{C}_{{\rm DB}}^{w}}{\sqrt{n}}(\log p)^{\frac{5}{2}+\frac{1}{2r}}$,
we have 
\begin{eqnarray}
\dfrac{1}{n}\|\ddot{W}\hat{\theta}^{{\rm S}}-\left(\ddot{X}_{t-1}^{\top}\beta^{*}+\ddot{Z}_{t-1}^{\top}\gamma^{*(1)}\right)\|_{2}^{2} & = & O_{p}\left(\dfrac{s^{2}}{n}(\log p)^{9+\frac{1}{r}}\right)\label{eq:Coint_fit}\\
\left\Vert \hat{\phi}^{{\rm S}}\right\Vert _{1}+\left\Vert \hat{\beta}^{{\rm S}}-\beta^{*}\right\Vert _{1}+\sqrt{\dfrac{\log p}{n}}\left\Vert \hat{\gamma}^{{\rm S}}-\gamma^{*(1)}\right\Vert _{1} & = & O_{p}\left(\dfrac{s^{2}}{n}(\log p)^{10+\frac{1}{2r}}\right).\label{eq:Coint_I1}
\end{eqnarray}
\end{thm}
The in-sample fitting performance (\ref{eq:Coint_fit}) shows that
Slasso for the feasible regression (\ref{eq:cointegration_feasible})
effectively learns the information in (\ref{eq:cointegration_y_2}).
The parameter estimation performance in (\ref{eq:Coint_I1}) illustrates
in terms of the $L_{1}$-norm that $\hat{\beta}^{{\rm S}}$ is consistent
for the pure unit root predictors $X_{t}$, whereas for the $I(0)$
regressors $\hat{\gamma}^{{\rm S}}$ consistently estimates $\gamma^{*(1)}$
to absorb the information in $v_{t}^{(1)}$. As explained in Example
\ref{exa:coint_Slasso}, Slasso shrinks $\hat{\phi}^{{\rm S}}$ all
the way to 0 due to the excessive penalty after scale-normalization.
Such over-penalization violates the conditions in Lemma \ref{lem:gene-Lasso}
and we must devise a new technique to cope with the variable rotation
in the proof (see Appendix \ref{subsec:ProofMain}). In summary, the
Slasso estimator $\hat{\theta}^{{\rm S}}$ converges in probability
to the parameters in (\ref{eq:cointegration_y_2}), and the component
associated with cointegrated variables is suppressed to 0 asymptotically. 
\begin{rem}
Under a fixed $p$, the twin-adaptive LASSO is proposed by \citet{lee2022lasso}
to deal with cointegrated predictors. However, it is difficult to
extend the twin-adaptive LASSO into high dimension. First, a consistent
initial estimator, which is essential for adaptive LASSO, is unavailable
for the original model (\ref{eq:cointegration_feasible}). Even if
a consistent initial estimator is provided, high dimensional adaptive
LASSO requires the \emph{adaptive irrepresentable condition} \citep[Condition A3]{huang2008adaptive},
which does not hold for nonstationary time series. Under $p\ll n$,
\citet{smeekes2021automated} handle cointegration in the framework
of the vector error correction model (VECM). To the best of our knowledge,
in the $p\gg n$ regime there is no method yet that achieves consistent
estimation for the parameter in (\ref{eq:cointegration_feasible}). 
\end{rem}

\section{Simulations\label{sec:Simulations}}

In this section we carry out Monte Carlo simulations with mixed roots
and pure unit roots.\textbf{}\footnote{We perform additional simulations based on DGP (\ref{eq:cointegration X1 X2})
and (\ref{eq:cointegration_y}) with the presence of cointegrated
variables. To save space, we defer the designs and the results to
Section \ref{subsec:simu_coint} in the Appendix.} We first consider the DGP (\ref{eq:DGP_XZ}) and generate the innovation
$v_{t}=(e_{t}^{\top},Z_{t}^{\top},u_{t})^{\top}$ by a (vector) autoregressive
(AR) process 
\begin{gather}
v_{t}=0.4v_{t-1}+\varepsilon_{t},\text{ for }\varepsilon_{t}\sim i.i.d.\,\mathcal{N}(0,\,0.84\Omega),\label{eq:sim_DGP_v}\\
\text{where }\Omega_{ij}=0.8^{|j-j'|}\times\boldsymbol{1}(\text{if }(j,j')\text{ is not associated with \ensuremath{Z_{t}} and \ensuremath{u_{t}}}),\nonumber 
\end{gather}
where ${\bf 1}(\cdot)$ is the indicator function. The AR(1) coefficient
is chosen to set the unconditional variance $0.84\Omega/(1-0.4^{2})=\Omega$,
and the indicator function ensures Assumption \ref{assu:EZu} with
uncorrelated $Z_{j,t-1}$ and $u_{t}$. We consider $n\in\{120,240,360\}$
and $p=2n$. We try $p_{x}=\{0.5n,0.8n,1.2n,1.5n\}$ for each $n$,
and thus $p_{z}=2n-p_{x}$ for each pair of $(n,p_{x})$$.$ The sparsity
indices are $s_{x}=s_{z}=2\lceil\log n\rceil$, so that $s=4\lceil\log n\rceil$.
We set the true coefficients of the stationary component $\gamma^{*}=(0.3\times[s_{z}]^{\top},0_{p_{z}-s_{z}}^{\top})^{\top}$.
We specify two cases for the unit root regressors, which vary only
in the coefficients $\beta_{(1)}^{*}=(n^{-1/2}1_{s_{x}}^{\top},0_{p_{x}-s_{x}}^{\top})^{\top}$
where the factor $n^{-1/2}$ bound the dependent variable $y_{t}$
to be of non-explosive, and $\beta_{(2)}^{*}=(1,n^{-1/2}1_{s_{x}-1}^{\top},0_{p_{x}-s_{x}}^{\top})^{\top}$where
the coefficient of the first regressor is invariant with the sample
size $n$. We label the data generated by the following coefficients
as
\begin{description}
\item [{DGP1}] $\theta_{(1)}^{*}=(\beta_{(1)}^{*\top},\gamma^{*\top})^{\top}$
\item [{DGP2}] $\theta_{(2)}^{*}=(\beta_{(2)}^{*\top},\gamma^{*\top})^{\top}$. 
\end{description}
For each DGP, we report the one-period-ahead out-of-sample \emph{root
mean squared prediction error} (RMSPE), defined as $\left\{ \mathbb{E}\left[(y_{n+1}-\hat{y}_{n+1})^{2}\right]\right\} ^{1/2}$,
and the parameter estimation root mean squared error, defined as $\left[\mathbb{E}\left(\Vert\hat{\theta}-\theta^{*}\Vert{}_{2}^{2}\right)\right]^{1/2}$.
The estimation method is either Plasso or Slasso. The expectations
are approximated by the empirical average over 5000 replications.
As a benchmark, we compare LASSO with an oracle estimator---the OLS
with known active variables. 

\begin{table}[]
\begin{center}
\caption{RMSPE for Mixed Regressors}
\label{tab:mix_rmse_ARMA}
\small
\begin{tabular}{ccc|r|rr|rr|r|rr|rr}
\hline\hline 
\multirow{3}{*}{$n$}   & \multicolumn{1}{c}{\multirow{3}{*}{$p_x$}} & \multicolumn{1}{c|}{\multirow{3}{*}{$p_z$}} & \multicolumn{5}{c|}{RMSPE}                                                                                                                                   & \multicolumn{5}{c}{RMSE for estimated coefficients}                                                                                                                         \\ \cline{4-13}
                     & \multicolumn{1}{c}{}                    & \multicolumn{1}{c|}{}                    & \multicolumn{1}{c|}{\multirow{2}{*}{Oracle}} & \multicolumn{2}{c|}{CV $\lambda$} & \multicolumn{2}{c|}{Calibrated $\lambda$} & \multicolumn{1}{c|}{\multirow{2}{*}{Oracle}} & \multicolumn{2}{c|}{CV $\lambda$} & \multicolumn{2}{c}{Calibrated $\lambda$} \\ \cline{5-8} \cline{10-13}
                     & \multicolumn{1}{c}{}                    & \multicolumn{1}{c|}{}                    & \multicolumn{1}{c|}{}                        & Plasso                & Slasso                       & Plasso                     & Slasso                    & \multicolumn{1}{c|}{}                        & Plasso                & Slasso                       & Plasso            & Slasso                             \\
                     \hline 
\multicolumn{13}{c}{DGP1}                                                                                                                                                     \\
\hline 
\multirow{4}{*}{120} & 60  & 180 & 1.140 & 1.699 & \textit{1.268} & 1.547 & \textit{\textbf{1.255}} & 0.848 & 1.234 & \textit{0.912} & 1.107 & \textit{\textbf{0.897}} \\
                     & 96  & 144 & 1.133 & 1.762 & \textit{1.248} & 1.522 & \textit{\textbf{1.229}} & 0.848 & 1.306 & \textit{0.907} & 1.124 & \textit{\textbf{0.892}} \\
                     & 144 & 96  & 1.141 & 1.807 & \textit{1.258} & 1.565 & \textit{\textbf{1.239}} & 0.847 & 1.315 & \textit{0.894} & 1.131 & \textit{\textbf{0.879}} \\
                     & 180 & 60  & 1.158 & 1.865 & \textit{1.252} & 1.552 & \textit{\textbf{1.239}} & 0.843 & 1.346 & \textit{0.878} & 1.132 & \textit{\textbf{0.863}} \\
                     \hline 
\multirow{4}{*}{240} & 120 & 360 & 1.063 & 2.095 & \textit{1.223} & 1.517 & \textit{\textbf{1.164}} & 0.609 & 1.420 & \textit{0.708} & 0.964 & \textit{\textbf{0.684}} \\
                     & 192 & 288 & 1.072 & 2.228 & \textit{1.232} & 1.527 & \textit{\textbf{1.173}} & 0.612 & 1.455 & \textit{0.707} & 0.973 & \textit{\textbf{0.679}} \\
                     & 288 & 192 & 1.073 & 2.266 & \textit{1.216} & 1.524 & \textit{\textbf{1.152}} & 0.610 & 1.523 & \textit{0.704} & 0.978 & \textit{\textbf{0.671}} \\
                     & 360 & 120 & 1.067 & 2.376 & \textit{1.231} & 1.572 & \textit{\textbf{1.166}} & 0.609 & 1.541 & \textit{0.701} & 0.985 & \textit{\textbf{0.664}} \\
                     \hline 
\multirow{4}{*}{360} & 180 & 540 & 1.057 & 2.400 & \textit{1.206} & 1.538 & \textit{\textbf{1.143}} & 0.482 & 1.443 & \textit{0.572} & 0.863 & \textit{\textbf{0.552}} \\
                     & 288 & 432 & 1.049 & 2.408 & \textit{1.207} & 1.518 & \textit{\textbf{1.142}} & 0.480 & 1.495 & \textit{0.572} & 0.868 & \textit{\textbf{0.548}} \\
                     & 432 & 288 & 1.055 & 2.558 & \textit{1.201} & 1.545 & \textit{\textbf{1.130}} & 0.477 & 1.543 & \textit{0.568} & 0.874 & \textit{\textbf{0.539}} \\
                     & 540 & 180 & 1.041 & 2.601 & \textit{1.194} & 1.550 & \textit{\textbf{1.125}} & 0.482 & 1.551 & \textit{0.570} & 0.878 & \textit{\textbf{0.537}} \\
                     \hline 
\multicolumn{13}{c}{DGP2}                                                                                                                                                                                                                                                                                                                                                                                          \\
\hline 
\multirow{4}{*}{120} & 60  & 180 & 1.139 & 2.571 & \textit{1.313} & 2.038 & \textit{\textbf{1.297}} & 0.846 & 1.852 & \textit{0.966} & 1.514 & \textit{\textbf{0.943}} \\
                     & 96  & 144 & 1.147 & 2.646 & \textit{1.318} & 2.058 & \textit{\textbf{1.300}} & 0.842 & 1.900 & \textit{0.961} & 1.529 & \textit{\textbf{0.938}} \\
                     & 144 & 96  & 1.115 & 2.747 & \textit{1.289} & 2.089 & \textit{\textbf{1.271}} & 0.845 & 1.945 & \textit{0.953} & 1.543 & \textit{\textbf{0.931}} \\
                     & 180 & 60  & 1.120 & 2.809 & \textit{1.267} & 2.080 & \textit{\textbf{1.249}} & 0.843 & 2.010 & \textit{0.943} & 1.558 & \textit{\textbf{0.920}} \\
                     \hline 
\multirow{4}{*}{240} & 120 & 360 & 1.096 & 3.893 & \textit{1.303} & 2.074 & \textit{\textbf{1.221}} & 0.610 & 2.538 & \textit{0.762} & 1.385 & \textit{\textbf{0.714}} \\
                     & 192 & 288 & 1.076 & 4.094 & \textit{1.293} & 2.094 & \textit{\textbf{1.210}} & 0.611 & 2.608 & \textit{0.762} & 1.389 & \textit{\textbf{0.710}} \\
                     & 288 & 192 & 1.099 & 4.308 & \textit{1.317} & 2.127 & \textit{\textbf{1.229}} & 0.611 & 2.685 & \textit{0.764} & 1.390 & \textit{\textbf{0.706}} \\
                     & 360 & 120 & 1.066 & 4.446 & \textit{1.277} & 2.162 & \textit{\textbf{1.188}} & 0.605 & 2.810 & \textit{0.758} & 1.410 & \textit{\textbf{0.695}} \\
                     \hline 
\multirow{4}{*}{360} & 180 & 540 & 1.051 & 5.123 & \textit{1.270} & 2.090 & \textit{\textbf{1.162}} & 0.479 & 2.936 & \textit{0.617} & 1.261 & \textit{\textbf{0.568}} \\
                     & 288 & 432 & 1.043 & 5.378 & \textit{1.279} & 2.132 & \textit{\textbf{1.165}} & 0.479 & 3.123 & \textit{0.620} & 1.274 & \textit{\textbf{0.566}} \\
                     & 432 & 288 & 1.061 & 5.631 & \textit{1.277} & 2.129 & \textit{\textbf{1.174}} & 0.480 & 3.211 & \textit{0.619} & 1.280 & \textit{\textbf{0.561}} \\
                     & 540 & 180 & 1.063 & 5.735 & \textit{1.274} & 2.144 & \textit{\textbf{1.170}} & 0.478 & 3.297 & \textit{0.619} & 1.291 & \textit{\textbf{0.555}} \\
                     \hline \hline  
\end{tabular}
\end{center}
\footnotesize{
Note: 
Italic numbers indicate the better performance between Plasso and Slasso with the same tuning method. 
Bold numbers indicate the best LASSO performance. }
\end{table}

A key ingredient in implementing LASSO is the choice of the tuning
parameter $\lambda$. One common data-driven approach is \emph{cross
validation} (CV). In our time series context, we cut $t\in[n]$ into
10 chronically ordered blocks and choose the $\lambda$ that minimizes
the CV means squared error as each block serves as a validation dataset
in turn whereas the other 9 blocks work as the training data. We refer
to the $\lambda$ chosen by this time series 10-fold CV as ``CV $\lambda$''.
CV $\lambda$ is completely data-driven. 

Alternatively, to evaluate our theoretical statement where $\lambda$
is specified as a constant multiplied by an expansion rate determined
by $n$, $p$ and $r$, we follow \citet{lee2022lasso} to use a small-scale
experiment to calibrate an initial choice. We try 100 replications
with $(n_{0},p_{x0})$, the smallest $n$ and $p_{x}$ considered
in the simulations, save $\lambda$ in each replication according
to the 10-fold CV described in the previous paragraph, and let $\widehat{\lambda}_{0}$
be the median of these $\lambda$'s. We then scale up $\widehat{\lambda}_{0}$
based on the theoretical expansion rate. We refer to this scheme as
``calibrated $\lambda$''. For Slasso with an initial $\widehat{\lambda}_{0}^{(\mathrm{s})}$
obtained from $(n_{0},p_{x0},p_{z0})=(120,60,180)$, we use 
\begin{equation}
\widehat{\lambda}^{(\mathrm{s})}=\widehat{\lambda}_{0}^{(\mathrm{s})}\left(\dfrac{n^{-1/2}\log p}{n_{0}^{-1/2}\log p_{0}}\right)^{2}\label{eq:calibrate_s}
\end{equation}
 to adhere to the rate $n^{-1/2}(\log p)^{\frac{3}{2}+\frac{1}{2r}}$
in Theorem \ref{thm:SlassoError-Mix} when $p$ is proportional to
$n$, where $r=1$ for the AR(1) innovation. Section \ref{sec:Mixed-Regressors}
has elaborated that Slasso enjoys theoretical guarantees whereas the
convergence of Plasso with mixed regressors is unknown. As a numerical
exercise we naively borrow the rate $(\log p)^{1+\frac{1}{2r}}$ in
Theorems \ref{thm:LassoError} after obtaining the initial calibrated
$\widehat{\lambda}_{0}^{(\mathrm{p})}$ and then calculate 
\begin{equation}
\widehat{\lambda}^{(\mathrm{p})}=\widehat{\lambda}_{0}^{(\mathrm{p})}\left(\log p/\log p_{0}\right)^{3/2}.\label{eq:calibrate_p}
\end{equation}

Consistent with our theory, in Table \ref{tab:mix_rmse_ARMA} we find
that Slasso with the calibrated $\lambda$ performs well. In both
DGP1 and DGP2 we observe that the prediction error and parameter estimation
error decrease as $n$ increases. Similar error reduction is observed
under CV $\lambda$. On the contrary, the simulation evidence suggests
possible inconsistency of Plasso under either the calibrated $\lambda$
or the CV $\lambda$. Similar patterns are found in terms of mean
absolute prediction error (MAPE) and mean absolute parameter estimation
error in Table \ref{tab:mix_mae_ARMA} in the Appendix. 

To better understand the unsatisfactory performance of Plasso, Table
\ref{tab:cate_sel} in the Appendix shows the percentage of variables
selected from the active and inactive coefficients. For example, under
CV $\lambda$ Plasso selects fewer active $\beta^{*}$ than Slasso,
and the gap is particularly big when $n=360$. In the meantime, it
makes more mistakes in selecting the inactive $\beta^{*}$. Plasso
faces an inherent dilemma concerning the suitable tuning parameter
levels of the stationary and nonstationary components, as discussed
in Remark \ref{rem:plasso-mix}. Nonstationary variables have larger
variations and are more influential in prediction. In order to achieve
the variable selection effect amongst the nonstationary variables,
Plasso requires a large $\lambda$ as in (\ref{eq:P_adm_rate}). Such
emphasis in the nonstationary component imposes a heavy cost in the
stationary component, where a non-trivial proportion of the active
$\gamma^{*}$ is eliminated, although it also rules out almost all
the inactive $\gamma^{*}$. This observation echoes the discussion
in Remark \ref{rem:var-sel} about variable selection in relation
to the scales. 

The scale normalization in Slasso balances the two types of time series,
which allows it to choose active variables in both components. Slasso
improves upon Plasso in $\beta^{*}$ for both the active and inactive
ones. Moreover, in Table \ref{tab:cate_sel} Slasso produces nearly
perfect variable selection in the active $\gamma^{*}$, and in the
meantime it controls the estimation error in the inactive one, as
shown in Table \ref{tab:cate_rmse} in the Appendix about the parameter
estimation RMSE of each subset of the coefficients. In this table,
the most prominent estimation error comes from the active $\gamma^{*}$
by Plasso, where its large $\lambda$ that accommodates the nonstationary
component results in substantial shrinkage bias in the estimation.

\bigskip

\begin{table}[]
\begin{center}
\caption{RMSPE for Pure Unit Root Regressors}
\label{tab:unit_rmse_ARMA}
\small
\begin{tabular}{cc|r|rr|rr|r|rr|rr}
  \hline   \hline 
\multirow{3}{*}{$n$}   & \multicolumn{1}{c|}{\multirow{3}{*}{$p_x$}} & \multicolumn{5}{c|}{RMSPE}                                                                                                                                   & \multicolumn{5}{c}{RMSE for estimated coefficients}                                                                                                                          \\ \cline{3-12}
                     & \multicolumn{1}{c|}{}                   & \multicolumn{1}{c|}{\multirow{2}{*}{Oracle}} & \multicolumn{2}{c|}{CV $\lambda$} & \multicolumn{2}{c|}{Calibrated $\lambda$} & \multicolumn{1}{c|}{\multirow{2}{*}{Oracle}} & \multicolumn{2}{c|}{CV $\lambda$} & \multicolumn{2}{c}{Calibrated $\lambda$} \\  \cline{4-7} \cline{9-12}
                     & \multicolumn{1}{c|}{}                   & \multicolumn{1}{c|}{}                        & Plasso                    & Slasso                   & Plasso                     & Slasso                    & \multicolumn{1}{c|}{}                        & Plasso                    & Slasso                   & Plasso                             & Slasso            \\
                     \hline 
\multicolumn{12}{c}{DGP3}                                                                                                                                                        \\
 \hline 
\multirow{4}{*}{120} & 60  & 1.098 & \textit{1.104} & 1.122          & \textit{\textbf{1.081}} & 1.096 & 0.383 & \textit{0.328} & 0.348          & \textit{\textbf{0.282}} & 0.305 \\
                     & 96  & 1.080 & \textit{1.094} & 1.115          & \textit{\textbf{1.068}} & 1.080 & 0.384 & \textit{0.324} & 0.350          & \textit{\textbf{0.281}} & 0.311 \\
                     & 144 & 1.069 & 1.131          & \textit{1.109} & \textit{\textbf{1.062}} & 1.074 & 0.385 & \textit{0.285} & 0.322          & \textit{\textbf{0.281}} & 0.315 \\
                     & 180 & 1.063 & 1.126          & \textit{1.109} & \textit{\textbf{1.074}} & 1.082 & 0.385 & \textit{0.288} & 0.326          & \textit{\textbf{0.282}} & 0.317 \\
                     \hline 
\multirow{4}{*}{240} & 120 & 1.041 & \textit{1.055} & 1.065          & \textit{\textbf{1.039}} & 1.052 & 0.227 & \textit{0.210} & 0.233          & \textit{\textbf{0.195}} & 0.217 \\
                     & 192 & 1.060 & \textit{1.069} & 1.091          & \textit{\textbf{1.056}} & 1.075 & 0.226 & \textit{0.212} & 0.236          & \textit{\textbf{0.195}} & 0.221 \\
                     & 288 & 1.044 & 1.129          & \textit{1.090} & \textit{\textbf{1.051}} & 1.070 & 0.227 & \textit{0.206} & 0.231          & \textit{\textbf{0.195}} & 0.226 \\
                     & 360 & 1.049 & 1.164          & \textit{1.103} & \textit{\textbf{1.074}} & 1.080 & 0.225 & \textit{0.207} & 0.234          & \textit{\textbf{0.196}} & 0.229 \\
                     \hline 
\multirow{4}{*}{360} & 180 & 1.023 & \textit{1.039} & 1.051          & \textit{\textbf{1.025}} & 1.041 & 0.149 & \textit{0.155} & 0.176          & \textit{\textbf{0.146}} & 0.166 \\
                     & 288 & 1.033 & \textit{1.050} & 1.073          & \textit{\textbf{1.041}} & 1.057 & 0.150 & \textit{0.157} & 0.180          & \textit{\textbf{0.147}} & 0.171 \\
                     & 432 & 1.037 & 1.142          & \textit{1.083} & \textit{\textbf{1.047}} & 1.062 & 0.150 & \textit{0.160} & 0.178          & \textit{\textbf{0.148}} & 0.174 \\
                     & 540 & 1.019 & 1.126          & \textit{1.072} & \textit{\textbf{1.035}} & 1.055 & 0.150 & \textit{0.161} & 0.181          & \textit{\textbf{0.149}} & 0.177 \\
                      \hline 
\multicolumn{12}{c}{DGP4}                                                                                                                                                        \\
 \hline 
\multirow{4}{*}{120} & 60  & 1.106 & \textit{1.113} & 1.127          & \textit{\textbf{1.087}} & 1.112 & 0.388 & \textit{0.348} & 0.379          & \textit{\textbf{0.299}} & 0.326 \\
                     & 96  & 1.087 & \textit{1.102} & 1.129          & \textit{\textbf{1.082}} & 1.108 & 0.386 & \textit{0.353} & 0.385          & \textit{\textbf{0.304}} & 0.341 \\
                     & 144 & 1.078 & 1.275          & \textit{1.141} & \textit{\textbf{1.079}} & 1.113 & 0.384 & \textit{0.360} & 0.368          & \textit{\textbf{0.309}} & 0.358 \\
                     & 180 & 1.100 & 1.304          & \textit{1.171} & \textit{\textbf{1.114}} & 1.145 & 0.387 & \textit{0.368} & 0.372          & \textit{\textbf{0.312}} & 0.363 \\
                     \hline 
\multirow{4}{*}{240} & 120 & 1.060 & \textit{1.080} & 1.093          & \textit{\textbf{1.067}} & 1.089 & 0.227 & \textit{0.222} & 0.251          & \textit{\textbf{0.201}} & 0.230 \\
                     & 192 & 1.043 & \textit{1.075} & 1.091          & \textit{\textbf{1.056}} & 1.086 & 0.225 & \textit{0.223} & 0.260          & \textit{\textbf{0.203}} & 0.242 \\
                     & 288 & 1.036 & 1.406          & \textit{1.131} & \textit{\textbf{1.056}} & 1.090 & 0.224 & 0.283          & \textit{0.272} & \textit{\textbf{0.206}} & 0.255 \\
                     & 360 & 1.035 & 1.430          & \textit{1.146} & \textit{\textbf{1.058}} & 1.104 & 0.229 & 0.290          & \textit{0.279} & \textit{\textbf{0.208}} & 0.262 \\
                     \hline 
\multirow{4}{*}{360} & 180 & 1.056 & \textit{1.075} & 1.084          & \textit{\textbf{1.060}} & 1.076 & 0.150 & \textit{0.162} & 0.191          & \textit{\textbf{0.149}} & 0.176 \\
                     & 288 & 0.997 & \textit{1.023} & 1.039          & \textit{\textbf{1.009}} & 1.040 & 0.149 & \textit{0.165} & 0.201          & \textit{\textbf{0.150}} & 0.187 \\
                     & 432 & 1.041 & 1.530          & \textit{1.178} & \textit{\textbf{1.063}} & 1.104 & 0.149 & 0.242          & \textit{0.221} & \textit{\textbf{0.152}} & 0.197 \\
                     & 540 & 1.027 & 1.546          & \textit{1.152} & \textit{\textbf{1.049}} & 1.087 & 0.148 & 0.246          & \textit{0.222} & \textit{\textbf{0.153}} & 0.200     \\
                       \hline   \hline 
\end{tabular}
\end{center}
\footnotesize{
Note: 
Italic numbers indicate the better performance between Plasso and Slasso with the same tuning method. 
Bold numbers indicate the best LASSO performance in each row. }
\end{table}

 For completeness, we check LASSO's performance under the prototype
pure unit root case. We consider the same set of $n$ and $p_{x}$
following (\ref{eq:DGP_X}):
\begin{description}
\item [{DGP3}] $\theta_{(3)}^{*}=\beta_{(1)}^{*\top}$
\item [{DGP4}] $\theta_{(4)}^{*}=\beta_{(2)}^{*\top}$
\end{description}
where we simply remove all the stationary regressors from DGP1 and
2, respectively. The innovation of $v_{t}=(e_{t}^{\top},u_{t}^{\top})^{\top}$
is again generated according to (\ref{eq:sim_DGP_v}), where $\Omega_{ij}$
is update to $0.8^{|j-j'|}$ for all $i,j\in[p_{x}]$ to allow correlation
between $e_{t}$ and $u_{t}$. Table \ref{tab:unit_rmse_ARMA} reports
the RMSPE and Table \ref{tab:unit_mae_ARMA} in the Appendix displays
MAPE. Plasso is slightly stronger than Slasso, reflecting the tighter
rates of convergence in Theorem \ref{thm:LassoError} than those in
Theorem \ref{thm:SlassoError} as Slasso involves extra randomness
in $\widehat{\sigma}_{j}$. 

\section{Empirical Application \label{sec:Empirical-demo} }

Faced with multiple regressors, some applied econometricians may be
inclined to avoid nonstationary regressors in view of the resulting
nonstandard asymptotic inference; they may prefer transforming them
into stationary ones. Whether we use the stationarized variable or
the nonstationary original variable count on the true DGP. The advantage
of nonstationary data arises from the super-consistency as the large
variation of the stochastic trend can accelerate the rate of convergence,
making the parameter estimation more accurate and thereby improving
prediction. There is little compelling justification for excluding
nonstationary variables \emph{a priori }in predictive regressions.

\begin{figure}
\begin{centering}
\subfloat[\label{fig:UNRATE} Unemployment Rate (\texttt{UNRATE})]{\centering{}\includegraphics[scale=0.3]{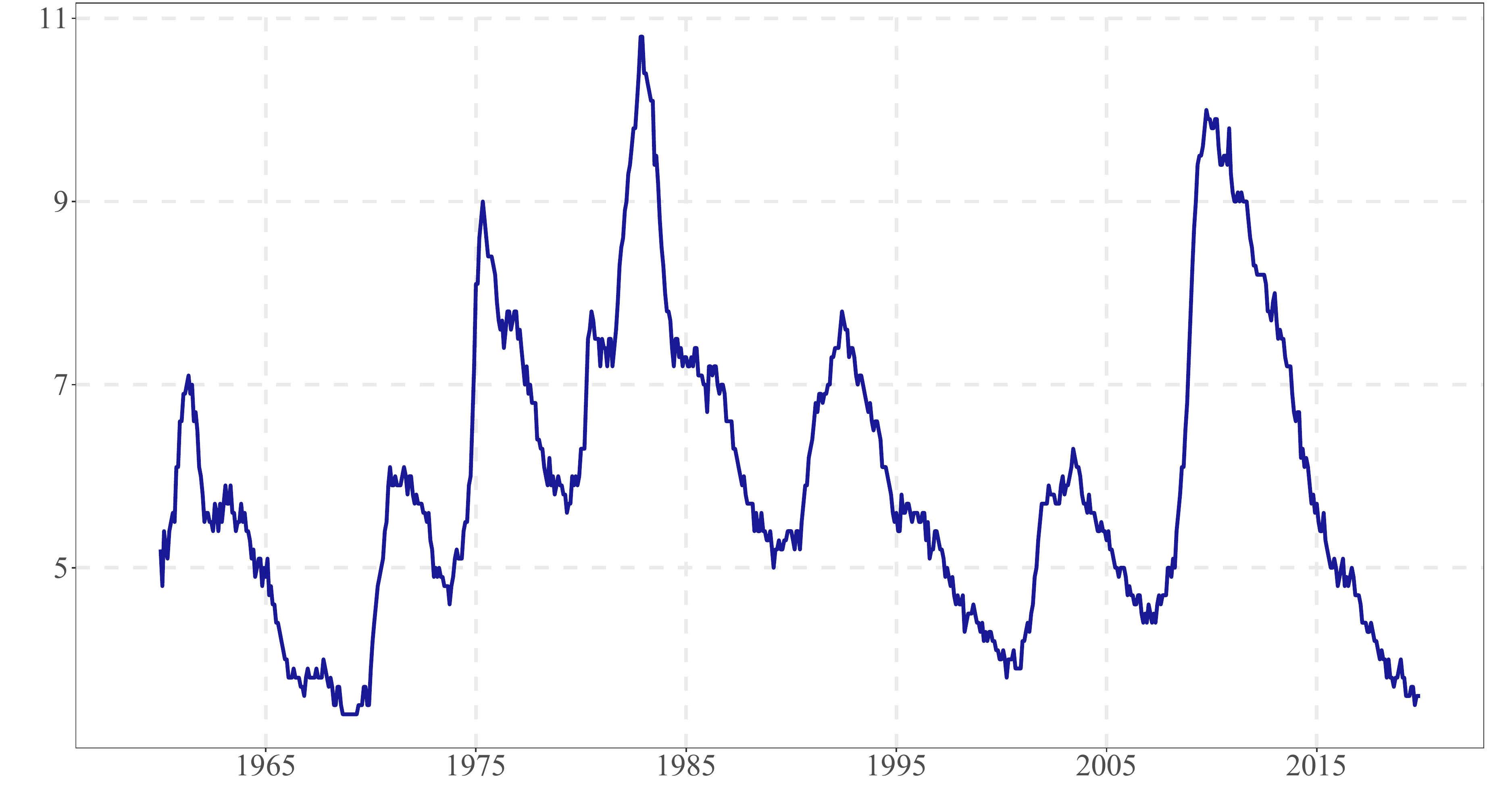}}
\par\end{centering}
\begin{centering}
\subfloat[\label{fig:Representative-Time-Series} Representative Time Series]{\begin{centering}
\includegraphics[scale=0.3]{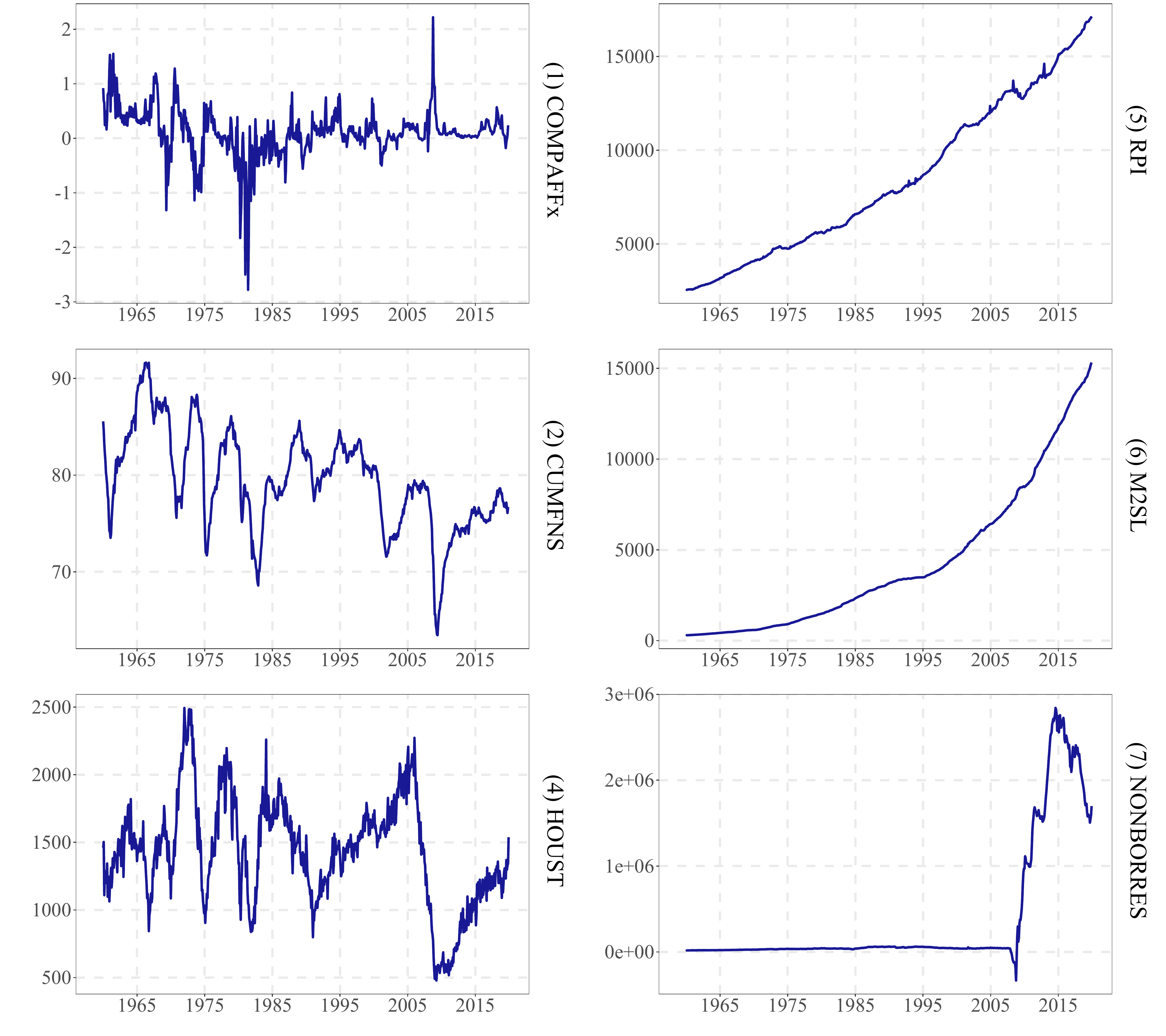}
\par\end{centering}
\raggedright{}}
\par\end{centering}
\begin{raggedright}
\medskip \footnotesize Notes: The number in the parenthesis is the
\texttt{TCODE}. (1):\texttt{COMPAPFFx} 3-Month Commercial Paper Minus
FEDFUNDS (Effective Federal Funds Rate); (2): \texttt{CUMFNS} Capacity
Utilization: Manufacturing; (4): \texttt{HOUST} Housing Starts: Total
New Privately Owned; (5): \texttt{RPI} Real Personal Income; (6):
\texttt{M2SL} M2 Money Stock; (7): \texttt{NONBORRES} Reserves Of
Depository Institutions. No variable in our predictors is of \texttt{TCODE}
(3).
\par\end{raggedright}
\centering{}\caption{\texttt{UNRATE} and Representative Time Series by \texttt{TCODE}}
\end{figure}

We use the FRED-MD macroeconomic database \citep{McCracken2016} to
check the predictability of the unemployment rate of the United States.
Given that the data cover 6 decades from 1960:Jan to 2019:Dec, we
adopt a rolling window of length 10 years, 20 years or 30 years, and
to make the results comparable we set the entire testing sample as
1990:Jan to 2019:Dec. 

The dependent variable, labeled as \texttt{UNRATE} in the database,
is plotted in Figure \ref{fig:UNRATE}. It ranges from 3.4\% to 10.8\%,
and peaks in the early 1980s recession and the 2008 Global Financial
Crisis. It is a persistent time series. If we run a simple AR(1) regression
in the entire sample, the AR coefficient 0.995 is close to unity. 

We include as potential predictors all the other 121 variables in
the database which have no missing values during the sample period.\textbf{
}Each variable in FRED-MD is accompanied with a \emph{transformation
code} (\texttt{TCODE}), which suggests a way to transform the raw
sequence into a stationary time series. There are 7 categories in
total. For a generic scalar time series $(w_{t})_{t=1}^{n}$, the
labels 1--7 correspond to the following transformations: (1) null
(10 variables out of our 121 predictors); (2) $\Delta w_{t}$ (17
variables); (3) $\Delta^{2}w_{t}$ (none); (4) $\log(w_{t})$ (10
variables); (5) $\Delta\log(w_{t})$ (50 variables); (6) $\Delta^{2}\log(w_{t})$
(33 variables); (7) $\Delta(w_{t}/w_{t-1}-1)$ (1 variable). \texttt{UNRATE}
is classified into (2). Figure \ref{fig:Representative-Time-Series}
draws a representative time series in each category. Obviously the
dynamic patterns vary substantively. For example, the one labeled
(5) exhibits a clear upward trend, the one labeled (6) shows exponential
acceleration, and the one labeled (7) has a dramatic structural break
after 2008. One option to avoid nonstationary time series is to stationarize
all the raw sequences according to the \texttt{TCODE}. We call this
practice stationarization transformation (ST). 

Figure \ref{fig:sd} compares the scale of the variables with no transformation
(NT) and those with ST. Each dot on the left panel represents the
sample s.d.~$\widehat{\sigma}_{j}$ (in logarithm base 10 along the
y-axis) of each variable, ordered from lower to high for every \texttt{TCODE}
marked along the x-axis. The right panel shows the histogram of all
variables (the axis again in logarithm base 10). We observe enormous
diversity in the upper sub-figures where the data are at their original
scales. The smallest sample s.d.~is about $10^{-2}$ and the biggest
goes over $10^{6}$. Large-scale variables are particularly common
in \texttt{TCODE} (4) and (5). Under ST, the variables are much more
concentrated. In particular, ST pulls down considerably the scale
of all variables that need to be stationarized. 

\begin{figure}
\begin{centering}
\includegraphics[scale=0.73]{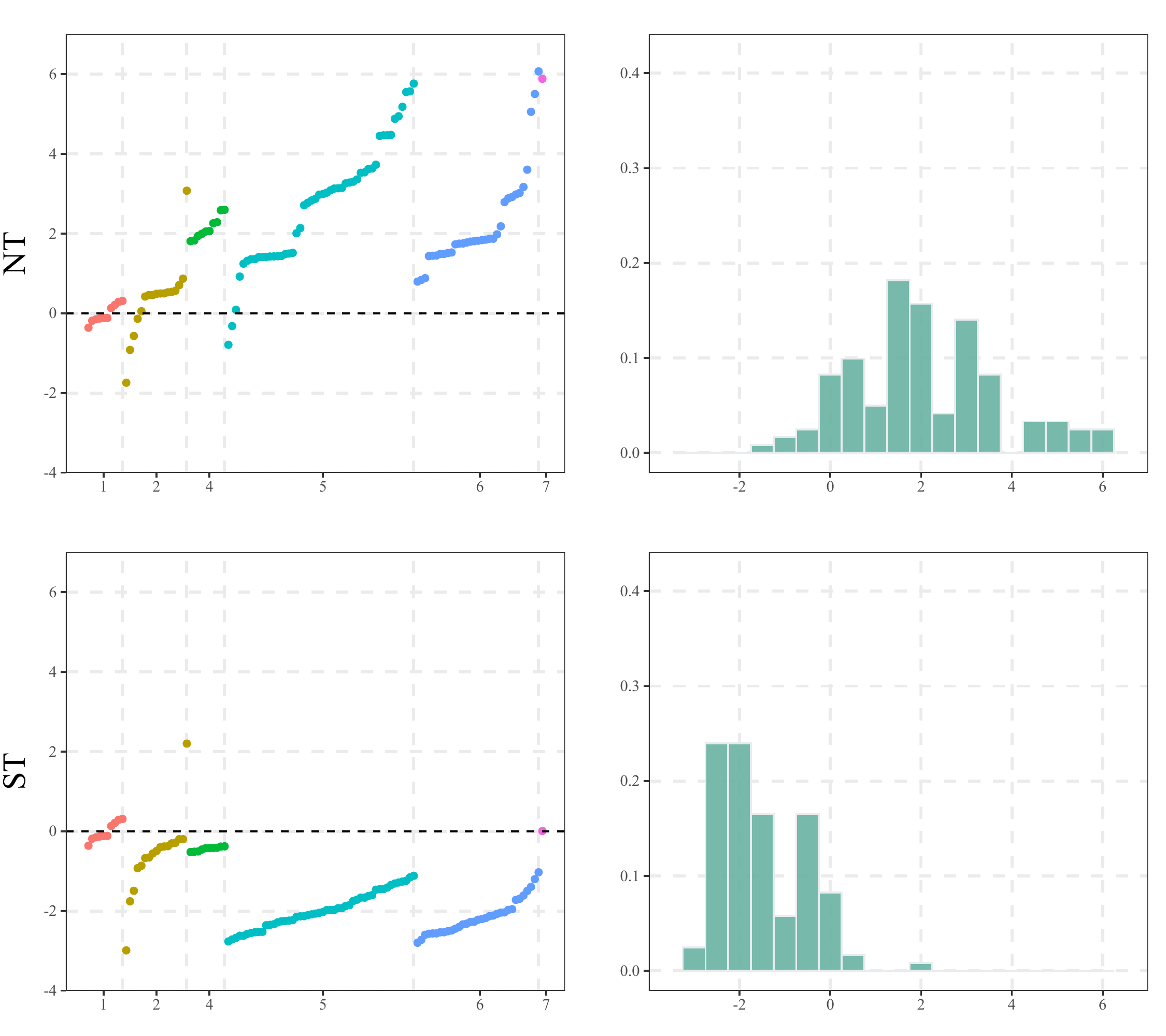}
\par\end{centering}
\footnotesize Note: In the left column, the y-axis is \emph{logarithm
base 10, }and the x-axis is the \texttt{TCODE}. For example, the single
point with \texttt{TCODE} (7) has a sample s.d.~as large as $10^{6}$.
The right column is the histogram of all dots on the left (x-axis
in logarithm base 10). 

\caption{\label{fig:sd} Standard Deviations and Histograms}
\end{figure}

\begin{table}[]
\caption{RMSPE for \texttt{UNRATE}}
\label{tab:unrate_RMSPE_504}
\small
\begin{center}
\begin{tabular}{cc|rr|rr|rr|rr|rr}
\hline\hline 
\multirow{3}{*}{$h$} & \multicolumn{1}{c|}{\multirow{3}{*}{$n$}} & \multicolumn{2}{c|}{\multirow{2}{*}{Benchmarks}} & \multicolumn{4}{c|}{121  Predictors}                                                        & \multicolumn{4}{c}{504 Predictors}                                                         \\ \cline{5-12}
                   & \multicolumn{1}{c|}{}                   & \multicolumn{2}{c|}{}                            & \multicolumn{2}{c}{NT}           & \multicolumn{2}{c|}{ST}                                  & \multicolumn{2}{c}{NT}           & \multicolumn{2}{c}{ST}                                  \\ \cline{3-12}
                   & \multicolumn{1}{c|}{}                   & \multicolumn{1}{c}{RWwD}                   & \multicolumn{1}{c|}{AR}                    & \multicolumn{1}{c}{Plasso} & \multicolumn{1}{c|}{Slasso}       & \multicolumn{1}{c}{Plasso} & \multicolumn{1}{c|}{Slasso}  & \multicolumn{1}{c}{Plasso} & \multicolumn{1}{c|}{Slasso}                  & \multicolumn{1}{c}{Plasso} & \multicolumn{1}{c}{Slasso} \\
                   \hline 
\multicolumn{12}{c}{\textbf{Entire testing sample: 1990--2019}}                                                                                                                                                                                                                                                                          \\

\multirow{3}{*}{1} & 120 & 0.154 & 0.150 & 0.639 & \textit{0.144} & 0.889 & 0.511 & 0.578 & \textit{\textbf{0.139}} & 0.467 & 0.148                   \\
                   & 240 & 0.154 & 0.149 & 0.614 & \textit{0.145} & 0.632 & 0.647 & 0.766 & \textit{\textbf{0.128}} & 0.238 & 0.133                   \\
                   & 360 & 0.154 & 0.144 & 0.518 & \textit{0.150} & 1.864 & 1.920 & 0.736 & \textit{\textbf{0.129}} & 0.192 & 0.134                   \\
                   \hline 
\multirow{3}{*}{2} & 120 & 0.230 & 0.214 & 0.689 & \textit{0.195} & 0.903 & 0.536 & 0.642 & \textit{\textbf{0.186}} & 0.556 & 0.204                   \\
                   & 240 & 0.230 & 0.205 & 0.821 & \textit{0.173} & 0.635 & 0.643 & 0.878 & \textit{\textbf{0.164}} & 0.303 & 0.176                   \\
                   & 360 & 0.229 & 0.199 & 0.600 & \textit{0.189} & 0.744 & 1.561 & 0.753 & \textit{\textbf{0.172}} & 0.255 & 0.176                   \\
                   \hline 
\multirow{3}{*}{3} & 120 & 0.306 & 0.281 & 0.732 & \textit{0.266} & 0.953 & 0.563 & 0.710 & \textit{\textbf{0.264}} & 0.667 & 0.266                   \\
                   & 240 & 0.306 & 0.262 & 0.726 & \textit{0.242} & 0.641 & 0.654 & 1.011 & 0.245                   & 0.393 & \textit{\textbf{0.212}} \\
                   & 360 & 0.305 & 0.255 & 0.654 & \textit{0.225} & 0.741 & 1.177 & 0.786 & \textit{\textbf{0.213}} & 0.326 & 0.218                   \\
                   \hline  
\\                   
\multicolumn{12}{c}{Testing sub-sample:  1990--1999}                                                                                                                                                                                                                                                                          \\
\multirow{3}{*}{1} & 120 & 0.141 & 0.138 & 0.469 & \textit{0.139} & 0.558 & 0.443 & 0.411 & \textit{\textbf{0.134}} & 0.258 & 0.143                   \\
                   & 240 & 0.141 & 0.138 & 0.204 & \textit{0.146} & 0.666 & 0.735 & 0.530 & \textit{\textbf{0.131}} & 0.193 & 0.133                   \\
                   & 360 & 0.141 & 0.141 & 0.221 & \textit{0.156} & 0.640 & 0.593 & 0.597 & 0.134                   & 0.181 & \textit{\textbf{0.131}} \\
                   \hline 
\multirow{3}{*}{2} & 120 & 0.189 & 0.180 & 0.483 & \textit{0.178} & 0.621 & 0.449 & 0.441 & \textit{\textbf{0.174}} & 0.294 & 0.184                   \\
                   & 240 & 0.189 & 0.184 & 0.266 & \textit{0.165} & 0.665 & 0.706 & 0.555 & \textit{\textbf{0.163}} & 0.269 & 0.175                   \\
                   & 360 & 0.190 & 0.187 & 0.272 & \textit{0.185} & 0.644 & 0.597 & 0.610 & \textit{\textbf{0.167}} & 0.244 & 0.168                   \\
                   \hline 
\multirow{3}{*}{3} & 120 & 0.234 & 0.224 & 0.503 & \textit{0.222} & 0.653 & 0.458 & 0.454 & \textit{\textbf{0.224}} & 0.340 & 0.233                   \\
                   & 240 & 0.234 & 0.226 & 0.393 & \textit{0.217} & 0.656 & 0.715 & 0.572 & \textit{\textbf{0.201}} & 0.338 & 0.212                   \\
                   & 360 & 0.237 & 0.229 & 0.336 & \textit{0.215} & 0.645 & 0.594 & 0.632 & 0.207                   & 0.307 & \textit{\textbf{0.200}} \\                
                   \hline 
\multicolumn{12}{c}{Testing sub-sample:  2000--2009}                                                                                                                                                                                                                                                                          \\
\multirow{3}{*}{1} & 120 & 0.168 & 0.150 & 0.439 & \textit{0.146} & 0.990 & 0.444 & 0.333 & \textit{\textbf{0.137}} & 0.648 & 0.149                   \\
                   & 240 & 0.169 & 0.145 & 0.663 & \textit{0.147} & 0.659 & 0.626 & 0.504 & \textit{\textbf{0.122}} & 0.236 & 0.122                   \\
                   & 360 & 0.169 & 0.141 & 0.388 & \textit{0.160} & 3.053 & 3.186 & 0.485 & \textit{\textbf{0.122}} & 0.197 & 0.123                   \\
                   \hline 
\multirow{3}{*}{2} & 120 & 0.282 & 0.237 & 0.576 & \textit{0.210} & 1.037 & 0.498 & 0.459 & \textit{\textbf{0.188}} & 0.700 & 0.225                   \\
                   & 240 & 0.283 & 0.219 & 1.105 & \textit{0.173} & 0.679 & 0.635 & 0.854 & \textit{\textbf{0.162}} & 0.282 & 0.174                   \\
                   & 360 & 0.282 & 0.212 & 0.515 & \textit{0.202} & 0.751 & 2.525 & 0.481 & 0.186                   & 0.269 & \textit{\textbf{0.174}} \\
                     \hline 
\multirow{3}{*}{3} & 120 & 0.399 & 0.333 & 0.652 & \textit{0.318} & 1.049 & 0.510 & 0.600 & \textit{\textbf{0.311}} & 0.718 & 0.313                   \\
                   & 240 & 0.400 & 0.305 & 0.709 & \textit{0.292} & 0.679 & 0.606 & 1.174 & 0.311                   & 0.310 & \textit{\textbf{0.212}} \\
                   & 360 & 0.399 & 0.293 & 0.555 & \textit{0.251} & 0.738 & 1.790 & 0.511 & 0.231                   & 0.334 & \textit{\textbf{0.228}}  \\
                   \hline  
\multicolumn{12}{c}{Testing sub-sample:  2010--2019}                                                                                                                                                                                                                                                                          \\
\multirow{3}{*}{1} & 120 & 0.151 & 0.160 & 0.902 & \textit{0.147} & 1.038 & 0.623 & 0.851 & \textit{\textbf{0.147}} & 0.408 & 0.151                   \\
                   & 240 & 0.150 & 0.163 & 0.806 & \textit{0.142} & 0.567 & 0.568 & 1.106 & \textit{\textbf{0.133}} & 0.278 & 0.145                   \\
                   & 360 & 0.149 & 0.151 & 0.779 & \textit{0.132} & 0.834 & 0.748 & 1.015 & \textit{\textbf{0.131}} & 0.199 & 0.146                   \\
                     \hline 
\multirow{3}{*}{2} & 120 & 0.208 & 0.221 & 0.926 & \textit{0.197} & 0.993 & 0.641 & 0.911 & \textit{\textbf{0.194}} & 0.592 & 0.200                   \\
                   & 240 & 0.207 & 0.210 & 0.855 & \textit{0.180} & 0.554 & 0.582 & 1.128 & \textit{\textbf{0.169}} & 0.352 & 0.180                   \\
                   & 360 & 0.204 & 0.197 & 0.861 & \textit{0.178} & 0.825 & 0.762 & 1.047 & \textit{\textbf{0.163}} & 0.250 & 0.185                   \\
                     \hline 
\multirow{3}{*}{3} & 120 & 0.259 & 0.274 & 0.964 & \textit{0.249} & 1.096 & 0.694 & 0.973 & 0.249                   & 0.839 & \textit{\textbf{0.246}} \\
                   & 240 & 0.257 & 0.248 & 0.962 & \textit{0.208} & 0.583 & 0.636 & 1.168 & \textit{\textbf{0.209}} & 0.502 & 0.213                   \\
                   & 360 & 0.253 & 0.240 & 0.928 & \textit{0.206} & 0.829 & 0.773 & 1.092 & \textit{\textbf{0.200}} & 0.335 & 0.225       \\
                        \hline \hline       
\end{tabular}
\end{center}
\footnotesize{Notes: NT and ST are abbreviations for no transformation and stationarization transformation, respectively. Bold numbers indicate the best performance in each row. Italic numbers indicate the best LASSO performance with the same number of predictors.}
\end{table}

\begin{figure}[h]
\begin{centering}
\includegraphics[scale=0.6]{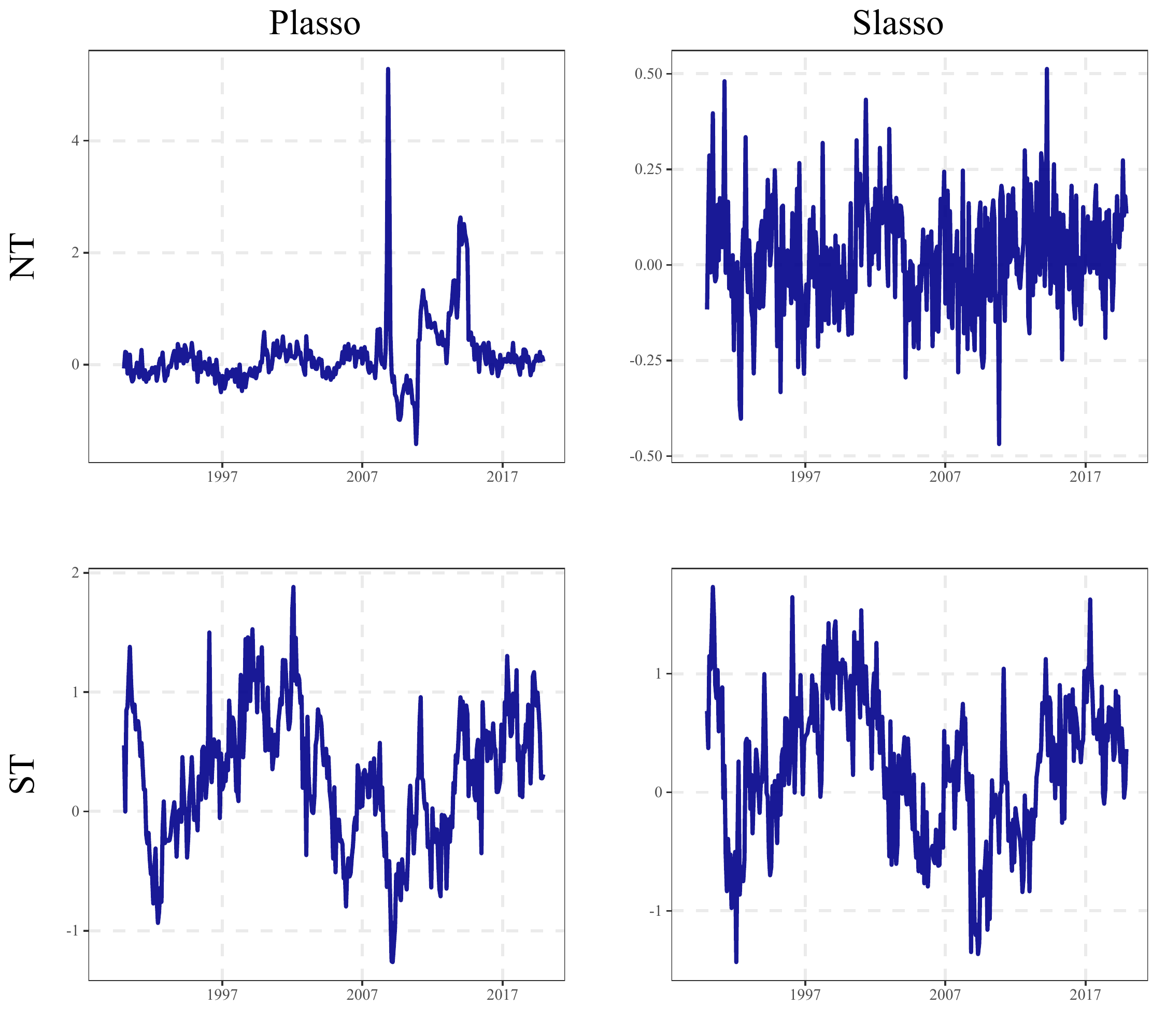}
\par\end{centering}
\caption{\label{fig:res_UNRATE} Prediction Errors under $h=1$ and 20-year
Rolling Windows}
\end{figure}

We conduct 1, 2, or 3-month ahead out-of-sample prediction, denoted
by $h=1,2$, or $3$. We set two simple benchmark models: (i) \emph{Random
walk with drift} (RWwD), where $\widehat{y}_{n+h}=y_{n}+\frac{h}{n}(y_{n}-y_{0})$;
and (ii) AR model $\hat{y}_{n+h}=\hat{\pi}_{0,h}+\hat{\pi}_{1,h}y_{n}+\cdots+\hat{\pi}_{q,h}y_{n-q+1}$
where the AR coefficients are estimated by OLS and the number of lags
$q$ is determined by the Bayesian information criterion. All these
models use information up to time $n$. Table \ref{tab:unrate_RMSPE_504}
shows RMSPE averaged over the entire testing sample 1990:Jan--2019:Dec,
and three testing sub-samples for each decade. Across the lengths
of the rolling windows, a 30-year rolling window does not necessarily
improve RMSPE, indicating potential model uncertainty over a long
training sample. Across the testing sub-samples, RMSPE is the largest
during 2000--2010, which includes the Great Recession. Across the
forecast horizons, the estimation error increases as the horizon gets
farther in the future. 

When implementing LASSO, we use the data-driven 10-fold CV as introduced
in Section \ref{sec:Simulations}. Regarding the potential regressors,
we first throw all the 121 predictors into the linear regression.
With the rich mix of time series of various temporal patterns, we
find in Table \ref{tab:unrate_RMSPE_504} that Plasso is much worse
than Slasso, and Slasso under NT outperforms the best benchmark model
in most cases. If we transform all variables to ST according to \texttt{TCODE},
the outcomes deteriorate. Relative performance is similar when the
error is measured by MAPE, reported in Table \ref{tab:unrate_MAPE_504}
in the Appendix. These empirical results echo \citet{smeekes2020unit},
who find that the best forecast strategy should be devised based on
the nature of the target time series as well as the predictors; the
information contained in NT often has an edge over ST.\footnote{\citet{smeekes2020unit} carry out empirical exercises targeting several
variables in FRED-MD and they also use Google Trend to nowcast Dutch
unemployment. These empirical applications are elaborated in \citet{smeekes2018macroeconomic}
and \citet{smeekes2021automated}.}\textbf{ }

We plot the prediction errors $(\widehat{y}_{n+1}-y_{n+1})$ in Figure
\ref{fig:res_UNRATE} under NT and ST with $h=1$ and 20-year rolling
window. The graphs are similar under other $h$ and rolling window
lengths. The prediction errors based on Slasso with NT (upper right
panel) fluctuate around 0 in a narrow range between $\pm0.5$, as
the 121 regressors form a linear combination that predicts well the
one-month-ahead unemployment rate. The errors produced by Plasso (upper
left panel) remain persistent, swinging wildly between 2008 and 2016.
Under ST (lower panels) the outliers are not as pronounced, but the
prediction errors go beyond the range of $\pm1.2$ and appear persistent.
Furthermore, under the same $h$ and $n$ we check LASSO's selected
variables under NT. The FRED database classifies all time series into
8 categories based on economic implications, and \texttt{UNRATE} belongs
to the \emph{Labor Market (LM)} group. Table \ref{tab:emp_select}
reports the top 10 most frequently selected variables over the rolling
window estimation. Among these 10 variables that Slasso picks out,
9 are from LM group, which showcase the economic relevance of the
variables chosen by Slasso. In contrast, Plasso is inclined to select
the variables with large s.d.

\begin{table}[]
\caption{Most Frequently Selected Variables}
\label{tab:emp_select}
\begin{center}
\small
\begin{tabular}{c|lccc|lccc}
\hline\hline 
\multicolumn{1}{c|}{\multirow{2}{*}{Rank}} & \multicolumn{4}{c}{Plasso}                             & \multicolumn{4}{c}{Slasso}                           \\
\cline{2-9}
\multicolumn{1}{c|}{}                      & Mnemonics    & LM          & Freq & s.d.~Rank & Mnemonics  & LM          & Freq & s.d.~Rank \\
\hline
1                                         &   \texttt{BOGMBASE}   &  & 360  & 1         & \texttt{CLAIMSx}   & \checkmark & 360  & 9         \\
2                                         &  \texttt{BUSINVx}   &             & 360  & 3         & \texttt{UEMP15OV}  & \checkmark & 360  & 24        \\
3                                         &  \texttt{CLAIMSx}   & \checkmark & 360  & 9         & \texttt{UEMP5TO14} & \checkmark & 360  & 39        \\
4                                         & \texttt{CMRMTSPLx}   &                       & 356  & 4         & \texttt{UEMPLT5}   & \checkmark & 353  & 41        \\
5                                         & \texttt{DTCTHFNM}    &                       & 342  & 6         & \texttt{HWI}       & \checkmark & 314  & 30        \\
6                                         & \texttt{AMDMUOx}     &                       & 315  & 5         & \texttt{AWOTMAN}   & \checkmark & 297  & 113       \\
7                                         & \texttt{NONBORRES}   &                       & 306  & 2         & \texttt{USTRADE}   & \checkmark & 262  & 19        \\
8                                         & \texttt{DTCOLNVHFNM} &                       & 301  & 8         & \texttt{UEMP27OV}  & \checkmark & 256  & 28        \\
9                                         & \texttt{UEMP15OV}    & \checkmark & 294  & 24        & \texttt{USCONS}    & \checkmark & 215  & 26        \\
10                                        & \texttt{PAYEMS}      & \checkmark & 275  & 13        & \texttt{PERMITW}  &                       & 200  & 49        \\
\hline\hline
\end{tabular}
\end{center}
\footnotesize{Notes: 
The estimation is conducted under NT with $h=1$ and $n=20$.
The variable names follow FRED's mnemonics; See \citet{McCracken2016}. 
The ``LM'' column ticks a variable if it belongs to the \emph{Labor Market} group. ``Freq'' displays the frequency of each variable being selected among 
the 360 regressions over the rolling windows. ``s.d.~Rank'' marks the ranking of each variable based on its sample s.d.~from high to low.
}
\end{table}

That the linear combination of 121 predictors under Slasso can outperform
the benchmarks is encouraging. It illustrates the value of a high
dimensional model estimated by an off-the-shelf machine learning method.
Next, we experiment with an augmented model. \citet{stock2002forecasting}
propose computing diffusion indices---the principal components from
many potential predictors, and \citet{bai2008forecasting} further
add lagged dependent variables into predictive regressions. Following
\citet{medeiros2021forecasting}, we incorporate the lagged dependent
variable and four diffusion indices, making 126 unique regressors,
and to allow potential delayed effects we include four time lags of
each predictor, totaling $126\times4=504$ regressors. 

Columns under ``504 predictors'' in Table \ref{tab:unrate_RMSPE_504}
show the corresponding RMPSE. While Plasso remains worse than the
simple benchmarks, we observe improvement in Slasso. First, under
ST the additional lagged dependent variables and diffusion indices
mitigate the imbalance in the predictive regression and therefore
strengthen the performance of Slasso. Moreover, these additional regressors
improve Slasso under NT, which is the overall best performer. It reduces
the RMPSE in 29 out of the 36 instances relative to the counterpart
with ``121 predictors'', and beats the last column in most instances
except $h=3$, where the errors mainly occur during 2000--2009. These
results indicate that macroeconomic domain knowledge is instrumental
in guiding the initial specification to determine the pool of regressors,
and then Slasso takes care of the estimation of many coefficients.
This fusion of field expertise and machine learning is more effective
than simply LASSOing with all variables in the database. 

\section{Conclusion\label{sec:Conclusion} }

This paper studies asymptotic properties of LASSO in predictive regressions
where many nonstationary time series are present. We establish new
bounds for the RE, which  allows us to derive convergence rates for
Plasso and Slasso. The consistency of Slasso is extended to the model
of mixed stationary and nonstationary regressors, and it can further
digest information from cointegrated variables. The simulations and
the empirical application provide numerical evidence that supports
the merits of Slasso, which we recommend for practice.

As a first step of exploration, this paper uses the unit root process
as a representative of nonstationary time series. There are other
popular models that characterize persistence, for example, local-to-unity
and fractional integration. Future investigation of these nonstationary
time series will generalize the theory and further guide practical
implementation. Moreover, the theoretical results of this paper rely
on a tuning parameter expanding at some rate based on the same size.
It will be important to explore the behaviors of LASSO involving nonstationary
regressors when the tuning parameter is selected by a data-driven
method, such as the CV. Last but not least, a formal testing procedure
for coefficients will be feasible in high dimensions if we debias
the LASSO estimator.

\bigskip \bigskip

\bibliographystyle{chicagoa}
\bibliography{HDPredictive}

\newpage{}

\setcounter{footnote}{0}
\setcounter{table}{0} 
\setcounter{figure}{0} 
\setcounter{equation}{0} 
\renewcommand{\thefootnote}{\thesection.\arabic{footnote}} 
\renewcommand{\theequation}{\thesection.\arabic{equation}} 
\renewcommand{\thefigure}{\thesection.\arabic{figure}} 
\renewcommand{\thetable}{\thesection.\arabic{table}} 
\setcounter{thm}{0} 
\setcounter{lem}{0} 
\setcounter{rem}{0}
\setcounter{cor}{0} 
\setcounter{prop}{0} 
\renewcommand{\thethm}{\thesection.\arabic{thm}}
\renewcommand{\thelem}{\thesection.\arabic{lem}} 
\renewcommand{\therem}{\thesection.\arabic{rem}} 
\renewcommand{\thecor}{\thesection.\arabic{cor}} 
\renewcommand{\theprop}{\thesection.\arabic{prop}} 

\begin{appendices}
\begin{center}
{\LARGE Online Appendix for \\
``On LASSO for High Dimensional Predictive Regression''}
\end{center}
\begin{center} \large
Ziwei Mei and Zhentao Shi
\end{center}
Section \ref{sec:RE_example} provides probabilistic calculation and
numerical evidence to demonstrate the behavior of the minimum eigenvalues
of the Gram matrix under stationary and unit root processes. Section
\ref{sec:Proofs} collects the proofs of all the theoretical statements
in the main text, and the supporting lemmas along with their proofs.
Section \ref{sec:Additional-Numerical-Results} contains Monte Carlo
simulations for the data generating processes (DGP) with cointegrated
variables, and additional supporting results for the numerical study
in the main text. 

\section{Technical Calculation \label{sec:RE_example}}

This section illustrates the behavior of the Gram matrix when the
underlying processes are stationary or unit roots. For simplicity,
we assume the unit root vector $X_{t}$ is generated by the $s$-dimensional
innovation $e_{t}\sim\mathrm{i.i.d.}\mathcal{N}\left(0,I_{s}\right)$.
When $n\to\infty$, the $j$th diagonal entry of the (scaled) Gram
matrix $\hat{\Sigma}_{s}/n$ is 
\[
\mathcal{D}_{jj}=\int_{0}^{1}\mathcal{B}_{j}^{2}(r)dr-\left(\int_{0}^{1}\mathcal{B}_{j}(r)dr\right)^{2}
\]
 where $\mathcal{B}_{j}(r)$ is the standard Brownian motion (Wiener
process). Proposition \ref{prop:RE_KL} shows that $\mathcal{D}_{jj}$
is smaller than any fixed positive constant with non-trivial probability
that is bounded away from 0. This result is proved at the end of this
section.
\begin{prop}
\label{prop:RE_KL} (a) For any $\delta>0$, there exists a $\zeta_{2}>0$
such that $\Pr\left\{ \min_{j\in[s]}\mathcal{D}_{jj}\geq\delta\right\} \leq\left(1-\zeta_{2}\right)^{s}$.
(b) $\mathbb{E}\left[\mathcal{D}\right]=\frac{1}{6}I_{s}$. 
\end{prop}
The above Proposition \ref{prop:RE_KL} (a) implies 
\[
\Pr\left\{ \phi_{\min}\left(\mathcal{D}\right)\geq\delta\right\} \leq\Pr\left\{ \min_{j\in[s]}\mathcal{D}_{jj}\geq\delta\right\} \leq\left(1-\zeta_{2}\right)^{s}
\]
and the right-hand side shrinks to 0 as $s\to\infty$; in other words
$\phi_{\min}\left(\mathcal{D}\right)\stackrel{p}{\to}0$. It characterizes
the behavior of the minimum eigenvalue that is suitable for the case
$s\ll n$ as in our analysis. Part (b) highlights that the behavior
of the population expectation is in sharp contrast with the minimum
diagonal element. The difference stems from the fact that $\mathcal{D}$
is a random matrix, not a constant matrix. 

We conduct a simulation exercise to provide further numerical evidence.
For comparison, we generate i.i.d.~sequence $e_{t}\sim\mathcal{N}\left(0,I_{s}\right)$,
compute the $s\times s$ Gram matrix of $\left(e_{t}\right)_{t=1}^{n}$,
and denote it as $\hat{\Sigma}_{s}^{\natural}$, to be distinguished
with the Gram matrix $\hat{\Sigma}_{s}$ when we generate the underlying
$X_{t}=\sum_{r=1}^{t}e_{r}$ as independent unit root processes. In
theory we should set the sample size $n$ as large as possible to
mimic the continuous path of the Brownian motion, whereas in practice
we find $n=5000$ is sufficiently large for our purpose. 

Figure \ref{fig:RE} shows the numerical evidence averaged over 1000
replications. Panel (A) displays the logarithm of the minimum diagonal
entries of the Gram matrix $\hat{\Sigma}_{s}^{\natural}$ (i.i.d.~regressors)
and $\hat{\Sigma}_{s}/n$ (unit root regressors). As $s$ growing,
$\min_{i\in[s]}(\hat{\Sigma}_{s,ii}^{\natural})$ is stable around
1 ($\log1=0$) as $\mathbb{E}\left[e_{t}e_{t}^{\top}\right]=I_{s}$
in theory. In contrast, as $s$ grows $\min_{i\in[s]}(\hat{\Sigma}_{s}/n)$
declines; for example, its value falls below ${\rm e}^{-3}=0.05$
when $s=128$. The phenomenon supports Proposition \ref{prop:RE_KL}.

Panel (B) displays parallel results on the minimum eigenvalue. Compared
with Panel (A) counterparts, we observe that $\lambda_{\min}(\hat{\Sigma}_{s}^{\natural})$
remains stable near 1, whereas $\lambda_{\min}(\hat{\Sigma}_{s}/n)$
vanishes much faster and becomes smaller than ${\rm e}^{-7}=0.001$
as $s=128$. Moreover, from $s=4$ to $s=128$ the points largely
align on a straight line, which echos the rate with respect to $s$
on the right-hand side of (\ref{eq:RE-unit}) in Lemma \ref{lem:Normal_RE}. 

The derivation and the numerical results provide clear evidence of
the drastically different behavior of the minimum eigenvalue of the
Gram matrix when the underlying regressors are i.i.d\@.~or unit
roots. Despite the shrinking minimum eigenvalue toward zero, the relatively
slow rate in terms of $s$ can be compensated by the super-consistency
due to the strong signal of unit roots in terms of $n$, making it
possible for LASSO to maintain consistency, as shown in the main text.

\begin{figure}
\begin{centering}
\includegraphics[scale=0.73]{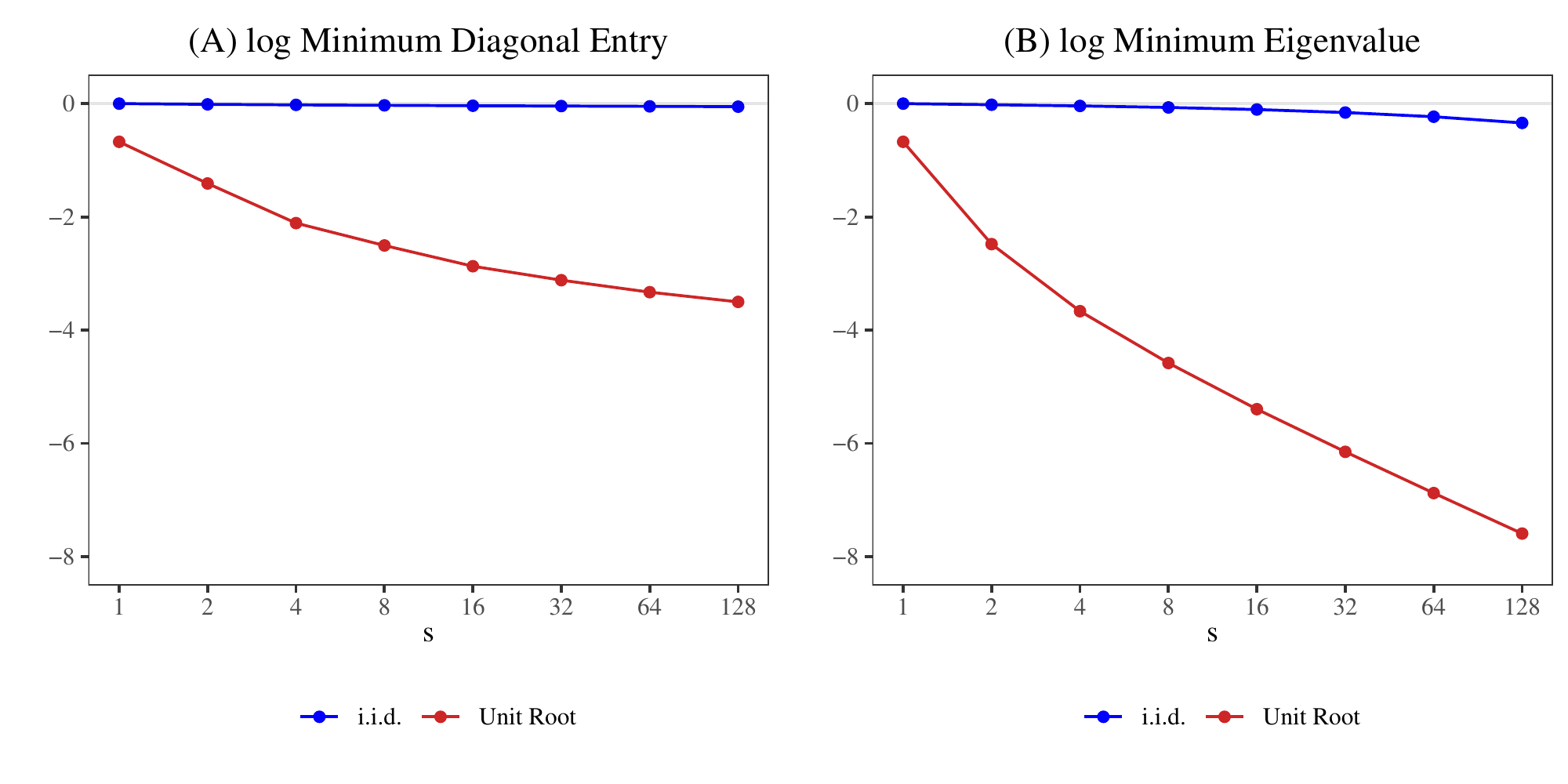}
\par\end{centering}
\footnotesize Note: The y-axis is the \emph{logarithm} \emph{of the
corresponding value (averaged over 1000 replications), }and the x-axis
is the dimension $s$.

\caption{\label{fig:RE} Numerical Illustrations for the RE Condition}
\end{figure}
\medskip
\begin{proof}[Proof of Proposition \ref{prop:RE_KL}]
 In this proof we discuss fixed $j,k\in[s]$ as $n\to\infty$. To
simplify the notations, for a diagonal element we denote $\mathcal{D}_{jj}$
as $d^{\diamondsuit}$ by suppressing its dependence on $j$, and
for an off-diagonal element we denote $\mathcal{D}_{j,k}$ as $d^{\ddagger}$. 

\textbf{Part (a).} Notice that $d^{\diamondsuit}=\int_{0}^{1}\mathcal{B}^{2}(r)dr-\left(\int_{0}^{1}\mathcal{B}(r)dr\right)^{2}\leq\int_{0}^{1}\mathcal{B}^{2}(r)dr.$
The standard Brownian motion (Wiener process) admits the Karhunen-Lo\`{e}ve
representation 
\[
\mathcal{B}\left(r\right)=\sqrt{2}\sum_{k=1}^{\infty}\frac{\sin\left((k-0.5)\pi r\right)}{\left(k-0.5\right)\pi}\xi_{k}
\]
 where $\xi_{k}\sim N\left(0,1\right)$ are i.i.d.~random coefficients
and $\left\{ \sin\left((k-0.5)\pi r\right)\right\} _{k=1}^{\infty}$
an orthogonal basis \citep{phillips1998new}. We can thus bound $d^{\diamondsuit}$
by the random series $\sum_{k=1}^{\infty}\left(\frac{\xi_{k}}{\left(k-0.5\right)\pi}\right)^{2}.$

Define an event $\mathcal{G}\left(\delta_{1}\right):=\bigcup_{k=1}^{\infty}\left\{ \left|\xi_{k}\right|\leq\delta_{1}\left[(k-0.5)\pi\right]^{1/4}\right\} $
for a fixed $\delta_{1}>0$. Its probability
\begin{align}
\Pr\left\{ \mathcal{G}\left(\delta_{1}\right)\right\}  & =\Pr\left\{ \bigcup_{k=1}^{\infty}\left\{ \frac{\left|\xi_{k}\right|}{\left(k-0.5\right)\pi}\leq\frac{\delta_{1}}{\left[(k-0.5)\pi\right]^{3/4}}\right\} \right\} \nonumber \\
 & \leq\Pr\left\{ \sum_{k=1}^{\infty}\left(\frac{\xi_{k}}{\left(k-0.5\right)\pi}\right)^{2}\leq\sum_{k=1}^{\infty}\frac{\delta_{1}^{2}}{\left[(k-0.5)\pi\right]^{3/2}}\right\} \leq\Pr\left\{ d^{\diamondsuit}\leq\delta_{1}^{2}m_{1}\right\} \label{eq:G2}
\end{align}
where $m_{1}:=\sum_{k=1}^{\infty}\left[(k-0.5)\pi\right]^{-3/2}$
is convergent. 

Since $\xi_{k}\sim\mathrm{i.i.d.}\mathcal{N}\left(0,1\right)$, we
use $\Phi\left(\cdot\right)$ to denote the cumulative distribution
function of $N\left(0,1\right)$. The probability of the event is
bounded below by 
\begin{eqnarray}
\Pr\left\{ \mathcal{G}\left(\delta_{1}\right)\right\}  & = & \prod_{k=1}^{\infty}\Pr\left\{ \left\{ \left|\xi_{k}\right|\leq\delta_{1}\left[(k-0.5)\pi\right]^{1/4}\right\} \right\} =\prod_{k=1}^{\infty}\left[1-2\Phi\left(-\delta_{1}\left[(k-0.5)\pi\right]^{1/4}\right)\right]\nonumber \\
 & = & \exp\left\{ \sum_{k=1}^{\infty}\log\left[1-2\Phi\left(-\delta_{1}\left[(k-0.5)\pi\right]^{1/4}\right)\right]\right\} \nonumber \\
 & \geq & \exp\left\{ -2\sum_{k=1}^{\infty}\Phi\left(-\delta_{1}\left[(k-0.5)\pi\right]^{1/4}\right)\right\} \nonumber \\
 & \geq & \exp\left\{ -2\sum_{k=1}^{\infty}\exp\left(-\frac{\delta_{1}^{2}}{2}\sqrt{(k-0.5)\pi}\right)\right\} ,\label{eq:G1}
\end{eqnarray}
where the first inequality is due to $\log(1-x)\geq-x$ for $x\in[0,1)$,
and the last inequality by $\Phi\left(-x\right)\leq\exp\left(-\frac{1}{2}x^{2}\right)$
for all $x>0$. Since $\sum_{k=1}^{\infty}\exp\left(-\frac{\delta_{1}^{2}}{2}\sqrt{(k-0.5)\pi}\right)$
is a convergent series, the probability $\Pr\left\{ \mathcal{G}\left(\delta_{1}\right)\right\} >0$
for any fixed $\delta_{1}>0$. Combine (\ref{eq:G2}) and (\ref{eq:G1}):
\[
\Pr\left\{ d^{\diamondsuit}\leq\delta_{1}^{2}m_{1}\right\} \geq\Pr\left\{ \mathcal{G}\left(\delta_{1}\right)\right\} \geq\exp\left\{ -2\sum_{k=1}^{\infty}\exp\left(-\frac{\delta_{1}^{2}}{2}\sqrt{(k-0.5)\pi}\right)\right\} >0.
\]
In other words, for any $\delta>0$, there exists an absolute constant
$\zeta_{2}>0$ such that $\Pr\left\{ d^{\diamondsuit}\leq\delta\right\} \geq\zeta_{2}$.
Since the diagonal elements of $\mathcal{D}$ are independent, 
\[
\Pr\left\{ \min\mathcal{D}\geq\delta\right\} =\left(\Pr\left\{ d^{\diamondsuit}\geq\delta\right\} \right)^{s}\leq\left(1-\zeta_{2}\right)^{s}.
\]

\textbf{Part (b)}. The behavior the population expectation of $\mathcal{D}$
is very different from the minimum diagonal element $\mathcal{D}$.
For the diagonal element we have $\mathbb{E}\left[d^{\diamondsuit}\right]=1/6$
as the difference between 
\[
\mathbb{E}\left[\int_{0}^{1}\mathcal{B}_{j}^{2}(r)dr\right]=\sum_{k=1}^{\infty}\frac{\mathbb{E}\left[\xi_{k}^{2}\right]}{\left(k-0.5\right)^{2}\pi^{2}}=\frac{1}{2}
\]
and 
\begin{align*}
\mathbb{E}\left[\left(\int_{0}^{1}\mathcal{B}(r)dr\right)^{2}\right] & =\sum_{k=1}^{\infty}\left(\int_{0}^{1}\sqrt{2}\frac{\sin\left((k-0.5)\pi r\right)}{\left(k-0.5\right)\pi}dr\right)^{2}\mathbb{E}\left[\xi_{k}^{2}\right]\\
 & =\sum_{k=1}^{\infty}\left(\int_{0}^{1}\sqrt{2}\frac{\sin\left((k-0.5)\pi r\right)}{\left(k-0.5\right)\pi}dr\right)^{2}=\frac{1}{3}.
\end{align*}

On the other hand, the off-diagonal element is
\[
d^{\ddagger}=\int_{0}^{1}\mathcal{B}_{1}(r)\mathcal{B}_{2}(r)dr-\int_{0}^{1}\mathcal{B}_{1}(r)dr\int_{0}^{1}\mathcal{B}_{2}(r)dr
\]
where $\mathcal{B}_{1}(r)$ and $\mathcal{B}_{2}(r)$ are two independent
Wiener processes. Its population expectation $\mathbb{E}\left[d^{\ddagger}\right]=0$,
because 
\[
\mathbb{E}\left[\int_{0}^{1}\mathcal{B}_{1}(r)dr\right]=\sum_{k=1}^{\infty}\int_{0}^{1}\sqrt{2}\frac{\sin\left((k-0.5)\pi r\right)}{\left(k-0.5\right)\pi}dr\mathbb{E}\left[\xi_{k}\right]=0
\]
and 
\[
\mathbb{E}\left[\int_{0}^{1}\mathcal{B}_{1}(r)\mathcal{B}_{2}(r)dr\right]=\sum_{k=1}^{\infty}\frac{\mathbb{E}\left[\xi_{1k}\xi_{2k}\right]}{\left(k-0.5\right)^{2}\pi^{2}}=0.
\]
We complete the proof.
\end{proof}

\section{Proofs\label{sec:Proofs}}

Section \ref{subsec:Preliminary-Lemmas} provides several preliminary
lemmas. Section \ref{subsec:Prop} includes the preparatory propositions
for DB and RE. Section \ref{subsec:ProofMain} collects the proofs
of the results in the main text. Section \ref{subsec:Proofs-of-Lemmas}
proves the Lemmas. We use $c$ and $C$, with no superscript or subscript,
to denote generic positive constants that may vary in occasions. 

For notational simplicity, in the proofs we assume $p\geq n^{\nu_{1}}$
for some absolute constant $\nu_{1}$. This is reasonable as we focus
on the high-dimensional case with large $p$ relative to $n$. There
is no technical difficulty in allowing $p$ to grow either slowly
at a logarithmic or quickly at an exponential rate of $n$, but without
the polynomial rate lower bound $\nu_{1}$ we have to compare $\log p$
and $\log n$ in many places, and in many conditions and rates the
term ``$\log p$'' has to be changed into $\log(np)$. 

\subsection{\label{subsec:Preliminary-Lemmas} Lemmas }

Lemma \ref{lem:mixing} shows the mixing properties of $\varepsilon=(\varepsilon_{t})_{t\in\mathbb{Z}}$
where $\varepsilon_{t}$ is the linear process defined in (\ref{eq:linrProc}).
Lemma \ref{lem:BernsteinSum} establishes a Bernstein-type concentration
inequality for the partial sums of independent sub-exponential variables.
Lemma \ref{lem:LinComb} and Corollary \ref{cor:BNtail} maintain
the sub-exponential property of linear combinations of independent
sub-exponential variables, which help bound the errors term from the
Wold decomposition of $\varepsilon_{t}$, and also the noises $e_{t}$
and $u_{t}$ which are linear transformations of $\varepsilon_{t}.$
The Gaussian approximation in Lemma \ref{lem:GaussianApprox} carries
RE over into non-Gaussian variables. The proofs of the lemmas are
relegated to Section \ref{subsec:Proofs-of-Lemmas}. 

The $\alpha$-mixing and $\rho$-mixing coefficients of two generic
$\sigma$-fields $\mathcal{A}$ and $\mathcal{B}$ are defined as
\begin{align}
\alpha(\mathcal{A},\mathcal{B}) & :=\sup_{A\in\mathcal{A},B\in\mathcal{B}}|\Pr(A\cap B)-\Pr(A)\Pr(B)|,\nonumber \\
\rho(\mathcal{A},\mathcal{B}) & :=\sup_{X\in\mathcal{A},Y\in\mathcal{B}}\left|\ensuremath{\mathbb{E}XY-\mathbb{E}X\mathbb{E}Y}\right|\big/\sqrt{\mathbb{E}X^{2}\mathbb{E}Y^{2}}\ \ \ \text{ for }\mathbb{E}X^{2},\mathbb{E}Y^{2}<\infty.\label{eq:def_rho}
\end{align}

In this section, we use the lowercase $x$ to denote a generic random
vector, and let $\sigma(x)$ be the $\sigma$-field generated by $x.$
For $d\in\mathbb{N}$, the $\alpha$-mixing and $\rho$-mixing coefficients
$x=(x_{t})_{t\in\mathbb{Z}}$ are defined as 
\begin{align*}
\alpha(x,d) & :=\sup_{s\in\mathbb{Z}}\alpha(\sigma((x_{t})_{t\leq s}),\sigma((x_{t})_{t\geq s+d}))\\
\rho(x,d) & :=\sup_{s\in\mathbb{Z}}\rho(\sigma((x_{t})_{t\leq s}),\sigma((x_{t})_{t\geq s+d})).
\end{align*}
The following Lemma \ref{lem:mixing} states that the linear process
$\varepsilon$ is geometric $\rho$-mixing.
\begin{lem}
\label{lem:mixing}Suppose that Assumptions \ref{assu:tail} and \ref{assu:alpha}
hold. Let $\varepsilon=\left(\varepsilon_{t}\right){}_{t\in\mathbb{Z}}$.
Then we have 
\begin{equation}
\alpha(\text{\ensuremath{\varepsilon}},d)\leq\rho(\text{\ensuremath{\varepsilon}},d)\leq C_{\alpha}\exp\left(-c_{\alpha}d^{r}\right)\label{eq:mixingBound}
\end{equation}
for $d$ sufficiently large, where $C_{\alpha}$ and $c_{\alpha}$
are two absolute constants.
\end{lem}
Lemma \ref{lem:BernsteinSum} provides a probabilistic order for the
maximum of the partial sum along its path when the innovations are
sub-exponential and geometric $\alpha$-mixing. 
\begin{lem}
\label{lem:BernsteinSum} Let $x_{t}$ be a $p\times1$ random vector
strictly stationary over $t$ and $\|\mathbb{E}x_{t}\|_{\infty}<\infty$.
Assume there exist absolute constants $C_{x}$ and $b_{x}$ such that
\[
\max_{j\in[p]}\Pr\left\{ |x_{jt}|>\mu\right\} \leq C_{x}\exp\left(-\mu/b_{x}\right)
\]
for all $\mu>0$. Moreover, assume the $\alpha$-mixing coefficient
of $x=(x_{t})_{t\in\mathbb{Z}}$ satisfies $\alpha(x,d)\leq C_{\alpha}\exp\left(-c_{\alpha}d^{r}\right)$
for some absolute constants $C_{\alpha},c_{\alpha}$ and $r$. If
$(\log p)^{1+2/r}=o(n)$, then 
\begin{equation}
\max_{j\in[p],t\in[n]}|\sum_{s=1}^{t}(x_{js}-\mathbb{E}x_{js})|\stackrel{\mathrm{p}}{\preccurlyeq}\sqrt{n\log p}.\label{eq:sumProbBound}
\end{equation}
\end{lem}
Lemma \ref{lem:LinComb} gives the tail bounds of linear combinations
of generic independent mean-zero noises with sub-exponential tails. 
\begin{lem}
\label{lem:LinComb} Let $x_{1},x_{2},\cdots$ be independent random
variables with $\mathbb{E}x_{i}=0$ for all $i\in\mathbb{N}$. Suppose
there exist absolute constants $C_{x}$ and $b_{x}$ such that 
\begin{equation}
\max_{i\in\mathbb{N}}\Pr\left\{ |x_{i}|>\mu\right\} \leq C_{x}\exp\left(-\mu/b_{x}\right)\label{eq:xprobineq}
\end{equation}
for all $\mu>0$. Then there exists an absolute constant $K_{x}$
such that for any vector $a=(a_{i}\in\mathbb{R})_{i\in\mathbb{N}}$
and $\mu>0$, we have 
\begin{equation}
\Pr\left\{ \sum_{i\in\mathbb{N}}|a_{i}x_{i}|>\mu\right\} \leq\exp\left(-\frac{1}{\|a\|_{\infty}}\left(\dfrac{\mu}{2K_{x}{\rm e}}-\left\Vert a\right\Vert _{1}\right)\right).\label{eq:lincomb_p}
\end{equation}
\end{lem}
Corollary \ref{cor:BNtail} applies Lemma \ref{lem:LinComb} to deduce
the tail bounds of the variables used in the main text. For example,
the stationary components in (\ref{eq:def-I0-mix}), where the linear
processes $\varepsilon_{jt}$ is an (infinite) linear combination
of the innovations $\eta_{jt}$ that appears in the Beveridge-Nelson
decomposition.
\begin{cor}
\label{cor:BNtail}Suppose that Assumptions \ref{assu:tail} and \ref{assu:alpha}
hold. Then there are absolute constants $C_{\eta}^{\prime},b_{\eta}^{\prime},\widetilde{C}_{\eta}$
and $\widetilde{b}_{\eta}$ such that 
\begin{eqnarray}
\Pr\left\{ |\varepsilon_{jt}|>\mu\right\}  & \leq & C_{\eta}^{\prime}\exp(-\mu/b_{\text{\ensuremath{\eta}}}^{\prime})\label{eq:epsBound}\\
\Pr\left\{ |\sum_{d=0}^{\infty}\widetilde{\psi}_{jd}\eta_{j,t-d}|>\mu\right\}  & \leq & \widetilde{C}_{\eta}\exp(-\mu/\widetilde{b}_{\eta})\label{eq:etaWoldSumBound}
\end{eqnarray}
for any $j\in\left[p+1\right]$, $t\in\mathbb{Z}$ and $\mu>0,$ where
$\widetilde{\psi}_{jd}=\sum_{\ell=d+1}^{\infty}\psi_{j\ell}.$ In
addition, if Assumption \ref{assu:covMat} holds and $v_{t}=(e_{t}^{\top},Z_{t}^{\top},u_{t})^{\top}$
is generated by (\ref{eq:def-I0-mix}), then there are absolute constants
$C_{v}$ and $b_{v}$ such that for any $\mu>0$
\begin{equation}
\sup_{j\in[p+1]}\Pr\left\{ |v_{jt}|>\mu\right\} \leq C_{v}\exp\left[-\left(\mu/b_{v}\right)\right].\label{eq:euBound}
\end{equation}
\end{cor}
Finally, if the innovations are non-Gaussian, Lemma \ref{lem:GaussianApprox}
provides Gaussian approximation when $n$ is large. 
\begin{lem}
\label{lem:GaussianApprox}Under the Assumptions in Proposition \ref{prop:UnitRE}
and Assumption \ref{assu:asym_n}, there exists standard Brownian
motions $\{\mathcal{B}_{j}(t)\}_{j\in[p+1]}\}$ with independent increment
$\mathcal{B}_{j}(t)-\mathcal{B}_{j}(s)\sim\mathcal{N}\left(t-s\right)$
for $t\geq s\geq0$ such that 

\[
\sup_{j\in[p+1],t\in[n]}\left|\dfrac{1}{\sqrt{n}}\left(\sum_{s=0}^{t-1}\varepsilon_{js}-\psi_{j}(1)\mathcal{B}_{j}\left(t\right))\right)\right|\stackrel{\mathrm{p}}{\preccurlyeq}\dfrac{\log p}{\sqrt{n}}.
\]
\end{lem}
The basic idea behind the proof of Lemma \ref{lem:GaussianApprox}
is that the innovation's temporal dependence can be handled by the
Beveridge-Nelson decomposition involving the long-run effect $\psi_{j}(1)$.
As a result, there exists a Gaussian process that behaves like the
underlying independent shocks $\frac{1}{\sqrt{n}}\sum_{s=0}^{t-1}\eta_{js}$,
as indicated by the Koml\'{o}s-Major-Tusn\'{a}dy coupling. The non-asymptotic
Koml\'{o}s-Major-Tusn\'{a}dy inequality allows us to extend by the
union bound for a uniform convergence over $p$. 

\subsection{Preparatory Propositions for DB and RE \label{subsec:Prop} }

The above lemmas have prepared for propositions that lead to the key
results on DB and RE.

\subsubsection{DB of Unit Root Components}
\begin{prop}
\label{prop:DB-All} Under Assumptions \ref{assu:tail}-\ref{assu:covMat}
and \ref{assu:asym_n}(a), there exits some absolute constant $C_{{\rm DB}}$
such that 
\begin{equation}
\left\{ \max_{j\in[p_{x}],k\in[p_{z}]}\dfrac{1}{n^{3/2}}\left|\sum_{t=1}^{n}X_{j,t-1}Z_{kt}\right|\vee\max_{j\in[p_{x}]}\dfrac{1}{n^{3/2}}\left|\sum_{t=1}^{n}X_{j,t-1}u_{t}\right|\right\} \stackrel{\mathrm{p}}{\preccurlyeq}\dfrac{(\log p)^{1+\frac{1}{2r}}}{\sqrt{n}}.\label{eq:DB-All-origin}
\end{equation}
\end{prop}
\begin{proof}[Proof of Proposition \ref{prop:DB-All}]
 Since $u_{t}$ and $Z_{kt}$ are both stationary and geometrically
$\alpha$-mixing, it suffices to show the order of the first term
$\max_{j\in[p_{x}],k\in[p_{z}]}n^{-3/2}\left|\sum_{t=1}^{n}X_{j,t-1}Z_{kt}\right|$,
and then the same order applies to the cross product involving $u_{t}$.
Let $G=\lfloor(2c_{\alpha}^{-1}\log(np))^{1/r}\rfloor$, and the triangular
inequality gives
\begin{eqnarray*}
 &  & \max_{j\in[p_{x}],k\in[p_{z}]}\left|\sum_{t=1}^{n}X_{j,t-1}Z_{kt}\right|\\
 & \leq & \max_{j\in[p_{x}],k\in[p_{z}]}\left|\sum_{t=1}^{G}X_{j,t-1}Z_{kt}\right|+\max_{j\in[p_{x}],k\in[p_{z}]}\left|\sum_{t=G+1}^{n}Z_{kt}\sum_{r=t-G+1}^{t-1}e_{jr}\right|+\max_{j\in[p_{x}],k\in[p_{z}]}\left|\sum_{t=G+1}^{n}X_{j,t-G}Z_{kt}\right|\\
 & =: & T_{1}+T_{2}+T_{3}.
\end{eqnarray*}
We will analyze one by one the three terms on the right-hand side. 

\textbf{Bound of $T_{1}.$ }Repeatedly applying Lemma \ref{lem:BernsteinSum}
yields 
\begin{equation}
\max_{j\in[p_{x}],t\in[n]}|X_{j,t-1}|=\max_{j\in[p_{x}],t\in[n]}|\sum_{s=0}^{t-1}e_{j,s}|\stackrel{\mathrm{p}}{\preccurlyeq}\sqrt{n\log p}\label{eq:I1-bound}
\end{equation}
 and $\max_{k\in[p_{z}]}|\sum_{t=1}^{G}(|Z_{kt}|-\mathbb{E}|Z_{kt}|)|\stackrel{\mathrm{p}}{\preccurlyeq}\sqrt{n\log p}.$
We deduce by the triangular inequality 
\begin{align}
T_{1} & \leq\max_{j\in[p_{x}]}|X_{j,t-1}|\cdot\max_{k\in[p_{z}]}\left|\sum_{t=1}^{G}Z_{kt}\right|\nonumber \\
 & \stackrel{\mathrm{p}}{\preccurlyeq}\sqrt{n\log p}\left[\max_{k\in[p_{z}]}\left|\sum_{t=1}^{G}(|Z_{kt}|-\mathbb{E}|Z_{kt}|)\right|+\max_{k\in[p_{z}]}\sum_{t=1}^{G}\mathbb{E}|Z_{kt}|\right]\nonumber \\
 & \stackrel{\mathrm{p}}{\preccurlyeq}\sqrt{n\log p}\left(\sqrt{n\log p}+G\right)=O(n\log p).\label{eq:T1bound}
\end{align}

\textbf{Bound of $T_{2}.$ }For any $d=1,2,\cdots,G-1,$ divide $\{1,2,\cdots,n-G\}$
into $d+1$ groups and assume $A_{d}=(n-G)/(d+1)$ is an integer for
simplicity of the notations. The $a$th group is given as $\mathcal{I}_{a}=\{a,a+(d+1),\cdots,a+(A_{d}-1)(d+1)\}$
for $a=1,2,\cdots,d+1.$ The triangular inequality gives 
\begin{align*}
T_{2} & \leq\max_{j\in[p_{x}],k\in[p_{z}]}\sum_{d=1}^{G-1}\left|\sum_{t=G+1}^{n}e_{j,t-d}Z_{kt}\right|\\
 & \leq\max_{j\in[p_{x}],k\in[p_{z}]}\sum_{d=1}^{G-1}\left|\sum_{t=G+1}^{n}\mathbb{E}(e_{j,t-d}Z_{kt})\right|+\max_{j\in[p_{x}],k\in[p_{z}]}\sum_{d=1}^{G-1}\left|\sum_{t=G+1}^{n}\left(e_{j,t-d}Z_{kt}-\mathbb{E}(e_{j,t-d}Z_{kt})\right)\right|\\
 & =:T_{21}+T_{22}.
\end{align*}
For $T_{21}$, by $\mathbb{E}(e_{j,t-d})=\mathbb{E}(Z_{kt})=0$ we
have 
\begin{align*}
T_{21} & =\max_{j\in[p_{x}],k\in[p_{z}]}\sum_{d=1}^{G-1}\sum_{t=G+1}^{n}\left|\mathbb{E}(e_{j,t-d}Z_{kt})-\mathbb{E}(e_{j,t-d})\mathbb{E}(Z_{kt})\right|\\
 & \leq\max_{j\in[p_{x}],k\in[p_{z}]}\sum_{d=1}^{G-1}\sum_{t=G+1}^{n}\rho(\varepsilon,d)\sqrt{\mathbb{E}(e_{j,t-d}^{2})\mathbb{E}(Z_{kt}^{2})}\\
 & =O\left(n\sum_{d=1}^{G-1}\rho(\varepsilon,d)\right)=O(n)
\end{align*}
where the inequality follows by the definition of $\rho$-mixing coefficient
in (\ref{eq:def_rho}), and the order follows as $\mathbb{E}(e_{j,t-d}^{2})$
and $\mathbb{E}(Z_{kt}^{2})$ are uniformly bounded for all $j$ and
$k$, with $\sum_{d=1}^{\infty}\rho(\varepsilon,d)$ being convergent
in view of (\ref{eq:mixingBound}). For $T_{22}$ we have 
\begin{align*}
T_{22} & \leq\max_{j\in[p_{x}],k\in[p_{z}]}\sum_{d=1}^{G-1}\sum_{a=1}^{d+1}\left|\sum_{t\in\mathcal{I}_{a}}\left(e_{j,t-d}Z_{kt}-\mathbb{E}(e_{j,t-d}Z_{kt})\right)\right|\\
 & \leq G^{2}\max_{\substack{\substack{j\in[p_{x}],k\in[p_{z}]\\
a\in[d+1],d\in[G-1]
}
}
}\left|\sum_{t\in\mathcal{I}_{a}}\left(e_{j,t-d}Z_{kt}-\mathbb{E}(e_{j,t-d}Z_{kt})\right)\right|.
\end{align*}
Note that for all $t\in\mathcal{I}_{a},$ the cross term $e_{j,t-d}Z_{kt}\in\sigma((\varepsilon_{s})_{t-d\leq s\leq t})$
and $|t-s|=d+1>d$ for $t,s\in\mathcal{I}_{a}$ and $t\neq s.$ Thus
$(e_{j,t-d}Z_{kt})_{t\in\mathcal{I}_{a}}$ is $\alpha$-mixing with
its coefficient bounded by (\ref{eq:mixingBound}). Moreover, for
any $\mu>0$ we can bound
\begin{align*}
\Pr\left\{ |e_{j,t-d}Z_{kt}|>\mu\right\}  & \leq\Pr\left\{ |e_{j,t-d}|>\sqrt{\mu}\right\} +\Pr\left\{ |Z_{kt}|>\sqrt{\mu}\right\} \\
 & \leq2C_{v}\exp\left[-\sqrt{\mu}/b_{v}\right]=2C_{v}\exp\left[-\left(\mu/b_{v}^{2}\right)^{1/2}\right]
\end{align*}
where the second inequality applies (\ref{eq:euBound}) in Corollary
\ref{cor:BNtail}. 

Let $r^{**}=\left(2+\frac{1}{r}\right)^{-1}<1$ and $|\mathcal{I}_{a}|\leq n$.
\citet[Theorem 1]{merlevede2011bernstein} yields 
\begin{align}
 & \Pr\left\{ \max_{\substack{\substack{j\in[p_{x}],k\in[p_{z}]\\
a\in[d+1],d\in[G-1]
}
}
}|\sum_{t\in\mathcal{I}_{a}}\left(e_{j,t-d}Z_{kt}-\mathbb{E}(e_{j,t-d}Z_{kt})\right)|>\mu\right\} \nonumber \\
\leq & \sum_{j=1}^{p_{x}}\sum_{k=1}^{p_{z}}\sum_{d=1}^{G-1}\sum_{a=1}^{d+1}\Pr\left\{ \sum_{t\in\mathcal{I}_{a}}\left(e_{j,t-d}Z_{kt}-\mathbb{E}(e_{j,t-d}Z_{kt})\right)|>\mu\right\} \nonumber \\
\leq & G^{2}p^{2}\left(n\exp\left(-\dfrac{\mu^{r^{**}}}{C_{1}}\right)+\exp\left(-\dfrac{\mu^{2}}{C_{2}(1+nV)}\right)+\exp\left(-\dfrac{\mu^{2}}{C_{3}n_{x}}\exp\left(\dfrac{\mu^{r^{**}(1-r^{**})}}{C_{4}(\log\mu)^{r^{**}}}\right)\right)\right)\label{eq:probupper}
\end{align}
 for any $\mu>0$, where $C_{1}$, $C_{2}$ and $C_{3}$ are absolute
constants. 

Given $p=O(n^{\nu_{2}})$, we have $(\log\left(p^{2}G^{2}\right))^{2/r^{**}-1}=o(n).$
Specify $\mu=C\sqrt{n\log\left(p^{2}G^{2}\right)}$ with $C^{2}>2(C_{2}+1)V+2C_{3}.$
Similar to the proof of Lemma \ref{lem:BernsteinSum}, (\ref{eq:probupper})
approaches zero as $n\to\infty$. Collecting the stochastic order
of $T_{12}$ and $T_{22}$, we obtain 
\begin{align}
T_{2} & \leq G^{2}\max_{\substack{\substack{j\in[p_{x}],k\in[p_{z}]\\
a\in[d+1],d\in[G-1]
}
}
}|\sum_{t\in\mathcal{I}_{a}}\left(e_{j,t-d}Z_{kt}-\mathbb{E}(e_{j,t-d}Z_{kt})\right)|+O(n)\nonumber \\
 & \stackrel{\mathrm{p}}{\preccurlyeq}G^{2}\sqrt{n\log\left(p^{2}G^{2}\right)}+n=O(n\log p).\label{eq:T2bound}
\end{align}
 \textbf{Bound of $T_{3}.$} Divide $\{1,2,\cdots,n\}$ samples into
$G$ groups. There exists some integers $A\geq1$ and $0\leq B\leq G$
such that $n=(A-1)G+B$. The $g$th group is given by $\mathcal{I}_{g}=\{g,g+G,\cdots,g+(A-1)G\}$
with $|\mathcal{I}_{g}|=A$ for $g\leq B$, and $\mathcal{I}_{g}=\{g,g+G,\cdots,g+(A-2)G\}$
with $|\mathcal{I}_{g}|=A-1$ for $g>B$. We can express 
\[
T_{3}=\sum_{g=1}^{G}\max_{j\in[p_{x}],k\in[p_{z}]}\left|\sum_{t\in\mathcal{I}_{g}}X_{j,t-G}Z_{kt}\right|.
\]
Let $\mathcal{F}_{a}^{b}$ be the $\sigma$-field generated by $(\varepsilon_{t})_{a\leq t\leq b}.$
For simplicity, we use $\mathbb{E}_{t}\left[\cdot\right]$ to denote
conditional expectation $\mathbb{E}\left[\cdot|\mathcal{F}_{-\infty}^{t}\right]$. 

The $\rho$-mixing coefficient defined in (\ref{eq:def_rho}) can
be equivalently written as 
\begin{align}
\rho(\mathcal{A},\mathcal{B}) & =\sup_{\substack{X\in\mathcal{A},\:\mathbb{E}X^{2}<\infty}
}\sqrt{\mathbb{E}\left(\ensuremath{\mathbb{E}(X|\mathcal{B})-\mathbb{E}X}\right)^{2}\big/\mathbb{E}X^{2}}.\label{eq:rho_cond}
\end{align}
For any $t\geq A+1$ and $d\in\mathbb{N}$ we have for any $k\in[p_{z}]$:
\[
\mathbb{E}\left[\left(\mathbb{E}_{t-G}(Z_{kt}^{d})-\mathbb{E}(Z_{kt}^{d})\right)^{2}\right]\big/\mathbb{E}|Z_{kt}^{2d}|\leq\rho^{2}(\varepsilon,G)
\]
by the $\rho$-mixing coefficient defined as (\ref{eq:rho_cond}).
Define $\mathcal{X}_{t}=\{\max_{j\in[p_{x}]}|X_{jt}|\leq C_{X}\sqrt{n\log p}\}$
for some large enough constant $C_{X}>0$, and further define
\[
\mathcal{E}_{t-G}=\left\{ \max_{k\in[p_{z}]}\left\{ \left|\mathbb{E}_{t-G}(Z_{kt}^{d})-\mathbb{E}(Z_{kt}^{d})\right|\big/\sqrt{\mathbb{E}(Z_{kt}^{2d})}\right\} \leq d\sqrt{\rho(\varepsilon,G)}\text{ for all }d>0\right\} ,
\]
and $\mathcal{H}_{t}=\mathcal{X}_{t}\cap\mathcal{E}_{t}$. The fact
${\bf 1}(\mathcal{H}_{t-G})\in\mathcal{F}_{-\infty}^{t-G}$ implies
\begin{align}
 & \mathbb{E}_{t-G}\left(\exp\left[\text{\ensuremath{\tau}}X_{j,t-G}Z_{kt}\right]{\bf 1}(\mathcal{H}_{t-G})\right)\nonumber \\
= & {\bf 1}(\mathcal{H}_{t-G})\mathbb{E}_{t-G}\left(\exp\left[\text{\ensuremath{\tau}}X_{j,t-G}Z_{kt}\right]\right)\nonumber \\
= & {\bf 1}(\mathcal{H}_{t-G})\left(1+\tau X_{j,t-G}\mathbb{E}_{t-G}(Z_{kt})+\sum_{d=2}^{\infty}\dfrac{|\tau X_{t-G}|^{d}}{d!}\mathbb{E}_{t-G}(|Z_{kt}^{d}|)\right)\nonumber \\
\leq & {\bf 1}(\mathcal{H}_{t-G})\left(1+\tau|X_{j,t-G}|\sqrt{\ensuremath{\rho(\varepsilon,G)}\mathbb{E}(Z_{kt}^{2})}+\sum_{d=2}^{\infty}\dfrac{|\tau X_{t-G}|^{d}}{d!}\mathbb{E}_{t-G}\left[|Z_{kt}^{d}|\right]\right).\label{eq:mgfX_0}
\end{align}
By\textbf{ }\citet[Lemma 5]{wong2020lasso} there exists some $K_{z}>0$
such that
\begin{equation}
\max_{k\in[p_{z}]}\mathbb{E}(|Z_{kt}|^{a})\leq K_{z}^{a}a^{a}\leq K_{z}^{a}a^{a}\label{eq:EZbound}
\end{equation}
 for any $a\geq1$. Under the event $\mathcal{E}_{t}$, we have 
\[
|\mathbb{E}_{t-G}(Z_{kt})|=|\mathbb{E}_{t-G}(Z_{kt})-\mathbb{E}(Z_{kt})|\leq\sqrt{\ensuremath{\rho(\varepsilon,G)}\mathbb{E}(Z_{kt}^{2})}
\]
for any $k\in[p_{z}]$, by specifying $d=1$ and using $\mathbb{E}(Z_{kt})=0$.
Thus for any $j\in[p]$, $t>G$ and 
\begin{equation}
\ensuremath{\tau}\in\left(0,\ \big(C_{X}^{2}n\log p\text{\ensuremath{\max_{k\in[p_{z}]}}}\mathbb{E}(Z_{kt}^{2})\big)^{-1/2}\right)\label{eq:tau_range}
\end{equation}
 under the event $\mathcal{H}_{t}$ we bound 
\begin{equation}
\text{\ensuremath{\tau}}|X_{j,t-G}|\sqrt{\ensuremath{\rho(\varepsilon,G)}\mathbb{E}(Z_{kt}^{2})}\leq\sqrt{\ensuremath{\rho(\varepsilon,G)}}=\exp\left(-\log(np)\right)=(np)^{-1}\label{eq:boundlambdaXE}
\end{equation}
and 
\begin{align}
\sum_{d=2}^{\infty}\dfrac{|\tau X_{t-G}|^{d}}{d!}\mathbb{E}_{t-G}(|Z_{kt}^{d}|) & \leq\sum_{d=2}^{\infty}\dfrac{|\tau X_{t-G}|^{d}}{d!}\left(\mathbb{E}|Z_{kt}|^{d}+d\sqrt{\ensuremath{\rho(\varepsilon,G)}\mathbb{E}(Z_{kt}^{2d})}\right)\nonumber \\
 & \leq\sum_{d=2}^{\infty}\dfrac{|\tau X_{t-G}|^{d}}{(d/{\rm e})^{d}}\left(K_{z}^{d}d^{d}+d\sqrt{K_{z}^{2d}(2d)^{2d}}\right)\nonumber \\
 & \leq\sum_{d=2}^{\infty}|{\rm e}\tau X_{t-G}|^{d}\cdot(3K_{z})^{d}\leq2|3K_{z}{\rm e}\tau X_{t-G}|^{2}\label{eq:mgfX_1}
\end{align}
as $\rho(\varepsilon,G)\leq1$ and $d!>(d/{\rm e})^{d}.$ Plug (\ref{eq:boundlambdaXE})
and (\ref{eq:mgfX_1}) into (\ref{eq:mgfX_0}): 
\begin{eqnarray}
\mathbb{E}_{t-G}\left(\exp\left[\text{\ensuremath{\tau}}X_{j,t-G}Z_{kt}\right]{\bf 1}(\mathcal{H}_{t-G})\right) & \leq & {\bf 1}(\mathcal{X}_{t-G})\left(1+(np)^{-1}+2|3K_{z}{\rm e}\tau X_{t-G}|^{2}\right)\nonumber \\
 & \leq & \text{\ensuremath{\left(1+(np)^{-1}\right)}}\exp\left[C\tau^{2}n\log p\right]\label{eq:mgfX}
\end{eqnarray}
with $C=18\left({\rm e}K_{z}C_{X}\right)^{2}$.

For any $g\in\{2,\cdots,G\}$, 
\begin{align*}
 & \Pr\left\{ \sum_{t\in\mathcal{I}_{g}}X_{j,t-G}Z_{kt}>\mu,\bigcap_{t\in\mathcal{I}_{g}}\mathcal{H}_{t-G}\right\} =\Pr\left\{ \exp\left[\tau\sum_{t\in\mathcal{I}_{g}}X_{j,t-G}Z_{kt}\right]>e^{\mu\tau},\bigcap_{t\in\mathcal{I}_{g}}\mathcal{H}_{t-G}\right\} \\
= & \Pr\left\{ \exp\left[\tau\sum_{s=1}^{|\mathcal{I}_{g}|}X_{j,g+(s-2)G}u_{g+(s-1)G}\right]>e^{\mu\tau},\bigcap_{s=1}^{|\mathcal{I}_{g}|}\mathcal{H}_{g+(s-2)-G}\right\} \\
= & \Pr\left\{ \exp\left[\tau\sum_{s=1}^{|\mathcal{I}_{g}|}X_{j,g+(s-2)G}u_{g+(s-1)G}\right]{\bf 1}(\bigcap_{s=1}^{|\mathcal{I}_{g}|}\mathcal{H}_{g+(s-2)-G})>e^{\mu\tau}\right\} .
\end{align*}
The Markov inequality implies that the above probability is bounded
by 
\begin{align*}
 & e^{-\mu\tau}\mathbb{E}\left[\exp\left[\tau\sum_{s=1}^{|\mathcal{I}_{g}|}X_{j,g+(s-2)G}u_{g+(s-1)G}\right]\prod_{s=1}^{|\mathcal{I}_{g}|}\boldsymbol{1}(\mathcal{H}_{g+(s-2)-G})\right]\\
= & e^{-\mu\tau}\mathbb{E}\left[\prod_{s=1}^{|\mathcal{I}_{g}|}\left(\exp\left[\tau X_{j,g+(s-2)G}u_{g+(s-1)G}\right]\cdot\boldsymbol{1}(\mathcal{H}_{g+(s-2)-G})\right)\right]\\
\leq & e^{-\mu\tau}\mathbb{E}\Bigg[\mathbb{E}_{g+(|\mathcal{I}_{g}|-2)-G}\left(\exp\left[\tau X_{j,g+(|\mathcal{I}_{g}|-2)G}u_{g+(|\mathcal{I}_{g}|-1)G}\right]\cdot\boldsymbol{1}(\mathcal{H}_{g+(|\mathcal{I}_{g}|-2)-G})\right)\\
 & \ \ \ \ \ \ \ \ \ \ \ \ \ \ \ \times\prod_{s=1}^{|\mathcal{I}_{g}|-1}\exp\left[\tau X_{j,g+(s-2)G}u_{g+(s-1)G}\right]\cdot\boldsymbol{1}(\mathcal{H}_{g+(s-2)-G})\Bigg]\\
\leq & e^{-\mu\tau}\left(1+(np)^{-1}\right)\exp(C\tau^{2}n\log p)\cdot\mathbb{E}\Bigg[\prod_{s=1}^{|\mathcal{I}_{g}|-1}\exp\left[\tau X_{j,g+(s-2)G}u_{g+(s-1)G}\right]\cdot\boldsymbol{1}(\mathcal{H}_{g+(s-2)-G})\Bigg].
\end{align*}
By induction, 
\begin{align*}
\Pr\left\{ \sum_{t\in\mathcal{I}_{g}}X_{j,t-G}Z_{kt}>\mu,\bigcap_{t\in\mathcal{I}_{g}}\mathcal{H}_{t-G}\right\}  & \leq\left(1+(np)^{-1}\right)^{A}\exp\left[-\mu\tau+A\cdot C\tau^{2}\cdot n\log p\right]\\
 & \leq2\exp\left[-\mu\tau+2C\tau^{2}\log p\cdot n^{2}G^{-1}\right]
\end{align*}
where the last inequality applies 
\[
\left(1+(np)^{-1}\right)^{A}\leq\left(1+(np)^{-1}\right)^{n}\leq2
\]
 and $A\leq nG^{-1}+1\leq2nG^{-1}$ with $n$ sufficiently large.
Let $\mu=4\sqrt{C}n(\log p+\log G)\cdot G^{-1/2}$ and 
\[
\tau=\dfrac{\mu G}{4Cn^{2}\log p}=\dfrac{\sqrt{C}(\log p+\log n)\sqrt{G}}{\sqrt{2}n\log p}.
\]
 When $n$ is sufficiently large $\tau$ falls into the interval of
(\ref{eq:tau_range}) and hence (\ref{eq:mgfX}) holds. Repeating
this argument for $-X_{j,t-G}Z_{kt}$, we obtain the same bound for
$\Pr\left\{ \sum_{t\in\mathcal{I}_{g}}X_{j,t-G}Z_{kt}<-\mu,\,\bigcap_{t\in\mathcal{I}_{g}}\mathcal{H}_{t-G}\right\} $.
 Therefore, we have 
\begin{align*}
 & \Pr\left\{ \sum_{g=1}^{G}\max_{j\in[p_{x}],k\in[p_{z}]}\left|\sum_{t\in\mathcal{I}_{g}}X_{j,t-G}Z_{kt}\right|>G\mu,\bigcap_{t=G+1}^{n}\mathcal{H}_{t-G}\right\} \\
\leq & \sum_{g=1}^{G}\sum_{j=1}^{p_{x}}\sum_{k=1}^{p_{z}}\Pr\left\{ \left|\sum_{t\in\mathcal{I}_{g}}X_{j,t-G}Z_{kt}\right|>\mu,\bigcap_{t\in\mathcal{I}_{g}}\mathcal{H}_{t-G}\right\} \\
\leq & 2\sum_{g=1}^{G}\sum_{j=1}^{p_{x}}\sum_{k=1}^{p_{z}}\Pr\left\{ \sum_{t\in\mathcal{I}_{g}}X_{j,t-G}Z_{kt}>\mu,\bigcap_{t\in\mathcal{I}_{g}}\mathcal{H}_{t-G}\right\} \\
\leq & 4Gp^{2}\exp\left[-\mu\tau+2C\tau^{2}\log p\cdot n^{2}G^{-1}\right]\\
= & 4Gp^{2}\exp\left[-\dfrac{\mu^{2}G}{8Cn^{2}\log p}\right]=4Gp^{2}\exp\left[-2(\log p+\log G)\right]
\end{align*}
and it follows that 
\begin{eqnarray}
 &  & \Pr\left\{ T_{3}>\sqrt{6C}n(\log p+\log G)\text{\ensuremath{\sqrt{G}}}\right\} \nonumber \\
 & \leq & \Pr\left\{ \sum_{g=1}^{G}\max_{j\in[p_{x}],k\in[p_{z}]}\left|\sum_{t\in\mathcal{I}_{g}}X_{j,t-G}Z_{kt}\right|>G\mu,\bigcap_{t\in G+1}^{n}\mathcal{H}_{t-G}\right\} +\Pr\left\{ \bigcup_{t=G+1}^{n}\mathcal{H}_{t-G}^{c}\right\} \nonumber \\
 & \leq & 4Gp^{2}\exp\left[-2(\log p+\log G)\right]+\Pr\left\{ \bigcup_{t=G+1}^{n}\mathcal{H}_{t-G}^{c}\right\} .\label{eq:P_T3}
\end{eqnarray}
Because the Chebyshev inequality and the union bound imply 
\begin{align}
\Pr\left\{ \bigcup_{t=A+1}^{n}\left(\mathcal{E}_{t-G}\right)^{c}\right\}  & \leq\sum_{t=A+1}^{n}\Pr\left\{ \max_{k\in[p_{z}]}\dfrac{\left|\mathbb{E}_{t-G}(Z_{kt}^{d})-\mathbb{E}(Z_{kt}^{d})\right|}{\sqrt{\mathbb{E}|Z_{kt}^{2d}|}}>d\text{\ensuremath{\sqrt{\rho(\varepsilon,G)}}}\text{ for some }d>0\right\} \nonumber \\
 & \leq\sum_{t=A+1}^{n}\sum_{k=1}^{p_{z}}\sum_{d=1}^{\infty}\Pr\left\{ \dfrac{\left|\mathbb{E}_{t-G}(Z_{kt}^{d})-\mathbb{E}(Z_{kt}^{d})\right|}{\sqrt{\mathbb{E}(Z_{kt}^{2d})}}>d\sqrt{\rho(\varepsilon,G)}\right\} \nonumber \\
 & \leq O(np)\sum_{d=1}^{\infty}\dfrac{\rho(\varepsilon,G)}{d^{2}}=O(np)\exp\left(-2\log(np)\right)\to0\label{eq:eventE}
\end{align}
given that the series $\sum_{d=1}^{\infty}d^{-2}$ is convergent,
and by Lemmas \ref{lem:mixing} and \ref{lem:BernsteinSum} $\Pr\left\{ \bigcup_{t=1}^{n}\mathcal{X}_{t}^{c}\right\} \to0$,
we have $\Pr\left\{ \bigcup_{t=1}^{n}\mathcal{H}_{t}^{c}\right\} =o(1)$
as well. It follows by (\ref{eq:P_T3}) that 
\[
T_{3}\stackrel{\mathrm{p}}{\preccurlyeq}n(\log p+\log G)\text{\ensuremath{\sqrt{G}}}=O\left(n(\log p)^{1+\frac{1}{2r}}\right)
\]
 in view of $G=O\left((\log(np))^{1/r}\right)=O\left((\log p)^{1/r}\right)$
where the second step applies Assumption \ref{assu:asym_n}.

Collecting the stochastic order of $T_{1}$, $T_{2}$ and $T_{3}$,
we complete the proof. 
\end{proof}
\begin{rem}
\label{rem:DB-demean}We can deduce similar upper bounds for the demeaned
variables. First, by the triangular inequality the demeaned cross
product is bounded by 
\[
\dfrac{1}{n^{3/2}}\left|\sum_{t=1}^{n}\ddot{X}_{j,t-1}Z_{kt}\right|\leq n^{-3/2}\left|\sum_{t=1}^{n}X_{j,t-1}Z_{kt}\right|+n^{-3/2}\left|\bar{X_{j}}\right|\left|\sum_{t=1}^{n}Z_{kt}\right|.
\]
 Lemmas \ref{lem:mixing} and \ref{lem:BernsteinSum} implies 
\begin{align}
\max_{k\in[p_{z}]}\left|\sum_{t=1}^{n}Z_{kt}\right| & \stackrel{\mathrm{p}}{\preccurlyeq}\sqrt{n\log p}\label{eq:boundSumZ}\\
\max_{j\in[p_{x}]}\left|\bar{X_{j}}\right| & \leq\max_{j\in[p_{x}],t\in[n]}|X_{j,t-1}|=\max_{j\in[p_{x}],t\in[n]}|\sum_{s=0}^{t-1}e_{js}|\stackrel{\mathrm{p}}{\preccurlyeq}\sqrt{n\log p}\label{eq:boundSumX}
\end{align}
 w.p.a.1. For all $j$ and $k$ we have 
\[
n^{-3/2}\left|\bar{X_{j}}\right|\left|\sum_{t=1}^{n}Z_{kt}\right|\stackrel{\mathrm{p}}{\preccurlyeq}\dfrac{\log p}{\sqrt{n}}=o\left(\dfrac{1}{\sqrt{n}}(\log p)^{1+\frac{1}{2r}}\right)
\]
with $n$ large enough given that $r>0.$ The same argument applies
to the cross product involving $u_{t}.$ As a result, we have 
\begin{equation}
\max_{j\in[p_{x}],k\in[p_{z}]}\dfrac{1}{n^{3/2}}\left|\sum_{t=1}^{n}\ddot{X}_{j,t-1}Z_{kt}\right|+\max_{j\in[p_{x}]}\dfrac{1}{n^{3/2}}\left|\sum_{t=1}^{n}\ddot{X}_{j,t-1}u_{t}\right|\stackrel{\mathrm{p}}{\preccurlyeq}\dfrac{1}{\sqrt{n}}(\log p)^{1+\frac{1}{2r}}.\label{eq:DB-All-demean}
\end{equation}
\end{rem}
Proposition \ref{prop:DB-All} shows the small order of the interaction
terms between the stationary and nonstdationary components. The upper
bound of the first term of (\ref{eq:DB-All-demean}) is used for the
RE condition of mixed regressors, while the second term is for the
DB condition of the unit root regressors. For simplicity of the notations
in the proofs, we use the lagged time subscript $t-1$ for both terms.
It is trivial to handle the case with the same time subscript for
both $X$ and $Z$ in the RE condition. Recall that $X_{jt}=X_{j,t-1}+e_{jt}$
and hence $n^{-3/2}\left|\sum_{t=1}^{n}X_{jt}Z_{kt}\right|\leq n^{-3/2}\left|\sum_{t=1}^{n}X_{j,t-1}Z_{kt}\right|+n^{-3/2}\left|\sum_{t=1}^{n}e_{jt}Z_{kt}\right|$.
It is easy to show the small order of the second term in the upper
bound given that both $e_{t}$ and $Z_{t}$ are stationary. See the
proof of Proposition \ref{prop:MixRE-cn} for more details. 

\subsubsection{DB of Stationary Components }

Proposition \ref{prop:DB-Ze} gives the DB condition for the stationary
components. 
\begin{prop}
\label{prop:DB-Ze}Suppose that Assumptions \ref{assu:tail}-\ref{assu:covMat},
\ref{assu:asym_n}(a) and \ref{assu:EZu} hold, and $v_{t}=(e_{t}^{\top},Z_{t}^{\top},u_{t})^{\top}$
follows (\ref{eq:def-I0-mix}). Then
\begin{gather}
\max_{k\in[p_{z}]}\max_{j\in[p_{x}]}\left|n^{-1}\sum_{t=1}^{n}Z_{k,t-1}e_{j,t-1}\right|\stackrel{\mathrm{p}}{\preccurlyeq}1\label{eq:ZeDB}\\
\max_{k\in[p_{z}]}\left|n^{-1}\sum_{t=1}^{n}\ddot{Z}_{k,t-1}u_{t}\right|\stackrel{\mathrm{p}}{\preccurlyeq}\sqrt{\dfrac{\log p}{n}}\label{eq:ZuDB}\\
\|\overline{\Sigma}^{(z)}-\Sigma^{(z)}\|_{\max}\stackrel{\mathrm{p}}{\preccurlyeq}\sqrt{\dfrac{\log p}{n}}\label{eq:maxSqZ}
\end{gather}
where $\overline{\Sigma}^{(z)}:=n^{-1}\sum_{t=1}^{n}Z_{t-1}Z_{t-1}^{\top}$
and $\Sigma^{(z)}=\mathbb{E}[Z_{t-1}Z_{t-1}^{\top}]$. 
\end{prop}
\begin{proof}[Proof of Proposition \ref{prop:DB-Ze}]

To prove (\ref{eq:ZeDB}), by (\ref{eq:EZbound}) $\mathbb{E}[Z_{k,t-1}^{2}]$
is uniformly bounded for all $k$, and so is $\mathbb{E}[e_{j,t-1}^{2}]$
for all $j$. As a result, $\left|\mathbb{E}[Z_{k,t-1}e_{j,t-1}]\right|\leq\sqrt{\mathbb{E}[Z_{k,t-1}^{2}]\mathbb{E}[e_{j,t-1}^{2}]}$
is uniformly bounded for all $k$ and $j$. By Lemma \ref{lem:mixing},
$(Z_{k,t-1}e_{j,t-1})_{t\geq1}$ is strong mixing with an $\alpha$-mixing
coefficient bounded by $\rho(\varepsilon,d)\leq C_{\alpha}\exp\left(-c_{\alpha}d{}^{r}\right)$,
and $(Z_{k,t-1}u_{t})_{t\geq1}$ is strong mixing with an $\alpha$-mixing
coefficient 
\[
\rho(\varepsilon,d-1)\leq C_{\alpha}\exp\left(-c_{\alpha}(d-1)^{r}\right)\leq C_{\alpha}\exp\left(-0.5c_{\alpha}d{}^{r}\right)
\]
 with sufficiently large $d.$ Besides, for any $\mu>0$
\begin{align*}
\Pr\left(|Z_{k,t-1}e_{j,t-1}|>\mu\right) & \leq\Pr\left(|e_{j,t}|>\sqrt{\mu}\right)+\Pr\left(|Z_{k,t-1}|>\sqrt{\mu}\right)\\
 & \leq2C_{v}\exp\left[-\sqrt{\mu}/b_{v}\right]=2C_{v}\exp\left[-\left(\mu/b_{v}^{2}\right)^{1/2}\right]
\end{align*}
where the second inequality applies (\ref{eq:euBound}) in Corollary
\ref{cor:BNtail}. It follows from the proof of (\ref{eq:sumProbBound})
in Lemma \ref{lem:BernsteinSum} that 
\begin{equation}
\max_{j\in[p_{x}],k\in[p_{z}]}\left|n^{-1}\sum_{t=1}^{n}\left(Z_{k,t-1}e_{j,t-1}-\mathbb{E}\left[Z_{k,t-1}e_{j,t-1}\right]\right)\right|\stackrel{\mathrm{p}}{\preccurlyeq}\sqrt{\dfrac{\log p}{n}}.\label{eq:Zelesprop}
\end{equation}
Then 
\begin{eqnarray*}
 &  & \max_{k\in[p_{z}]}\max_{j\in[p_{x}]}\left|n^{-1}\sum_{t=1}^{n}Z_{k,t-1}e_{j,t-1}\right|\\
 & \leq & \max_{j\in[p_{x}],k\in[p_{z}]}\left|n^{-1}\sum_{t=1}^{n}\left(Z_{k,t-1}e_{j,t-1}-\mathbb{E}\left[Z_{k,t-1}e_{j,t-1}\right]\right)\right|+\max_{j\in[p_{x}],k\in[p_{z}]}\left|\mathbb{E}\left[Z_{k,t-1}e_{j,t-1}\right]\right|\\
 & \stackrel{\mathrm{p}}{\preccurlyeq} & \sqrt{\dfrac{\log p}{n}}+1=O\left(1\right).
\end{eqnarray*}
The same bound applies to $\max_{k\text{\ensuremath{\in}[\ensuremath{p_{z}}]}}\left|n^{-1}\sum_{t=1}^{n}Z_{k,t-1}u_{t}\right|$
as $\mathbb{E}[Z_{k,t-1}u_{t}]=0$. Then (\ref{eq:ZeDB}) follows
by 
\[
\left|n^{-1}\sum_{t=1}^{n}\ddot{Z}_{k,t-1}u_{t}\right|\leq\left|n^{-1}\sum_{t=1}^{n}Z_{k,t-1}u_{t}\right|+\left|\bar{Z}_{k}\bar{u}\right|\stackrel{\mathrm{p}}{\preccurlyeq}\sqrt{\frac{\log p}{n}}+\frac{\log p}{n}=O\left(\sqrt{\frac{\log p}{n}}\right).
\]

To prove (\ref{eq:maxSqZ}), notice $\{Z_{k,t-1}Z_{m,t-1}\}_{t\geq1}$
is $\alpha$-mixing with 
\[
\Pr\left(Z_{k,t-1}Z_{m,t-1}-\mathbb{E}[Z_{k,t-1}Z_{m,t-1}]>\mu\right)\leq C_{v}\exp\left[-\left(\mu/b_{v}^{2}\right)^{1/2}\right].
\]
 Furthermore, by normalization $\psi_{j0}=1$ we have for all $j\in[p]$
that $\mathbb{E}\varepsilon_{jt}^{2}=\sum_{d=0}^{\infty}\psi_{jd}^{2}\geq\psi_{j0}^{2}=1$
and by Assumption \ref{assu:alpha} and (\ref{eq:gene_ineq1}) we
have 
\[
\mathbb{E}\left[\varepsilon_{jt}^{2}\right]=\sum_{d=0}^{\infty}\psi_{jd}^{2}\leq C_{\psi}^{2}\sum_{d=0}^{\infty}\exp\left(-2c_{\psi}d^{r}\right)\leq C_{\psi}^{2}\left[M+\sum_{d=M}^{\infty}\exp\left(-2c_{\psi}d^{r}\right)\right]\leq\dfrac{C_{\psi}^{2}}{2c_{\psi}}\exp\left(-c_{\psi}C^{r}\right)
\]
with $C$ large enough so that $C^{\frac{1}{r}-1}\leq\exp\left(c_{\psi}C\right)$. 

(\ref{eq:Zelesprop}) follows by similar procedures as for (\ref{eq:maxSqZ}).
\end{proof}

\subsubsection{RE for Demeaned Unit Roots}

For any square matrix $A$, define 
\begin{equation}
\phi_{\min}(A,s+m):=\inf_{\delta\in\mathbb{R}^{p},\|\delta\|_{0}\leq s+m}\dfrac{\delta^{\top}A\delta}{\delta^{\top}\delta},\ \ \phi_{\max}(A,s+m):=\sup_{\delta\in\mathbb{R}^{p},\|\delta\|_{0}\leq s+m}\dfrac{\delta^{\top}A\delta}{\delta^{\top}\delta}.\label{eq:def-phi-minmax}
\end{equation}
Let 
\begin{equation}
C_{m}=C_{m}(L):=\left\lceil 4L^{2}\widetilde{C}/\widetilde{c}\right\rceil \label{eq:Cm-def}
\end{equation}
for some $L\geq1,$ where $\widetilde{c}=0.5(1-\sqrt{1/2})^{2}c_{\Omega}$
and $\widetilde{C}=2(1+\sqrt{1/2})^{2}C_{\Omega}$ for $c_{\Omega}$
and $C_{\Omega}$ in Assumption \ref{assu:covMat}. Define 
\begin{equation}
m:=C_{m}s.\label{eq:m-def}
\end{equation}
Furthermore, recall that $v_{t}=(e_{t}^{\top},u_{t})^{\top}=\Phi\varepsilon_{t}$
as defined in (\ref{eq:def-error}). Let $\Phi=(\Phi_{e}^{\top},\Phi_{u}^{\top})^{\top}$
with $\Phi_{e}$ being $p\times(p+1)$ and $\Phi_{u}$ being $1\times(p+1).$
We consider the pure unit root case in Section \ref{sec:UnitRoot}.
We have the following RE for $\text{\ensuremath{\widehat{\Sigma}}}$,
under normality as in Part (a) and non-normal innovation in Part (b). 
\begin{prop}
\label{prop:UnitRE-cn} Suppose that $(1+C_{m}(L))s=o(n\wedge p)$
as $n\to\infty$.

(a) If Assumption \ref{assu:covMat} (a) holds and $\varepsilon_{t}\sim i.i.d.\ \mathcal{N}(0,I_{p})$,
then there exists some absolute constant $\widetilde{c}_{\kappa}$
such that 
\begin{equation}
\frac{\kappa_{I}(\hat{\Sigma},L,s)}{n}\geq\dfrac{\widetilde{c}_{\kappa}}{L^{2}s\log p}\label{eq:RE-unit-1}
\end{equation}
holds w.p.a.1.~for any $L\geq1$.\footnote{Here we use a generic $L\geq1$ to unify the proofs. Plasso applies
this result with $L=3$. }

(b) If Assumptions \ref{assu:tail}-\ref{assu:covMat} holds and in
addition $s^{2}L^{4}(\log p)^{5/2}=o(n^{1/2})$, then (\ref{eq:RE-unit-1})
is satisfied w.p.a.1. 
\end{prop}
\begin{proof}[Proof of Proposition \ref{prop:UnitRE-cn}]
 \textbf{Part (a).} The normal distribution $\varepsilon_{t}\sim i.i.d.\ \mathcal{N}(0,I_{p})$
implies $e_{t}\sim i.i.d.\,\mathcal{N}(0,\Omega_{e})$ with $\Omega_{e}=\Phi_{e}\Phi_{e}^{\top}.$
 Let $R$ be an $n\times n$ lower triangular matrix of ones on and
below the diagonal, and $J_{n}=n^{-1}1_{n}1_{n}^{\top}.$ Note that
$\underset{(n\times p)}{X}=\underset{(n\times n)}{R}\underset{(n\times p)}{e}$
with $e=(e_{0},e_{1},\cdots,e_{n-1})^{\top}$, we write 
\[
\text{\ensuremath{\widehat{\Sigma}}}=n^{-1}X^{\top}(I_{n}-J_{n})^{2}X=n^{-1}e^{\top}R^{\top}(I_{n}-J_{n})Re.
\]
Let $\lambda_{1}\geq\lambda_{2}\geq\cdots\geq\lambda_{n}\geq0$ and
$\widetilde{\lambda}_{1}\geq\widetilde{\lambda}_{2}\geq\cdots\geq\widetilde{\lambda}_{n}\geq0$
be the eigenvalues of $R^{\top}(I_{n}-J_{n})R$ and $R^{\top}R$,
respectively, ordered from large to small. 

Let $\mu_{\ell}$ be the $\ell$th largest singular value of the idempotent
matrix $I_{n}-J_{n}$. Recall $\boldsymbol{1}\left(\cdot\right)$
is the indicator function, and obviously $\mu_{\ell}=\boldsymbol{1}(1\leq\ell\leq n-1)$
for $\ell\in[n]$. Denote the $\ell$th eigenvalue values of $R^{\top}(I_{n}-J_{n})R$
and $R^{\top}R$ be $\lambda_{\ell}$ and $\widetilde{\lambda}_{\ell}$,
respectively. When $\ell\in[n-1]$, the first inequality of Eq.(15)
in \citet[Theorem 9]{merikoski2004inequalities} gives $\lambda_{\ell}\geq\widetilde{\lambda}_{\ell+1}\mu_{n-1}=\widetilde{\lambda}_{\ell+1}$. 

Following the technique used to prove Remark 3.5 in \citet{zhang2019identifying},
which is also used for Theorem B.2 in \citet{smeekes2021automated},
we diagonalize $R(I_{n}-J_{n})R^{\top}=V{\rm diag}(\lambda_{1},\lambda_{2},\cdots,\lambda_{n})V^{\top}$,
where $V$ is an orthonormal matrix. For any $\delta\in\mathbb{R}^{p}$,
$\delta\neq0$, the quadratic form 
\begin{align}
\delta^{\top}\text{\ensuremath{\widehat{\Sigma}}}\delta & =\dfrac{1}{n}e^{\top}R^{\top}(I_{n}-J_{n})Re=\frac{1}{n}\delta^{\top}e^{\top}V{\rm diag}(\lambda_{1},\lambda_{2},\cdots,\lambda_{n})V^{\top}e\delta\nonumber \\
 & \geq\frac{1}{n}\delta^{\top}e^{\top}V_{\cdot[\ell]}{\rm diag}(\lambda_{1},\cdots,\lambda_{\ell})V_{\cdot[\ell]}^{\top}e\delta\geq\frac{\lambda_{\ell}}{n}\delta^{\top}e^{\top}V_{\cdot[\ell]}V_{\cdot[\ell]}^{\top}e\delta\nonumber \\
 & \geq\frac{\ell\widetilde{\lambda}_{\ell+1}}{n}\cdot\delta^{\top}\Gamma_{\ell}\delta\label{eq:LBeigen}
\end{align}
for any $\ell\in[n-1]$, where $V_{\cdot[\ell]}$ is the submatrix
composed of the first $\ell$ columns of $V$ and $\Gamma_{\ell}=\ell^{-1}e^{\top}V_{\cdot[\ell]}V_{\cdot[\ell]}^{\top}e$. 

We first work with the first factor $\ell\widetilde{\lambda}_{\ell+1}/n$
in (\ref{eq:LBeigen}). \citet{smeekes2021automated} provide the
exact formula of $\widetilde{\lambda}_{\ell}$: 
\begin{equation}
\widetilde{\lambda}_{\ell}=\left[2\left(1-\cos\left(\dfrac{(2\ell-1)\pi}{2n+1}\right)\right)\right]^{-1}\text{ for all }\ell\in[n].\label{eq:lambda_exact}
\end{equation}
 A Taylor expansion of $\cos\left(x\pi\right)$ around $x=0$ yields
\[
\widetilde{\lambda}_{\ell+1}^{-1}=2\left(1-\cos\left(\dfrac{(2\ell+1)\pi}{2n+1}\right)\right)=\left(\dfrac{(2\ell+1)\pi}{2n+1}\right)^{2}\left(1+o\left(\frac{\ell}{n}\right)\right)=\left(\dfrac{\ell\pi}{n}\right)^{2}\left(1+o\left(\frac{\ell}{n}\right)\right)
\]
whenever $\ell=o\left(n\right)$. This implies 
\begin{equation}
\dfrac{\ell\widetilde{\lambda}_{\ell+1}}{n}=\dfrac{n}{\pi^{2}\ell\left(1+o\left(\ell/n\right)\right)}\geq\dfrac{n}{2\pi^{2}\ell}\label{eq:eigen_LB}
\end{equation}
for $\ell=o\left(n\right)$ when $n$ is sufficiently large. 

Next, we focus on the second factor $\delta^{\top}\Gamma_{\ell}\delta$
in (\ref{eq:LBeigen}). For any $\mathcal{M}\subseteq[p]$, the submatrix
of $\Gamma_{\ell}$ with the rows and columns indexed by $\mathcal{M}$
is 
\[
\Gamma_{\ell}(\mathcal{M}):=(\Gamma_{\ell,ij})_{i,j\in\mathcal{M}}=\dfrac{1}{\ell}e_{\cdot\mathcal{M}}^{\top}V_{\cdot[\ell]}V_{\cdot[\ell]}^{\top}e_{\cdot\mathcal{M}}\sim\dfrac{1}{\ell}W_{p}\left(\Omega_{e}(\mathcal{M}),\ell\right),
\]
following a Wishart distribution $\mathcal{W}_{p}\left(\Omega_{e}(\mathcal{M}),\ell\right)$
divided by $\ell$. There are as many as 
\[
K=\begin{pmatrix}p\\
s+m
\end{pmatrix}\leq p^{s+m}
\]
 submatrices $\mathcal{M}$ of the dimension $(s+m)\times(s+m)$ for
$\Gamma_{\ell}(\mathcal{M})$. Index these matrices by $k=1,\ldots,K$
and denote them as $\Gamma_{\ell}(\mathcal{M}_{k})$. 

To establish uniformity over all $\mathcal{M}_{k}$, we invoke Theorem
6.1 of \citet{Wainwright2019high}: for all $c\in(0,1)$ and $k\in[K]$,
we have the non-asymptotic deviation bounds for Wishart random matrices:
\begin{align*}
\Pr\left\{ \sqrt{\lambda_{\max}\left(\Gamma_{\ell}(\mathcal{M}_{k})\right)}\geq\sqrt{\lambda_{\max}(\Omega_{e}(\mathcal{M}_{k}))}\cdot(1+c)+\sqrt{\dfrac{\text{tr}(\Omega_{e}(\mathcal{M}_{k}))}{\ell}}\right\}  & \leq\exp(-\ell c^{2}/2)\\
\Pr\left\{ \sqrt{\lambda_{\min}\left(\Gamma_{\ell}(\mathcal{M}_{k})\right)}\leq\sqrt{\lambda_{\min}(\Omega_{e}(\mathcal{M}_{k}))}\cdot(1-c)-\sqrt{\dfrac{\text{tr}(\Omega_{e}(\mathcal{M}_{k}))}{\ell}}\right\}  & \leq\exp\left(-\ell c^{2}/2\right).
\end{align*}
Since $c_{\Omega}\leq\lambda_{\min}(\Omega_{e}(\mathcal{M}_{k}))\leq\lambda_{\max}(\Omega_{e}(\mathcal{M}_{k}))\leq C_{\Omega}$
and $c_{\Omega}(s+m)\leq\text{tr}(\Omega(\mathcal{M}_{k}))\leq C_{\Omega}(s+m)$
for all $k\in[K]$ in our context, we bound 
\begin{align*}
 & \Pr\left\{ \sqrt{\phi_{\max}(\Gamma_{\ell},s+m)}\geq\sqrt{C_{\Omega}}\left(1+c\right)+\sqrt{C_{\Omega}}\sqrt{\dfrac{s+m}{\ell}}\right\} \\
\leq & \sum_{k=1}^{K}\Pr\left\{ \sqrt{\lambda_{\max}\left(\Gamma_{\ell}(\mathcal{M}_{k})\right)}\geq\sqrt{\lambda_{\max}(\Omega_{e}(\mathcal{M}_{k}))}\cdot(1+c)+\sqrt{\dfrac{\text{tr}(\Omega(\mathcal{M}_{k}))}{\ell}}\right\} \\
\leq & K\exp\left(-\ell c^{2}/2\right)\leq p^{s+m}\exp\left(-\ell c^{2}/2\right)
\end{align*}
and similarly 
\begin{align*}
 & \Pr\left\{ \sqrt{\phi_{\min}(\Gamma_{\ell},s+m)}\leq\sqrt{c_{\Omega}}\left(1-c\right)-\sqrt{C_{\Omega}}\sqrt{\dfrac{s+m}{\ell}}\right\} \\
\leq & \sum_{k=1}^{K}\Pr\left\{ \sqrt{\lambda_{\min}\left(\Gamma_{\ell}(\mathcal{M}_{k})\right)}\leq\sqrt{\lambda_{\min}(\Omega_{e}(\mathcal{M}_{k}))}\cdot(1-c)-\sqrt{\dfrac{\text{tr}(\Omega_{e}(\mathcal{M}_{k}))}{\ell}}\right\} \\
\leq & K\exp\left(-\ell c^{2}/2\right)\leq p^{s+m}\exp\left(-\ell c^{2}/2\right).
\end{align*}

Let $\ell=16(s+m)\log p$ and $c=0.5$. When $p$ is sufficiently
large, 
\begin{align*}
\sqrt{c_{\Omega}}(1-c)-\sqrt{C_{\Omega}}\sqrt{\dfrac{s+m}{\ell}} & =(1-0.5)\sqrt{c_{\Omega}}-\dfrac{\sqrt{C_{\Omega}}}{\sqrt{\log p}}>0.4\sqrt{c_{\Omega}}\\
\sqrt{C_{\Omega}}\left(1+c\right)+\sqrt{C_{\Omega}}\sqrt{\dfrac{s+m}{\ell}} & =(1+0.5)\sqrt{C_{\Omega}}+\dfrac{\sqrt{C_{\Omega}}}{\sqrt{\log p}}<1.6\sqrt{c_{\Omega}}.
\end{align*}
These two inequalities give us 
\begin{eqnarray*}
 &  & \Pr\left\{ \left\{ \phi_{\min}(\Gamma_{\ell},s+m)\leq0.16c_{\Omega}\right\} \cup\left\{ \phi_{\max}(\Gamma_{\ell},s+m)\geq2.56C_{\Omega}\right\} \right\} \\
 & \leq & 2p^{s+m}\exp\left(-\ell\cdot0.5^{2}/2\right)=2p^{s+m}\exp\left(-2(s+m)\log p\right)=2p^{-(s+m)}\to0.
\end{eqnarray*}
In other words, 
\[
0.16c_{\Omega}=\tilde{c}\leq\phi_{\min}(\Gamma_{\ell},s+m)\leq\phi_{\max}(\Gamma_{\ell},s+m)\leq\tilde{C}=2.56C_{\Omega}
\]
 holds w.p.a.1. As a result, 
\begin{equation}
m\phi_{\min}(\Gamma_{\ell},s+m)\geq m\widetilde{c}=C_{m}s\widetilde{c}\geq4L^{2}s\widetilde{C}>4L^{2}s\phi_{\max}(\Gamma_{\ell},m)\label{eq:Bickel_assu_2}
\end{equation}
holds w.p.a.1.~as well. Under the condition $s+m=(1+C_{m})s=o(p)$,
the inequality (\ref{eq:Bickel_assu_2}) verifies \citet{bickel2009simultaneous}'s
Assumption 2 $m\phi_{\min}(s+m)>L^{2}s\phi_{\max}(m)$ w.p.a.1.. Let
$\mathcal{S}_{01}=\mathcal{S}_{0}\cup\mathcal{S}_{1}$ where $\mathcal{S}_{0}=\{j\in[p]:\delta_{j}\neq0\}$
and $\mathcal{S}_{1}\subset[p]$ is another index set corresponding
to the $m$ largest (in terms of absolute value) coordinates of $\delta$
outside of $\mathcal{S}_{0}$. Let $P_{01}$ be the projection matrix
that maps any $p\times1$ vector onto the linear space spanned by
the columns of $V_{\cdot[\ell]}^{\top}e$ indexed by the set $\mathcal{S}_{01}$,
i.e.~$(V_{\cdot[\ell]}^{\top}e)_{\cdot\mathcal{S}_{01}}$. We have
\[
\delta^{\top}\Gamma_{\ell}\delta=\dfrac{1}{\ell}\|V_{\cdot[\ell]}^{\top}e\delta\|_{2}^{2}\geq\dfrac{1}{\ell}\|P_{01}V_{\cdot[\ell]}^{\top}e\delta\|_{2}^{2}.
\]
For all $\delta\in\mathcal{R}(L,s)$\textbf{ }in the restricted set
defined below (\ref{eq:RE}), we have
\begin{equation}
\dfrac{1}{\ell}\|P_{01}V_{\cdot[\ell]}^{\top}e\delta\|_{2}^{2}\geq\tilde{\phi}^{2}\|\delta_{\mathcal{S}_{01}}\|_{2}^{2}\label{eq:quadraticLower1}
\end{equation}
w.p.a.1.~by\textbf{ }\citet{bickel2009simultaneous}'s Lemma 4.1
(ii), where 
\begin{align}
\tilde{\phi} & =\sqrt{\phi_{\min}(\Gamma_{\ell},s+m)}\left(1-L\sqrt{\dfrac{s\phi_{\max}(\Gamma_{\ell},m)}{m\phi_{\min}(\Gamma_{\ell},s+m)}}\right)\geq\sqrt{\widetilde{c}}\left(1-L\sqrt{\dfrac{\widetilde{C}}{C_{m}\widetilde{c}}}\right)\geq\frac{\sqrt{\tilde{c}}}{2}.\label{eq:tildekappa}
\end{align}
\citet{bickel2009simultaneous}'s Eq.(B.28) yields 
\begin{equation}
\|\delta_{\mathcal{S}_{01}}\|_{2}\geq\dfrac{1}{1+L\sqrt{s/m}}\|\delta\|_{2}=\frac{1}{1+L/\sqrt{C_{m}}}\|\delta\|_{2}\geq\dfrac{\sqrt{\widetilde{c}}}{\sqrt{\widetilde{c}}+2\sqrt{\widetilde{C}}}\|\delta\|_{2}\label{eq:delta_s01}
\end{equation}
where the last inequality follows by $C_{m}=\left\lceil 4L^{2}\widetilde{C}/\widetilde{c}\right\rceil \geq4L^{2}\widetilde{C}/\widetilde{c}$.
 Inserting (\ref{eq:tildekappa}) and (\ref{eq:delta_s01}) into (\ref{eq:quadraticLower1}),
the second factor of (\ref{eq:LBeigen}) is bounded from below by
\begin{align}
\delta^{\top}\Gamma_{\ell}\delta & \geq\tilde{\phi}^{2}\|\delta_{\mathcal{S}_{01}}\|_{2}^{2}\geq\ensuremath{\dfrac{\widetilde{c}}{4}\left(\dfrac{\sqrt{\widetilde{c}}}{\sqrt{\widetilde{c}}+2\sqrt{\widetilde{C}}}\right)^{2}}\|\delta\|_{2}^{2}=C_{\kappa}\|\delta\|_{2}^{2}\label{eq:LB-RE}
\end{align}
w.p.a.1., where $C_{\kappa}=\text{\ensuremath{\widetilde{c}\left(\sqrt{\widetilde{c}}/(\sqrt{\widetilde{c}}+2\sqrt{\widetilde{C}})\right)^{2}/4.}}$
Insert (\ref{eq:eigen_LB}) and (\ref{eq:LB-RE}) into (\ref{eq:LBeigen})
and rearrange: 
\begin{align}
\frac{\delta^{\top}\text{\ensuremath{\widehat{\Sigma}}}\delta}{n\|\delta\|_{2}^{2}} & \geq\dfrac{C_{\kappa}}{2\pi^{2}\ell}\geq\dfrac{C_{\kappa}}{32\pi^{2}\cdot(s+m)\log p}\nonumber \\
 & \geq\dfrac{C_{\kappa}}{32\pi^{2}(1+C_{m})s\log p}=\dfrac{C_{\kappa}}{32\pi^{2}\left(1+\left\lceil 4L^{2}\widetilde{C}/\widetilde{c}\right\rceil \right)s\log p}\nonumber \\
 & \geq\dfrac{C_{\kappa}}{32\pi^{2}\left(2+4L^{2}\widetilde{C}/\widetilde{c}\right)s}\geq\dfrac{\widetilde{c}\cdot C_{\kappa}}{128\pi^{2}(\widetilde{c}+\widetilde{C})L^{2}s\log p}=\dfrac{\tilde{c}_{\kappa}}{L^{2}s\log p}\label{eq:LB-RE2}
\end{align}
w.p.a.1., where $\tilde{c}_{\kappa}=\widetilde{c}\cdot C_{\kappa}/[128\pi^{2}(\widetilde{c}+\widetilde{C})]$. 

\textbf{Part (b).} When $e_{t}$ is non-normal, we define $\xi_{j,t-1}:=\sum_{s=0}^{t-1}\varepsilon_{j,s}$
and a companion Brownian motion $\zeta_{t}:=\{\zeta_{j,t-1}=\psi_{j}(1)\mathcal{B}_{j}(t)\}_{j\in[p+1]}$.
Let $\Phi=(\Phi_{e}^{\top},\Phi_{u}^{\top})^{\top}$ where $\Phi_{e}$
is $p\times(p+1)$ and $\Phi_{u}$ is $1\times(p+1)$, and $\hat{\Upsilon}=\Phi_{e}(n^{-1}\sum_{t=1}^{n}\ddot{\zeta}_{t-1}\ddot{\zeta}_{t-1}{}^{\top})\Phi_{e}^{\top}$.
The triangular inequality yields
\begin{align}
\delta^{\top}\text{\ensuremath{\widehat{\Sigma}}}\delta & \geq\delta^{\top}\hat{\Upsilon}\delta-\left|\delta^{\top}(\text{\ensuremath{\widehat{\Sigma}}}-\hat{\Upsilon})\delta\right|.\label{eq:norm_app_1}
\end{align}
 The procedures as in Part (a) bounds the first term on the right-hand
side of the above expression 
\begin{equation}
\delta^{\top}\hat{\Upsilon}\delta\geq\dfrac{c_{\kappa}^{\prime}}{L^{2}s\log p}n\|\delta\|_{2}^{2}\label{eq:boundmaxtemp5}
\end{equation}
w.p.a.1 for some absolute constant $c_{\kappa}^{\prime}$. We move
on to the second term
\begin{align}
\left|\delta^{\top}(\text{\ensuremath{\widehat{\Sigma}}}-\hat{\Upsilon})\delta\right| & \leq\|\delta\|_{1}^{2}\|\text{\ensuremath{\widehat{\Sigma}}}-\hat{\Upsilon}\|_{\max}\leq\left(\|\delta_{\mathcal{S}}\|_{1}+\|\delta_{\mathcal{S}^{c}}\|_{1}\right)^{2}\|\text{\ensuremath{\widehat{\Sigma}}}-\hat{\Upsilon}\|_{\max}\nonumber \\
 & \leq(1+L)^{2}\|\delta_{\mathcal{S}}\|_{1}^{2}\|\text{\ensuremath{\widehat{\Sigma}}}-\hat{\Upsilon}\|_{\max}\leq4L^{2}s\|\delta_{\mathcal{S}}\|_{2}^{2}\|\text{\ensuremath{\widehat{\Sigma}}}-\hat{\Upsilon}\|_{\max}\nonumber \\
 & \leq4L^{2}s\|\delta\|_{2}^{2}\|\text{\ensuremath{\widehat{\Sigma}}}-\hat{\Upsilon}\|_{\max}\label{eq:boundmaxtemp4}
\end{align}
for any $L\geq1$, where the third inequality applies the restriction
$\delta\in\mathcal{R}(L,s).$ 

Since $X_{t}=\sum_{s=0}^{t}e_{s}=\Phi_{e}\sum_{s=0}^{t}\varepsilon_{s}=\Phi_{e}\xi_{t-1}$,
it follows that 
\begin{eqnarray*}
\|\text{\ensuremath{\widehat{\Sigma}}}-\hat{\Upsilon}\|_{\text{\ensuremath{\max}}} & = & \|\Phi_{e}n^{-1}\sum_{t=1}^{n}(\ddot{\xi}_{t-1}\ddot{\xi}{}_{t-1}^{\top}-\ddot{\zeta}_{t-1}\ddot{\zeta}_{t-1}^{\top})\Phi_{e}\|_{\max}\\
 & \leq & \left(\max_{j\in[p]}\sum_{\ell=1}^{p+1}\left|\Phi_{j\ell}\right|\right)^{2}\|n^{-1}\sum_{t=1}^{n}(\ddot{\xi}_{t-1}\ddot{\xi}{}_{t-1}^{\top}-\ddot{\zeta}_{t-1}\ddot{\zeta}_{t-1}^{\top})\|_{\max}\\
 & \leq & C_{L}^{2}\|n^{-1}\sum_{t=1}^{n}(\ddot{\xi}_{t-1}\ddot{\xi}{}_{t-1}^{\top}-\ddot{\zeta}_{t-1}\ddot{\zeta}_{t-1}^{\top})\|_{\max}\\
 & = & C_{L}^{2}\|n^{-1}\sum_{t=1}^{n}(\ddot{\xi}_{t-1}\ddot{\xi}{}_{t-1}^{\top}-\bar{\xi}_{t-1}\bar{\xi}{}_{t-1}^{\top}-\ddot{\zeta}_{t-1}\ddot{\zeta}_{t-1}^{\top}+\bar{\zeta}_{t-1}\bar{\zeta}_{t-1}^{\top})\|_{\max}\\
 & \leq & C_{L}^{2}\|n^{-1}\sum_{t=1}^{n}(\xi_{t-1}\xi{}_{t-1}^{\top}-\zeta_{t-1}\zeta{}_{t-1}^{\top})\|_{\max}+n^{-1}\|\bar{\xi}\bar{\xi}^{\top}-\bar{\zeta}\bar{\zeta}^{\top}\|_{\max},
\end{eqnarray*}
 where the second inequality follows by Assumption \ref{assu:covMat}.
Notice
\begin{align}
 & \sup_{j,\ell\in\left[p+1\right]}\left|\frac{1}{n}\sum_{t=1}^{n}(\xi_{j,t-1}\xi{}_{\ell,t-1}^{\top}-\zeta_{j,t-1}\zeta{}_{\ell,t-1}^{\top})\right|\nonumber \\
\leq & \ensuremath{\sup_{j,\ell\in\left[p+1\right]}\dfrac{1}{n}\sum_{t=0}^{n-1}(\text{\ensuremath{\left|\xi_{jt}-\zeta_{jt}\right|}}\cdot|\xi_{\ell t}|\text{\ensuremath{+}\ensuremath{\left|\xi_{\ell t}-\zeta_{\ell t}\right|}}\cdot|\zeta_{jt}|)}\nonumber \\
\leq & \sup_{j,\ell\in\left[p+1\right]}\sup_{0\leq t\leq n-1}\left(|\xi_{\ell t}|+|\zeta_{jt}|\right)\cdot\sup_{j\in\left[p+1\right]}\dfrac{1}{n}\sum_{t=0}^{n-1}\text{\ensuremath{\left|\xi_{jt}-\zeta_{jt}\right|}}\nonumber \\
= & O_{p}\left(\sqrt{n\log p}\right)\cdot\sqrt{n}\sup_{j\in\left[p+1\right]}\sup_{0\leq t\leq n-1}n^{-1/2}\left|\xi_{jt}-\zeta_{jt}\right|\label{eq:boundmaxtemp}
\end{align}
where the last equality applies Lemmas \ref{lem:mixing} and \ref{lem:BernsteinSum}
for the stochastic order of $\sup_{j,\ell,t}\left(|\xi_{\ell t}|+|\zeta_{jt}|\right)$.
We invoke Lemma \ref{lem:GaussianApprox} to obtain 
\[
\sup_{j\in\left[p+1\right]}\sup_{0\leq t\leq n-1}n^{-1/2}\left|\xi_{jt}-\zeta_{jt}\right|=O_{p}(n^{-1/2}\log p).
\]
These bounds imply 
\[
\|n^{-1}\sum_{t=1}^{n}\xi_{t-1}\xi{}_{t-1}^{\top}-n^{-1}\sum_{t=1}^{n}\zeta_{t-1}\zeta{}_{t-1}^{\top}\|_{\max}=O_{p}\left(\dfrac{n(\log p)^{3/2}}{n^{1/2}}\right)=O_{p}\left(n^{1/2}(\log p)^{3/2}\right).
\]
Similar derivation also shows $\|\bar{\xi}\bar{\xi}^{\top}-\bar{\zeta}\bar{\zeta}^{\top}\|_{\max}=O_{p}\left(n^{1/2}(\log p)^{3/2}\right)$
and therefore 
\[
\|\text{\ensuremath{\widehat{\Sigma}}}-\hat{\Upsilon}\|_{\text{\ensuremath{\max}}}=O_{p}\left(n^{1/2}(\log p)^{3/2}\right).
\]

Inserting the above expression into (\ref{eq:boundmaxtemp4}), we
have 
\begin{equation}
\frac{|\delta^{\top}(\text{\ensuremath{\widehat{\Sigma}}}-\hat{\Upsilon})\delta|}{n\|\delta\|_{2}^{2}}\leq4L^{2}sO_{p}\left(n^{-1/2}(\log p)^{3/2}\right)=o_{p}\left(\frac{L^{-2}}{s\log p}\right)\label{eq:boundmaxtemp6}
\end{equation}
given $s^{2}L^{4}(\log p)^{5/2}=o(n^{1/2})$. (\ref{eq:boundmaxtemp5})
and (\ref{eq:boundmaxtemp6}) then provide 
\begin{align*}
\frac{\delta^{\top}\text{\ensuremath{\widehat{\Sigma}}}\delta}{n\|\delta\|_{2}^{2}} & \geq\dfrac{c_{\kappa}^{\prime}}{L^{2}s\log p}-o_{p}\left(\frac{L^{-2}}{s\log p}\right)\geq\dfrac{\widetilde{c}_{\kappa}}{L^{2}s\log p}
\end{align*}
w.p.a.1.~when $n$ is large enough, where $\widetilde{c}_{\kappa}=0.5c_{\kappa}^{\prime}.$
\end{proof}

\subsubsection{RE and DB for Standardized Unit Roots}

The Slasso estimator is equivalent to $\hat{\theta}^{\mathrm{S}}:=D^{-1}\check{\theta}$,
where 
\[
\check{\theta}:=\arg\min_{\theta}\left\{ n^{-1}\left\Vert \ddot{Y}-\ddot{W}D^{-1}\theta\right\Vert _{2}^{2}+\lambda\left\Vert \theta\right\Vert _{1}\right\} .
\]
The scale-normalization transforms $W_{j}$ to $\tilde{W}_{j}=\ddot{W}_{j}/\widehat{\sigma}_{j}$.
Here we deduce RE and DB for the standardized time series. Define
$\widehat{\varsigma}:=\widehat{\sigma}_{\max}/\widehat{\sigma}_{\text{\ensuremath{\min}}}$
as the ratio of the maximum and the minimum standard deviation. Recall
$\widehat{\kappa}_{D}=\kappa_{D}(\hat{\Sigma},3,s)$.
\begin{prop}
\label{prop:REDBstdUnit} Suppose that Assumptions \ref{assu:tail}-\ref{assu:covMat}
hold, $(1+C_{m}(3\widehat{\varsigma}))s=o(n\wedge p)$ with $C_{m}(L)$
defined as (\ref{eq:Cm-def}), and $s^{2}\widehat{\varsigma}^{4}(\log p)^{5/2+1/(2r)}=o_{p}(n^{1/2})$.
Then we have 
\begin{equation}
n^{-1}\|\sum_{t=1}^{n}D^{-1}\ddot{W}_{t-1}u_{t}\|_{\infty}\leq\widehat{\sigma}_{\min}^{-1}C_{{\rm DB}}(\log p)^{1+\frac{1}{2r}}\label{eq:DB-SL-sigma}
\end{equation}
w.p.a.1., and 
\begin{equation}
\Pr\left\{ \widehat{\kappa}_{D}\geq\dfrac{nc_{\kappa}}{s\log p\cdot\widehat{\varsigma}^{2}\widehat{\sigma}_{{\rm \max}}^{2}}\right\} \to1\label{eq:RE-SL-sigma}
\end{equation}
 for some absolute constant $c_{\kappa}$.
\end{prop}
\begin{proof}
\textbf{DB}. It follows from (\ref{eq:DB-All-demean}) that w.p.a.1.:
\begin{align*}
n^{-1}\|\sum_{t=1}^{n}D^{-1}\ddot{W}_{t-1}u_{t}\|_{\infty} & =\max_{j\in[p]}n^{-1}\left|\sum_{t=1}^{n}\frac{\ddot{X}_{j,t-1}}{\widehat{\sigma}_{j}}u_{t}\right|\leq\widehat{\sigma}_{\text{\ensuremath{\min}}}C_{{\rm DB}}(\log p)^{1+\frac{1}{2r}}.
\end{align*}

\textbf{RE}. Define $\widetilde{\delta}:=D^{-1}\delta=(\widehat{\sigma}_{j}^{-1}\delta_{j}){}_{j\in[p]}.$
Obviously, $\|\delta_{\mathcal{M}}\|_{1}\leq\widehat{\sigma}_{\max}\|\widetilde{\delta}_{\mathcal{M}}\|_{1}$
and $\widehat{\sigma}_{\min}\|\widetilde{\delta}_{\mathcal{M}^{c}}\|_{1}\leq\|\delta_{\mathcal{M}^{c}}\|_{1}$.
Whenever $\delta\in\mathcal{R}(3,s)$ such that for any $|\mathcal{M}|\leq s$
we have $\|\delta_{\mathcal{M}^{c}}\|_{1}\leq3\|\delta_{\mathcal{M}}\|_{1},$
and thus $\widetilde{\delta}\in\mathcal{R}\left(\widehat{\varsigma},s\right)$.
Then 
\begin{align*}
\hat{\kappa}_{D} & =\inf_{\delta\in\mathcal{R}(3,s)}\dfrac{\delta^{\top}D^{-1}\hat{\Sigma}D^{-1}\delta}{\delta^{\top}\delta}=\inf_{\delta\in\mathcal{R}(3,s)}\dfrac{\delta^{\top}D^{-1}\text{\ensuremath{\widehat{\Sigma}}}D^{-1}\delta}{\delta^{\top}D^{-1}D^{2}D^{-1}\delta}=\inf_{\widetilde{\delta}\in\mathcal{R}\left(3,s\right)}\dfrac{\widetilde{\delta}^{\top}\text{\ensuremath{\widehat{\Sigma}}}\widetilde{\delta}}{\widetilde{\delta}^{\top}D^{2}\widetilde{\delta}}\\
 & \geq\inf_{\widetilde{\delta}\in\mathcal{R}\left(3\widehat{\varsigma},s\right)}\dfrac{\widetilde{\delta}^{\top}\text{\ensuremath{\widehat{\Sigma}}}\widetilde{\delta}}{\widetilde{\delta}^{\top}D^{2}\widetilde{\delta}}\geq\widehat{\sigma}_{{\rm \max}}^{-2}\inf_{\widetilde{\delta}\in\mathcal{R}\left(3\widehat{\varsigma},s\right)}\dfrac{\widetilde{\delta}^{\top}\text{\ensuremath{\widehat{\Sigma}}}\widetilde{\delta}}{\widetilde{\delta}^{\top}\widetilde{\delta}}=\widehat{\sigma}_{{\rm \max}}^{-2}\kappa_{I}(\text{\ensuremath{\widehat{\Sigma}}},3\widehat{\varsigma},s).
\end{align*}
Taking $L=3\widehat{\varsigma}$. By Proposition \ref{prop:UnitRE-cn}
$\kappa_{I}(\text{\ensuremath{\widehat{\Sigma}}},3\widehat{\varsigma},s)\geq\dfrac{\widetilde{c}_{\kappa}n}{9s\log p\widehat{\varsigma}^{2}}$
w.p.a.1., the result holds as stated in  (\ref{eq:RE-SL-sigma}).
\end{proof}

\subsubsection{RE for Demeaned Mixed Regressors}

The following proposition considers the case of mixed regressors formulated
in Section \ref{sec:Mixed-Regressors}. Here we scale-normalize $X_{t}$
by $\sqrt{n}$ and define $W_{t}^{*}:=(n^{-1/2}X_{t}^{\top},Z_{t}^{\top})^{\top}$.
We also define a corresponding Gram matrix of the $\sqrt{n}$-scaled
regressors as 
\[
\widehat{\Sigma}^{*}:=n^{-1}\sum_{t=1}^{n}(W_{t}^{*}-\bar{W}^{*})(W_{t}^{*}-\bar{W}^{*})^{\top}=n^{-1}\sum_{t=1}^{n}\ddot{W}_{t}^{*}\ddot{W}_{t}^{*\top}.
\]

Recall $v_{t}=(e_{t}^{\top},Z_{t}^{\top},u_{t})^{\top}=\Phi\varepsilon_{t}$
in (\ref{eq:def-error}). Partition $\Phi=(\Phi_{e}^{\top},\Phi_{z}^{\top},\Phi_{u}^{\top})^{\top}$
with $\Phi_{e}$ being $p_{x}\times(p+1)$, $\Phi_{z}$ being $p_{z}\times(p+1)$,
and $\Phi_{u}$ being $1\times(p+1).$ We have the following RE for
$\widehat{\Sigma}^{*}.$ 
\begin{prop}
\label{prop:MixRE-cn} Suppose that Assumptions \ref{assu:tail}-\ref{assu:covMat}
hold. As $n\to\infty$, if $L^{2}s=o(n\wedge p)$ and $s^{2}L^{4}(\log p)^{5/2+1/(2r)}=o(n^{1/2})$,
then w.p.a.1 
\begin{equation}
\kappa_{I}(\widehat{\Sigma}^{*},L,s)\geq\dfrac{\widetilde{c}_{\kappa}}{L^{2}s\log p}\label{eq:RE-unit-1-1}
\end{equation}
for any $L\geq1$, where $\widetilde{c}_{\kappa}$ is an absolute
constant. 
\end{prop}
\begin{proof}[Proof of Proposition \ref{prop:MixRE-cn}]
For any $\delta\in\mathbb{R}^{p}$, write $\delta=(\delta_{x}^{\top},\delta_{z}^{\top})^{\top}$
with $\delta_{x}\in\mathbb{R}^{p_{x}}$ and $\delta_{z}\in\mathbb{R}^{p_{z}}$.
Then 
\begin{equation}
\delta^{\top}\widehat{\Sigma}^{*}\delta=\dfrac{1}{n}\delta_{x}^{\top}\text{\ensuremath{\widehat{\Sigma}}}^{(x)}\delta_{x}+\delta_{z}^{\top}\widehat{\Sigma}^{(z)}\delta_{z}+\frac{1}{n^{3/2}}\delta_{x}^{\top}\sum_{t=1}^{n}(\ddot{X}_{t-1}\ddot{Z}_{t-1}^{\top}+\ddot{Z}_{t-1}\ddot{X}_{t-1}^{\top})\delta_{z}\label{eq:delta2Lower}
\end{equation}
consists of two quadratic terms and a cross term. The third term in
(\ref{eq:delta2Lower}) is bounded by 
\begin{align*}
\left\Vert \sum_{t=1}^{n}\ddot{X}_{t-1}\ddot{Z}_{t-1}^{\top}\right\Vert _{\max} & \leq\left\Vert \sum_{t=1}^{n}X_{t-2}Z_{t-1}^{\top}\right\Vert _{\max}+\left\Vert \sum_{t=1}^{n}e_{t-1}Z_{t-1}^{\top}\right\Vert _{\max}+\left\Vert \sum_{t=1}^{n}X_{t-1}\bar{Z}{}^{\top}\right\Vert _{\max}\\
 & =O_{p}\left(n(\log p)^{1+1/(2r)}\right)+O_{p}\left(n\right)+O_{p}\left(n\log p\right)
\end{align*}
where the stochastic order of the first term follows by Proposition
\ref{prop:DB-All}, that of the second term by (\ref{eq:ZuDB}), and
that of the third term by (\ref{eq:boundSumZ}) and (\ref{eq:boundSumX}).
As $\delta\in\mathcal{R}(L,s)$ implies $\|\delta\|_{1}^{2}\leq(1+L)^{2}\|\delta_{\mathcal{S}}\|_{1}^{2}\leq4L^{2}s\|\delta\|_{2}^{2}$,
we have 
\begin{align}
\delta_{x}^{\top}\frac{1}{n^{3/2}}\sum_{t=1}^{n}(\ddot{X}_{t-1}\ddot{Z}_{t-1}^{\top}+\ddot{Z}_{t-1}\ddot{X}_{t-1}^{\top})\delta_{z} & \leq\frac{2}{n^{3/2}}\|\delta_{x}\|_{1}\|\delta_{z}\|_{1}\left\Vert \ddot{X}_{t-1}\ddot{Z}_{t-1}^{\top}\right\Vert _{\max}\nonumber \\
 & \leq\frac{2}{n^{3/2}}\|\delta\|_{1}^{2}\left\Vert \ddot{X}_{t-1}\ddot{Z}_{t-1}^{\top}\right\Vert _{\max}\nonumber \\
 & \stackrel{\mathrm{p}}{\preccurlyeq}\dfrac{\|\delta\|_{1}^{2}}{4\sqrt{n}}(\log p)^{1+\frac{1}{2r}}\leq\|\delta\|_{2}^{2}\dfrac{L^{2}s}{\sqrt{n}}(\log p)^{1+\frac{1}{2r}}.\label{eq:delta2Bound1}
\end{align}

The second term in (\ref{eq:delta2Lower}) can be decomposed into
\[
\delta_{z}^{\top}\overline{\Sigma}^{(z)}\delta_{z}=\delta_{z}^{\top}\overline{\Sigma}^{(z)}\delta_{z}-\delta_{z}^{\top}\bar{Z}\bar{Z}^{\top}\delta_{z}=\delta_{z}^{\top}\Sigma^{(z)}\delta_{z}-\delta_{z}^{\top}\bar{Z}\bar{Z}^{\top}\delta_{z}-\delta_{z}^{\top}\left(\Sigma^{(z)}-\overline{\Sigma}^{(z)}\right)\delta_{z}
\]
where 
\[
\delta_{z}^{\top}\bar{Z}\bar{Z}^{\top}\delta_{z}\leq\|\delta_{z}\|_{1}^{2}\max_{k\in[p_{z}]}\left|\bar{Z}_{k}\right|\stackrel{\mathrm{p}}{\preccurlyeq}\|\delta\|_{2}^{2}\frac{L^{2}s}{\sqrt{n}}\sqrt{\log p}
\]
 by (\ref{eq:boundSumZ}), and 
\[
\delta_{z}^{\top}(\Sigma^{(z)}-\overline{\Sigma}^{(z)})\delta_{z}\leq\|\delta_{z}\|_{1}^{2}\|\Sigma^{(z)}-\overline{\Sigma}^{(z)}\|_{\max}\stackrel{\mathrm{p}}{\preccurlyeq}\|\delta\|_{2}^{2}\frac{L^{2}s}{\sqrt{n}}\sqrt{\log p}
\]
by (\ref{eq:maxSqZ}). We thus continue (\ref{eq:delta2Lower}):
\begin{align*}
\delta^{\top}\widehat{\Sigma}^{*}\delta & \geq\delta_{x}^{\top}\text{\ensuremath{\widehat{\Sigma}}}^{(x)}\delta_{x}+\delta_{z}^{\top}\Sigma^{(z)}\delta_{z}-\|\delta\|_{2}^{2}O_{p}\left(\frac{L^{2}s}{\sqrt{n}}\sqrt{\log p}+\dfrac{L^{2}s}{\sqrt{n}}(\log p)^{1+1/(2r)}\right)\\
 & =\delta_{x}^{\top}\text{\ensuremath{\widehat{\Sigma}}}^{(x)}\delta_{x}+\delta_{z}^{\top}\Sigma^{(z)}\delta_{z}-\|\delta\|_{2}^{2}O_{p}\left(\dfrac{L^{2}s}{\sqrt{n}}(\log p)^{1+1/(2r)}\right)\\
 & =\delta_{x}^{\top}\text{\ensuremath{\widehat{\Sigma}}}^{(x)}\delta_{x}+\delta_{z}^{\top}\Sigma^{(z)}\delta_{z}-\|\delta\|_{2}^{2}o_{p}\left(1/(L^{2}s\log p)\right)
\end{align*}
where the last line follows by the condition $L^{4}s^{2}(\log p)^{2+1/(2r)}=o(n^{1/2})$
. Therefore there exists an absolute constant $\widetilde{c}_{\kappa}$
such that 
\[
\delta_{x}^{\top}\text{\ensuremath{\widehat{\Sigma}}}^{(x)}\delta_{x}+\delta_{z}^{\top}\Sigma^{(z)}\delta_{z}\geq\|\delta\|_{2}^{2}\dfrac{\widetilde{c}_{\kappa}}{L^{2}s\log p}
\]
 w.p.a.1. 

Parallel to Proposition \ref{prop:UnitRE-cn}, in the rest of the
proof Step 1 will establish RE under normal innovations, and Step
2 will allow non-normal innovations.

\textbf{Step 1.} If $\varepsilon_{jt}\sim i.i.d.\ \mathcal{N}(0,1)$,
then $e_{t}\sim i.i.d.\ \mathcal{N}(0,\Omega_{e})$ with $\Omega_{e}=\Phi_{e}\Phi_{e}^{\top}$,
and $Z_{t}\sim i.i.d.\ \mathcal{N}(0,\Sigma^{(z)})$ with $\Sigma^{(z)}=\Omega_{z}:=\Phi_{z}\Phi_{z}^{\top}.$
Similar to (\ref{eq:LBeigen}) and (\ref{eq:eigen_LB}) in the proof
of Proposition \ref{prop:UnitRE-cn}, we deduce that 
\begin{align*}
\delta_{x}^{\top}\text{\ensuremath{\widehat{\Sigma}}}^{(x)}\delta_{x} & \geq\dfrac{n}{2\pi^{2}\ell}\cdot\delta^{\top}\Gamma_{\ell}\delta
\end{align*}
for any $\ell\leq n-1$, where $V$ is the orthonormal matrix used
in (\ref{eq:LBeigen}), $V_{\cdot[\ell]}$ is the submatrix composed
of the first $\ell$ columns of $V$ and $\Gamma_{\ell}:=\dfrac{1}{\ell}e^{\top}V_{\cdot[\ell]}V_{\cdot[\ell]}^{\top}e$.
Then 
\begin{align*}
\dfrac{1}{n}\delta_{x}^{\top}\text{\ensuremath{\widehat{\Sigma}}}^{(x)}\delta_{x}+\delta_{z}^{\top}\Sigma^{(z)}\delta_{z} & \geq\dfrac{1}{2\pi^{2}\ell}\cdot\delta^{\top}\Gamma_{\ell}\delta+\delta_{z}^{\top}\Sigma^{(z)}\delta_{z}\\
 & \geq\dfrac{1}{2\pi^{2}\ell}\cdot\left(\delta^{\top}\Gamma_{\ell}\delta+\delta_{z}^{\top}\Sigma^{(z)}\delta_{z}\right)=\dfrac{1}{2\pi^{2}\ell}\delta^{\top}\Lambda_{\ell}\delta
\end{align*}
where $\Lambda_{\ell}={\rm diag}(\Gamma_{\ell},\Sigma^{(z)})$. The
second inequality follows by $(2\pi^{2}\ell)^{-1}<1$ as $\ell\geq1$.

The proof of Proposition \ref{prop:UnitRE-cn} has shown that when
$m=C_{m}s$ with $C_{m}\geq1$ and $\ell=16(s+m)\log p$, there are
absolute constants $\widetilde{c}$ and $\widetilde{C}$ such that
\[
\tilde{c}\leq\phi_{\min}(\Gamma_{\ell},s+m)\leq\phi_{\max}(\Gamma_{\ell},s+m)\leq\tilde{C}.\ \ \text{w.p.a.1.}
\]
Similarly, for the stationary part $c_{\Omega}\leq\lambda_{\min}(\Phi_{z}\Phi_{z}^{\top})\leq\lambda_{\max}(\Phi_{z}\Phi_{z}^{\top})\leq C_{\Omega}$
and therefore the bounds are also applicable. It follows 
\begin{equation}
\dfrac{1}{n}\delta_{x}^{\top}\text{\ensuremath{\widehat{\Sigma}}}^{(x)}\delta_{x}+\delta_{z}^{\top}\Sigma^{(z)}\delta_{z}\geq\dfrac{1}{2\pi^{2}\ell}\delta^{\top}\Lambda_{\ell}\delta\geq\dfrac{\widetilde{c}_{\kappa}\|\delta\|_{2}^{2}}{L^{2}s\log p}\ \ \text{w.p.a.1}\label{eq:LB-RE2-Mix}
\end{equation}
 for some absolute constant $\widetilde{c}_{\kappa}$. 

\textbf{Step 2.} When $e_{t}$ is non-normal, we decompose 
\begin{align*}
\dfrac{1}{n}\delta_{x}^{\top}\text{\ensuremath{\widehat{\Sigma}}}^{(x)}\delta_{x}+\delta_{z}^{\top}\Sigma^{(z)}\delta_{z} & \geq\delta_{x}^{\top}\hat{\Upsilon}^{(x)}\delta_{x}^{\top}+\delta_{z}^{\top}\Sigma^{(z)}\delta_{z}-\left|\delta_{x}^{\top}\left(\text{\ensuremath{\widehat{\Sigma}}}^{(x)}-\hat{\Upsilon}^{(x)}\right)\delta_{x}\right|.
\end{align*}
Following the same argument for Part (b) of Proposition \ref{prop:UnitRE-cn},
under non-normality (\ref{eq:LB-RE2-Mix}) remains valid under the
specified orders of $L$, $s$, $n$, $p$, and we further bound $\|\text{\ensuremath{\widehat{\Sigma}}}^{(x)}-\hat{\Upsilon}^{(x)}\|_{\text{\ensuremath{\max}}}$
to obtain the conclusion.
\end{proof}

\subsection{\label{subsec:ProofMain}Proofs of Results in Main Text}
\begin{proof}[Proof of Lemma \ref{lem:gene-Lasso}]
 The minimization of (\ref{eq:Lasso.theta.origin}) with respect
to $\left(\alpha,\theta\right)$ is numerically equivalent to a two-step
minimization 
\[
\min_{\theta}\min_{\alpha\left(\theta\right)}\left\{ \frac{1}{n}\left\Vert Y\boldsymbol{-}\alpha\left(\theta\right)1_{N}-W\theta\right\Vert _{2}^{2}+\lambda\left\Vert H\theta\right\Vert _{1}\right\} 
\]
where the outer step is with respect to $\theta$ and the inner step
is with respect to $\alpha$ under a given $\theta$, which we denote
as $\alpha\left(\theta\right)$. Since the $L_{1}$-penalty term is
irrelevant to inner optimization, we immediately get a closed-form
solution $\alpha\left(\theta\right)=\bar{Y}-\bar{W}^{\top}\theta$.
Substituting this inner solution back to eliminate the inner optimization,
the criterion function is
\[
\frac{1}{n}\left\Vert Y\boldsymbol{-}\alpha\left(\theta\right)1_{N}-W\theta\right\Vert _{2}^{2}+\lambda\left\Vert H\theta\right\Vert _{1}=\frac{1}{n}\left\Vert \ddot{Y}-\ddot{W}\theta\right\Vert _{2}^{2}+\lambda\left\Vert H\theta\right\Vert _{1}.
\]
From now on, we focus on
\begin{equation}
\min_{\theta}\left\{ \frac{1}{n}\left\Vert \ddot{Y}-\ddot{W}\theta\right\Vert _{2}^{2}+\lambda\left\Vert \theta\right\Vert _{1}\right\} .\label{eq:Lasso.theta}
\end{equation}

The following steps are known from \citet{buhlmann2011statistics};
here we include the proof for completeness. Since the estimator minimizes
the criterion function, we have 
\[
\dfrac{1}{n}\|\ddot{Y}-\ddot{W}\hat{\theta}\|_{2}^{2}+\lambda\|H\hat{\theta}\|_{1}\leq\dfrac{1}{n}\|\ddot{Y}-\ddot{W}\theta^{*}\|_{2}^{2}+\lambda\|H\theta^{*}\|_{1}.
\]
Define $\check{W}:=\ddot{W}H^{-1},$ $\check{\theta}:=H\hat{\theta}$,
$\check{\theta}^{*}:=H\theta^{*}$ and $\check{\Sigma}:=n^{-1}\check{W}^{\top}\check{W}$.
The inequality above can be written as

\[
\dfrac{1}{n}\|\ddot{Y}-\check{W}\check{\theta}\|_{2}^{2}+\lambda\|\check{\theta}\|_{1}\leq\dfrac{1}{n}\|\ddot{Y}-\check{W}\check{\theta}^{*}\|_{2}^{2}+\lambda\|\check{\theta}^{*}\|_{1},
\]
which implies the basic inequality 
\[
(\check{\theta}-\check{\theta}^{*})^{\top}\check{\Sigma}(\check{\theta}-\check{\theta}^{*})+\lambda\|\check{\theta}\|_{1}\leq\dfrac{2}{n}\ddot{u}^{\top}\check{W}(\check{\theta}-\check{\theta}^{*})+\lambda\|\check{\theta}^{*}\|_{1}.
\]
By the Holder's inequality $\ddot{u}^{\top}\check{W}(\check{\theta}-\theta^{*})\leq\|\check{W}^{\top}\ddot{u}\|_{\infty}\|\check{\theta}-\check{\theta}^{*}\|_{1}$
and the specified condition for the tuning parameter $n^{-1}\|\check{W}^{\top}\ddot{u}\|_{\infty}\leq\lambda/4$,
we have 
\begin{align}
(\check{\theta}-\check{\theta}^{*})^{\top}\check{\Sigma}(\check{\theta}-\check{\theta}^{*})+\lambda\|\check{\theta}\|_{1} & \leq\dfrac{1}{n}\|\check{W}^{\top}\ddot{u}\|_{\infty}\|\check{\theta}-\check{\theta}^{*}\|_{1}+\lambda\|\check{\theta}^{*}\|_{1}\leq\dfrac{\lambda}{2}\|\check{\theta}-\check{\theta}^{*}\|_{1}+\lambda\|\check{\theta}^{*}\|_{1}\nonumber \\
 & =\dfrac{\lambda}{2}\|(\check{\theta}-\check{\theta}^{*})_{\mathcal{S}}\|_{1}+\|(\check{\theta}-\check{\theta}^{*})_{\mathcal{S}^{c}}\|_{1}+\lambda\|\check{\theta}^{*}\|_{1}\nonumber \\
 & =\dfrac{\lambda}{2}\left(\|(\check{\theta}-\check{\theta}^{*})_{\mathcal{S}}\|_{1}+\|\check{\theta}_{\mathcal{S}^{c}}\|_{1}\right)+\lambda\|\check{\theta}_{\mathcal{S}}^{*}\|_{1}.\label{eq:eqPl1}
\end{align}
We substitute the following triangular inequality 
\begin{equation}
\|\check{\theta}\|_{1}=\|\check{\theta}_{\mathcal{S}}\|_{1}+\|\check{\theta}_{\mathcal{S}^{c}}\|_{1}\geq\|\check{\theta}_{\mathcal{S}}^{*}\|_{1}-\|(\check{\theta}-\check{\theta}^{*})_{\mathcal{S}}\|_{1}+\|\check{\theta}_{\mathcal{S}^{c}}\|_{1}\label{eq:eqPl2}
\end{equation}
into (\ref{eq:eqPl1}); after rearrangement we obtain 
\begin{align*}
(\check{\theta}-\check{\theta}^{*})^{\top}\check{\Sigma}(\check{\theta}-\check{\theta}^{*})+\lambda\|\check{\theta}_{\mathcal{S}^{c}}\|_{1} & \leq\lambda\|(\check{\theta}-\check{\theta}^{*})_{\mathcal{S}}\|_{1}+\dfrac{\lambda}{2}\left(\|(\check{\theta}-\check{\theta}^{*})_{\mathcal{S}}\|_{1}+\|\check{\theta}_{\mathcal{S}^{c}}\|_{1}\right)\\
 & =\dfrac{3\lambda}{2}\|(\check{\theta}-\check{\theta}^{*})_{\mathcal{S}}\|_{1}+\dfrac{\lambda}{2}\|\check{\theta}_{\mathcal{S}^{c}}\|_{1}
\end{align*}
or equivalently 
\begin{equation}
2(\check{\theta}-\check{\theta}^{*})^{\top}\check{\Sigma}(\check{\theta}-\check{\theta}^{*})+\lambda\|(\check{\theta}-\check{\theta}^{*})_{\mathcal{S}^{c}}\|_{1}\leq3\lambda\|(\check{\theta}-\check{\theta}^{*})_{\mathcal{S}}\|_{1}.\label{eq:A60}
\end{equation}
Add $\lambda\|(\check{\theta}-\check{\theta}^{*})_{\mathcal{S}}\|_{1}$
on both sides of the above inequality:
\begin{align*}
2(\check{\theta}-\check{\theta}^{*})^{\top}\check{\Sigma}(\check{\theta}-\check{\theta}^{*})+\lambda\|\check{\theta}-\check{\theta}^{*}\|_{1} & \leq4\lambda\|(\check{\theta}-\check{\theta}^{*})_{\mathcal{S}}\|_{1}\leq4\lambda\sqrt{s}\|(\check{\theta}-\check{\theta}^{*})_{\mathcal{S}}\|_{2}\\
 & \leq4\lambda\sqrt{s}\|\check{\theta}-\check{\theta}^{*}\|_{2}\leq4\lambda\sqrt{\dfrac{s}{\kappa_{H}}(\check{\theta}-\check{\theta}^{*})^{\top}\check{\Sigma}(\check{\theta}-\check{\theta}^{*})}\\
 & \leq4\lambda s/\kappa_{H}+(\check{\theta}-\check{\theta}^{*})^{\top}\check{\Sigma}(\check{\theta}-\check{\theta}^{*})
\end{align*}
where the fourth inequality follows the fact that $(\check{\theta}-\check{\theta}^{*})\in\mathcal{R}(3,s)$
implied by (\ref{eq:A60}), and the last inequality applies the generic
inequality $4ab\leq4a^{2}+b^{2}.$ Rearrange the above inequality
into 
\[
(\check{\theta}-\check{\theta}^{*})^{\top}\check{\Sigma}(\check{\theta}-\check{\theta}^{*})+\lambda\|\check{\theta}-\check{\theta}^{*}\|_{1}\leq4\lambda^{2}s/\kappa_{H}.
\]
The first and the second inequalities in the statement of the lemma
immediately follow. The last inequality is deduced by $\|\check{\theta}-\check{\theta}^{*}\|_{2}^{2}\leq\kappa_{H}^{-1}(\check{\theta}-\check{\theta}^{*})^{\top}\check{\Sigma}(\check{\theta}-\theta^{*})\leq4\lambda^{2}s/\kappa_{H}^{2}.$
\end{proof}
\medskip
\begin{proof}[Proof of Proposition \ref{prop:UnitDB}]
 The DB for pure unit root regressors is a special case of (\ref{eq:DB-All-demean})
with $p_{x}=p$ and $p_{z}=0.$
\end{proof}
\medskip
\begin{proof}[Proof of Lemma \ref{lem:Normal_RE}]
 When $n$ is sufficiently large, Assumption \ref{assu:asym_n} implies
$(1+C_{m}(3))s=o(n\wedge p)$ for Part (a) of Proposition \ref{prop:UnitRE-cn}.
Lemma \ref{lem:Normal_RE} is a direct consequence of Proposition
\ref{prop:UnitRE-cn} (a) by taking $L=3$ and $c_{\kappa}=\widetilde{c}_{\kappa}/9$.
\end{proof}
\medskip
\begin{proof}[Proof of Proposition \ref{prop:UnitRE}]
 Assumption \ref{assu:asym_n} implies $(1+C_{m}(3))s=o(n\wedge p)$
and that $s^{2}(\log p)^{5/2}=o(n^{1/2})$ for Proposition \ref{prop:UnitRE-cn}
(b). Proposition \ref{prop:UnitRE} is a direct result of Proposition
\ref{prop:UnitRE-cn} (b) by taking $L=3$ and $c_{\kappa}=\widetilde{c}_{\kappa}/9$.
\end{proof}
\medskip
\begin{proof}[Proof of Proposition \ref{prop:MinMax-Unit}]
 We first show Part (a). Since $\hat{\sigma}_{\min}^{2}$ is the
minimum diagonal matrix of $\hat{\Sigma}$ and $(1+C_{m}(1))\leq p$
with $n$ large enough, it can be bounded below by a special restricted
eigenvalue 
\[
\widehat{\sigma}_{\min}^{2}=\inf_{\delta\in\mathcal{R}(1,1)}\dfrac{\delta^{\top}\text{\ensuremath{\widehat{\Sigma}}}\delta}{\delta^{\top}\delta}=\kappa_{I}(\text{\ensuremath{\widehat{\Sigma}}},1,1)\stackrel{\mathrm{p}}{\succcurlyeq}n\left(\log p\right)^{-1}
\]
where the last inequality applies Proposition \ref{prop:UnitRE-cn}
with $L=1$ and $s=1$. Next, we bound the maximum sample variance
from above by 
\[
\widehat{\sigma}_{\max}^{2}\leq\max_{j\in[p]}n^{-1}\sum_{t=1}^{n}X_{j,t-1}^{2}\leq\max_{j\in[p],t\in[n]}X_{j,t-1}^{2}\stackrel{\mathrm{p}}{\preccurlyeq}n\log p.
\]

The DB in Part (b) is implied by (\ref{eq:DB-SL-sigma}) and (\ref{eq:sigma-minmax}).
Regarding the RE, note that Assumption \ref{assu:asym_n} implies
$s^{2}(\log p)^{13/2}=o(n^{1/2})$ and hence by (\ref{eq:sigma-minmax})
we have $s^{2}\text{\ensuremath{\widehat{\varsigma}}}^{4}(\log p)^{5/2}=o_{p}(n^{1/2})$.
Besides, Assumption \ref{assu:asym_n} and $s(\log p)^{2}=o(p)$ also
implies $(1+C_{m}(3\widehat{\varsigma}))s\log p\stackrel{\mathrm{p}}{\preccurlyeq}s(\log p)^{2}=o(n\wedge p).$
Taking $L=3\widehat{\varsigma},$ by (\ref{eq:RE-SL-sigma}) we have
w.p.a.1 that 
\[
\hat{\kappa}_{D}\geq\dfrac{nc_{\kappa}}{s\log p\cdot\widehat{\varsigma}^{2}\widehat{\sigma}_{{\rm \max}}^{2}}\stackrel{\mathrm{p}}{\succcurlyeq}\frac{1}{s(\log p)^{4}}
\]
 for some absolute constant $c_{\kappa}^{\prime}>0$. 
\end{proof}
\medskip
\begin{proof}[Proof of Theorem \ref{thm:LassoError}]
 Propositions \ref{prop:UnitDB} and \ref{prop:UnitRE} have constructed
the DB and RE for $\hat{\beta}^{\mathrm{P}}$, respectively. We plug
them into Lemma \ref{lem:gene-Lasso} and the rates of convergence
follow.
\end{proof}
\medskip
\begin{proof}[Proof of Theorem \ref{thm:SlassoError}]
 Proposition \ref{prop:MinMax-Unit} has constructed the DB and RE
for $\hat{\beta}^{{\rm S}}$. We plug them into Lemma \ref{lem:gene-Lasso}
and the rates of convergence follow by 
\[
\|\hat{\beta}^{{\rm S}}-\beta^{*}\|_{q}\leq\max_{j\in[p]}\hat{\sigma}_{j}\|D^{-1}(\hat{\beta}^{{\rm S}}-\beta^{*})\|_{q}\stackrel{\mathrm{p}}{\preccurlyeq}\sqrt{\log p}\|D^{-1}(\hat{\beta}^{{\rm S}}-\beta^{*})\|_{q}
\]
for $q=1,2$.
\end{proof}
\medskip
\begin{proof}[Proofs of Theorem \ref{thm:SlassoError-Mix}]
 Define $\mathcal{M}_{x}:=[p_{x}]$ be the index set for the unit
root regressors and $\mathcal{M}_{z}:=\{p_{x}+1,p_{x}+2,\cdots,p\}$
be the index set for the stationary regressors. Let $\widehat{\sigma}_{j}^{*}=\frac{\hat{\sigma}_{j}}{\sqrt{n}}\cdot\boldsymbol{1}(j\in\mathcal{M}_{x})+\hat{\sigma}_{j}\cdot\boldsymbol{1}(j\in\mathcal{M}_{z}),$
and $\widehat{\sigma}_{{\rm \max}}^{*}=\max_{j\in[p]}\hat{\sigma}_{j}^{*}$
and $\widehat{\sigma}_{{\rm \min}}^{*}=\min_{j\in[p]}\hat{\sigma}_{j}^{*}$.
Using Proposition \ref{prop:MinMax-Unit}, we have for $j\in\mathcal{M}_{x}$,
the sample variances are bounded by 
\begin{align*}
\left(\log p\right)^{-1}\stackrel{\mathrm{p}}{\preccurlyeq}\min_{j\in\mathcal{M}_{x}}\ensuremath{\widehat{\sigma}_{j}^{*2}}\leq & \max_{j\in\mathcal{M}_{x}}\ensuremath{\widehat{\sigma}_{j}^{*2}}=\frac{1}{n}\min_{j\in\mathcal{M}_{x}}\ensuremath{\widehat{\sigma}_{j}^{2}}\leq\frac{1}{n}\max_{j\in\mathcal{M}_{x}}\ensuremath{\widehat{\sigma}_{j}^{2}}\stackrel{\mathrm{p}}{\preccurlyeq}\log p.
\end{align*}
For $j\in\mathcal{M}_{z}$, it is easy to show $\widehat{\sigma}_{j}^{2}$
are uniformly bounded away from 0 and $\infty$ as 
\begin{gather}
\max_{j\in\mathcal{M}_{z}}\ensuremath{\widehat{\sigma}_{j}^{2}\leq\max_{k}\dfrac{1}{n}\sum_{t=1}^{n}Z_{k,t-1}^{2}}\leq\max_{k\in[p_{z}]}\dfrac{1}{n}\sum_{t=1}^{n}\mathbb{E}(Z_{k,t-1}^{2})+C\sqrt{\dfrac{\log p}{n}}\stackrel{\mathrm{p}}{\preccurlyeq}1,\label{eq:max_sigma_z}\\
\min_{j\in\mathcal{M}_{z}}\ensuremath{\widehat{\sigma}_{j}^{2}\geq}\min_{k}\mathbb{E}(Z_{k,t-1}^{2})-C\sqrt{(\log p)/n}\stackrel{\mathrm{p}}{\succcurlyeq}1.\label{eq:min_sigma_z}
\end{gather}
As a result, the sample variances of the mixed regressors are bounded
by 
\begin{equation}
\left(\log p\right)^{-1}\stackrel{\mathrm{p}}{\preccurlyeq}\widehat{\sigma}_{\text{\ensuremath{\min}}}^{*}\leq\widehat{\sigma}_{\max}^{*}\stackrel{\mathrm{p}}{\preccurlyeq}\log p.\label{eq:sigma_star}
\end{equation}

\textbf{DB}.\textbf{ }The bounds for $\widehat{\sigma}_{j}^{*}$ implies\textbf{
\begin{align}
\text{\ensuremath{n^{-1}\|\sum_{t=1}^{n}\widetilde{W}_{t-1}u_{t}\|_{\infty}}} & \ensuremath{\leq\widehat{\sigma}_{\min}^{*-1}\left(\max_{j}\dfrac{1}{n^{3/2}}\left|\sum_{t=1}^{n}\ddot{X}_{j,t-1}u_{t}\right|+\max_{k}\dfrac{1}{n}\left|\sum_{t=1}^{n}\ddot{Z}_{k,t-1}u_{t}\right|\right)}\nonumber \\
 & \stackrel{\mathrm{p}}{\preccurlyeq}\widehat{\sigma}_{\min}^{*-1}\left(\dfrac{1}{\sqrt{n}}(\log p)^{1+\frac{1}{2r}}+\sqrt{\dfrac{\log p}{n}}\right)\stackrel{\mathrm{p}}{\preccurlyeq}\dfrac{1}{\sqrt{n}}(\log p)^{\frac{3}{2}+\frac{1}{2r}}\label{eq:DB-mix-std}
\end{align}
}by (\ref{eq:DB-All-demean}) abd (\ref{eq:ZeDB}). 

\textbf{RE}.\textbf{ }Setting $L=3\widehat{\varsigma}$, we have w.p.a.1.
\begin{equation}
\hat{\kappa}_{D}\geq\frac{\kappa_{I}(\widehat{\Sigma}^{*},3\widehat{\varsigma},s)}{\widehat{\sigma}_{{\rm \max}}^{*2}}\geq\dfrac{\widetilde{c}_{\kappa}/(9s\log p\widehat{\varsigma}^{2})}{\widehat{\sigma}_{{\rm \max}}^{*2}}\stackrel{\mathrm{p}}{\succcurlyeq}\dfrac{1}{s(\log p)^{4}},\label{eq:RE-mix-std}
\end{equation}
where the first inequality follows the proof of Proposition \ref{prop:REDBstdUnit},
the second by Proposition \ref{prop:MixRE-cn}, and the last one by
the relative size of $n,$ $s$ and $p$ in Assumption \ref{assu:asym_n}. 

We plug these two building blocks, DB and RE, into Lemma \ref{lem:gene-Lasso}
and the rates of convergence follow by
\begin{gather*}
\|\hat{\beta}^{{\rm S}}-\beta^{*}\|_{q}\leq\max_{j\in\mathcal{M}_{x}}\hat{\sigma}_{j}\|D^{-1}(\hat{\theta}^{{\rm S}}-\theta^{*})\|_{q}\stackrel{\mathrm{p}}{\preccurlyeq}\sqrt{\log p}\|D^{-1}(\hat{\theta}^{{\rm S}}-\theta^{*})\|_{q}\\
\|\hat{\gamma}^{{\rm S}}-\gamma^{*}\|_{q}\leq\max_{j\in\mathcal{M}_{z}}\hat{\sigma}_{j}\|D^{-1}(\hat{\theta}^{{\rm S}}-\theta^{*})\|_{q}\stackrel{\mathrm{p}}{\preccurlyeq}\|D^{-1}(\hat{\theta}^{{\rm S}}-\theta^{*})\|_{q}
\end{gather*}
for $q=1,2$.
\end{proof}
\bigskip
Before moving to the proof of Theorem \ref{thm:cointegration}, we
introduce some additional notations for the model with cointegrated
variables.  Let $\mathcal{M}_{1}$, $\mathcal{M}_{2}$, $\mathcal{M}_{x}$,
and $\mathcal{M}_{z}$ be the index sets of the location of $X_{t}^{co(1)}$,
$X_{t}^{co(2)}$, $X_{t}$, and $Z_{t}$, respectively. Let $D^{{\rm co}(1)}:={\rm diag}(\hat{\sigma}_{j}^{{\rm co}(1)})_{j\in\mathcal{M}_{1}}$
where $\hat{\sigma}_{j}^{{\rm co}(1)}$ be the sample s.d.~of $X_{j}^{\mathrm{co}(1)}$.
Obviously under the conditions in Theorem \ref{thm:cointegration}
we have 
\begin{equation}
\min_{j\in[p_{c1}]}\hat{\sigma}_{j}^{{\rm co}(1)}\stackrel{\mathrm{p}}{\succcurlyeq}\sqrt{n/\log p},\label{eq:coint_variance}
\end{equation}
as these $X_{j}^{\mathrm{co}(1)}$ behaves as a unit root process
individually. Similarly we define the sample s.d~$\hat{\sigma}_{j}^{{\rm co}(2)}$,
$\hat{\sigma}_{j}^{X}$, and $\hat{\sigma}_{j}^{Z}$ to be embedded
into the diagonal matrices $D^{{\rm co}(2)}$, $D^{X}$ and $D^{Z}$.
Define a big diagonal matrix $D:={\rm diag}(D^{{\rm co}(1)},D^{{\rm co}(2)},D^{X},D^{Z})$
to concatenate all variables.

Denote a lower-triangular $p\times p$ matrix
\[
\Pi:=\begin{pmatrix}I_{p_{c1}}\\
A^{\top} & I_{p_{c2}}\\
 &  & I_{p_{x}}\\
 &  &  & I_{p_{z}}
\end{pmatrix}
\]
as the rotation matrix, where the blank entries are zeros. Its inverse
rotates the observed regressor matrix $W$ into the infeasible counterpart
$W_{\Pi}:=W\Pi^{-1}=(V^{(1)},W^{(0)})$, where $W^{(0)}:=(X^{{\rm co}(2)\top},X^{\top},Z^{\top})^{\top}$
is defined as the regressor matrix for the components invariant to
the rotation. 

The benchmark model for Slasso is (\ref{eq:cointegration_y_2}), with
the pseudo true coefficients $\theta^{*}=\left(0_{p_{c}}^{\top},\beta^{*\top},\gamma^{(1)*\top}\right)^{\top}$
associated with $W$, and then the true coefficients associated with
$W^{(0)}$ is $\theta^{*(0)}=\theta_{[p]\backslash\mathcal{M}_{1}}^{*}$.
Also, for a generic $\theta\in\mathbb{R}^{p}$, we define $\theta^{(0)}:=\theta_{[p]\backslash\mathcal{M}_{1}}$,
$\theta^{(1)}:=\theta_{\mathcal{M}_{1}}$, and $\theta^{(2)}:=\theta_{\mathcal{M}_{2}}.$ 
\begin{proof}[Proof of Theorem \ref{thm:cointegration}]
 This proof works exclusively with Slasso; therefore for conciseness
we use $\widehat{\theta}$ to denote the Slasso estimator by suppressing
the superscript ``$\mathrm{S}$''. This proof contains three new
lemmas given new notations are defined as the deduction advances.

As the minimizer of the criterion function, Slasso gives
\begin{equation}
\dfrac{1}{n}\|\ddot{Y}-\ddot{W}\hat{\theta}\|_{2}^{2}+\lambda\|D\hat{\theta}\|_{1}\leq\dfrac{1}{n}\|\ddot{Y}-\ddot{W}\theta^{*}\|_{2}^{2}+\lambda\|D\theta^{*}\|_{1},\label{eq:Slasso_cri}
\end{equation}
where $\theta^{*}$ has been defined as the pseudo-true coefficient
in the benchmark model (\ref{eq:cointegration_y_2}).\textbf{\textcolor{red}{{} }}

Notice for a generic $\theta$, rotation and scaling yield the fitted
value 
\[
\ddot{W}\theta=\ddot{W}\Pi^{-1}D^{-1}\cdot D\Pi\theta=\tilde{W}_{\Pi}\tilde{\theta},
\]
 where $\tilde{W}_{\Pi}:=\ddot{W}\Pi^{-1}D^{-1}$ and $\tilde{\theta}:=D\Pi\theta$,
and the corresponding penalized vector $D\theta=D\Pi^{-1}D^{-1}\tilde{\theta}=Q\tilde{\theta}$
with $Q:=D\Pi^{-1}D^{-1}$. Then (\ref{eq:Slasso_cri}) is equivalent
to 
\[
\dfrac{1}{n}\|\ddot{Y}-\tilde{W}_{\Pi}\tilde{\theta}\|_{2}^{2}+\lambda\|Q\tilde{\theta}\|_{1}\leq\dfrac{1}{n}\|\ddot{Y}-\tilde{W}_{\Pi}\tilde{\theta}^{*}\|_{2}^{2}+\lambda\|Q\tilde{\theta}^{*}\|_{1}=\dfrac{1}{n}\|\ddot{Y}-\tilde{W}_{\Pi}\tilde{\theta}^{*}\|_{2}^{2}+\lambda\|\tilde{\theta}^{*}\|_{1},
\]
where the equality applies the fact that $Q\tilde{\theta}^{*}=\tilde{\theta}^{*}$
because the first $p_{c}$ entries of $\theta^{*}$ are zeros. Given
$\ddot{Y}=\tilde{W}_{\Pi}\tilde{\theta}+\ddot{u}^{(1)}$, we have
the basic inequality 
\begin{align*}
\dfrac{1}{n}\|\tilde{W}_{\Pi}(\tilde{\theta}-\tilde{\theta}^{*})\|_{2}^{2}+\lambda\|Q\tilde{\theta}\|_{1} & \leq\dfrac{2}{n}\ddot{u}{}^{(1)\top}\tilde{W}_{\Pi}(\tilde{\theta}-\tilde{\theta}^{*})+\lambda\|\tilde{\theta}^{*}\|_{1}\\
 & \leq\dfrac{2}{n}\left(\ddot{u}{}^{(1)\top}\tilde{v}^{(1)}\tilde{\theta}_{1}+\ddot{u}{}^{(1)\top}\tilde{W}^{(0)}(\tilde{\theta}^{(0)}-\tilde{\theta}^{\text{(0)}*})\right)+\lambda\|\tilde{\theta}^{*}\|_{1}\\
 & \leq\dfrac{2}{n}\left(\|\tilde{v}^{(1)\top}\ddot{u}^{(1)}\|_{\infty}\|\tilde{\theta}^{(1)}\|_{1}+\|\tilde{W}^{(0)\top}\ddot{u}^{(1)}\|_{\infty}\|\tilde{\theta}^{(0)}-\tilde{\theta}^{\text{(0)}*}\|_{1}\right)+\lambda\|\tilde{\theta}^{*}\|_{1}.
\end{align*}

Recall that $\|\cdot\|_{r1}$ is the maximum row-wise norm defined
above Assumption \ref{assu:cointegration}. For a generic matrix
$A$, define the maximum column-wise norm as $\|A\|_{c1}:=\|A^{\top}\|_{r1}$.
We have the following Lemma.

\begin{lem}\label{lem:Coint-DB}Under the conditions of Theorem \ref{thm:cointegration},
we have w.p.a.1. 
\begin{equation}
1\leq\|Q^{-1}\|_{c1}\stackrel{\mathrm{p}}{\preccurlyeq}\log p\label{eq:L1Q-1}
\end{equation}
and 
\begin{equation}
\left\{ n^{-1}\|\tilde{V}^{(1)\top}\ddot{u}^{(1)}\|_{\infty}\vee n^{-1}\|\tilde{W}^{(0)\top}\ddot{u}^{(1)}\|_{\infty}\right\} \leq\frac{\lambda}{4\|Q^{-1}\|_{c1}}.\label{eq:Coint-DB}
\end{equation}

\end{lem} Let $\mathcal{S}_{(0)}$ denote the active set of $\tilde{\theta}^{\text{(0)}*}$.
Note that $\mathcal{S}_{(0)}=\mathcal{S}_{x}\cup\mathcal{S}_{z}$
where $\mathcal{S}_{x}$ and $\mathcal{S}_{z}$ are the active sets
of $\beta^{*}$ and $\gamma^{(1)*}$ respectively, and thus $\|\tilde{\theta}^{\text{(0)}*}\|_{0}=\|\beta^{*}\|_{0}+\|\psi^{*}\|_{0}\leq s$
by Assumption \ref{assu:cointegration}. Also, define $\mathcal{S}_{(0)}^{c}:=[p]\backslash\mathcal{S}_{(0)}.$
By Lemma \ref{lem:Coint-DB}, we have 
\begin{eqnarray}
 &  & \dfrac{1}{n}\|\tilde{W}_{\Pi}(\tilde{\theta}-\tilde{\theta}^{*})\|_{2}^{2}+\lambda\|Q\tilde{\theta}\|_{1}\leq\dfrac{\lambda}{2\|Q^{-1}\|_{c1}}\left(\|\tilde{\theta}^{(1)}\|_{1}+\|\tilde{\theta}^{(0)}-\tilde{\theta}^{\text{(0)}*}\|_{1}\right)+\lambda\|\tilde{\theta}^{*}\|_{1}\nonumber \\
 & = & \dfrac{\lambda}{2\|Q^{-1}\|_{c1}}\left(\|(\tilde{\theta}^{(0)}-\tilde{\theta}^{\text{(0)}*})_{\mathcal{S}_{(0)}}\|_{1}+\|(\tilde{\theta}^{(0)}-\tilde{\theta}^{\text{(0)}*})_{\mathcal{S}_{(0)}^{c}}\|_{1}+\|\tilde{\theta}^{(1)}\|_{1}\right)+\lambda\|\tilde{\theta}^{*}\|_{1}\nonumber \\
 & = & \dfrac{\lambda}{2\|Q^{-1}\|_{c1}}\left(\|(\tilde{\theta}^{(0)}-\tilde{\theta}^{\text{(0)}*})_{\mathcal{S}_{(0)}}\|_{1}+\|\tilde{\theta}_{\mathcal{S}_{(0)}^{c}}^{(0)}\|_{1}+\|\tilde{\theta}^{(1)}\|_{1}\right)+\lambda\|\tilde{\theta}^{*}\|_{1}.\label{eq:eqPl1-1}
\end{eqnarray}
Recall that for a generic $\theta$, we have defined $\theta^{(1)}:=\theta_{\mathcal{M}_{1}}$
and $\theta^{(2)}:=\theta_{\mathcal{M}_{2}}.$ Define $\tilde{\theta}^{{\rm co}}:=\left(\tilde{\theta}^{(1)\top},\tilde{\theta}^{(2)\top}\right)^{\top}$
and 
\[
Q^{{\rm co}}:=\begin{pmatrix}I_{p_{c1}} & 0\\
-D^{{\rm co}(2)}A^{\top}[D^{{\rm co}(1)}]^{-1} & I_{p_{c2}}
\end{pmatrix}.
\]
Thus $\tilde{\theta}=(\tilde{\theta}^{{\rm co}\top},\tilde{\beta}^{\top},\tilde{\gamma}^{\top})^{\top}$
and $Q={\rm diag}\left(Q^{{\rm co}},I_{p_{x}},I_{p_{z}}\right)$ is
a block diagonal matrix. 

Further define $\mathcal{S}_{x}^{c}:=\mathcal{M}_{x}\backslash\mathcal{S}_{x}$
and $\mathcal{S}_{z}^{c}:=\mathcal{M}_{z}\backslash\mathcal{S}_{z}$.
We then derive 
\begin{align}
\|Q\tilde{\theta}\|_{1} & =\|Q^{{\rm co}}\tilde{\theta}^{{\rm co}}\|_{1}+\|\tilde{\beta}\|_{1}+\|\tilde{\gamma}\|_{1}\nonumber \\
 & \geq\dfrac{\|\tilde{\theta}^{{\rm co}}\|_{1}}{\|(Q^{{\rm co}})^{-1}\|_{c1}}+\|\tilde{\beta}_{\mathcal{S}_{x}^{c}}\|_{1}+\|\tilde{\gamma}_{\mathcal{S}_{z}^{c}}\|_{1}+\|\tilde{\beta}_{\mathcal{S}_{x}}\|_{1}+\|\tilde{\gamma}_{\mathcal{S}_{z}}\|_{1}\nonumber \\
 & \geq\left(\|\tilde{\theta}^{{\rm co}}\|_{1}+\|\tilde{\beta}_{\mathcal{S}_{x}^{c}}\|_{1}+\|\tilde{\gamma}_{\mathcal{S}_{z}^{c}}\|_{1}\right)\big/\|(Q^{{\rm co}})^{-1}\|_{c1}-\|(\tilde{\beta}-\tilde{\beta}^{*})_{\mathcal{S}_{x}}\|_{1}-\|(\tilde{\gamma}-\tilde{\psi}^{*})_{\mathcal{S}_{z}}\|_{1}+\|\tilde{\beta}^{*}\|_{1}+\|\tilde{\psi}^{*}\|_{1}\nonumber \\
 & =\left(\|\tilde{\theta}^{(1)}\|_{1}+\|\tilde{\theta}^{(2)}\|_{1}+\|\tilde{\beta}_{\mathcal{S}_{x}^{c}}\|_{1}+\|\tilde{\gamma}_{\mathcal{S}_{z}^{c}}\|_{1}\right)\big/\|Q{}^{-1}\|_{c1}-\|(\tilde{\theta}^{(0)}-\tilde{\theta}^{\text{(0)}*})_{\mathcal{S}_{(0)}}\|_{1}+\|\tilde{\theta}^{*}\|_{1}\nonumber \\
 & =\left(\|\tilde{\theta}^{(1)}\|_{1}+\|\tilde{\theta}_{\mathcal{S}_{(0)}^{c}}^{(0)}\|_{1}\right)\big/\|Q^{-1}\|_{c1}-\|(\tilde{\theta}^{(0)}-\tilde{\theta}^{\text{(0)}*})_{\mathcal{S}_{(0)}}\|_{1}+\|\tilde{\theta}^{*}\|_{1}\label{eq:Qtheta ineq}
\end{align}
where the second inequality follows by $\|\tilde{\theta}^{{\rm co}}\|_{1}=\|(Q^{{\rm co}})^{-1}Q^{{\rm co}}\tilde{\theta}^{{\rm co}}\|_{1}\leq\|(Q^{{\rm co}})^{-1}\|_{c1}\|Q^{{\rm co}}\tilde{\theta}^{{\rm co}}\|_{1},$
the third line applies the fact that $\|(Q^{c})^{-1}\|_{1}\geq1$
and the triangular inequality, the fourth line applies $\|(Q^{{\rm co}})^{-1}\|_{c1}=\|Q^{-1}\|_{c1}$,
and the last line applies the fact that the inactive set of $\tilde{\theta}^{\text{(0)}*}$
includes all entries corresponding to $\tilde{\theta}^{(2)}$ and
the inactive entries in $\beta^{*}$ and $\gamma^{(1)*}$, so $\|\tilde{\theta}_{\mathcal{S}_{(0)}^{c}}^{(0)}\|_{1}=\|\tilde{\theta}^{(2)}\|_{1}+\|\tilde{\beta}_{\mathcal{S}_{x}^{c}}\|_{1}+\|\tilde{\gamma}_{\mathcal{S}_{z}^{c}}\|_{1}$. 

We substitute (\ref{eq:Qtheta ineq}) into (\ref{eq:eqPl1-1}) and
rearrange 
\begin{eqnarray*}
 &  & \dfrac{1}{n}\|\tilde{W}_{\Pi}(\tilde{\theta}-\tilde{\theta}^{*})\|_{2}^{2}+\dfrac{\lambda}{\|Q^{-1}\|_{c1}}\left(\|\tilde{\theta}_{\mathcal{S}_{(0)}^{c}}^{(0)}\|_{1}+\|\tilde{\theta}^{(1)}\|_{1}\right)\\
 & \leq & \lambda\|(\tilde{\theta}^{(0)}-\tilde{\theta}^{\text{(0)}*})_{\mathcal{S}_{(0)}}\|_{1}+\dfrac{\lambda}{2\|Q^{-1}\|_{c1}}\left(\|(\tilde{\theta}^{(0)}-\tilde{\theta}^{\text{(0)}*})_{\mathcal{S}_{(0)}}\|_{1}+\|\tilde{\theta}_{\mathcal{S}_{(0)}^{c}}^{(0)}\|_{1}+\|\tilde{\theta}^{(1)}\|_{1}\right)\\
 & = & \dfrac{3\lambda}{2}\|(\tilde{\theta}^{(0)}-\tilde{\theta}^{\text{(0)}*})_{\mathcal{S}_{(0)}}\|_{1}+\dfrac{\lambda}{2\|Q^{-1}\|_{c1}}\left(\|\tilde{\theta}_{\mathcal{S}_{(0)}^{c}}^{(0)}\|_{1}+\|\tilde{\theta}^{(1)}\|_{1}\right),
\end{eqnarray*}
and further rearranging the above inequality yields 
\begin{equation}
\dfrac{2}{n}\|\tilde{W}_{\Pi}(\tilde{\theta}-\tilde{\theta}^{*})\|_{2}^{2}+\dfrac{\lambda}{\|Q^{-1}\|_{c1}}\left(\|\tilde{\theta}_{\mathcal{S}_{(0)}^{c}}^{(0)}\|_{1}+\|\tilde{\theta}^{(1)}\|_{1}\right)\leq3\lambda\|(\tilde{\theta}^{(0)}-\tilde{\theta}^{\text{(0)}*})_{\mathcal{S}_{(0)}}\|_{1}.\label{eq:inequality-Lasso-proof}
\end{equation}
This expression bounds the in-sample fitting 
\begin{equation}
\dfrac{1}{n}\left\Vert \ddot{W}\hat{\theta}-(\ddot{X}_{t-1}^{\top}\beta^{*}+\ddot{Z}_{t-1}^{\top}\gamma^{*(1)})\right\Vert _{2}^{2}=\dfrac{\|\tilde{W}_{\Pi}(\tilde{\theta}-\tilde{\theta}^{*})\|_{2}^{2}}{n}\leq\dfrac{3\lambda}{2}\|\tilde{\theta}^{(0)}-\tilde{\theta}^{\text{(0)}*}\|_{1}.\label{eq:coint-in-sample}
\end{equation}

To characterize the rate of convergence, we expand the first term
\begin{align*}
\dfrac{1}{n}\|\tilde{W}_{\Pi}(\tilde{\theta}-\tilde{\theta}^{*})\|_{2}^{2} & =(\tilde{\theta}^{(0)}-\tilde{\theta}^{\text{(0)}*})^{\top}\tilde{\Sigma}^{(0)}(\tilde{\theta}^{(0)}-\tilde{\theta}^{\text{(0)}*})+2(\tilde{\theta}^{(0)}-\tilde{\theta}^{\text{(0)}*})^{\top}\tilde{\Sigma}^{(01)}\tilde{\theta}^{(1)}+\tilde{\theta}^{(1)}\tilde{\Sigma}^{(1)}\tilde{\theta}^{(1)}\\
 & \ge(\tilde{\theta}^{(0)}-\tilde{\theta}^{\text{(0)}*})^{\top}\tilde{\Sigma}^{(0)}(\tilde{\theta}^{(0)}-\tilde{\theta}^{\text{(0)}*})-2\left|(\tilde{\theta}^{(0)}-\tilde{\theta}^{\text{(0)}*})^{\top}\tilde{\Sigma}^{(01)}\tilde{\theta}^{(1)}\right|+\tilde{\theta}^{(1)}\tilde{\Sigma}^{(1)}\tilde{\theta}^{(1)}
\end{align*}
where $\tilde{\Sigma}^{(0)}:=\tilde{W}^{(0)\top}\tilde{W}^{(0)}/n$
, $\tilde{\Sigma}^{(01)}:=\tilde{W}^{(0)\top}\tilde{v}^{(1)}/n$,
and $\tilde{\Sigma}^{(1)}:=\tilde{v}^{(1)\top}\tilde{v}^{(1)}/n$.
The following lemma controls the magnitude of the cross term. 

\begin{lem}\label{lem:Coint-DB-2}Under the conditions of Theorem
\ref{thm:cointegration}, there exists an absolute constant $C$ such
that $\|\tilde{\Sigma}^{(01)}\|_{\max}\leq C(\log p)^{\frac{3}{2}+\frac{1}{2r}}/\sqrt{n}$
w.p.a.1. 

\end{lem}By Lemma \ref{lem:Coint-DB-2}, w.p.a.1\@.~we have
\begin{align}
\left|(\tilde{\theta}^{(0)}-\tilde{\theta}^{\text{(0)}*})^{\top}\tilde{\Sigma}^{(01)}\tilde{\theta}^{(1)}\right| & \leq\|\tilde{\Sigma}^{(01)}\|_{\max}\|\tilde{\theta}^{(0)}-\tilde{\theta}^{\text{(0)}*}\|_{1}\cdot\|\tilde{\theta}^{(1)}\|_{1}\nonumber \\
 & \leq Cn^{-1/2}(\log p)^{\frac{3}{2}+\frac{1}{2r}}\left(\|(\tilde{\theta}^{(0)}-\tilde{\theta}^{\text{(0)}*})_{\mathcal{S}_{(0)}}\|_{1}+\|\tilde{\theta}_{\mathcal{S}_{(0)}^{c}}^{(0)}\|_{1}\right)\cdot\|\tilde{\theta}^{(1)}\|_{1}\nonumber \\
 & \leq Cn^{-1/2}(\log p)^{\frac{3}{2}+\frac{1}{2r}}\cdot(1+3\|Q^{-1}\|_{c1})\cdot3\|Q^{-1}\|_{c1}\cdot\|(\tilde{\theta}^{(0)}-\tilde{\theta}^{\text{(0)}*})_{\mathcal{S}_{(0)}}\|_{1}^{2}\nonumber \\
 & \leq Cq_{n}\|(\tilde{\theta}^{(0)}-\tilde{\theta}^{\text{(0)}*})_{\mathcal{S}_{(0)}}\|_{2}^{2}.\label{eq:coint-cross-prod-bound}
\end{align}
where $q_{n}:=s\cdot n^{-1/2}(\log p)^{\frac{3}{2}+\frac{1}{2r}}\cdot(1+3\|Q^{-1}\|_{c1})\cdot3\|Q^{-1}\|_{c1}.$ 

We substitute 
\[
\dfrac{1}{n}\|\tilde{W}_{\Pi}(\tilde{\theta}-\tilde{\theta}^{*})\|_{2}^{2}\geq(\tilde{\theta}^{(0)}-\tilde{\theta}^{\text{(0)}*})^{\top}\tilde{\Sigma}^{(0)}(\tilde{\theta}^{(0)}-\tilde{\theta}^{\text{(0)}*})+\tilde{\theta}^{(1)}\tilde{\Sigma}^{(1)}\tilde{\theta}^{(1)}-2Cq_{n}\|(\tilde{\theta}^{(0)}-\tilde{\theta}^{\text{(0)}*})_{\mathcal{S}_{(0)}}\|_{2}^{2}
\]
 into (\ref{eq:inequality-Lasso-proof}) and add $\frac{\lambda}{\|Q^{-1}\|_{c1}}\|(\tilde{\theta}^{(0)}-\tilde{\theta}^{\text{(0)}*})_{\mathcal{S}_{(0)}}\|_{1}$
to both sides: 
\begin{align}
 & 2(\tilde{\theta}^{(0)}-\tilde{\theta}^{\text{(0)}*})^{\top}\tilde{\Sigma}^{(0)}(\tilde{\theta}^{(0)}-\tilde{\theta}^{\text{(0)}*})+2\tilde{\theta}^{(1)}\tilde{\Sigma}^{(1)}\tilde{\theta}^{(1)}+\dfrac{\lambda}{\|Q^{-1}\|_{c1}}\left(\|\tilde{\theta}^{(0)}-\tilde{\theta}^{\text{(0)}*}\|_{1}+\|\tilde{\theta}^{(1)}\|_{1}\right)\nonumber \\
\leq & \left(3+\dfrac{1}{\|Q^{-1}\|_{c1}}\right)\lambda\|(\tilde{\theta}^{(0)}-\tilde{\theta}^{\text{(0)}*})_{\mathcal{S}_{(0)}}\|_{1}+4Cq_{n}\|(\tilde{\theta}^{(0)}-\tilde{\theta}^{\text{(0)}*})_{\mathcal{S}_{(0)}}\|_{2}^{2}\nonumber \\
\leq & 4\lambda\sqrt{s}\|(\tilde{\theta}^{(0)}-\tilde{\theta}^{\text{(0)}*})_{\mathcal{S}_{(0)}}\|_{2}+4Cq_{n}\|(\tilde{\theta}^{(0)}-\tilde{\theta}^{\text{(0)}*})_{\mathcal{S}_{(0)}}\|_{2}^{2}\label{eq:inequality-Lasso-proof-2}
\end{align}
where the second inequality applies $\|(\tilde{\theta}^{(0)}-\tilde{\theta}^{\text{(0)}*})_{\mathcal{S}_{(0)}}\|_{1}\leq\sqrt{s}\|(\tilde{\theta}^{(0)}-\tilde{\theta}^{\text{(0)}*})_{\mathcal{S}_{(0)}}\|_{2}$
and $\|Q^{-1}\|_{c1}\geq1$.

The first term of the left-hand side of the above display inequality
is governed by the restricted eigenvalue of $\tilde{\Sigma}^{(0)}$,
which we denote as $\tilde{\kappa}^{(0)}=\kappa_{I}(\tilde{\Sigma}^{(0)},3\|Q^{-1}\|_{c1},s).$ 

\medskip

\begin{lem}\label{lem:Coint-RE}Under the conditions of Theorem \ref{thm:cointegration},
we have $\tilde{\kappa}^{(0)}\stackrel{\mathrm{p}}{\succcurlyeq}1/\left(s(\log p)^{6}\right).$

\end{lem}

\medskip

Recall that $\|\tilde{\theta}_{\mathcal{S}_{(0)}^{c}}^{(0)}\|_{1}\leq3\|Q^{-1}\|_{c1}\cdot\|(\tilde{\theta}^{(0)}-\tilde{\theta}^{\text{(0)}*})_{\mathcal{S}_{(0)}}\|_{1}$
by (\ref{eq:inequality-Lasso-proof}), and thus $\tilde{\theta}^{(0)}-\tilde{\theta}^{\text{(0)}*}\in\mathcal{R}(3\|Q^{-1}\|_{c1},s).$
Then Lemma \ref{lem:Coint-RE} implies
\begin{equation}
\|(\tilde{\theta}^{(0)}-\tilde{\theta}^{\text{(0)}*})_{\mathcal{S}_{(0)}}\|_{2}^{2}\leq\dfrac{1}{\tilde{\kappa}^{(0)}}(\tilde{\theta}^{(0)}-\tilde{\theta}^{\text{(0)}*})^{\top}\tilde{\Sigma}^{(0)}(\tilde{\theta}^{(0)}-\tilde{\theta}^{\text{(0)}*}).\label{eq:RE_temp}
\end{equation}
We continue (\ref{eq:inequality-Lasso-proof-2}):
\begin{align}
 & 2(\tilde{\theta}^{(0)}-\tilde{\theta}^{\text{(0)}*})^{\top}\tilde{\Sigma}^{(0)}(\tilde{\theta}^{(0)}-\tilde{\theta}^{\text{(0)}*})+2\tilde{\theta}^{(1)}\tilde{\Sigma}^{(1)}\tilde{\theta}^{(1)}+\dfrac{\lambda}{\|Q^{-1}\|_{c1}}\left(\|\tilde{\theta}^{(0)}-\tilde{\theta}^{\text{(0)}*}\|_{1}+\|\tilde{\theta}^{(1)}\|_{1}\right)\nonumber \\
\leq & 4\lambda\sqrt{\dfrac{s}{\tilde{\kappa}^{(0)}}}\cdot\sqrt{(\tilde{\theta}^{(0)}-\tilde{\theta}^{\text{(0)}*})^{\top}\tilde{\Sigma}^{(0)}(\tilde{\theta}^{(0)}-\tilde{\theta}^{\text{(0)}*})}+\frac{4Cq_{n}}{\tilde{\kappa}^{(0)}}\cdot(\tilde{\theta}^{(0)}-\tilde{\theta}^{\text{(0)}*})^{\top}\tilde{\Sigma}^{(0)}(\tilde{\theta}^{(0)}-\tilde{\theta}^{\text{(0)}*})\nonumber \\
\leq & \dfrac{4\lambda^{2}s}{\tilde{\kappa}^{(0)}}+\left(1+4C\frac{q_{n}}{\tilde{\kappa}^{(0)}}\right)(\tilde{\theta}^{(0)}-\tilde{\theta}^{\text{(0)}*})^{\top}\tilde{\Sigma}^{(0)}(\tilde{\theta}^{(0)}-\tilde{\theta}^{\text{(0)}*})\label{eq:before the key ineq}
\end{align}
holds w.p.a.1, where the first inequality follows by (\ref{eq:RE_temp}),
and the second inequality by the generic inequality $4ab\leq4a^{2}+b^{2}.$
Note that 
\begin{equation}
\frac{q_{n}}{\tilde{\kappa}^{(0)}}=\frac{3s}{\sqrt{n}\tilde{\kappa}^{(0)}}(\log p)^{\frac{3}{2}+\frac{1}{2r}}(1+3\|Q^{-1}\|_{c1})\|Q^{-1}\|_{c1}\stackrel{\mathrm{p}}{\preccurlyeq}\dfrac{s^{2}}{\sqrt{n}}\left(\log p\right)^{\frac{19}{2}+\frac{1}{2r}}\to0\label{eq:proof_logpn}
\end{equation}
where the second inequality applies (\ref{eq:L1Q-1}) and Lemma \ref{lem:Coint-RE},
and the limit applies the relative size of $n$, $p$ and $s$ specified
in Assumption \ref{assu:asym_n}. Thus, $1+4Cq_{n}/\tilde{\kappa}^{(0)}\leq1.5$
holds with w.p.a.1.~as sample size is sufficiently large. Rearranging
(\ref{eq:before the key ineq}) yields
\[
\frac{1}{2}(\tilde{\theta}^{(0)}-\tilde{\theta}^{\text{(0)}*})^{\top}\tilde{\Sigma}^{(0)}(\tilde{\theta}^{(0)}-\tilde{\theta}^{\text{(0)}*})+2\tilde{\theta}^{(1)}\tilde{\Sigma}^{(1)}\tilde{\theta}^{(1)}+\dfrac{1}{\|Q^{-1}\|_{c1}}\lambda\left(\|\tilde{\theta}^{(0)}-\tilde{\theta}^{\text{(0)}*}\|_{1}+\|\tilde{\theta}^{(1)}\|_{1}\right)\leq\dfrac{4\lambda^{2}s}{\tilde{\kappa}^{(0)}},
\]
which immediately implies 
\[
\|\tilde{\theta}^{(0)}-\tilde{\theta}^{\text{(0)}*}\|_{1}+\|\tilde{\theta}^{(1)}\|_{1}\leq\frac{4\lambda s}{\tilde{\kappa}^{(0)}}\|Q^{-1}\|_{c1}\stackrel{\mathrm{p}}{\preccurlyeq}\frac{4\lambda s}{\tilde{\kappa}^{(0)}}\log p
\]
and the quality of the in-sample fitting 
\[
\dfrac{1}{n}\left\Vert \ddot{W}\hat{\theta}-(\ddot{X}_{t-1}^{\top}\beta^{*}+\ddot{Z}_{t-1}^{\top}\gamma^{*(1)})\right\Vert _{2}^{2}\leq\dfrac{3\lambda}{2}\|\tilde{\theta}^{(0)}-\tilde{\theta}^{\text{(0)}*}\|_{1}\stackrel{\mathrm{p}}{\preccurlyeq}\dfrac{s^{2}}{n}(\log p)^{9+\frac{1}{r}}
\]
in view of (\ref{eq:coint-in-sample}). 

In terms of parameter estimation, the unit root and the stationary
components are governed by 
\begin{eqnarray*}
\|\hat{\beta}-\beta^{*}\|_{1} & \leq & \|\tilde{\theta}^{(0)}-\tilde{\theta}^{\text{(0)}*}\|_{1}/\min_{j\in\mathcal{M}_{x}}\hat{\sigma}_{j}^{X}\stackrel{\mathrm{p}}{\preccurlyeq}\dfrac{4\lambda s/\tilde{\kappa}^{(0)}\log p}{\sqrt{n/\log p}}\stackrel{\mathrm{p}}{\preccurlyeq}\dfrac{s^{2}}{n}(\log p)^{10+\frac{1}{2r}}\\
\|\hat{\gamma}-\gamma^{(1)*}\|_{1} & \leq & \|\tilde{\theta}^{(0)}-\tilde{\theta}^{\text{(0)}*}\|_{1}/\min_{j\in\mathcal{M}_{z}}\hat{\sigma}_{j}^{Z}\stackrel{\mathrm{p}}{\preccurlyeq}4\lambda s/\tilde{\kappa}^{(0)}\log p\stackrel{\mathrm{p}}{\preccurlyeq}\dfrac{s^{2}}{\sqrt{n}}(\log p)^{\frac{19}{2}+\frac{1}{2r}},
\end{eqnarray*}
respectively. The coefficients for the cointegrated variables shrink
toward zero as
\begin{align*}
\|\hat{\phi}_{1}\|_{1} & \leq\dfrac{\|\tilde{\theta}^{(1)}\|_{1}}{\min_{j\in[k_{1}]}\hat{\sigma}_{j}^{(1)}}\stackrel{\mathrm{p}}{\preccurlyeq}\dfrac{4\lambda s/\tilde{\kappa}^{(0)}\log p}{\sqrt{n/\log p}}\stackrel{\mathrm{p}}{\preccurlyeq}\dfrac{s^{2}}{n}(\log p)^{10+\frac{1}{2r}}
\end{align*}
and by the triangular inequality 
\begin{align*}
\|\hat{\phi}_{2}\|_{1} & \leq\|\hat{\phi}_{2}+A^{\top}\hat{\phi}_{1}\|_{1}+\|A^{\top}\hat{\phi}_{1}\|_{1}\stackrel{\mathrm{p}}{\preccurlyeq}\|\tilde{\theta}^{(0)}-\tilde{\theta}^{\text{(0)}*}\|_{1}/\min_{j\in\mathcal{M}_{2}}\hat{\sigma}_{j}^{(2)}+\|A\|_{r1}\cdot\|\hat{\phi}_{1}\|_{1}\\
 & \stackrel{\mathrm{p}}{\preccurlyeq}\dfrac{4\lambda s/\tilde{\kappa}^{(0)}\log p}{\sqrt{n/\log p}}+\dfrac{s^{2}}{n}(\log p)^{10+\frac{1}{2r}}\stackrel{\mathrm{p}}{\preccurlyeq}\dfrac{s^{2}}{n}(\log p)^{10+\frac{1}{2r}}.
\end{align*}
Hence $\|\hat{\phi}\|_{1}\stackrel{\mathrm{p}}{\preccurlyeq}\dfrac{s^{2}}{n}(\log p)^{10+\frac{1}{2r}}$. 
\end{proof}

\subsection{Proofs of Lemmas and Their Corollaries\label{subsec:Proofs-of-Lemmas}}
\begin{proof}[Proof of Lemma \ref{lem:mixing}]
 We first derive a generic inequality. Let $b>0$, $r>0$ and $a\in\mathbb{N}$.
If $a^{\frac{1}{r}-1}\leq\exp\left(\frac{b}{2}a\right)$, we have
\begin{align}
\sum_{t=a}^{\infty}\exp\left(-bt^{r}\right) & \leq\int_{a}^{\infty}\exp\left(-bt^{r}\right)dt=\int_{a^{r}}^{\infty}y^{\frac{1}{r}-1}\exp\left(-by\right)dy\nonumber \\
 & \leq\int_{a^{r}}^{\infty}\exp\left(-\frac{by}{2}\right)dy=\frac{1}{b}\exp\left(-\frac{ba^{r}}{2}\right).\label{eq:gene_ineq1}
\end{align}
Since $a\in\mathbb{N}$, the condition $a^{\frac{1}{r}-1}\leq\exp\left(\frac{b}{2}a\right)$
is trivial whenever $r\geq1$, while it also holds for a sufficiently
large $a$ if $r\in(0,1).$

We apply \citet[p.411]{gorodetskii1978strong}'s Theorem: When $d$
is large enough, there exists a constant $C_{g}$ such that 
\[
\alpha(\text{(\ensuremath{\varepsilon_{jt})_{t\in\mathbb{Z}}}},d)\leq C_{g}\sum_{m=d}^{\infty}\left(\sum_{q=m}^{\infty}|\psi_{jq}|\right)^{1/2}\leq C_{g}\sqrt{C_{\psi}}\sum_{m=d}^{\infty}\left(\sum_{q=m}^{\infty}\exp\left(-c_{\psi}q^{r}\right)\right)^{1/2}
\]
where the second inequality follows by Assumption \ref{assu:alpha}.
When $d$ is sufficiently large so that $d^{\frac{1}{r}-1}\leq\exp\left(\frac{b}{2}d\right)$,
we apply (\ref{eq:gene_ineq1}) to yield
\begin{eqnarray*}
 &  & C_{g}\sqrt{C_{\psi}}\sum_{m=d}^{\infty}\left(\sum_{q=m}^{\infty}\exp\left(-c_{\psi}q^{r}\right)\right)^{1/2}\leq C_{g}\sqrt{C_{\psi}}\sum_{m=d}^{\infty}\left(\frac{1}{c_{\psi}}\exp\left(-\frac{c_{\psi}}{2}m^{r}\right)\right)^{1/2}\\
 & = & C_{g}\sqrt{\frac{C_{\psi}}{c_{\psi}}}\sum_{m=d}^{\infty}\exp\left(-\frac{c_{\psi}}{4}m^{r}\right)\leq C_{g}\sqrt{\frac{C_{\psi}}{c_{\psi}}}\frac{4}{c_{\psi}}\exp\left(-\frac{c_{\psi}}{8}m^{r}\right)=\widetilde{C}_{\alpha}\exp\left(-c_{\alpha}k^{r}\right)
\end{eqnarray*}
where $\widetilde{C}_{\alpha}=4C_{g}C_{\psi}^{1/2}c_{\psi}^{-3/2}$
and $c_{\alpha}=c_{\psi}/8$. By \citet[Theorem 1]{bradley1993equivalent}
and the i.i.d.~of $\eta_{jt}$, we have 
\[
\rho(\text{(\ensuremath{\varepsilon_{jt})_{t\in\mathbb{Z}}}},d)\leq2\pi\alpha(\text{(\ensuremath{\varepsilon_{jt})_{t\in\mathbb{Z}}}},d)\leq C_{\alpha}\exp\left(-c_{\alpha}d^{r}\right)
\]
where $C_{\alpha}=2\pi\widetilde{C}_{\alpha}.$ \citet[Theorem 5.2 (b)]{bradley2005basic},
together with the independence of the components in $\varepsilon_{t}$,
implies 
\[
\alpha(\text{(\ensuremath{\varepsilon_{jt})_{t\in\mathbb{Z}}}},d)\leq\rho(\text{(\ensuremath{\varepsilon_{jt})_{t\in\mathbb{Z}}}},d)\leq\rho(\varepsilon,d)\leq C_{\alpha}\exp\left(-c_{\alpha}d^{r}\right)
\]
when $d$ is sufficiently large. 
\end{proof}
\medskip
\begin{proof}[Proof of Lemma \ref{lem:BernsteinSum}]
 The triangular inequality and the Markov inequality give 
\[
\max_{j\in[p]}\Pr\left\{ |x_{jt}-\mathbb{E}x_{jt}|>\mu\right\} \leq\max_{j\in[p]}\Pr\left\{ |x_{jt}|>\mu-|\mathbb{E}x_{jt}|\right\} \leq C_{x}\exp\left[-\left((\mu-|\mathbb{E}x_{jt})/b_{x}\right)\right]
\]
for all $\mu>|\mathbb{E}x_{jt}|$. This tail bound allows us to invoke
\citet[p.441]{merlevede2011bernstein}'s Theorem 1 and Remark 1 to
obtain that along the path of $t\in[n]$ the partial sum
\begin{align*}
 & \Pr\left\{ \max_{t\in[n]}|\sum_{s=1}^{t}(x_{js}-\mathbb{E}x_{js})|>\mu\right\} \\
\leq & n\exp\left(-\dfrac{\mu^{r^{*}}}{C_{1}}\right)+\exp\left(-\dfrac{\mu^{2}}{C_{2}(1+nV)}\right)+\exp\left(-\dfrac{\mu^{2}}{C_{3}n}\exp\left(\dfrac{\mu^{r^{*}(1-r^{*})}}{C_{4}(\log\mu)^{r^{*}}}\right)\right)
\end{align*}
for all $j\in[p]$, where $C_{1},C_{2},C_{3},C_{4}$ and $V$ are
absolute constants, and $r^{*}=(1+1/r)^{-1}\in(0,1)$. The union bound
is 
\begin{eqnarray*}
 &  & \Pr\left\{ \max_{j\in[p]}\max_{t\in[n]}|\sum_{s=1}^{t}(x_{js}-\mathbb{E}x_{js})|>\mu\right\} \\
 & \leq & p\left(n\exp\left(-\dfrac{\mu^{r^{*}}}{C_{1}}\right)+\exp\left(-\dfrac{\mu^{2}}{C_{2}(1+nV)}\right)+\exp\left(-\dfrac{\mu^{2}}{C_{3}n}\exp\left(\dfrac{\mu^{r^{*}(1-r^{*})}}{C_{4}(\log\mu)^{r^{*}}}\right)\right)\right)
\end{eqnarray*}
Set $\text{\ensuremath{\mu=C\sqrt{n\log p}}}$ with $C^{2}\geq2(C_{2}+1)V+2C_{3}.$
Recall that $(\log p)^{2/r^{*}-1}=(\log p)^{1+2/r}=o(n)$ and thus
$\log n+\log p=o\left(\left(n\log p\right)^{r^{*}/2}\right).$ When
$n$ is sufficiently large, all the three terms on the right-hand
side of the above expression shrinks to zero as 
\begin{align*}
pn\exp\left(-\dfrac{\mu^{r^{*}}}{C_{1}}\right) & =\exp\left(\log p+\log n-\dfrac{\left(C^{2}n\log p\right)^{r^{*}/2}}{C_{1}}\right)\to0\\
p\exp\left(-\dfrac{\mu^{2}}{C_{2}(1+nV)}\right) & =\exp\left(\left(1-\dfrac{C^{2}n}{C_{2}nV+C_{2}}\right)\log p\right)\leq p^{-1}\to0\\
p\exp\left(-\dfrac{\mu^{2}}{C_{3}n}\exp\left(\dfrac{\mu^{r^{*}(1-r^{*})}}{C_{4}(\log\mu)^{r^{*}}}\right)\right) & \leq\exp\left(\left(1-\dfrac{C^{2}n}{C_{3}n}\right)\log p\right)\leq p^{-1}\to0.
\end{align*}
We complete the proof.
\end{proof}
\medskip
\begin{proof}[Proof of Lemma \ref{lem:LinComb} ]
 \citet[Lemma 5]{wong2020lasso} gives the following inequality:
for a generic random variable $x_{i}$ with finite $\mathbb{E}(|x_{i}|^{b})$
for some $b>0$, there exists some $K_{x}>0$ such that 
\begin{equation}
\mathbb{E}(|x_{i}|^{b})\leq K_{i}^{b}b^{b}.\label{eq:Wong_lemma5}
\end{equation}

Now, let $K_{x}=\max_{i\in\mathbb{N}}K_{i}$. For any $a\in\mathbb{R}$,
it implies 
\begin{align*}
\mathbb{E}\left[\exp\left(\left|ax\right|\right)\right] & =1+\sum_{j=1}^{\infty}\frac{1}{j!}\left|a\right|^{j}\mathbb{E}\left[\left|x\right|^{j}\right]\leq1+\sum_{j=1}^{\infty}\left(\frac{\mathrm{e}}{j}\right)^{j}\left|a\right|^{j}\mathbb{E}\left[\left|x\right|^{j}\right]\\
 & \leq1+\sum_{j=1}^{\infty}\left(\frac{\mathrm{e}}{j}\right)^{j}\left|a\right|^{j}(jK_{x})^{j}=1+\sum_{j=1}^{\infty}\left(K_{x}\text{{\rm e}}\left|a\right|\right)^{j}.
\end{align*}
Let $\tau_{x}=(2K_{x}{\rm e}\|a\|_{\infty})^{-1}$ and then for each
$i$ we have 
\begin{align*}
\mathbb{E}\left[\exp\left[\tau_{x}\left(|a_{i}x_{i}|\right)\right]\right] & \leq1+\sum_{d=1}^{\infty}\left[|\tau_{x}K_{x}a_{i}|{\rm e}\right]^{d}\leq1+\sum_{d=1}^{\infty}\left[\frac{1}{2}\left(\left|a_{i}\right|/\|a\|_{\infty}\right)\right]^{d}\\
 & =1+\dfrac{\frac{1}{2}\left(\left|a_{i}\right|/\|a\|_{\infty}\right)}{1-\frac{1}{2}\left(\left|a_{i}\right|/\|a\|_{\infty}\right)}\leq1+\ensuremath{\dfrac{\left|a_{i}\right|}{\|a\|_{\infty}}}\leq\exp\left(\ensuremath{\dfrac{\left|a_{i}\right|}{\|a\|_{\infty}}}\right).
\end{align*}
Since $\left|\sum_{i\in\mathbb{N}}a_{i}x_{i}\right|\leq\sum_{i\in\mathbb{N}}|a_{i}x_{i}|$,
it further implies 
\begin{align*}
\mathbb{E}\left[\exp\left(\tau_{x}\left|\sum_{i\in\mathbb{N}}a_{i}x_{i}\right|\right)\right] & \leq\mathbb{E}\left[\exp\left(\tau_{x}\sum_{i\in\mathbb{N}}|a_{i}x_{i}|\right)\right]=\prod_{i\in\mathbb{N}}\mathbb{E}\left[\exp\left(\tau_{x}|a_{i}x_{i}|\right)\right]\\
 & \leq\prod_{i\in\mathbb{N}}\exp\left(\ensuremath{\dfrac{\left|a_{i}\right|}{\|a\|_{\infty}}}\right)\leq\exp\left(\ensuremath{\dfrac{\sum_{i\in\mathbb{N}}\left|a_{i}\right|}{\|a\|_{\infty}}}\right)=\exp\left(\ensuremath{\frac{\left\Vert a\right\Vert _{1}}{\|a\|_{\infty}}}\right)
\end{align*}
where the equality follows by the independence of $\{x_{i}\}_{i\in\mathbb{N}}$.
By the Markov inequality we have 
\begin{align*}
\Pr\left\{ \left|\sum_{i\in\mathbb{N}}a_{i}x_{i}\right|>\mu\right\}  & \leq\mathrm{e}^{-\mu\tau_{x}}\cdot\mathbb{E}\left[\exp\left(\tau_{x}\left|\sum_{i\in\mathbb{N}}a_{i}x_{i}\right|\right)\right]\leq\exp\left(-\dfrac{\mu}{2K_{x}{\rm e}\|a\|_{\infty}}+\ensuremath{\dfrac{\left\Vert a\right\Vert _{1}}{\|a\|_{\infty}}}\right).
\end{align*}
\end{proof}
\medskip
\begin{proof}[Proof of Corollary \ref{cor:BNtail}]
 Recall that $\varepsilon_{jt}=\sum_{d=0}^{\infty}\psi_{jd}\eta_{j,t-d}$
is a linear process with $(\eta_{j,t-d})_{d\in\mathbb{N}}$ independent
over the cross section and the time, and $\eta_{jt}$ satisfies (\ref{eq:xprobineq})
with $C_{x}=C_{\eta}$ and $b_{x}=b_{\eta}$. By Lemma \ref{lem:LinComb},
to verify (\ref{eq:epsBound}) it suffices to show that $\sum_{d=0}^{\infty}|\psi_{jd}|$
are uniformly bounded by some absolute constant for all $j\in[p+1]$.
Under Assumption \ref{assu:alpha}, the uniform bound holds as 
\begin{align}
\sum_{d=0}^{\infty}|\psi_{jd}| & \leq C_{\psi}\sum_{d=0}^{\infty}\exp\left(-c_{\psi}d^{r}\right)\leq C_{\psi}\left(M+\sum_{d=M}^{\infty}\exp\left(-c_{\psi}d^{r}\right)\right)\nonumber \\
 & \leq C_{\psi}\left(M+\dfrac{1}{c_{\psi}}\exp\left(-\dfrac{c_{\psi}M^{r}}{2}\right)\right)\label{eq:psi_inf_sum_bound}
\end{align}
where the last inequality applies (\ref{eq:gene_ineq1}) with a sufficiently
large integer $M$ so that $M^{\frac{1}{r}-1}\le\exp\left(\frac{c_{\psi}}{2}M\right)$. 

Similarly by Lemma \ref{lem:LinComb}, to verify (\ref{eq:etaWoldSumBound})
it suffices to show that $\sum_{d=0}^{\infty}|\widetilde{\psi}_{jd}|$
\begin{align*}
\sum_{d=0}^{\infty}|\widetilde{\psi}_{jd}| & \leq\sum_{d=0}^{\infty}\sum_{\ell=d+1}^{\infty}|\psi_{j\ell}|\leq C_{\psi}\sum_{d=0}^{\infty}\sum_{\ell=d+1}^{\infty}\exp\left(-c_{\psi}\ell^{r}\right)\leq\dfrac{C_{\psi}}{c_{\psi}}\sum_{d=0}^{\infty}\exp\left(-\frac{c_{\psi}}{2}(d+1)^{r}\right)\\
 & \leq\dfrac{C_{\psi}}{c_{\psi}}\left(\tilde{M}-1+\sum_{d=M}^{\infty}\exp\left(-\frac{c_{\psi}}{2}d^{r}\right)\right)\leq\dfrac{C_{\psi}}{c_{\psi}}\left(\tilde{M}-1+\frac{2}{c_{\psi}}\exp\left(-\frac{c_{\psi}\tilde{M}^{r}}{4}\right)\right)
\end{align*}
is uniformly bounded for all $j\in[p+1]$, where the third and the
last inequalities apply (\ref{eq:gene_ineq1}) with sufficiently large
integer $\tilde{M}$ so that $\tilde{M}^{\frac{1}{r}-1}\le\exp\left(\frac{c_{\psi}}{4}\tilde{M}\right)$. 

Finally we verify (\ref{eq:euBound}). Recall that $v_{jt}=\sum_{\ell=1}^{p+1}\Phi_{j\ell}\varepsilon_{\ell t}$
for all $j\in[p+1]$ as defined by (\ref{eq:def-error}). By (\ref{eq:epsBound})
and the independence of $\varepsilon_{\ell t}$ across all $\ell\in[p+1]$,
we have for any $t$, $(\varepsilon_{\ell t})_{\ell\in[p+1]}$ consists
of $p+1$ independent variables satisfying (\ref{eq:xprobineq}) with
$C_{x}=C_{\eta}^{\prime}$ and $b_{x}=b_{\eta}^{\prime}$. Assumption
\ref{assu:covMat} ensures that 
\[
\max_{\ell\in[p+1]}|\Phi_{j\ell}|\leq\sum_{\ell=1}^{p+1}|\Phi_{j\ell}|\leq C_{L}
\]
for all $j\in[p+1]$, so that we can invoke (\ref{eq:lincomb_p})
by specifying $\|a\|_{\infty}=\max_{\ell\in[p+1]}|\Phi_{j\ell}|$
and $\|a\|_{1}=\sum_{\ell=1}^{p+1}|\Phi_{j\ell}|$ and choose a sufficient
large $\mu$. 
\end{proof}
\medskip
\begin{proof}[Proof of Lemma \ref{lem:GaussianApprox}]
The Beveridge-Nelson decomposition makes $\left\{ \varepsilon_{jt}\right\} $
as 
\[
\varepsilon_{jt}=\psi_{j}(1)\eta_{jt}-(\widetilde{\varepsilon}_{jt}-\widetilde{\varepsilon}_{j,t-1}),\ \text{ where }\widetilde{\varepsilon}_{jt}=\sum_{d=0}^{\infty}\widetilde{\psi}_{jd}\eta_{j,t-d},\ \text{and }\widetilde{\psi}_{jd}=\sum_{\ell=d+1}^{\infty}\psi_{j\ell}
\]
and thus the partial sum is
\[
\sum_{s=0}^{t-1}\varepsilon_{js}=\psi_{j}(1)\sum_{s=0}^{t-1}\eta_{js}-\widetilde{\varepsilon}_{j,t-1}+\widetilde{\varepsilon}_{j,-1}.
\]
which deduces 
\[
\dfrac{1}{\sqrt{n}}\left|\sum_{s=0}^{t-1}\varepsilon_{js}-\psi_{j}(1)\mathcal{B}_{j}\left(t\right)\right|\leq|\psi_{j}(1)|\cdot\left|\dfrac{1}{\sqrt{n}}\left(\sum_{s=0}^{t-1}\eta_{js}-\mathcal{B}_{j}\left(t\right)\right)\right|+\dfrac{|\widetilde{\varepsilon}_{j,t-1}|+|\widetilde{\varepsilon}_{j,-1}|}{\sqrt{n}}.
\]
By (\ref{eq:etaWoldSumBound}) in Corollary \ref{cor:BNtail}, taking
$\mu=2\tilde{b}_{\eta}(\log n+\log p),$
\[
\Pr\left\{ |\widetilde{\varepsilon}_{j,t-1}|>2\tilde{b}_{\eta}(\log n+\log p)\right\} \leq\tilde{C}_{\eta}(np)^{-2}
\]
and by the union bound 
\[
\Pr\left\{ \sup_{j\in[p+1],t\in[n]}|\widetilde{\varepsilon}_{j,t-1}|>2\tilde{b}_{\eta}(\log n+\log p)\right\} \leq\tilde{C}_{\eta}n^{-1}p^{-2}(p+1)\to0
\]
and thus 
\[
\sup_{j\in[p+1],t\in[n]}n^{-1/2}\left(|\widetilde{\varepsilon}_{j,t-1}|\vee|\widetilde{\varepsilon}_{j,-1}|\right)=O_{p}\left(n^{-1/2}\left(\log n+\log p\right)\right)=O_{p}\left(\frac{\log p}{\sqrt{n}}\right)
\]
where the last inequality applies $\log n\leq\nu_{1}^{-1}\log p$
stated in the beginning of Section \ref{sec:Proofs}. Furthermore,
$\sup_{j\in[p+1]}|\psi_{j}(1)|=O(1)$ by (\ref{eq:psi_inf_sum_bound}). 

Next we work with $\sup_{j\in[p+1],t\in[n]}\left|\frac{1}{\sqrt{n}}\left(\sum_{s=0}^{t-1}\eta_{js}-\mathcal{B}_{j}\left(t\right)\right)\right|.$
Note that $\eta_{js}$ is sub-exponential by Assumption \ref{assu:tail}
and thus has a finite moment generating function within a compact
interval \citep[Proposition 2.7.1]{vershynin2018high}. Given this
fact, we use the Koml\'{o}s-Major-Tusn\'{a}dy coupling coupling
inequality \citep[Theorem 1]{komlos1976approximation}: for any $\tau>0$
and $j\in[p+1]$, there are absolute constants $C$, $K_{1}$ and
$K_{2}$ such that the following non-asymptotic inequality holds:
\[
\Pr\left\{ \sup_{t\in[n]}\left|\sum_{s=0}^{t-1}\eta_{js}-\mathcal{B}_{j}\left(t\right)\right|>C\log n+\tau\right\} \leq K_{1}\exp(-K_{2}\tau).
\]

Applying the union bound, we obtain 
\[
\Pr\left\{ \sup_{j\in[p+1],t\in[n]}\left|\sum_{s=0}^{t-1}\eta_{js}-\mathcal{B}_{j}\left(t\right)\right|>C\log n+\tau\right\} \leq K_{1}(p+1)\exp(-K_{2}\tau).
\]
Set $\tau=\frac{2}{K_{2}}\log p$ and we obtain
\[
\Pr\left\{ \sup_{j\in[p+1],t\in[n]}\left|\sum_{s=0}^{t-1}\eta_{js}-\mathcal{B}_{j}\left(t\right)\right|>C\log n+\frac{2}{K_{2}}\log p\right\} \leq K_{1}(p+1)p^{-2}\to0.
\]
We thus conclude 
\[
\sup_{j\in[p+1],t\in[n]}\dfrac{1}{\sqrt{n}}\left|\sum_{s=0}^{t-1}\eta_{js}-\mathcal{B}_{j}\left(t\right)\right|=O_{p}\left(\frac{1}{\sqrt{n}}C\log n+\frac{2}{K_{2}}\log p\right)=O_{p}\left(\frac{\log p}{\sqrt{n}}\right).
\]
We complete the proof.
\end{proof}
\medskip
\begin{proof}[Proof of Lemma \ref{lem:Coint-DB}]
\uline{ Proof of (\mbox{\ref{eq:L1Q-1}})}. It is easy to verify
\begin{align*}
Q^{-1} & =D\Pi D^{-1}=\begin{pmatrix}I_{k_{1}}\\
D^{{\rm co}(2)}A^{\top}[D^{{\rm co}(1)}]^{-1} & I_{k_{2}}\\
 &  & I_{p_{x}}\\
 &  &  & I_{p_{z}}
\end{pmatrix}
\end{align*}
and thus $\|Q^{-1}\|_{c1}=\|D^{{\rm co}(2)}A^{\top}[D^{{\rm co}(1)}]^{-1}\|_{c1}+1\geq1$
gives the lower bound. 

Notice that 
\[
\|D^{{\rm co}(2)}A^{\top}[D^{{\rm co}(1)}]^{-1}\|_{c1}\leq\|A\|_{r1}\cdot\frac{\max_{j\in[p_{c2}]}\hat{\sigma}_{j}^{{\rm co}(2)}}{\min_{j\in[p_{c1}]}\hat{\sigma}_{j}^{{\rm co}(1)}}.
\]
Besides, $\min_{j\in[p_{c1}]}\hat{\sigma}_{j}^{{\rm co}(1)}\stackrel{\mathrm{p}}{\succcurlyeq}\sqrt{n/\log p}$
by (\ref{eq:sigma-minmax}), and $\max_{j\in[p_{c2}]}\hat{\sigma}_{j}^{{\rm co}(2)}\stackrel{\mathrm{p}}{\preccurlyeq}\sqrt{n\log p}$
following (\ref{eq:coint_variance}) as $X_{t}^{{\rm co}(2)}$ is
a pure $I(1)$ vector. Then with $p$ large enough,
\[
\|Q^{-1}\|_{c1}\stackrel{\mathrm{p}}{\preccurlyeq}\dfrac{\sqrt{n\log p}}{\sqrt{n/\log p}}+1\leq2\log p.
\]

\uline{Proof of (\mbox{\ref{eq:Coint-DB}}).} We check the orders
of the two terms on the left-hand side of (\ref{lem:Coint-DB}). Recall
that the error term in (\ref{eq:cointegration_y_2}) is $u_{t}^{(1)}=u_{t}+v_{t-1}^{(1)\top}\phi_{1}^{*}-Z_{t-1}^{\top}\omega^{*}$
and thus
\begin{align*}
n^{-1}\|\tilde{v}^{(1)\top}\ddot{u}^{(1)}\|_{\infty} & \leq\dfrac{\|\frac{1}{n}\sum_{t=1}^{n}\ddot{v}_{t-1}^{(1)}\ddot{u}_{t}^{(1)}\|_{\infty}}{\min_{j\in[p_{c1}]}\hat{\sigma}_{j}^{(1)}}\stackrel{\mathrm{p}}{\preccurlyeq}\sqrt{\dfrac{\log p}{n}},
\end{align*}
where the denominator is governed by (\ref{eq:coint_variance}) and
the numerator is controlled by Lemma \ref{lem:LinComb} $\|\frac{1}{n}\sum_{t=1}^{n}\ddot{v}_{t-1}^{(1)}\ddot{u}_{t}^{(1)}\|_{\infty}=O_{p}(1)$
as $v_{t}^{(1)}$ and $u_{t}^{(1)}$ are both stationary, mixing and
sub-exponential.

Recall $\tilde{W}^{(0)}=\left(\tilde{W}_{j}\right)_{j\in[p]\backslash\mathcal{M}_{1}}=\left(\widehat{\sigma}_{j}^{-1}W_{j}\right)_{j\in[p]\backslash\mathcal{M}_{1}}$
is a scale-standardized vector of a mixture of $I(1)$ and $I(0)$
regressors, and $\mathbb{E}(Z_{t-1}u_{t}^{(1)})=0$. Thus 
\[
n^{-1}\|\tilde{W}^{(0)\top}\ddot{u}^{(1)}\|_{\infty}\stackrel{\mathrm{p}}{\preccurlyeq}\sqrt{n^{-1}(\log p)^{3+\frac{1}{r}}}=\frac{1}{\sqrt{n}}(\log p)^{\frac{3}{2}+\frac{1}{2r}}
\]
 following (\ref{eq:DB-mix-std}). We obtain (\ref{lem:Coint-DB})
in view of the choice of $\lambda$ and the order of $\|Q^{-1}\|_{c1}$
in (\ref{eq:L1Q-1}). 
\end{proof}
\medskip
\begin{proof}[Proof of Lemma \ref{lem:Coint-DB-2}]
 Recall that $\tilde{\Sigma}^{(01)}=n^{-1}\sum_{t=1}^{n}\tilde{W}_{t-1}^{(0)}\tilde{v}_{t-1}^{(1)\top}$
and thus 
\begin{equation}
\|\tilde{\Sigma}^{(01)}\|_{\max}\leq\dfrac{\|\sum_{t=1}^{n}\tilde{W}_{t-1}^{(0)}\ddot{v}_{t-1}^{(1)\top}\|_{\max}}{n\cdot\min_{j\in[p_{c1}]}\hat{\sigma}_{j}^{{\rm co}(1)}}\stackrel{\mathrm{p}}{\preccurlyeq}\frac{1}{n}\sqrt{\dfrac{\log p}{n}}\|\sum_{t=1}^{n}\tilde{W}_{t-1}^{(0)}\ddot{v}_{t-1}^{(1)\top}\|_{\max}\label{eq:Sigma_01}
\end{equation}
given the order of $\min_{j\in[p_{c1}]}\hat{\sigma}_{j}^{{\rm co}(1)}$.
Recall that $v_{t}^{(1)}$ is stationary and $\tilde{W}_{t}^{(0)}=(\tilde{X}_{t}^{{\rm co}(2)\top},\tilde{X}_{t}^{\top},\tilde{Z}_{t}^{\top})$
is a standardized vector collecting a mix of pure $I(1)$ and $I(0)$
regressors without cointegration; by (\ref{eq:DB-mix-std}), (\ref{eq:max_sigma_z})
and (\ref{eq:ZeDB}) we verify 
\begin{align*}
\|\sum_{t=1}^{n}\left(\tilde{X}_{t}^{{\rm co}(2)\top},\tilde{X}_{t}^{\top}\right)\ddot{v}_{t-1}^{(1)\top}\|_{\max} & \stackrel{\mathrm{p}}{\preccurlyeq}\sqrt{n}(\log p)^{\frac{3}{2}+\frac{1}{2r}}\\
\|\sum_{t=1}^{n}\tilde{Z}_{t-1}\ddot{v}_{t-1}^{(1)\top}\|_{\max} & \leq\max_{j\in[p_{z}]}\hat{\sigma}_{j}^{Z}\cdot\|\sum_{t=1}^{n}Z_{t-1}\ddot{v}_{t-1}^{(1)\top}\|_{\max}\stackrel{\mathrm{p}}{\preccurlyeq}n.
\end{align*}
When the sample size is sufficiently large such that $\log p\leq\sqrt{n},$
the leading term in $\|\sum_{t=1}^{n}\tilde{W}_{t-1}^{(0)}\ddot{v}_{t-1}^{(1)\top}\|_{\max}$
is $O_{p}\left((\log p)^{1+\frac{1}{2r}}\right)$. Insert it into
(\ref{eq:Sigma_01}) and we obtain the stated rate. 
\end{proof}
\medskip
\begin{proof}[Proof of Lemma \ref{lem:Coint-RE}]
Recall $\tilde{\Sigma}^{(0)}=n^{-1}\tilde{W}^{(0)\top}\tilde{W}^{(0)}$
where 
\[
\tilde{W}_{t}^{(0)}=(\ddot{X}_{t}^{{\rm co}(2)\top}[D^{{\rm co}(2)}]^{-1},\ddot{X}_{t}^{\top}[D^{X}]^{-1},\ddot{Z}_{t}^{\top}[D^{Z}]^{-1})^{\top}=(D^{(0)})^{-1}W_{t}^{(0)}
\]
is the scale-standardized vector that concatenate the components invariant
to the rotation. Define $W_{t}^{(0)*}:=(n^{-1/2}X_{t}^{(2)\top},n^{-1/2}X_{t}^{\top},Z_{t}^{\top})^{\top}$,
and $\hat{\Sigma}^{(0)*}=n^{-1}\sum_{t=1}^{n}W_{t}^{(0)*}W_{t}^{(0)*\top}$.
Further denote $\widehat{\sigma}_{{\rm \max}}^{(0)*}:=\max_{j}\hat{\sigma}_{j}^{(0)*}$
and $\widehat{\sigma}_{{\rm \min}}^{(0)*}:=\min_{j}\hat{\sigma}_{j}^{(0)*}$
where $\hat{\sigma}_{j}^{(0)*}$ is the sample s.d.~of $W_{jt}^{(0)*}$,
and their ratio $\widehat{\varsigma}^{(0)}:=\widehat{\sigma}_{{\rm \max}}^{(0)*}/\widehat{\sigma}_{{\rm \min}}^{(0)*}$.
Obviously $\widehat{\varsigma}^{(0)}\stackrel{\mathrm{p}}{\preccurlyeq}\log p$
By (\ref{eq:sigma_star}). 

Setting $L=3\|Q^{-1}\|_{c1}\widehat{\varsigma}^{(0)}\stackrel{\mathrm{p}}{\preccurlyeq}(\log p)^{2}$.
Following (\ref{eq:RE-mix-std}), w.p.a.1.~we have
\[
\tilde{\kappa}^{(0)}\geq\widehat{\sigma}_{{\rm \max}}^{*-2}\cdot\kappa_{I}(\widehat{\Sigma}^{(0)*},3\|Q^{-1}\|_{c1}\widehat{\varsigma}^{(0)},s)\geq\dfrac{\widehat{\sigma}_{{\rm \max}}^{*-2}\widetilde{c}_{\kappa}}{9\|Q^{-1}\|_{c1}^{2}s\log p[\hat{\varsigma}^{(0)}]^{2}}\stackrel{\mathrm{p}}{\succcurlyeq}\dfrac{1}{s(\log p)^{4}\|Q^{-1}\|_{c1}^{2}}\stackrel{\mathrm{p}}{\succcurlyeq}\dfrac{1}{s(\log p)^{6}}
\]
where the first inequality follows by the proof of Proposition \ref{prop:REDBstdUnit},
the second by Proposition \ref{prop:MixRE-cn} when $L^{2}\cdot s\stackrel{\mathrm{p}}{\preccurlyeq}s(\log p)^{4}=o(n\wedge p)$
and 
\[
s^{2}L^{4}(\log p)^{5/2+1/(2r)}\log(np)\preccurlyeq s^{2}(\log p)^{21/2+1/(2r)}=o(n^{1/2}),
\]
 the third by the relative size of $n,$ $s$ and $p$ in the condition
of Theorem \ref{thm:cointegration}, and the last one by (\ref{eq:L1Q-1}). 
\end{proof}

\section{Additional Numerical Results \label{sec:Additional-Numerical-Results}}

\subsection{Cointegration\label{subsec:simu_coint}}

This section provides Monte Carlo simulations to demonstrate the theoretical
results in Section \ref{subsec:Cointegration}. We consider the DGP
(\ref{eq:cointegration_y}) with cointegrated variables generated
by the triangular representation (\ref{eq:cointegration X1 X2}).
We set $p_{c1}=2$, $p_{c2}=p/2-p_{c1}$, $p_{x}=p/2$ and $p_{z}=2n-p$,
and let $A=1_{p_{c1}}\otimes(0.4\cdot1_{6}^{\top},0_{p_{c2}-6}^{\top})$
in (\ref{eq:cointegration X1 X2}), where ``$\otimes$'' denotes
the Kronecker product. The oracle includes only the first 6 predictors
in $X_{t-1}^{(2)}.$ We set the coefficients in (\ref{eq:cointegration_y})
as $\phi_{1}^{*}=0.8\cdot1_{p_{c1}}$, $\beta^{*}=(1,n^{-1/2}1_{s_{x}-1}^{\top},0_{p_{x}-s_{x}}^{\top})^{\top},$
and in particular $\gamma^{*}=0_{p_{z}}$ to highlight $Z_{t}$'s
partial digesting of the unobservable $v_{t}^{(1)}$ that is predicted
by the theory. 

We generate the innovation $v_{t}=(e_{t}^{(2)\top},e_{t}^{\top},v_{t}^{(1)\top},Z_{t}^{\top},u_{t})^{\top}$
by a (vector) autoregressive (AR) 
\begin{gather}
v_{t}=0.4v_{t-1}+\varepsilon_{t},\text{ for }\varepsilon_{t}\sim i.i.d.\,\mathcal{N}(0,\,0.84\Omega),\label{eq:sim_DGP_v-1}\\
\text{where }\Omega_{ij}=0.8^{|j-j'|}\times\boldsymbol{1}(\text{ }(j,j')\not\in\mathcal{O}).\nonumber 
\end{gather}
Here $\mathcal{O}$ is a subset of the two-dimensional index set that
marks the uncorrelated entries at the following locations: (a) $Z_{t}$
and $u_{t}$, (b) $v_{t}^{(1)}$ and $u_{t}$, (c) $v_{t}^{(1)}$
and $Z_{(s+1):p_{z},t}$, and (d) $Z_{1:s,t}$ and $Z_{(s+1):p_{z},t}$.
The above (a) and (b) parts of this design ensure $(Z_{t}^{\top},v_{t}^{(1)\top})^{\top}$
is orthogonal to $u_{t}$, and the (c) and (d) parts guarantee that
the oracle model, with the first $s$ predictors in $Z_{t}$ involved,
remains of low dimension. Notice that $v_{t}^{(1)}$ is correlated
with the first $s$ predictors in $Z_{t}$ so that the first $s$
entries in the coefficient vector $\gamma^{(1)*}=\omega^{*}$ in (\ref{eq:cointegration_feasible})
is nonzero. 

We consider the following three regressions: (1) Regressing $y_{t}$
on $X_{t-1}$ only; (2) Regressing $y_{t}$ on $\left(X_{t-1}^{\top},Z_{t-1}^{\top}\right)^{\top}$;
and (3) Regressing $y_{t}$ on all observable regressors $\left(X_{t-1}^{\mathrm{co}\top},X_{t-1}^{\top},Z_{t-1}^{\top}\right)^{\top}$.
In practice only Regression (3) is feasible for LASSO. Regression
(2) is infeasible by borrowing the oracle ``$X_{t-1}^{\mathrm{co}\top}$
is inactive in the benchmark DGP (\ref{eq:cointegration_y_2}).''
The oracle will strengthen its finite sample performance. 

Table \ref{tab:coint} reports the one-period-ahead prediction errors
where the tuning parameter is selected by cross validation described
in Section \ref{sec:Simulations}. Plasso slightly outperforms Slasso
in Regression (1), consistent with the theory that the convergence
rates in Theorem \ref{thm:LassoError} are faster than those in Theorem
\ref{thm:SlassoError} and the simulation results in Tables \ref{tab:unit_rmse_ARMA}
and \ref{tab:unit_mae_ARMA}. Regression (2) with the stationary $Z_{t-1}$
substantially improves the prediction of both the oracle OLS and Slasso
estimators, as $Z_{t-1}$ absorbs the correlated part of $v_{t-1}^{(1)}$
and turns the $\gamma^{*}=0$ to $\gamma^{*(1)}\neq0$. Plasso is
unsatisfactory due to the distinctive dynamic behaviors of the nonstationary
$X_{t-1}$ and the stationary $Z_{t-1}$, as explained in Remark \ref{rem:var-sel}. 

The oracle OLS achieves the best prediction in Regression (3) as the
correctly selected active $X_{t-1}^{\mathrm{co}}$ fully captures
the information in $v_{t-1}^{(1)}$. Slasso in the feasible Regression
(3), despite the lack of oracle information concerning $X_{t-1}^{\mathrm{co}}$,
is nearly as good as that in Slasso under Regression (2), and comparable
to the oracle OLS under Regression (2) when the same size is large.
The Plasso again performs poorly given the mixture of predictors in
this case. 

\begin{table}[]
\caption{Simulations with Cointegrated Data}
\label{tab:coint}
\small
\begin{tabular}{ccc|ccc|ccc|ccc}
\hline\hline
\multirow{2}{*}{$n$} & \multirow{2}{*}{$p_{c1},p_x$} & \multirow{2}{*}{$p_z$} & \multicolumn{3}{c}{Oracle}  & \multicolumn{3}{c}{Plasso}  & \multicolumn{3}{c}{Slasso}  \\
\cline{4-12}
                   &                         &                     & Reg(1) & Reg(2) & Reg(3) & Reg(1) & Reg(2) & Reg(3) & Reg(1) & Reg(2) & Reg(3) \\
                   \hline 
\multicolumn{12}{c}{RMPSE}                                                                                                                                   \\
\hline 
\multirow{4}{*}{120}               & 30                      & 180                 & 2.052   & 1.663   & 1.235   & 2.076   & 1.891   & 2.011   & 2.081   & 1.741   & 1.752   \\
               & 48                      & 144                 & 2.052   & 1.684   & 1.257   & 2.105   & 1.915   & 2.034   & 2.122   & 1.731   & 1.753   \\
                & 72                      & 96                  & 2.023   & 1.626   & 1.219   & 2.070   & 1.910   & 2.037   & 2.074   & 1.661   & 1.689   \\
               & 90                      & 60                  & 2.036   & 1.657   & 1.238   & 2.101   & 1.958   & 2.051   & 2.110   & 1.673   & 1.708   \\
               \hline 
\multirow{4}{*}{240}               & 60                      & 360                 & 1.963   & 1.568   & 1.125   & 1.981   & 1.994   & 2.101   & 2.003   & 1.641   & 1.655   \\
                 & 96                      & 288                 & 1.928   & 1.555   & 1.111   & 1.949   & 1.997   & 2.106   & 1.983   & 1.640   & 1.648   \\
                & 144                     & 192                 & 1.964   & 1.555   & 1.125   & 2.015   & 2.082   & 2.184   & 2.026   & 1.637   & 1.651   \\
                 & 180                     & 120                 & 1.939   & 1.551   & 1.116   & 1.989   & 2.044   & 2.131   & 2.007   & 1.609   & 1.626   \\
                 \hline 
\multirow{4}{*}{360}               & 90                      & 540                 & 1.922   & 1.511   & 1.076   & 1.955   & 2.066   & 2.163   & 1.968   & 1.618   & 1.625   \\
                & 144                     & 432                 & 1.882   & 1.498   & 1.082   & 1.899   & 2.059   & 2.187   & 1.917   & 1.592   & 1.602   \\
                & 216                     & 288                 & 1.915   & 1.527   & 1.097   & 1.957   & 2.121   & 2.213   & 1.978   & 1.620   & 1.622   \\
                & 270                     & 180                 & 1.902   & 1.497   & 1.083   & 1.952   & 2.121   & 2.224   & 1.976   & 1.583   & 1.602   \\
\hline 
\multicolumn{12}{c}{MPAE}                                                                                                                                    \\
\hline 
\multirow{4}{*}{120}              & 30                      & 180                 & 1.631   & 1.331   & 0.982   & 1.654   & 1.496   & 1.583   & 1.654   & 1.385   & 1.391   \\
           & 48                      & 144                 & 1.643   & 1.344   & 1.012   & 1.688   & 1.527   & 1.612   & 1.702   & 1.370   & 1.389   \\
              & 72                      & 96                  & 1.605   & 1.292   & 0.979   & 1.655   & 1.520   & 1.616   & 1.648   & 1.322   & 1.343   \\
               & 90                      & 60                  & 1.624   & 1.323   & 0.985   & 1.680   & 1.554   & 1.621   & 1.686   & 1.334   & 1.363   \\
               \hline 
\multirow{4}{*}{240}                & 60                      & 360                 & 1.559   & 1.247   & 0.892   & 1.579   & 1.576   & 1.664   & 1.591   & 1.299   & 1.307   \\
              & 96                      & 288                 & 1.542   & 1.246   & 0.889   & 1.561   & 1.592   & 1.678   & 1.589   & 1.311   & 1.319   \\
               & 144                     & 192                 & 1.560   & 1.235   & 0.898   & 1.601   & 1.638   & 1.712   & 1.610   & 1.303   & 1.311   \\
                & 180                     & 120                 & 1.542   & 1.236   & 0.889   & 1.589   & 1.620   & 1.683   & 1.597   & 1.284   & 1.297   \\
                \hline 
\multirow{4}{*}{360}               & 90                      & 540                 & 1.526   & 1.198   & 0.859   & 1.552   & 1.636   & 1.707   & 1.563   & 1.288   & 1.292   \\
            & 144                     & 432                 & 1.504   & 1.198   & 0.862   & 1.511   & 1.638   & 1.736   & 1.524   & 1.273   & 1.279   \\
               & 216                     & 288                 & 1.524   & 1.218   & 0.876   & 1.561   & 1.686   & 1.751   & 1.577   & 1.287   & 1.291   \\
              & 270                     & 180                 & 1.518   & 1.202   & 0.864   & 1.562   & 1.688   & 1.767   & 1.572   & 1.268   & 1.284   \\
\hline \hline 
\end{tabular} 
\end{table}

\subsection{Omitted Results from the Main Text\label{subsec:Omitted-num-Results}}

This section contains a few tables to which the main text has referred.
For the simulations in Section \ref{sec:Simulations}, Table \ref{tab:unit_mae_ARMA}
shows the MAPE for prediction and the MAE of parameter estimation
in the case of mixed regressors, and Table \ref{tab:unit_mae_ARMA}
reports those of the pure unit root regressors. To better understand
the performance of Plasso and Slasso in variable selection, Table
\ref{tab:cate_sel} shows the percentage of variables selected in
each category, and it is accompanied by the RMSE for each category
in Table \ref{tab:cate_rmse}. 

For the empirical application in Section \ref{sec:Empirical-demo},
Table \ref{tab:unrate_MAPE_504} displays the MAPE of the prediction
of unemployment rate.

\begin{table}[]
\begin{center}
\caption{MAPE for Mixed Regressors}
\label{tab:mix_mae_ARMA}
\small
\begin{tabular}{ccc|r|rr|rr|r|rr|rr}
\hline\hline 
\multirow{3}{*}{$n$}   & \multicolumn{1}{c}{\multirow{3}{*}{$p_x$}} & \multicolumn{1}{c|}{\multirow{3}{*}{$p_z$}} & \multicolumn{5}{c|}{MAPE}                                                                                                                                   & \multicolumn{5}{c}{MAE for estimated coefficients}                                                                                                                         \\ \cline{4-13}
                     & \multicolumn{1}{c}{}                    & \multicolumn{1}{c|}{}                    & \multicolumn{1}{c|}{\multirow{2}{*}{Oracle}} & \multicolumn{2}{c|}{CV $\lambda$} & \multicolumn{2}{c|}{Calibrated $\lambda$} & \multicolumn{1}{c|}{\multirow{2}{*}{Oracle}} & \multicolumn{2}{c|}{CV $\lambda$} & \multicolumn{2}{c}{Calibrated $\lambda$} \\ \cline{5-8} \cline{10-13}
                     & \multicolumn{1}{c}{}                    & \multicolumn{1}{c|}{}                    & \multicolumn{1}{c|}{}                        & Plasso                & Slasso                       & Plasso                     & Slasso                    & \multicolumn{1}{c|}{}                        & Plasso                & Slasso                       & Plasso            & Slasso                             \\
                     \hline 
\multicolumn{13}{c}{DGP1}                                                                                                                                                     \\
\hline 
\multirow{4}{*}{120} & 60  & 180 & 0.913 & 1.297 & \textit{1.006} & 1.228 & \textit{\textbf{0.996}} & 2.861 & 3.896 & \textit{3.572}          & 3.638 & \textit{\textbf{3.566}} \\
                     & 96  & 144 & 0.902 & 1.342 & \textit{0.988} & 1.213 & \textit{\textbf{0.975}} & 2.858 & 4.101 & \textit{3.496}          & 3.725 & \textit{\textbf{3.478}} \\
                     & 144 & 96  & 0.909 & 1.379 & \textit{0.999} & 1.244 & \textit{\textbf{0.986}} & 2.859 & 4.200 & \textit{3.375}          & 3.796 & \textit{\textbf{3.347}} \\
                     & 180 & 60  & 0.924 & 1.410 & \textit{1.001} & 1.243 & \textit{\textbf{0.990}} & 2.851 & 4.283 & \textit{3.259}          & 3.824 & \textit{\textbf{3.224}} \\
                     \hline 
\multirow{4}{*}{240} & 120 & 360 & 0.847 & 1.563 & \textit{0.971} & 1.205 & \textit{\textbf{0.921}} & 2.178 & 4.406 & \textit{\textbf{2.821}} & 3.372 & \textit{3.037}          \\
                     & 192 & 288 & 0.851 & 1.633 & \textit{0.971} & 1.209 & \textit{\textbf{0.925}} & 2.192 & 4.547 & \textit{\textbf{2.792}} & 3.427 & \textit{2.952}          \\
                     & 288 & 192 & 0.860 & 1.680 & \textit{0.973} & 1.217 & \textit{\textbf{0.924}} & 2.183 & 4.762 & \textit{\textbf{2.757}} & 3.478 & \textit{2.823}          \\
                     & 360 & 120 & 0.849 & 1.742 & \textit{0.973} & 1.244 & \textit{\textbf{0.923}} & 2.179 & 4.824 & \textit{\textbf{2.710}} & 3.514 & \textit{2.711}          \\
                     \hline 
\multirow{4}{*}{360} & 180 & 540 & 0.843 & 1.746 & \textit{0.961} & 1.222 & \textit{\textbf{0.911}} & 1.671 & 4.243 & \textit{\textbf{2.221}} & 2.946 & \textit{2.506}          \\
                     & 288 & 432 & 0.836 & 1.757 & \textit{0.959} & 1.206 & \textit{\textbf{0.908}} & 1.667 & 4.415 & \textit{\textbf{2.215}} & 2.988 & \textit{2.433}          \\
                     & 432 & 288 & 0.841 & 1.858 & \textit{0.956} & 1.229 & \textit{\textbf{0.902}} & 1.652 & 4.565 & \textit{\textbf{2.187}} & 3.031 & \textit{2.298}          \\
                     & 540 & 180 & 0.832 & 1.896 & \textit{0.955} & 1.243 & \textit{\textbf{0.901}} & 1.673 & 4.631 & \textit{\textbf{2.178}} & 3.063 & \textit{2.210}          \\
                     \hline 
\multicolumn{13}{c}{DGP2}                                                                                                                                                                                                                                                                                                                                                                                          \\
\hline 
\multirow{4}{*}{120} & 60  & 180 & 0.910 & 1.870 & \textit{1.046} & 1.628 & \textit{\textbf{1.034}} & 2.851 & 5.413 & \textit{3.815}          & 4.844 & \textit{\textbf{3.733}} \\
                     & 96  & 144 & 0.916 & 1.921 & \textit{1.047} & 1.630 & \textit{\textbf{1.032}} & 2.844 & 5.613 & \textit{3.732}          & 4.944 & \textit{\textbf{3.653}} \\
                     & 144 & 96  & 0.894 & 1.974 & \textit{1.022} & 1.645 & \textit{\textbf{1.011}} & 2.856 & 5.795 & \textit{3.628}          & 5.023 & \textit{\textbf{3.544}} \\
                     & 180 & 60  & 0.895 & 2.038 & \textit{1.009} & 1.656 & \textit{\textbf{0.996}} & 2.853 & 6.000 & \textit{3.524}          & 5.100 & \textit{\textbf{3.445}} \\
                     \hline 
\multirow{4}{*}{240} & 120 & 360 & 0.870 & 2.731 & \textit{1.039} & 1.646 & \textit{\textbf{0.974}} & 2.181 & 7.351 & \textit{\textbf{3.053}} & 4.691 & \textit{3.146}          \\
                     & 192 & 288 & 0.856 & 2.879 & \textit{1.030} & 1.662 & \textit{\textbf{0.961}} & 2.186 & 7.577 & \textit{\textbf{3.033}} & 4.744 & \textit{3.066}          \\
                     & 288 & 192 & 0.876 & 3.009 & \textit{1.052} & 1.689 & \textit{\textbf{0.980}} & 2.186 & 7.893 & \textit{3.012}          & 4.787 & \textit{\textbf{2.965}} \\
                     & 360 & 120 & 0.853 & 3.128 & \textit{1.020} & 1.715 & \textit{\textbf{0.947}} & 2.164 & 8.255 & \textit{2.957}          & 4.864 & \textit{\textbf{2.843}} \\
                     \hline 
\multirow{4}{*}{360} & 180 & 540 & 0.841 & 3.509 & \textit{1.004} & 1.665 & \textit{\textbf{0.922}} & 1.666 & 8.127 & \textit{\textbf{2.409}} & 4.166 & \textit{2.547}          \\
                     & 288 & 432 & 0.832 & 3.785 & \textit{1.013} & 1.691 & \textit{\textbf{0.927}} & 1.660 & 8.691 & \textit{\textbf{2.413}} & 4.238 & \textit{2.488}          \\
                     & 432 & 288 & 0.850 & 3.890 & \textit{1.016} & 1.699 & \textit{\textbf{0.939}} & 1.664 & 8.979 & \textit{2.395}          & 4.278 & \textit{\textbf{2.376}} \\
                     & 540 & 180 & 0.852 & 4.003 & \textit{1.014} & 1.718 & \textit{\textbf{0.933}} & 1.660 & 9.273 & \textit{2.387}          & 4.323 & \textit{\textbf{2.282}}  \\
                     \hline \hline  
\end{tabular}
\end{center}
\footnotesize{
Note: 
Italic numbers indicate the better performance between Plasso and Slasso with the same tuning method. 
Bold numbers indicate the best LASSO performance. }
\end{table}

\begin{table}[]
\begin{center}
\caption{MAPE for Pure Unit Root Regressors}
\label{tab:unit_mae_ARMA}
\small
\begin{tabular}{cc|r|rr|rr|r|rr|rr}
  \hline   \hline 
\multirow{3}{*}{$n$}   & \multicolumn{1}{c|}{\multirow{3}{*}{$p_x$}} & \multicolumn{5}{c|}{MAPE}                                                                                                                                   & \multicolumn{5}{c}{MAE for estimated coefficients}                                                                                                                          \\ \cline{3-12}
                     & \multicolumn{1}{c|}{}                   & \multicolumn{1}{c|}{\multirow{2}{*}{Oracle}} & \multicolumn{2}{c|}{CV $\lambda$} & \multicolumn{2}{c|}{Calibrated $\lambda$} & \multicolumn{1}{c|}{\multirow{2}{*}{Oracle}} & \multicolumn{2}{c|}{CV $\lambda$} & \multicolumn{2}{c}{Calibrated $\lambda$} \\  \cline{4-7} \cline{9-12}
                     & \multicolumn{1}{c|}{}                   & \multicolumn{1}{c|}{}                        & Plasso                    & Slasso                   & Plasso                     & Slasso                    & \multicolumn{1}{c|}{}                        & Plasso                    & Slasso                   & Plasso                             & Slasso            \\
                     \hline 
\multicolumn{12}{c}{DGP3}                                                                                                                                                        \\
\hline 
\multirow{4}{*}{120} & 60  & 0.880 & \textit{0.881} & 0.897          & \textit{\textbf{0.864}} & 0.876 & 0.943 & \textit{1.075}          & 1.133 & \textit{\textbf{0.925}} & 1.002 \\
                     & 96  & 0.857 & \textit{0.869} & 0.885          & \textit{\textbf{0.848}} & 0.858 & 0.942 & \textit{1.108}          & 1.177 & \textit{\textbf{0.939}} & 1.047 \\
                     & 144 & 0.855 & 0.898          & \textit{0.881} & \textit{\textbf{0.847}} & 0.858 & 0.944 & \textit{\textbf{0.905}} & 1.073 & \textit{0.948}          & 1.079 \\
                     & 180 & 0.850 & 0.899          & \textit{0.885} & \textit{\textbf{0.857}} & 0.864 & 0.943 & \textit{\textbf{0.924}} & 1.108 & \textit{0.953}          & 1.095 \\
                     \hline 
\multirow{4}{*}{240} & 120 & 0.829 & \textit{0.844} & 0.852          & \textit{\textbf{0.829}} & 0.841 & 0.613 & \textit{0.754}          & 0.844 & \textit{\textbf{0.694}} & 0.781 \\
                     & 192 & 0.846 & \textit{0.851} & 0.868          & \textit{\textbf{0.839}} & 0.853 & 0.610 & \textit{0.783}          & 0.878 & \textit{\textbf{0.703}} & 0.815 \\
                     & 288 & 0.829 & 0.896          & \textit{0.868} & \textit{\textbf{0.838}} & 0.854 & 0.612 & \textit{\textbf{0.691}} & 0.843 & \textit{0.706}          & 0.845 \\
                     & 360 & 0.836 & 0.925          & \textit{0.877} & \textit{\textbf{0.854}} & 0.861 & 0.609 & \textit{\textbf{0.701}} & 0.866 & \textit{0.711}          & 0.860 \\
                     \hline 
\multirow{4}{*}{360} & 180 & 0.818 & \textit{0.832} & 0.841          & \textit{\textbf{0.823}} & 0.834 & 0.403 & \textit{0.560}          & 0.648 & \textit{\textbf{0.525}} & 0.610 \\
                     & 288 & 0.824 & \textit{0.836} & 0.853          & \textit{\textbf{0.829}} & 0.840 & 0.406 & \textit{0.580}          & 0.681 & \textit{\textbf{0.530}} & 0.636 \\
                     & 432 & 0.830 & 0.911          & \textit{0.867} & \textit{\textbf{0.838}} & 0.849 & 0.404 & \textit{\textbf{0.534}} & 0.656 & \textit{0.537}          & 0.659 \\
                     & 540 & 0.811 & 0.891          & \textit{0.851} & \textit{\textbf{0.820}} & 0.838 & 0.406 & \textit{0.542}          & 0.681 & \textit{\textbf{0.539}} & 0.675 \\
                     \hline 
\multicolumn{12}{c}{DGP4}                                                                                                                                                        \\
\hline 
\multirow{4}{*}{120} & 60  & 0.886 & \textit{0.893} & 0.902          & \textit{\textbf{0.873}} & 0.892 & 0.951 & \textit{1.129}          & 1.221 & \textit{\textbf{0.957}} & 1.014 \\
                     & 96  & 0.869 & \textit{0.876} & 0.898          & \textit{\textbf{0.861}} & 0.884 & 0.945 & \textit{1.183}          & 1.286 & \textit{\textbf{0.977}} & 1.071 \\
                     & 144 & 0.854 & 1.006          & \textit{0.905} & \textit{\textbf{0.855}} & 0.885 & 0.943 & \textit{1.011}          & 1.133 & \textit{\textbf{0.995}} & 1.131 \\
                     & 180 & 0.874 & 1.037          & \textit{0.932} & \textit{\textbf{0.890}} & 0.913 & 0.947 & \textit{1.033}          & 1.161 & \textit{\textbf{1.008}} & 1.149 \\
                     \hline 
\multirow{4}{*}{240} & 120 & 0.839 & \textit{0.857} & 0.870          & \textit{\textbf{0.846}} & 0.867 & 0.615 & \textit{0.792}          & 0.900 & \textit{\textbf{0.703}} & 0.783 \\
                     & 192 & 0.830 & \textit{0.855} & 0.870          & \textit{\textbf{0.842}} & 0.868 & 0.608 & \textit{0.817}          & 0.964 & \textit{\textbf{0.712}} & 0.831 \\
                     & 288 & 0.824 & 1.108          & \textit{0.899} & \textit{\textbf{0.843}} & 0.867 & 0.607 & \textit{0.841}          & 0.894 & \textit{\textbf{0.726}} & 0.879 \\
                     & 360 & 0.827 & 1.122          & \textit{0.910} & \textit{\textbf{0.843}} & 0.880 & 0.619 & \textit{0.857}          & 0.923 & \textit{\textbf{0.735}} & 0.905 \\
                     \hline 
\multirow{4}{*}{360} & 180 & 0.847 & \textit{0.857} & 0.863          & \textit{\textbf{0.848}} & 0.859 & 0.404 & \textit{0.576}          & 0.690 & \textit{\textbf{0.526}} & 0.606 \\
                     & 288 & 0.790 & \textit{0.814} & 0.828          & \textit{\textbf{0.803}} & 0.829 & 0.403 & \textit{0.603}          & 0.751 & \textit{\textbf{0.533}} & 0.648 \\
                     & 432 & 0.834 & 1.200          & \textit{0.939} & \textit{\textbf{0.852}} & 0.883 & 0.404 & \textit{0.701}          & 0.716 & \textit{\textbf{0.541}} & 0.684 \\
                     & 540 & 0.821 & 1.205          & \textit{0.918} & \textit{\textbf{0.836}} & 0.869 & 0.401 & \textit{0.715}          & 0.730 & \textit{\textbf{0.546}} & 0.701 \\ 
                       \hline   \hline 
\end{tabular}
\end{center}
\footnotesize{
Note: 
Italic numbers indicate the better performance between Plasso and Slasso with the same tuning method. 
Bold numbers indicate the best LASSO performance. }
\end{table} 

\begin{sidewaystable}[ht]
\caption{Percentage of variables selected by LASSO} 
\label{tab:cate_sel}
\small
\vspace{-1.5em}
\small
\begin{center}
\begin{tabular}{ccc|rrrr|rrrr|rrrr|rrrr}
 \hline\hline 
\multirow{3}{*}{$n$}   & \multicolumn{1}{c}{\multirow{3}{*}{$p_x$}} & \multicolumn{1}{c|}{\multirow{3}{*}{$p_z$}} & \multicolumn{4}{c|}{Active $\beta^*$}                                                                          & \multicolumn{4}{c|}{Inactive $\beta^*$}                                                                        & \multicolumn{4}{c|}{Active $\gamma^*$}                                                                          & \multicolumn{4}{c}{Inactive $\gamma^*$}                                                                        \\ \cline{4-19}
                     & \multicolumn{1}{c}{}                    & \multicolumn{1}{c|}{}                    & \multicolumn{2}{c}{CV $\lambda$} & \multicolumn{2}{c|}{Calibrated $\lambda$} & \multicolumn{2}{c}{CV $\lambda$} & \multicolumn{2}{c|}{Calibrated $\lambda$} & \multicolumn{2}{c}{CV $\lambda$} & \multicolumn{2}{c|}{Calibrated $\lambda$} & \multicolumn{2}{c}{CV $\lambda$} & \multicolumn{2}{c}{Calibrated $\lambda$} \\
                     & \multicolumn{1}{c}{}                    & \multicolumn{1}{c|}{}                    & Plasso                & Slasso                & Plasso                    & Slasso                    & Plasso                & Slasso                & Plasso                    & Slasso                    & Plasso                & Slasso                & Plasso                    & Slasso                    & Plasso                & Slasso                & Plasso                    & Slasso                    \\
                     \hline 
                     \\
\multicolumn{19}{c}{DGP1}                                                                                                                                                                                                                                                                                                                                                                                                                                                                                                              \\
\multirow{4}{*}{120} & 60  & 180 & 55.70 & 56.43 & 57.61 & 58.96 & 7.77 & 4.32 & 7.03 & 4.73 & 92.66 & 98.00 & 93.17 & 98.18 & 0.32 & 5.96 & 0.06 & 6.56 \\
                     & 96  & 144 & 52.49 & 55.31 & 55.77 & 58.03 & 5.81 & 3.50 & 5.44 & 3.87 & 91.80 & 98.02 & 93.07 & 98.20 & 0.24 & 6.12 & 0.07 & 6.72 \\
                     & 144 & 96  & 49.99 & 54.46 & 53.63 & 57.17 & 4.49 & 3.02 & 4.29 & 3.32 & 91.50 & 97.94 & 92.83 & 98.13 & 0.23 & 6.60 & 0.11 & 7.15 \\
                     & 180 & 60  & 48.43 & 53.65 & 52.50 & 56.17 & 3.91 & 2.84 & 3.77 & 3.10 & 91.31 & 98.05 & 92.82 & 98.23 & 0.31 & 7.16 & 0.17 & 7.68 \\
                     \hline 
\multirow{4}{*}{240} & 120 & 360 & 54.15 & 56.04 & 63.06 & 61.98 & 4.27 & 2.37 & 4.42 & 3.05 & 91.51 & 99.18 & 94.22 & 99.53 & 0.01 & 1.92 & 0.01 & 4.51 \\
                     & 192 & 288 & 51.96 & 54.93 & 61.17 & 61.11 & 3.23 & 1.86 & 3.41 & 2.46 & 91.15 & 99.16 & 94.04 & 99.52 & 0.01 & 1.92 & 0.01 & 4.63 \\
                     & 288 & 192 & 48.66 & 53.88 & 59.31 & 60.05 & 2.47 & 1.56 & 2.67 & 2.06 & 90.44 & 99.11 & 94.03 & 99.50 & 0.01 & 2.01 & 0.01 & 4.84 \\
                     & 360 & 120 & 47.73 & 53.23 & 58.47 & 59.54 & 2.13 & 1.37 & 2.31 & 1.84 & 90.26 & 99.14 & 93.88 & 99.52 & 0.02 & 2.12 & 0.01 & 5.12 \\
                     \hline 
\multirow{4}{*}{360} & 180 & 540 & 55.48 & 57.77 & 67.00 & 65.23 & 2.97 & 1.65 & 3.12 & 2.19 & 90.28 & 99.58 & 93.88 & 99.80 & 0.00 & 0.91 & 0.00 & 3.43 \\
                     & 288 & 432 & 52.69 & 56.23 & 64.90 & 63.34 & 2.24 & 1.32 & 2.37 & 1.77 & 89.72 & 99.58 & 93.67 & 99.85 & 0.00 & 0.92 & 0.00 & 3.54 \\
                     & 432 & 288 & 50.16 & 55.33 & 63.30 & 62.37 & 1.70 & 1.03 & 1.83 & 1.42 & 89.41 & 99.57 & 93.68 & 99.80 & 0.00 & 0.95 & 0.00 & 3.68 \\
                     & 540 & 180 & 49.23 & 54.89 & 62.53 & 62.12 & 1.46 & 0.91 & 1.58 & 1.27 & 89.16 & 99.62 & 93.62 & 99.85 & 0.00 & 1.02 & 0.00 & 3.89 \\
                     \hline 
                     \\
\multicolumn{19}{c}{DGP2}                                                                                                                                                                                                                                                                                                                                                                                                                                                                                                                \\
\multirow{4}{*}{120} & 60  & 180 & 55.70 & 56.43 & 57.61 & 58.96 & 7.77 & 4.32 & 7.03 & 4.73 & 92.66 & 98.00 & 93.17 & 98.18 & 0.32 & 5.96 & 0.06 & 6.56 \\
                     & 96  & 144 & 52.49 & 55.31 & 55.77 & 58.03 & 5.81 & 3.50 & 5.44 & 3.87 & 91.80 & 98.02 & 93.07 & 98.20 & 0.24 & 6.12 & 0.07 & 6.72 \\
                     & 144 & 96  & 49.99 & 54.46 & 53.63 & 57.17 & 4.49 & 3.02 & 4.29 & 3.32 & 91.50 & 97.94 & 92.83 & 98.13 & 0.23 & 6.60 & 0.11 & 7.15 \\
                     & 180 & 60  & 48.43 & 53.65 & 52.50 & 56.17 & 3.91 & 2.84 & 3.77 & 3.10 & 91.31 & 98.05 & 92.82 & 98.23 & 0.31 & 7.16 & 0.17 & 7.68 \\
                     \hline 
\multirow{4}{*}{240} & 120 & 360 & 54.15 & 56.04 & 63.06 & 61.98 & 4.27 & 2.37 & 4.42 & 3.05 & 91.51 & 99.18 & 94.22 & 99.53 & 0.01 & 1.92 & 0.01 & 4.51 \\
                     & 192 & 288 & 51.96 & 54.93 & 61.17 & 61.11 & 3.23 & 1.86 & 3.41 & 2.46 & 91.15 & 99.16 & 94.04 & 99.52 & 0.01 & 1.92 & 0.01 & 4.63 \\
                     & 288 & 192 & 48.66 & 53.88 & 59.31 & 60.05 & 2.47 & 1.56 & 2.67 & 2.06 & 90.44 & 99.11 & 94.03 & 99.50 & 0.01 & 2.01 & 0.01 & 4.84 \\
                     & 360 & 120 & 47.73 & 53.23 & 58.47 & 59.54 & 2.13 & 1.37 & 2.31 & 1.84 & 90.26 & 99.14 & 93.88 & 99.52 & 0.02 & 2.12 & 0.01 & 5.12 \\
                     \hline 
\multirow{4}{*}{360} & 180 & 540 & 55.48 & 57.77 & 67.00 & 65.23 & 2.97 & 1.65 & 3.12 & 2.19 & 90.28 & 99.58 & 93.88 & 99.80 & 0.00 & 0.91 & 0.00 & 3.43 \\
                     & 288 & 432 & 52.69 & 56.23 & 64.90 & 63.34 & 2.24 & 1.32 & 2.37 & 1.77 & 89.72 & 99.58 & 93.67 & 99.85 & 0.00 & 0.92 & 0.00 & 3.54 \\
                     & 432 & 288 & 50.16 & 55.33 & 63.30 & 62.37 & 1.70 & 1.03 & 1.83 & 1.42 & 89.41 & 99.57 & 93.68 & 99.80 & 0.00 & 0.95 & 0.00 & 3.68 \\
                     & 540 & 180 & 49.23 & 54.89 & 62.53 & 62.12 & 1.46 & 0.91 & 1.58 & 1.27 & 89.16 & 99.62 & 93.62 & 99.85 & 0.00 & 1.02 & 0.00 & 3.89 \\        
                     \hline\hline            
\end{tabular} 
\end{center}
\end{sidewaystable}

\begin{sidewaystable}[ht]
\caption{Categorized RMSE of estimated coefficients} 
\label{tab:cate_rmse}
\small
\vspace{-1.5em}
\small
\begin{center}
\begin{tabular}{ccc|rrrr|rrrr|rrrr|rrrr}
 \hline\hline 
\multirow{3}{*}{$n$}   & \multicolumn{1}{c}{\multirow{3}{*}{$p_x$}} & \multicolumn{1}{c|}{\multirow{3}{*}{$p_z$}} & \multicolumn{4}{c|}{Active $\beta^*$}                                                                          & \multicolumn{4}{c|}{Inactive $\beta^*$}                                                                        & \multicolumn{4}{c|}{Active $\gamma^*$}                                                                          & \multicolumn{4}{c}{Inactive $\gamma^*$}                                                                        \\ \cline{4-19}
                     & \multicolumn{1}{c}{}                    & \multicolumn{1}{c|}{}                    & \multicolumn{2}{c}{CV $\lambda$} & \multicolumn{2}{c|}{Calibrated $\lambda$} & \multicolumn{2}{c}{CV $\lambda$} & \multicolumn{2}{c|}{Calibrated $\lambda$} & \multicolumn{2}{c}{CV $\lambda$} & \multicolumn{2}{c|}{Calibrated $\lambda$} & \multicolumn{2}{c}{CV $\lambda$} & \multicolumn{2}{c}{Calibrated $\lambda$} \\
                     & \multicolumn{1}{c}{}                    & \multicolumn{1}{c|}{}                    & Plasso                & Slasso                & Plasso                    & Slasso                    & Plasso                & Slasso                & Plasso                    & Slasso                    & Plasso                & Slasso                & Plasso                    & Slasso                    & Plasso                & Slasso                & Plasso                    & Slasso                    \\
                     \hline 
                     \\
\multicolumn{19}{c}{DGP1}                                                                                                                                                                                                                                                                                                                                                                                                                                                                                                              \\
\multirow{4}{*}{120} & 60  & 180 & 0.29 & 0.29 & 0.28 & 0.28 & 0.11 & 0.08 & 0.10 & 0.08 & 1.19 & 0.82 & 1.07 & 0.81 & 0.07 & 0.26 & 0.05 & 0.26 \\
                     & 96  & 144 & 0.29 & 0.29 & 0.28 & 0.28 & 0.12 & 0.09 & 0.11 & 0.10 & 1.27 & 0.82 & 1.08 & 0.81 & 0.05 & 0.23 & 0.04 & 0.24 \\
                     & 144 & 96  & 0.30 & 0.29 & 0.28 & 0.28 & 0.13 & 0.11 & 0.12 & 0.11 & 1.27 & 0.81 & 1.09 & 0.80 & 0.05 & 0.20 & 0.05 & 0.20 \\
                     & 180 & 60  & 0.30 & 0.29 & 0.28 & 0.28 & 0.14 & 0.11 & 0.12 & 0.12 & 1.31 & 0.80 & 1.09 & 0.79 & 0.05 & 0.17 & 0.04 & 0.17 \\
                     \hline 
\multirow{4}{*}{240} & 120 & 360 & 0.23 & 0.22 & 0.20 & 0.21 & 0.08 & 0.05 & 0.06 & 0.06 & 1.40 & 0.66 & 0.94 & 0.62 & 0.02 & 0.13 & 0.01 & 0.20 \\
                     & 192 & 288 & 0.23 & 0.22 & 0.20 & 0.21 & 0.09 & 0.06 & 0.07 & 0.06 & 1.44 & 0.66 & 0.95 & 0.61 & 0.01 & 0.12 & 0.01 & 0.18 \\
                     & 288 & 192 & 0.23 & 0.22 & 0.21 & 0.21 & 0.10 & 0.07 & 0.08 & 0.07 & 1.50 & 0.66 & 0.95 & 0.61 & 0.01 & 0.10 & 0.01 & 0.16 \\
                     & 360 & 120 & 0.23 & 0.22 & 0.21 & 0.21 & 0.10 & 0.07 & 0.08 & 0.07 & 1.52 & 0.65 & 0.96 & 0.61 & 0.01 & 0.09 & 0.01 & 0.13 \\
                     \hline 
\multirow{4}{*}{360} & 180 & 540 & 0.18 & 0.18 & 0.15 & 0.16 & 0.06 & 0.04 & 0.05 & 0.04 & 1.43 & 0.54 & 0.85 & 0.50 & 0.01 & 0.09 & 0.00 & 0.17 \\
                     & 288 & 432 & 0.18 & 0.18 & 0.16 & 0.17 & 0.07 & 0.05 & 0.05 & 0.05 & 1.48 & 0.54 & 0.85 & 0.50 & 0.01 & 0.08 & 0.00 & 0.15 \\
                     & 432 & 288 & 0.18 & 0.18 & 0.16 & 0.17 & 0.08 & 0.05 & 0.06 & 0.05 & 1.53 & 0.53 & 0.86 & 0.49 & 0.01 & 0.07 & 0.01 & 0.13 \\
                     & 540 & 180 & 0.18 & 0.18 & 0.16 & 0.17 & 0.08 & 0.05 & 0.06 & 0.06 & 1.54 & 0.54 & 0.86 & 0.50 & 0.00 & 0.06 & 0.00 & 0.11 \\
                     \hline 
                     \\
\multicolumn{19}{c}{DGP2}                                                                                                                                                                                                                                                                                                                                                                                                                                                                                                              \\
\multirow{4}{*}{120} & 60  & 180 & 0.42 & 0.37 & 0.38 & 0.35 & 0.13 & 0.09 & 0.11 & 0.09 & 1.80 & 0.85 & 1.46 & 0.83 & 0.05 & 0.27 & 0.04 & 0.27 \\
                     & 96  & 144 & 0.44 & 0.37 & 0.39 & 0.35 & 0.15 & 0.11 & 0.13 & 0.11 & 1.84 & 0.85 & 1.47 & 0.83 & 0.05 & 0.24 & 0.04 & 0.24 \\
                     & 144 & 96  & 0.45 & 0.38 & 0.40 & 0.36 & 0.17 & 0.12 & 0.14 & 0.12 & 1.88 & 0.84 & 1.48 & 0.83 & 0.05 & 0.21 & 0.04 & 0.21 \\
                     & 180 & 60  & 0.47 & 0.38 & 0.40 & 0.36 & 0.19 & 0.13 & 0.15 & 0.13 & 1.94 & 0.83 & 1.50 & 0.82 & 0.04 & 0.17 & 0.03 & 0.17 \\
                     \hline 
\multirow{4}{*}{240} & 120 & 360 & 0.35 & 0.28 & 0.26 & 0.25 & 0.12 & 0.07 & 0.08 & 0.06 & 2.51 & 0.69 & 1.36 & 0.63 & 0.01 & 0.14 & 0.01 & 0.20 \\
                     & 192 & 288 & 0.37 & 0.29 & 0.27 & 0.26 & 0.14 & 0.07 & 0.09 & 0.07 & 2.58 & 0.69 & 1.36 & 0.63 & 0.01 & 0.13 & 0.01 & 0.18 \\
                     & 288 & 192 & 0.38 & 0.30 & 0.27 & 0.26 & 0.16 & 0.08 & 0.09 & 0.08 & 2.65 & 0.69 & 1.36 & 0.63 & 0.03 & 0.11 & 0.02 & 0.16 \\
                     & 360 & 120 & 0.40 & 0.30 & 0.27 & 0.26 & 0.17 & 0.08 & 0.10 & 0.08 & 2.78 & 0.68 & 1.38 & 0.63 & 0.01 & 0.09 & 0.01 & 0.13 \\
                     \hline 
\multirow{4}{*}{360} & 180 & 540 & 0.30 & 0.23 & 0.20 & 0.19 & 0.11 & 0.05 & 0.06 & 0.05 & 2.92 & 0.56 & 1.24 & 0.51 & 0.00 & 0.09 & 0.00 & 0.16 \\
                     & 288 & 432 & 0.32 & 0.23 & 0.20 & 0.20 & 0.13 & 0.06 & 0.07 & 0.05 & 3.10 & 0.57 & 1.26 & 0.51 & 0.00 & 0.08 & 0.00 & 0.15 \\
                     & 432 & 288 & 0.33 & 0.23 & 0.20 & 0.20 & 0.14 & 0.06 & 0.07 & 0.06 & 3.19 & 0.57 & 1.26 & 0.51 & 0.00 & 0.07 & 0.00 & 0.13 \\
                     & 540 & 180 & 0.35 & 0.24 & 0.21 & 0.20 & 0.15 & 0.06 & 0.07 & 0.06 & 3.27 & 0.56 & 1.27 & 0.50 & 0.00 & 0.06 & 0.00 & 0.10   \\ 
                     \hline\hline            
\end{tabular} 
\end{center}
\end{sidewaystable}

\begin{table}[]
\caption{MAPE for \texttt{UNRATE}}
\label{tab:unrate_MAPE_504}
\small
\begin{center}
\begin{tabular}{cc|rr|rr|rr|rr|rr}
\hline\hline 
\multirow{3}{*}{$h$} & \multicolumn{1}{c|}{\multirow{3}{*}{$n$}} & \multicolumn{2}{c|}{\multirow{2}{*}{Benchmarks}} & \multicolumn{4}{c|}{121  Predictors}                                                        & \multicolumn{4}{c}{504 Predictors}                                                         \\ \cline{5-12}
                   & \multicolumn{1}{c|}{}                   & \multicolumn{2}{c|}{}                            & \multicolumn{2}{c}{NT}           & \multicolumn{2}{c|}{ST}                                  & \multicolumn{2}{c}{NT}           & \multicolumn{2}{c}{ST}                                  \\ \cline{3-12}
                   & \multicolumn{1}{c|}{}                   & \multicolumn{1}{c}{RWwD}                   & \multicolumn{1}{c|}{AR}                    & \multicolumn{1}{c}{Plasso} & \multicolumn{1}{c|}{Slasso}       & \multicolumn{1}{c}{Plasso} & \multicolumn{1}{c|}{Slasso}  & \multicolumn{1}{c}{Plasso} & \multicolumn{1}{c|}{Slasso}                  & \multicolumn{1}{c}{Plasso} & \multicolumn{1}{c}{Slasso} \\
                   \hline 
\multicolumn{12}{c}{\textbf{Entire testing sample: 1990--2019}}                                                                                                                                                                                                                                                                          \\
\multirow{3}{*}{1} & 120 & 0.114          & 0.116          & 0.454 & \textit{0.113}          & 0.647 & 0.418 & 0.408 & \textit{\textbf{0.109}} & 0.275 & 0.115                   \\
                   & 240 & 0.114          & 0.117          & 0.335 & \textit{0.113}          & 0.515 & 0.535 & 0.575 & \textit{\textbf{0.099}} & 0.177 & 0.105                   \\
                   & 360 & 0.114          & 0.114          & 0.304 & \textit{0.116}          & 0.722 & 0.646 & 0.577 & \textit{\textbf{0.101}} & 0.153 & 0.107                   \\
                   \hline 
\multirow{3}{*}{2} & 120 & 0.168          & 0.162          & 0.486 & \textit{0.151}          & 0.638 & 0.438 & 0.461 & \textit{\textbf{0.145}} & 0.332 & 0.162                   \\
                   & 240 & 0.167          & 0.161          & 0.394 & \textit{0.135}          & 0.517 & 0.531 & 0.618 & \textit{\textbf{0.127}} & 0.229 & 0.141                   \\
                   & 360 & 0.167          & 0.158          & 0.362 & \textit{0.146}          & 0.629 & 0.649 & 0.593 & \textit{\textbf{0.133}} & 0.202 & 0.141                   \\
                   \hline 
\multirow{3}{*}{3} & 120 & 0.218          & 0.211          & 0.514 & \textit{\textbf{0.200}} & 0.684 & 0.459 & 0.513 & \textit{0.201}          & 0.422 & 0.212                   \\
                   & 240 & 0.217          & 0.202          & 0.437 & \textit{0.178}          & 0.529 & 0.549 & 0.665 & \textit{\textbf{0.166}} & 0.297 & 0.174                   \\
                   & 360 & 0.217          & 0.200          & 0.417 & \textit{0.175}          & 0.629 & 0.627 & 0.621 & \textit{\textbf{0.165}} & 0.261 & 0.174                   \\
                   \hline 
\\
\multicolumn{12}{c}{Testing sub-sample:  1990--1999}                                                                                                                                                                                                                                                                          \\
\multirow{3}{*}{1} & 120 & 0.107          & 0.105          & 0.373 & \textit{0.109}          & 0.470 & 0.370 & 0.321 & \textit{\textbf{0.105}} & 0.213 & 0.111                   \\
                   & 240 & 0.107          & 0.111          & 0.172 & \textit{0.109}          & 0.534 & 0.614 & 0.450 & \textit{\textbf{0.099}} & 0.157 & 0.104                   \\
                   & 360 & 0.107          & 0.112          & 0.178 & \textit{0.119}          & 0.526 & 0.487 & 0.521 & 0.103                   & 0.144 & \textit{\textbf{0.101}} \\
                   \hline 
\multirow{3}{*}{2} & 120 & 0.153          & 0.141          & 0.388 & 0.140                   & 0.485 & 0.374 & 0.360 & \textit{\textbf{0.139}} & 0.238 & 0.151                   \\
                   & 240 & 0.153          & 0.151          & 0.224 & \textit{0.125}          & 0.547 & 0.584 & 0.468 & \textit{\textbf{0.125}} & 0.226 & 0.140                   \\
                   & 360 & 0.154          & 0.155          & 0.227 & \textit{0.144}          & 0.518 & 0.487 & 0.537 & \textit{\textbf{0.131}} & 0.198 & 0.136                   \\
                   \hline 
\multirow{3}{*}{3} & 120 & 0.188          & \textbf{0.178} & 0.406 & \textit{0.180}          & 0.519 & 0.387 & 0.381 & \textit{0.179}          & 0.272 & 0.190                   \\
                   & 240 & 0.188          & 0.185          & 0.292 & \textit{0.164}          & 0.548 & 0.597 & 0.472 & \textit{\textbf{0.155}} & 0.284 & 0.180                   \\
                   & 360 & 0.190          & 0.187          & 0.288 & \textit{0.168}          & 0.521 & 0.487 & 0.559 & 0.162                   & 0.246 & \textit{\textbf{0.160}} \\
                    \hline 
\multicolumn{12}{c}{Testing sub-sample:  2000--2009}                                                                                                                                                                                                                                                                          \\
\multirow{3}{*}{1} & 120 & 0.123          & 0.118          & 0.328 & \textit{0.116}          & 0.597 & 0.350 & 0.263 & \textit{\textbf{0.107}} & 0.334 & 0.117                   \\
                   & 240 & 0.123          & 0.115          & 0.298 & \textit{0.120}          & 0.538 & 0.516 & 0.354 & \textit{\textbf{0.097}} & 0.158 & 0.100                   \\
                   & 360 & 0.123          & 0.110          & 0.258 & \textit{0.125}          & 0.900 & 0.820 & 0.404 & \textit{\textbf{0.096}} & 0.153 & 0.103                   \\
                   \hline 
\multirow{3}{*}{2} & 120 & 0.192          & 0.180          & 0.380 & \textit{0.155}          & 0.655 & 0.390 & 0.323 & \textit{\textbf{0.142}} & 0.373 & 0.171                   \\
                   & 240 & 0.194          & 0.165          & 0.378 & \textit{0.143}          & 0.539 & 0.513 & 0.446 & \textit{\textbf{0.132}} & 0.189 & 0.138                   \\
                   & 360 & 0.193          & 0.160          & 0.321 & \textit{0.159}          & 0.642 & 0.817 & 0.403 & 0.141                   & 0.199 & \textit{\textbf{0.139}} \\
                   \hline 
\multirow{3}{*}{3} & 120 & 0.264          & 0.238          & 0.407 & \textit{\textbf{0.210}} & 0.675 & 0.390 & 0.392 & 0.222                   & 0.416 & \textit{0.243}          \\
                   & 240 & 0.266          & 0.224          & 0.356 & \textit{0.207}          & 0.545 & 0.500 & 0.540 & 0.185                   & 0.207 & \textit{\textbf{0.164}} \\
                   & 360 & 0.265          & 0.220          & 0.351 & \textit{0.199}          & 0.632 & 0.744 & 0.429 & 0.174                   & 0.247 & \textit{\textbf{0.173}} \\
                    \hline 
\multicolumn{12}{c}{Testing sub-sample:  2010--2019}                                                                                                                                                                                                                                                                          \\
\multirow{3}{*}{1} & 120 & \textbf{0.113} & 0.124          & 0.662 & \textit{0.113}          & 0.873 & 0.533 & 0.640 & \textit{0.115}          & 0.277 & 0.116                   \\
                   & 240 & 0.112          & 0.125          & 0.534 & \textit{0.112}          & 0.474 & 0.475 & 0.921 & \textit{\textbf{0.102}} & 0.216 & 0.113                   \\
                   & 360 & 0.111          & 0.120          & 0.476 & \textit{\textbf{0.103}} & 0.741 & 0.630 & 0.805 & \textit{0.103}          & 0.163 & 0.116                   \\
                   \hline 
\multirow{3}{*}{2} & 120 & 0.158          & 0.164          & 0.691 & \textit{0.158}          & 0.774 & 0.552 & 0.702 & \textit{\textbf{0.156}} & 0.384 & 0.163                   \\
                   & 240 & 0.156          & 0.169          & 0.580 & \textit{0.139}          & 0.465 & 0.495 & 0.941 & \textit{\textbf{0.125}} & 0.272 & 0.144                   \\
                   & 360 & 0.155          & 0.158          & 0.538 & 0.136                   & 0.727 & 0.643 & 0.840 & \textit{\textbf{0.128}} & 0.208 & 0.148                   \\
                   \hline 
\multirow{3}{*}{3} & 120 & \textbf{0.201} & 0.217          & 0.729 & \textit{0.209}          & 0.859 & 0.599 & 0.766 & 0.204                   & 0.578 & \textit{0.203}          \\
                   & 240 & 0.199          & 0.198          & 0.665 & \textit{0.164}          & 0.495 & 0.550 & 0.983 & \textit{\textbf{0.159}} & 0.401 & 0.178                   \\
                   & 360 & 0.196          & 0.191          & 0.612 & \textit{\textbf{0.159}} & 0.735 & 0.651 & 0.876 & \textit{0.160}          & 0.291 & 0.188                            \\
                        \hline \hline       
\end{tabular}
\end{center}
\footnotesize{Notes: NT and ST are abbreviations for no transformation and stationarization transformation respectively. Bold numbers indicate the best performance in each row. Italic numbers indicate the best LASSO performance with the same number of predictors.}
\end{table}

To dig in further, we count the number of selected variables in Table
\ref{tab:nonzero}(a) under each \texttt{TCODE} for $h=1$, averaged
over the entire testing sample. Under NT, Slasso selects more variables
than Plasso, for example, the stationary variables with \texttt{TCODE}
(1) and (4). The majority of variables selected by Plasso are of \texttt{TCODE}
(5) and (6), reflecting the issue we discussed in Remark \ref{rem:var-sel}
that Plasso tends to pick variables of large scale; Plasso under ST
further makes it clear as variables of \texttt{TCODE} (1), (2) and
(4) become of large scale after stationarization. The scale normalization
in Slasso allows all variables to have equal opportunities to start
with, and thus the selected ones are more evenly distributed. However,
despite that Slasso under ST selects many variables, the RMPSE is
still unsatisfactory due to the imbalance between the two sides of
this predictive regression where stationarized regressors on the right-hand
side do not match the persistent dependent variable on the left-hand
side. This imbalance explains the persistent prediction error in the
lower right panel of Figure \ref{fig:res_UNRATE}.

Table \ref{tab:nonzero} (b) presents the variable selection of the
126 unique regressors in each of the four lags, again averaged over
the entire testing sample. Across the lags, the numbers of selected
variables by Slasso monotonically decrease as the lags go farther
behind. Although the recent predictors are the most relevant, the
unemployment rate responds to further lags as well. Slasso under NT
is more parsimonious than that under ST. Overall, given 504 predictors
Slasso achieves smaller RMPSEs with fewer active variables than it
is fed with 121 predictors.

\begin{table}[t]
\caption{Average Numbers of Active Generic Predictors under $h=1$.}
\small
\label{tab:nonzero}
    \begin{minipage}{.5\textwidth}
      \centering
      \vspace{0.5em}
      \begin{tabular}{lc|rr|rr}
\multicolumn{6}{c}{(a) 121 Predictors} \\ \hline\hline 
TCODE & \multicolumn{1}{c|}{\multirow{2}{*}{$n$}} & \multicolumn{2}{c|}{NT} & \multicolumn{2}{c}{ST}                                                                    \\ \cline{3-6}
                       & \multicolumn{1}{c|}{}                   & \multicolumn{1}{r}{Plasso}    & \multicolumn{1}{r|}{Slasso}  & \multicolumn{1}{r}{Plasso} & \multicolumn{1}{r}{Slasso}  \\
                       \hline 
\multirow{3}{*}{All}   & 120                                    & 4.553                      & 16.206                     & 4.833                      & 26.228                     \\
                       & 240                                    & 12.381                     & 22.764                     & 21.275                     & 62.458                     \\
                       & 360                                    & 12.867                     & 32.808                     & 24.092                     & 66.156                     \\
                       \hline 
                       \\
\multicolumn{6}{c}{Each category of TCODE}                        \\
\multirow{3}{*}{1}     & 120                                    & 0                      & 0.953                      & 2.192                      & 4.058                      \\
                       & 240                                    & 0                      & 2.322                      & 7.519                      & 6.047                      \\
                       & 360                                    & 0                      & 3.569                      & 7.556                      & 6.819                      \\
                       \hline
\multirow{3}{*}{2}     & 120                                    & 0                      & 2.653                      & 1.461                      & 3.686                      \\
                       & 240                                    & 0.381                      & 4.144                      & 7.017                      & 9.536                      \\
                       & 360                                    & 0.425                      & 6.389                      & 8.519                      & 10.133                     \\
                       \hline
\multirow{3}{*}{4}     & 120                                    & 0                      & 1.569                      & 1.056                      & 4.294                      \\
                       & 240                                    & 0.006                      & 2.725                      & 5.100                      & 6.000                      \\
                       & 360                                    & 0                      & 3.972                      & 5.631                      & 6.117                      \\
                       \hline
\multirow{3}{*}{5}     & 120                                    & 2.736                      & 9.858                      & 0                      & 9.967                      \\
                       & 240                                    & 8.358                      & 11.925                     & 0.617                      & 25.381                     \\
                       & 360                                    & 8.875                      & 14.703                     & 1.131                      & 26.156                     \\
                       \hline
\multirow{3}{*}{6}     & 120                                    & 1.483                      & 1.117                      & 0                      & 4.108                      \\
                       & 240                                    & 2.786                      & 1.433                      & 0.642                      & 14.664                     \\
                       & 360                                    & 2.847                      & 3.672                      & 0.794                      & 16.064                     \\
                       \hline
\multirow{3}{*}{7}     & 120                                    & 0.333                      & 0.056                      & 0.125                      & 0.114                      \\
                       & 240                                    & 0.850                      & 0.214                      & 0.381                      & 0.831                      \\
                       & 360                                    & 0.719                      & 0.503                      & 0.461                      & 0.867                 \\
                       \hline \hline 
\end{tabular} 
    \end{minipage}
    \begin{minipage}{.6\textwidth}
      \centering
      \vspace{-9.25em}
      \vspace{0.5em}
      \begin{tabular}{lc|rr|rr}
\multicolumn{6}{c}{(b) 504 Predictors} \\ \hline\hline 
\multirow{2}{*}{Lag} & \multicolumn{1}{c|}{\multirow{2}{*}{$n$}} & \multicolumn{2}{c|}{NT} & \multicolumn{2}{c}{ST} \\ \cline{3-6}
                     & \multicolumn{1}{c|}{}                   & \multicolumn{1}{c}{Plasso}      & \multicolumn{1}{c|}{Slasso}     & \multicolumn{1}{c}{Plasso}       & \multicolumn{1}{c}{Slasso}   \\
                     \hline                      
\multirow{3}{*}{All} & 120 & 10.428 & 13.858 & 4.753 & 20.989 \\
                     & 240 & 9.494  & 10.472 & 4.033 & 22.756 \\
                     & 360 & 8.542  & 9.522  & 3.822 & 23.500 \\
                      \hline \\
\multicolumn{6}{c}{Each Lag Order}                        \\
\multirow{3}{*}{1}   & 120 & 3.164  & 5.028  & 1.397 & 6.686  \\
                     & 240 & 2.817  & 4.636  & 1.081 & 9.886  \\
                     & 360 & 2.747  & 5.464  & 0.919 & 10.964 \\
                     \hline
\multirow{3}{*}{2}   & 120 & 1.917  & 3.781  & 1.017 & 5.350  \\
                     & 240 & 1.767  & 2.608  & 0.981 & 5.178  \\
                     & 360 & 1.575  & 2.253  & 1.000 & 5.722  \\
                     \hline
\multirow{3}{*}{3}   & 120 & 2.231  & 2.703  & 1.178 & 5.322  \\
                     & 240 & 1.753  & 2.339  & 1.072 & 4.542  \\
                     & 360 & 1.289  & 1.036  & 0.989 & 3.194  \\
                     \hline
\multirow{3}{*}{4}   & 120 & 3.117  & 2.347  & 1.161 & 3.631  \\
                     & 240 & 3.158  & 0.889  & 0.900 & 3.150  \\
                     & 360 & 2.931  & 0.769  & 0.914 & 3.619    \\
                     \hline\hline 
\end{tabular}  
    \end{minipage}
  \end{table}

\end{appendices}
\end{document}